\documentclass[11pt, a4paper]{article}
\usepackage{a4wide}
\usepackage{url}
\usepackage{graphicx, afterpage}
\usepackage{natbib}
\usepackage{enumitem}
\usepackage{hyperref}
\usepackage[labelfont=bf]{caption}
\usepackage{setspace}
\usepackage{multirow}
\usepackage{amsmath, amsthm, amssymb, amsfonts}
\usepackage{bm}
\usepackage{multirow}
\usepackage{color}     
\usepackage[ruled, vlined]{algorithm2e}
\usepackage{longtable}
\usepackage{pifont}

\allowdisplaybreaks

\DeclareMathAlphabet\mathbfcal{OMS}{cmsy}{b}{n}

\newcommand{\mbf}{\mathbf}
\newcommand{\mc}{\mathcal}

\newcommand{\vep}{\varepsilon}

\renewcommand{\l}{\left}
\renewcommand{\r}{\right}

\def\wh{\widehat}
\def\wt{\widetilde}

\newcommand{\E}[0]{\mathsf{E}}
\newcommand{\Var}[0]{\mathsf{Var}}

\newcommand{\p}{\mathsf{P}}
\newcommand{\R}{\mathbb{R}}
\newcommand{\Z}{\mathbb{Z}}
\newcommand{\N}{\mathbb{N}}

\newcommand{\iid}{\mbox{\scriptsize{iid}}}
\newcommand{\nn}{\nonumber}
\newcommand{\sic}{\mbox{SC}}
\newcommand{\rss}{\mbox{RSS}}
\newcommand{\trim}{\mbox{\scriptsize trim}}

\newcommand{\rhow}{\rho_n^{\text{\scriptsize (W)}}}

\newcommand{\bbI}{\mathbb{I}}

\newcommand{\cp}{\theta}  
\newcommand{\Cp}{\Theta}  
\renewcommand{\c}{k} 
\newcommand{\C}{\mc K} 

\newcommand{\omegao}{\omega_n^{(1)}}
\newcommand{\omegat}{\omega_n^{(2)}}

\newcommand{\cmark}{\ding{51}}%
\newcommand{\xmark}{\ding{55}}%

\theoremstyle{definition}
\newtheorem{thm}{Theorem}[section]
\theoremstyle{definition}
\newtheorem{cor}[thm]{Corollary}
\theoremstyle{definition}
\newtheorem{lem}[thm]{Lemma}
\theoremstyle{definition}
\newtheorem{prop}[thm]{Proposition}
\theoremstyle{definition}
\newtheorem{assum}{Assumption}[section]
\theoremstyle{remark}
\newtheorem{rem}{Remark}[section]
\theoremstyle{definition}
\newtheorem{defn}{Definition}[section]
\theoremstyle{definition}

\title{Two-stage data segmentation permitting multiscale change points, heavy tails and dependence}
\author{Haeran Cho$^1$ and Claudia Kirch$^2$}

\begin{document}

\maketitle


\begin{abstract}
The segmentation of a time series into piecewise stationary segments, 
a.k.a. multiple change point analysis, is an important problem 
both in time series analysis and signal processing.
In the presence of {\it multiscale change points}
with both large jumps over short intervals and small changes over long stationary intervals,
multiscale methods achieve good adaptivity in their localisation
but at the same time, require the removal of false positives and duplicate estimators
via a model selection step.
In this paper, we propose a localised application of Schwarz information criterion which,
as a generic methodology, is applicable with any 
multiscale candidate generating procedure fulfilling mild assumptions.
We establish the theoretical consistency of the proposed localised pruning method
in estimating the number and locations of multiple change points 
under general assumptions permitting heavy tails and dependence.
Further, we show that combined with a MOSUM-based candidate generating procedure,
it attains minimax optimality in terms of detection lower bound and localisation for i.i.d.\ sub-Gaussian errors.
A careful comparison with the existing methods by means of 
(a) theoretical properties such as generality, optimality and algorithmic complexity, 
(b) performance on simulated datasets and run time, as well as 
(c) performance on real data applications,
confirm the overall competitiveness of the proposed methodology.
\end{abstract}

\footnotetext[1]{School of Mathematics, University of Bristol, UK.
Email: \url{haeran.cho@bristol.ac.uk}.}

\footnotetext[2]{Department of Mathematics, Otto-von-Guericke University; Center for Behavioral Brain Sciences (CBBS); Magdeburg, Germany.
Email: \url{claudia.kirch@ovgu.de}.} 

\section{Introduction}

Change point analysis has a long tradition in statistics since \citet{page1954}.
In recent years, there has been a surge of interest for 
computationally fast and statistically efficient methods for change point analysis
due to its importance in time series analysis, signal processing and 
many other applications where data is routinely collected over time in naturally nonstationary environments.  
In particular, many papers address the problem of testing for a change point, 
either retrospectively or sequentially,
when at most one change is expected, see \citet{csorgo1997} and \citet{horvath2014} for an overview.
Based on such tests, estimators for the location of a single change point 
can be derived with optimal localisation properties. 

However, it is often unknown how many structural changes are present in the data
and allowing for multiple change points, 
the goal of change point analysis is to estimate both the total number and locations of the change points. 
Examples where data segmentation is popularly employed include genomics 
(detecting chromosomal copy number aberrations, see \citet{olshen2004, li2016, niu2012, chan2017}),
neurophysiology 
(modelling the instabilities in the rate at which a neuron fires an action potential, \citet{messer2014}),
astronomy (detecting orbiting planets and their periodicity, 
\citet{fisch2018})
and finance (identifying and dating change points in financial time series, \citet{cho2012}), to name but a few. 

Broadly, approaches to retrospective change point analysis 
in the literature can be categorised into two:
One line of research relates to the aforementioned tests,
while the other aims at optimising objective functions,
constructed on the principle of penalised likelihood or minimum description length,
via dynamic programming \citep{killick2012, maidstone2017} or genetic algorithm \citep{davis2013}.
There are also methods based on hidden Markov models 
with algorithms for estimating the sequence of hidden states \citep{titsias2016}.

Recent algorithmic developments include multiscale methodologies 
which focus on isolating each change point within an interval sufficiently large for its detection,
whereby the tests and the estimators designed for the at-most-one-change alternatives 
are applicable to detect (possibly) multiple change points. 
The Wild Binary Segmentation (WBS) algorithm proposed in \citet{fryzlewicz2014}
accomplishes this by drawing a large number of random intervals.
\citet{kirch2014} investigate a moving sum (MOSUM) procedure which 
systematically tests for at most a single change point over moving windows at a single bandwidth,
and briefly discuss its multiscale extension for better adaptivity.
On the one hand, such multiscale methods enjoy the near-optimal localisation of change points
through scanning the same regions of the data at multiple resolutions.
On the other, this may result in conflicting (duplicate) estimators detected for the identical change point,
as well as false positives spuriously detected without any change points in their vicinity,
which makes a model selection step inevitable.

There exist post-processing and pruning procedures
specifically tailored for particular multiscale candidate generating methods and settings
to handle false positives and duplicates, however,
there is a lack of a unified approach to this task.
In this paper, we propose a generic methodology for this purpose,
which utilises the Schwarz criterion \citep{schwarz1978} 
and performs an exhaustive search for change point estimators in a {\it localised} way
on a candidate set generated by multiscale methods.
Contrary to the common usage of information criteria in change point problems,
the proposed localised pruning algorithm does not require the maximum number of change points as an input,
nor does it seek for the global minimiser of the criterion.

We show that as a generic tool, the localised pruning algorithm inherits 
the properties of the candidate generating method. 
Therefore, with a suitable candidate generating method, it consistently estimates 
the total number of change points as well as locating the change points with accuracy
while being computationally feasible.
In this paper, we verify the suitability of two candidate generating multiscale methods
based on the MOSUM and cumulative sum (CUSUM) statistics.

\subsection{Main contributions}
\label{sec:main}

Below, we summarise the main contributions made in this paper.

\begin{enumerate}[label = (\alph*)]
\item \label{ctrb:one} \textbf{Two-stage procedure.}
We explicitly separate the statistical analysis of the candidate generating method 
(Stage~1, see Section~\ref{sec:mosum}) 
from that of the model selection (pruning) methodology (Stage~2, see Section~\ref{sec:loc}). 
This allows us (i) to easily extend our statistical conclusions to different candidate generating methods,
and (ii) to gain insights into the assumptions required for each stage separately.
	
\item \label{ctrb:two} \textbf{Truly multiscale change points.}
In contrast to the assumptions commonly found in the literature
that require homogeneity on the change point structure, we adopt a truly multiscale setting
and thereby shed light upon the performance of the proposed change point methodology 
when both large changes over short stretches of stationarity 
as well as small changes over long stretches of stationarity 
are present simultaneously in the signal, see Definition~\ref{def_scenarios}.

\item \label{ctrb:three} \textbf{Minimax optimality.}
We show that the proposed localised pruning, combined with 
a MOSUM-based multiscale candidate generating mechanism,
achieves minimax optimality in change point localisation
as well as matching the rate of the minimax detection lower bound
when the errors are distributed as i.i.d.\ sub-Gaussian random variables,
see Corollary~\ref{cor:mosum:optimal}.

\item \label{ctrb:four} \textbf{Assumptions on the error distribution.}
While assumptions such as independence and (sub-)Gaussianity are 
often made in the literature, we allow for very general assumptions on the error distribution 
permitting both serial dependence and heavy tails. 
In addition, we explicitly state how the error distribution 
enters into our detection lower bound and localisation rate
(see Assumptions~\ref{assum:size} and~\ref{assum:cand}),
thus providing a guidance on the choice of tuning parameters.

\item \label{ctrb:five} \textbf{Universally competitive performance in simulations and data analysis.}
For a range of test signals of varying length, frequency of change points and error distributions,
the proposed method performs uniformly well in both model selection consistency and localisation accuracy, 
and within reasonable computation time (see Section~\ref{sec:sim}).
Applied to real data examples, our procedure is capable of handling the issues
often encountered in practice such as heteroscedasticity and low signal-to-noise ratio. 
We provide its implementation with a MOSUM-based candidate generating procedure
in the R package {\tt mosum} available on CRAN \citep{mosum}.

\item \label{ctrb:six} \textbf{Computational complexity.}
The computational complexity of the localised pruning algorithm 
with the MOSUM-based candidate generating method is given by $O(n\log(n))$,
which is comparable to or much lower than 
that of most competing methods (see Table~\ref{table:overview}). 
With other candidate generating methods, 
the computational complexity of the combined procedure 
will effectively be determined by that of the first-stage candidate generation.
\end{enumerate}

The multiple mean change problem has been extensively studied in the literature, 
often laying the groundwork for generalisations to more complex problems.
While this aspect has been well-established in the (at-most-one-change)
testing literature for both online and offline procedures,
the literature on data segmentation beyond $L^2$-based mean change point detection 
is relatively scarce.
The proposed localised pruning methodology has been  
constructed with such extensions in view, 
and we discuss these possibilities in Section~\ref{sec:conc}.

The rest of the paper is organised as follows:
In Section~\ref{sec:optimal}, we define a truly multiscale change point problem
and introduce the assumptions for theoretical consistency.
Also, we present the minimax optimality results available from the literature,
and provide a comparative study of our proposed methodology
and those shown to be near-minimax optimal.
In Section~\ref{sec:loc}, we motivate and propose 
the localised pruning as a generic methodology
applicable with a class of candidate generating mechanisms
and establish its theoretical consistency.
Section~\ref{sec:mosum} discusses 
a MOSUM-based candidate generating procedure
and shows the minimax optimality of the combined two-stage methodology.
In Section~\ref{sec:numeric}, we briefly summarise the simulation studies
and apply the proposed methodology to array comparative genomic hybridisation data.
Section~\ref{sec:conc} concludes the paper, followed by the main proofs. 
The rest of the proofs, discussion of an alternative, CUSUM-based candidate generating procedure
related to the WBS \citep{fryzlewicz2014}, 
complete simulation results as well as an application to Kepler light curve data 
are provided in Appendix.

\subsubsection*{Notations}

Throughout the paper, we adopt $\nu_n$ to denote a sequence
satisfying $\nu_n \to \infty$ at an arbitrarily slow rate,
which may differ from one occasion to another. 
We adopt the notation $a_n \asymp b_n$ to denote that $a_n = O(b_n)$ and $b_n = O(a_n)$.

\enlargethispage{1cm}

For convenience, the assumptions are formulated with asymptotic arguments 
but the proofs work directly with non-asymptotic conditions on the corresponding quantities 
on the set $\mc M_n$ defined in Theorem~\ref{thm:sbic} (collected in \eqref{eq:rates}) 
making constants trackable in principle.

\section{Multiscale change point analysis}
\label{sec:optimal}

\subsection{Multiscale change point detection problem}
\label{sec:model}

We consider the canonical change point model
\begin{align}
X_t &= f_t + \vep_t = f_0 + \sum_{j = 1}^{q_n} d_{j} \cdot \bbI_{t \ge \cp_j + 1} + \vep_t, \label{eq:model} 
\end{align}
where $\cp_1 < \cp_2 < \ldots < \cp_{q_n}$ with $\cp_j = \cp_{j, n}$ denote the $q_n$ change points 
(with $\cp_0 = 0$ and $\cp_{q_n + 1} = n$), 
at which the mean of $X_t$ undergoes changes of size $|d_j|$ where, again, $d_j = d_{j, n}$. 
We denote by $\delta_j = \delta_{j, n} = \min(\cp_j - \cp_{j - 1}, \cp_{j + 1} - \cp_j)$ 
the minimum distance of $\cp_j$ to its neighbouring change points.
The sequence of errors $\{\vep_t\}_{t=1}^n$ satisfies $\E(\vep_t) = 0$ and 
is allowed both serial dependence as well as heavy-tailedness as specified later.
We assume that $\max_{1 \le j \le q_n} \vert d_j \vert = O(1)$
as well as $\min_{1 \le j \le q_n} \delta_j \to \infty$,
separating the problem of change point detection under~\eqref{eq:model} 
from that of outlier detection.

In this paper, our interest lies in studying the performance of a change point detection methodology
in a truly multiscale, heterogeneous change point setting
where the signal $f_t$ may simultaneously contain both frequent large jumps
as well as small jumps over long stretches of stationarity. 
The following definitions distinguish multiscale formulations
of detection lower bound and localisation rate from their non-multiscale counterparts. 
\begin{defn}
\label{def_scenarios} \hfill
\begin{enumerate}[label = (\alph*)] 
\item \textbf{Sublinear change point problem.} There exists some $\kappa > 0$
such that $\max_{1 \le j \le q_n} \delta_j = O(n^{1 - \kappa})$.

\item \textbf{Detection lower bound.} 
We distinguish the following change point scenarios that are linked to different detection lower bounds.
\begin{enumerate}[label = (\roman*)]
	\item \textbf{Homogeneous change points:} \mbox{$\min_{1 \le j \le q_n} d_j^2 \,\min_{1\le j \le q_n} \delta_j \to \infty$.}

\item \textbf{Finite mixture of homogeneous change points:} 
There are $M < \infty$ subsets of change points, 
where change points within each subset are homogeneous as defined in~(a). 
The situation with finitely many changes, i.e., \ where $q_n \le M$, is a special case.

\item \textbf{Multiscale change points:} $\min_{1 \le j \le q_n} d_j^2 \, \delta_j \to \infty$.
\end{enumerate}

\item \textbf{Localisation rate:} We distinguish between a \textbf{homogeneous localisation rate} 
where the rate of localisation for the $j$-th change point is weighted globally
with $\min_{1\le j \le q_n}d_j^2$, and a \textbf{multiscale localisation rate} 
where the localisation rate is weighted locally with $d_j^2$.
\end{enumerate}
\end{defn}

The localised pruning methodology proposed in this paper 
is designed with the sublinear change point problem in view,
as the detection of (near-)linear change points is much easier
and does not usually require a multiscale procedure.
While it does not play any role in practice, 
the detection lower bound attained by our methodology
is minimax rate optimal only in the sublinear change point setting;
outside this setting, the detection lower bound is worse by at most $\log(n)$,
see also \cite{chan2013} where similar observations are made
on scan likelihood ratio statistic for a signal detection problem.
Extension of the result beyond this setting
would require the adoption of a scale-dependent penalty as in \cite{fromont2020} for the pruning methodology.
To the best of our knowledge, such a choice of penalty 
is available only for light-tailed errors,
and its extension to the general error distribution
we consider in this paper has not been investigated in the literature.

Definition~\ref{def_scenarios}~(b) shows different extensions of the assumption 
$d_1^2 \min(\cp_1, n - \cp_1) \to \infty$ commonly found in the change point testing literature 
(where $q_n = 1$ at most, see e.g., \citet{csorgo1997}).
Proceeding from (i) to (iii) therein,
the associated parameter space becomes more general 
and only (iii) truly requires multiscale methods
that scan the data for change points at more than finitely many scales. 
Nevertheless, most papers in the change point detection literature 
formulate the detection lower bound for the homogeneous setting only,
see Table~\ref{table:overview} and Section~\ref{sec:literature}.
On the other hand, we adopt the most general setting and impose 
an assumption on the size of changes correspondingly (see Assumption~\ref{assum:size} below).

The multiscale localisation rate in~(c) reflects that the difficulty 
in accurate localisation of each change point
depends on the size of the corresponding jump size only.

\subsection{Main assumptions}
\label{sec:err:dist}

The mathematical analysis in this paper is based on the following properties on the error distributions only, 
which makes the results very general permitting e.g., heavy tails and dependence 
and even non-stationarity.

\begin{assum}[Error distribution] 
\label{assum:vep}
We assume that $\{\vep_t\}_{t = 1}^n$ is ergodic with $\E(\vep_t) = 0$ and
$0 < c \le \Var(\vep_t) \le C < \infty$ for some $c, C >0$.
Further:
\begin{enumerate}[label = (\alph*)]
\item For some $\omega_n$ satisfying $\sqrt{\log(n)} = O(\omega_n)$, 
let $\p(\mc M_n^{(11)})\to 1$ where
\begin{align*}
\mc M_n^{(11)} &= \l\{ \max_{0 \le s < e \le n}
\frac{1}{\sqrt{e - s}} \Big\vert \sum_{t = s + 1}^e \vep_t \Big\vert \le \omega_n \r\}.
\end{align*}

\item 
For any sequences $1 \le a_n, b_n \le D_n$ with $D_n$ defined in Assumption~\ref{assum:size}, 
let $\p(\mc M_n^{(12)} \cap \mc M_n^{(13)}) \to 1$ where
\begin{align*}
\mc M_n^{(12)} &= \l\{ \max_{1 \le j \le q_n} \;
\max_{d_j^{-2} a_n \le \ell \le \cp_j - \cp_{j - 1}} 
\frac{\sqrt{d_j^{-2} a_n}}{\ell} \left\vert \sum_{t = \cp_j - \ell + 1}^{\cp_j} \vep_t \right\vert 
\le \omegao \r\}
\\
& \qquad \bigcap \l\{ \max_{1 \le j \le q_n} \; \max_{d_j^{-2}a_n \le \ell \le \cp_{j + 1} - \cp_j} 
\frac{\sqrt{d_j^{-2}a_n}}{\ell} \left\vert \sum_{t = \cp_j + 1}^{\cp_j + \ell} \vep_t \right\vert
 \le \omegao
\r\}, \quad \text{ and} \\
\mc M_n^{(13)} &= \l\{
\max_{1 \le j \le q_n} \max_{1 \le \ell \le d_j^{-2}b_n} \, \frac{1}{\sqrt{d_j^{-2}b_n}}
 \left\vert \sum_{t = \cp_j - \ell + 1}^{\cp_j} \vep_t \right\vert
\le \omegat
\r\}\\
&\qquad \bigcap \l\{
\max_{1 \le j \le q_n} \max_{1 \le \ell \le d_j^{-2}b_n} \, \frac{1}{\sqrt{d_j^{-2}b_n}}
 \left\vert \sum_{t = \cp_j + 1}^{\cp_j + \ell} \vep_t \right\vert
\le \omegat
\r\}.
\end{align*}
\end{enumerate}
\end{assum}

The lower bound on $\omega_n$ in Assumption~\ref{assum:vep}~(a) is quite natural in light of
Theorem~1 of \cite{shao1995} which derives the corresponding result
for i.i.d.\ random variables whose moment-generating function exists.
The bound is closely linked to the detection lower bound of our proposed methodology 
as we make the following assumption:

\begin{assum}[Multiscale lower bound on the size of changes]
\label{assum:size}
For $D_n :=\min_{1 \le j \le q_n} d_j^2 \, \delta_j$,
we require $D_n^{-1} \omega_n^2 \to 0$ for $\omega_n$ as in Assumption~\ref{assum:vep}. 
In addition, $D_n$ dominates the penalty used in the localised pruning algorithm, 
see Assumption~\ref{assum:penalty}.
\end{assum}

\begin{rem}
\label{rem:vep:one} \hfill
\begin{enumerate}[label = (\alph*)]
\item The rates $\omegao$ and $\omegat$ 
are closely connected with the localisation rate of the localised pruning method, 
see Assumption~\ref{assum:cand} for the precise statemenet.
Also, the bound $\omegao$ gives the rate of localisation
for the multiscale MOSUM procedure considered as
one of the candidate generating mechanisms in Section~\ref{sec:mosum}.
Note that $\omegao$ and $\omegat$ are always dominated by $\omega_n$
and are often much smaller, particularly in the presence of heavy tails and when $q_n$ is bounded,
see Proposition~\ref{prop:vep} below for specific examples. 

\item The bounds for the respective second set in  
$\mc M_n^{(12)}$ and $\mc M_n^{(13)}$ follow from the bounds of the first set 
in the case of i.i.d.\ errors but this is not necessarily so for time series errors. 
\end{enumerate}
\end{rem}

The next proposition provides the exact rates
for $\omega_n$, $\omegao$ and $\omegat$ in Assumption~\ref{assum:vep} 
for some special cases.
\begin{prop}
\label{prop:vep} 
In all follows, $\nu_n \to \infty$ arbitrarily slow.
\begin{enumerate}[label = (\alph*)]
\item \textbf{Sub-Gaussianity.} Let $\{\vep_t\}_{t = 1}^n$ be a sequence of i.i.d. \ random variables 
following a sub-Gaussian distribution as defined e.g.,\ in Section~2.5 of \cite{vershynin2018}.
Then, Assumption~\ref{assum:vep} holds with 
$\omega_n \asymp \sqrt{\log(n)}$ and $\omegao = \omegat \asymp \max(\sqrt{\log(q_n)}, \nu_n)$.

\item \textbf{Heavy tails.}  Let $\{\vep_t\}_{t = 1}^n$ be a sequence of i.i.d.  
regularly varying random variables with index of regular variation $\alpha > 0$ 
as defined e.g., \ in \cite{mikosch2010}.
Then Assumption~\ref{assum:vep} holds with
$\omega_n \asymp n^{1/\beta}$ and
$\omegao = \omegat \asymp \max(q_n^{1/\beta}, \nu_n)$ for any $\beta < \alpha$.

\item In the following general situations,
Assumption~\ref{assum:vep} holds with the rates given below which, however, are usually not tight:
\begin{enumerate}[label = (\roman*)]
\item \textbf{Invariance principle.} If there exists (possibly after changing the probability space) 
a standard Wiener process $W(\cdot)$ such that
$\sum_{t = 1}^{\ell} \vep_t - W(\ell) = O(\lambda_{\ell})$ a.s.
with $\lambda_{\ell} = o(\sqrt{\ell})$, then 
Assumption~\ref{assum:vep}~(a) holds with $\omega_n \asymp \max(\lambda_n\nu_n, \sqrt{\log(n)})$.

\item \textbf{Moment conditions.} If $\E \vert \sum_{t = l + 1}^r \vep_t \vert^{\gamma} \le C (r - l)^{\gamma/2}$ 
for any $-\infty < l < r < \infty$ and some constants $C > 0$ and $\gamma > 2$, then
Assumption~\ref{assum:vep}~(b) holds with 
$\omegao = \omegat \asymp q_n^{1/\gamma} \,\nu_n$.
\end{enumerate}
\end{enumerate}
\end{prop}

\begin{rem}
\label{rem:vep:two} \hfill
\begin{enumerate}[label = (\alph*)]
\item For regularly varying jump size distributions,  
$\omega_n$ in Proposition~\ref{prop:vep}~(b) cannot be improved beyond 
$\omega_n = n^{1/\alpha}L(n)$ for some slowly varying function $L$  
(see Theorem~1.1 of \citet{mikosch2010} and Proposition~B.1.9~(9) of \cite{de2007}). 
For dependent errors, similar results are derived in \citet{mikosch2013}.
Furthermore, in the special case of a $t$-distribution with $\alpha$ degrees of freedom, 
then Assumption~\ref{assum:vep} holds with $\omega_n \asymp n^{1/\alpha}$
\citep[Section~4.2]{schluter2009}.

\item 
Invariance principles as in Proposition~\ref{prop:vep}~(c.i) 
have been derived for a variety of situations including dependent data under
weak dependency conditions such as mixing \citep[Theorem~4]{kuelbs1980}
and functional dependence measure conditions \citep{berkes2014}, to name but a few. 
The rate $\lambda_n$ is typically directly linked to the number of moments that exist, 
e.g., \ for i.i.d.\ errors, $\lambda_{\ell} = \log(\ell)$ if the moment generating function exists, 
and $\lambda_{\ell} = \ell^{1/(2+\Delta)}$ if $\E(\vep_t^{2+\Delta}) < \infty$
\citep{komlos1975, komlos1976}.  
Comparing the rate of $\omega_n$ in Proposition~\ref{prop:vep}~(c.i) with the one in (b),
shows that the rates from the invariance principle are usually not tight.

Moment conditions as in Proposition~\ref{prop:vep}~(c.ii) 
have been shown for many time series; see e.g., Appendix~B.1 in \citet{kirch2006resampling}.
\end{enumerate}
\end{rem}

\subsection{Minimax optimality}
\label{sec:minimax}

In this section, we state the benchmark for the minimax detection lower bound
and optimal localisation rate. 

The following result is from Proposition~1 of \citet{arias2011}.
\begin{prop}[\textbf{Minimax optimal separation rate}] 
\label{prop:lb}
Under~\eqref{eq:model}, let $H_{0, n}: \, q_n = 0$ 
and $H_{1, n}$ the setting where $q_n= 2$, $d_n: = d_1 = -d_2$
and $\delta_n: = \cp_2 - \cp_1$ with $n^{-1} \delta_n \to 0$.
Then, $H_{0, n}$ and $H_{1, n}$ are asymptotically inseparable if
$\vert d_n \vert \sqrt{\delta_n} \le \sqrt{2\log(n/\delta_n)} - \nu_n$
where $\nu_n \to \infty$. 
\end{prop}

The next proposition is from Proposition~6 of \citet{fromont2020} 
and provides the minimax optimal rate of multiple change point localisation,
which is stated here with an enlarged parameter space for ready comparison.

\begin{prop}[\textbf{Minimax optimal localisation rate for possibly an unbounded number of change points}]
\label{prop:loc}
Under~\eqref{eq:model}, let $\vert d_j \vert =: d_n$ for all $j = 1, \ldots, q_n$ 
and denote by $\Xi = \{(\cp_1, \ldots, \cp_{q_n}): \,
0 \equiv \cp_0 < \cp_1 < \ldots < \cp_{q_n} < \cp_{q_n + 1} \equiv n \text{ and }
d_n^2 \min_{1 \le j \le q_n} (\cp_{j + 1} - \cp_{j-1}) > c_0 \log(q_n)\}$ for some $c_0 > 0$,
the parameter space for the locations of change points.
Then, 
\begin{align}
\inf_{\C \in \N^{q_n}} \sup_{\Cp \in \Xi}
\E_{\Cp}\{d_H(\C, \Cp)\} \ge C d_n^{-2} \log(q_n) \nn
\end{align}
for some $C > 0$, where $d_H$ denotes the Hausdorff distance, i.e.,
$d_H(\C, \Cp) = \max\{\max_{\c \in \C} \min_{\cp \in \Cp} |\c - \cp|, 
\max_{\cp \in \Cp} \min_{\c \in \C} |\cp - \c|\}$.
\end{prop}

Both Propositions~\ref{prop:lb}--\ref{prop:loc} are derived
under the special case where $\{\vep_t\}_{t = 1}^n$ are  i.i.d.\ random variables
following a (sub-)Gaussian distribution.
In Section~\ref{sec:mosum}, Corollary~\ref{cor:mosum:optimal} shows that
under sub-Gaussianity, 
the two-stage procedure combining a MOSUM-based candidate generating method
and the proposed localised pruning algorithm,
achieves minimax optimal rates in both localisation and detection lower bound 
(the latter in the sublinear change point setting
where $\log(n/\delta_j) \asymp \log(n)$ for each $j$).
Further, even in the presence of heavy-tailed errors and dependence,
it is shown to attain minimax optimal localisation rates
provided that there are finitely many change points (i.e., $q_n$ is finite).
To the best of our knowledge, there do not exist equivalent results 
on the detection lower bound or the localisation rate (when $q_n \to \infty$) 
beyond the i.i.d. sub-Gaussianity; once they become available,
the results we derive for the proposed methodology under Assumption~\ref{assum:vep} are general enough 
to be immediately compared to such a benchmark.

\subsection{Comparison with the existing literature}
\label{sec:literature}

\begin{table}[htb]
\caption{Comparison of change point detection methodologies on
the rates of detection lower bound and localisation derived under (sub-)Gaussianity
where $\delta_n = \min_{1 \le j \le q_n} \delta_j$,
and whether they are formulated in a multiscale way according to Definition~\ref{def_scenarios}.
We also provide their computational complexity, 
and whether their theoretical guarantee goes beyond the (sub-)Gaussian setting.
\citet{wang2018d} $\dagger$ refers to their $\ell_0$-penalised LSE estimator,
while \citet{wang2018d} $*$ refers to their modified WBS.}
\label{table:overview}
\resizebox{\columnwidth}{!}{
\begin{tabular}{l | cc | cc | c | c}
\hline\hline
& \multicolumn{2}{|c|}{Detection lower bound} & \multicolumn{2}{|c|}{Localisation} & Computational & Beyond  \\
Methodology  & Multiscale & Rate & Multiscale & Rate & complexity & sub-Gaussianity \\ \hline
{\bf MoLP} & \cmark & $\boldsymbol{\log(n)}$ & \cmark& $\boldsymbol{\log(q_n)}$ & $\boldsymbol{O(n\log(n))}$ & \cmark 
\\
\citet{chan2017} & \xmark & $\log(n/\delta_n)$ & \xmark & $\log(n)$ & $O(n\log(n))$ & \xmark  \\
\hline
Single-scale MOSUM & \xmark & $\log(n/\delta_n)$ & \cmark & $\log(q_n)$ & $O(n)$ & \cmark \\
\hline
\citet{fromont2020} & \cmark & $\log(n/\delta_j)$ & \cmark & $\log(q_n)$  &  $O(n^2)$ & \xmark \\
\citet{wang2018d} $\dagger$ & \xmark & $\log(n)$ & \cmark & $\log(n)$ & $O(n^2)$ & \xmark \\
\hline
\citet{wang2018d} $*$ & \xmark & $\log(n)$ & \cmark & $\log(n)$ & $O(n R_n)$ with & \xmark\\
\citet{baranowski2016} & \xmark & $\log(n)$ & \cmark & $\log(n)$ & $(n/\delta_n)^2 \log(n)/R_n \to 0$ & \xmark \\
\hline
\citet{frick2014} & \xmark & $\log(n/\delta_n)$ & \xmark & $\log(n)$ & $O(n^2)$ & \cmark \\
\citet{li2019} & \xmark & $q_n\log(n)$ & \xmark & $q_n\log(n)$ & -- & \xmark \\
\hline
\citet{fryzlewicz2017} & \xmark & $\log^2(n)$ & \xmark & $\log^2(n)$ & $O(n\log^2(n))$ & \xmark \\ 
\hline
\end{tabular}}
\end{table}

There exist various univariate time series segmentation algorithms 
which are shown to be near-minimax optimal in detecting and locating 
multiple change points: 
\citet{frick2014} and \citet{li2016} propose procedures that are termed as
multiscale change point segmentation methods in \citet{li2019};
noting empirical and theoretical limitations of the WBS as proposed in \cite{fryzlewicz2014},
\cite{baranowski2016} and \citet{wang2018d} propose  modifications of the WBS
which require additional tuning parameters such as a threshold or a lower bound on 
$\delta_n := \min_{1 \le j \le q_n} \delta_j$;
 \cite{boysen2009}, \citet{wang2018d} and \citet{fromont2020} 
investigate an $\ell_0$-penalised least squares (LSE) estimator,
the former two with the Schwarz criterion-type penalty and the latter with an adaptive one;
\citet{chan2017} propose two methods, where
one  bears some resemblance to a multiscale MOSUM procedure with `bottom-up' merging
(see also \citet{messer2014}) while the other 
to the tail-greedy unbalanced Haar (TGUH) method of \citet{fryzlewicz2017}.

All the papers discussed above present their theoretical findings under the assumption
that $\{\vep_t\}_{t = 1}^n$ is a sequence of i.i.d. (sub-)Gaussian random variables,
 with the exception of \cite{frick2014} allowing for i.i.d. errors following exponential family distributions;
an extension of their results to dependent error processes is studied in \cite{dette2018}.

Table~\ref{table:overview} provides an overview of these methodologies alongside
the localised pruning applied with a multiscale MOSUM procedure for candidate generation (referred to as `MoLP'),
on their theoretical performance, computational complexity and generality beyond the sub-Gaussian setting;
\cite{boysen2009} assume that $\vert d_j \vert$ and $\delta_j/n$ are 
bounded away from zero, and thus we exclude it from the table.
Apart from this paper and \citet{fromont2020}, 
all others formulate the detection lower bounds only for 
the homogeneous change points according to Definition~\ref{def_scenarios}~(b);
the detection lower bound in the latter paper is slightly lower 
outside the sublinear change point setting, 
requiring that $d_j^2\delta_j/\log(n/\delta_j) \to \infty$ for each $j$.
Also, the MoLP, the penalised LSE of \cite{fromont2020} 
and the single-scale MOSUM procedure are
the only methods known to achieve the exact minimax optimal localisation rate $\log(q_n)$
for multiple change point estimation (Proposition~\ref{prop:loc}).
Our proposed method achieves this with the computational complexity of $O(n\log(n))$
rather than $O(n^2)$ required for solving the $\ell_0$-penalised least squares estimation problem;
we defer a detailed discussion on the computational complexity to Appendix~\ref{sec:comp}.

\enlargethispage{1cm}

Additionally, the theoretical analysis in this paper is conducted in a much more general setting 
permitting heavy-tailed and serially correlated errors under Assumption~\ref{assum:vep},
which sets our paper apart from the above list.

\section{Localised pruning via Schwarz criterion}
\label{sec:loc}

Our goal is to estimate both the total number $q_n$ and 
the locations of the change points $\cp_j, \, j = 1, \ldots, q_n$ under~\eqref{eq:model}.
For this purpose, we introduce a generic, localised pruning methodology which,
applicable to a set of candidate change point estimators returned 
by multiscale change point procedures,
achieves consistent estimation of multiple change points in their total number and locations.

Many multiscale change point procedures are based on the principle of 
isolating each change point for its detection and estimation,
and typically attach extra information to change point estimators about their  detection intervals.
Such examples include the multiscale extension of the MOSUM procedure \citep{kirch2014}
and the WBS \citep{fryzlewicz2014}:
The MOSUM procedure scans a series of MOSUM statistics
\begin{align}
T_{b, n}(G; X):= 
\sqrt{\frac{G}{2}}  \left(\bar{X}_{(b-G+1):b}-\bar{X}_{(b+1):(b+G)}\right)
\label{eq:mosum}
\end{align}
where $\bar{X}_{s:e} = (e - s + 1)^{-1} \sum_{t = s}^e X_t$,
for a given bandwidth $G$ and $G \le b \le n - G$, and marks as change point candidates 
the locations where $|T_{b, n}(G; X)|$ simultaneously exceeds a critical value and forms local maxima;
thus each candidate estimator $\c$ is associated with its natural detection interval
$\mc I_N(\c) = (\c - G, \c + G]$.
The WBS examines the CUSUM statistics
\begin{align}
\mc X_{s, b, e} \equiv \mc X_{s, b, e}(X) = 
\sqrt{\frac{(b - s)(e - b)}{e - s}}
\l(\bar{X}_{(s+1):b}-\bar{X}_{(b+1):e}
\r) 
\label{eq:cusum}
\end{align}
for $s + 1 \le b \le e - 1$ over a large number of randomly drawn intervals $(s, e] \subset [1, n]$.
The maximiser of the CUSUM statistics $\c = \arg\max_{s < b < e} |\mc X_{s, b, e}|$
can be regarded as a change point candidate if the test statistic $|\mc X_{s, \c, e}|$ exceeds a certain threshold,
and the interval $\mc I_N(\c) = (s, e]$ is readily associated with its detection.

In what follows, we describe the proposed localised pruning methodology
assuming that a set of candidate estimators $\C$ is given.
Specific candidate generating methods are discussed in 
Section~\ref{sec:mosum} and Appendix~\ref{sec:wbs}.

\subsection{Methodology}
\label{sec:method}

Let $\C$ denote the set of all the candidate change point estimators to be pruned down.
For each $\c \in \C$, we denote the detection interval of $\c$ by
$\mc I(\c) \equiv (\c - G_L, \c + G_R]$,
where the left detection distance $G_L = G_L(\c)$ is the distance from $\c$ to the leftmost point of the interval,
and the right detection distance $G_R = G_R(\c)$ is defined analogously. 

Information criteria are frequently adopted for model selection in change point problems,
and we adopt the Schwarz criterion \citep[SC]{schwarz1978} for this purpose.
For a given set of change point candidates $\mc A = \{\tilde{\c}_1 < \ldots <\tilde{\c}_m\} \subset \C$,
the SC is evaluated as
\begin{align}
\label{eq_sic}
\sic(\mc A) = \frac{n}{2}\log\l\{\frac{\rss(\mc A)}{n}\r\} + \vert\mc A\vert \cdot \xi_n,
\end{align}
where it balances between the goodness-of-fit measured by the residual sum of squares
\begin{align*}
\rss(\mc A) = \sum_{j=0}^m \sum_{t = \tilde{\c}_j + 1}^{\tilde{\c}_{j + 1}} 
\Big( X_t - \bar{X}_{(\tilde{\c}_j+1) : \tilde{\c}_{j + 1}} \Big)^2
\quad \text{with} \quad
\tilde{\c}_0 = 0 \text{ and } \tilde{\c}_{m + 1} = n,
\end{align*}
and the penalty imposed on the model complexity $|\mc A|$.

\begin{assum}[Penalty] 
\label{assum:penalty}
The penalty parameter $\xi_n$ satisfies
\begin{align*}
\frac{\xi_n}{D_n} \to 0 \quad \text{and} \quad \frac{\omega_n^2}{\xi_n} \to 0,
\end{align*}
where $\omega_n$ and $D_n$ are as in Assumptions~\ref{assum:vep}~(a) and~\ref{assum:size}, respectively.
\end{assum}

The assumption shows the connection between the penalty parameter $\xi_n$,  
the noise level $\omega_n$ and the detection lower bound $D_n$.
For i.i.d. \ sub-Gaussian random variables,
the rate of $\omega_n$ in Proposition~\ref{prop:vep}~(a) cannot be improved \citep[Theorem~1]{shao1995}
and thus the (strengthened) Schwarz penalty of $\xi_n = \log^{1 + \Delta}(n)$
with some $\Delta > 0$ can be allowed by Assumption~\ref{assum:penalty} 
(see e.g., \citet{yao1988} and \citet{fryzlewicz2014}).
Proposition~\ref{prop:vep}~(b) and Remark~\ref{rem:vep:two}~(a) 
indicate that a penalty stronger than logarithmic in $n$ is required for heavy-tailed errors
in order to guarantee consistent estimation of the number of change points by means of the $\sic$, 
an observation also made by \citet{kuhn2001}.

In the literature,
exhaustive minimisation of an information criterion over all $\mc A \subset \C$
for a given candidate set $\C$, has been considered as
a model selection method,
see e.g., \citet{niu2012}, \citet{chan2014} and \citet{yau2016}.
Such an exhaustive approach may result in a computationally inhibitive search space
as its size grows exponentially with $|\C|$. 
Moreover, it does not utilise the information immediately available 
about the detection intervals of change point estimators.
For example, if the detection interval of a candidate $\c$ does not overlap with that of any other estimator, 
there is little to be gained by having $\c$ considered alongside other candidates in the evaluation of $\sic$.
On the other hand, if $\mc I(\c)$ overlaps with the detection interval of another candidate, say $\c^\prime$,
it is possible that $\c$ and $\c^\prime$ are conflicting estimators of the identical change point,
which justifies the joint consideration of the two.

Based on these observations, we propose the localised pruning methodology 
consisting of two nested algorithms, 
where the outer algorithm iteratively selects the local environment 
on which the inner algorithm performs the pruning.

\subsubsection{Outer algorithm: Localisation ({\tt LocAlg})}
\label{section_LA}

Taking the set of change point candidates $\C$ as an input, 
the outer algorithm for localisation iteratively selects a subset of candidates
to be pruned down by the inner algorithm ({\tt PrunAlg}) described in Section~\ref{section_PA}.
For this, the algorithm sorts the candidates in $\C$ according to a sorting function $h$.
One possibility is to use the jump size associated with each $\c \in \C$,
which is calculated within the detection interval $\mc I(\c) = (\c - G_L, \c + G_R]$ as
\begin{align}
\label{eq:jump:sort}
h_{\mc J}(\c) = \l\vert \bar{X}_{(\c - G_L + 1):\c} - \bar{X}_{(\c + 1):(\c + G_R)} \r\vert.
\end{align}
If (asymptotic) null distributions of the test statistics are available, another possibility is to use the
inverse of the $p$-values, say $h_{\mc P}$, as a sorting function. 
In practice, the use of $h_{\mc P}$ may slow down the pruning algorithm by generating many ties
when many of the $p$-values are artificially set to zero by the machine (see \citet{meier2018}).
Either with $h_{\mc J}$ or $h_{\mc P}$, additional tie-breaking rules can be employed,
e.g., by preferring the candidates associated with the smallest detection interval 
according to $G_L + G_R$, $G_L$ or $G_R$;
if there are still ties, an arbitrary choice can be made.
We note that the theoretical results do not depend on the choice of the sorting function or the tie-breaking rule.

Denote by $\mc C$ the candidates for which no decision has been reached yet,
and by $\wh{\Cp}$ the set of already accepted candidates. 
At the beginning of the algorithm, the active candidate set $\mc C$ 
is given by the complete candidate set $\C$ and $\wh{\Cp}$ is set to be empty.
Then, the  outer algorithm iteratively processes the candidates in the following way.

\begin{enumerate}[label=\textbf{Step \arabic*:}, itemindent = 25pt]
	\item \textbf{Find the most prominent candidate}.
		According to a sorting function $h$ (and tie-breakers if necessary), 
		find a candidate $\c_{\circ} \in \mc C$ from the active candidate set that maximises $h$. 
	
	\item \textbf{Define the local search environment}.
		Find $\c_L$ that is closest to $\c_\circ$ while being strictly left to $\c_\circ$
		from the candidates which either
		\begin{itemize}
			\item have already been accepted (and belong to $\wh{\Cp} \cup \{0\}$), or 
			\item are still to be either accepted or discarded ($\mc C$)  
			whose detection intervals do not overlap with that of $\c_{\circ}$,
			i.e.,\ $\mc I(\c_L)\cap\mc I(\c_\circ)=\emptyset$ or equivalently 
			$|\c_{\circ}-\c_L| \ge G_R(\c_L)+G_L(\c_{\circ})$.
		\end{itemize}
		Identify $\c_R$ strictly to the right of $\c_\circ$ from $\wh{\Cp} \cup \{n\} \cup \mc C$ 
		with analogous restrictions.
		Then, any candidates without decision that fall within $(\c_L, \c_R)$ are considered 
		as candidates competing with $\c_{\circ}$. 
		We denote this set of change point candidates by $\mc D$,
		i.e., \ $\mc D = \mc C \cap (\c_L, \c_R)$.
		
	\item  \textbf{Pruning Algorithm ({\tt PrunAlg})}. 
		Apply the inner algorithm for pruning, {\tt PrunAlg}, with the arguments 
		$(\mc D, \mc C, \wh{\Cp}, \c_L, \c_R)$. 
		As an output, we yield a subset $\wh{\mc A} \subset \mc D$ (possibly empty) which contains 
		candidates to be accepted in the next step.
	
	\item \textbf{Update the accepted ($\wh{\Cp}$) and active ($\mc C$) candidate sets.} 
		We accept all estimators from the output of {\tt PrunAlg}, $\wh{\mc A}$, 
		but not all of $\mc D \setminus \wh{\mc A}$ are discarded yet.
		This is because $\mc D$ may contain acceptable estimators 
		of change points that are too close to the boundaries $\c_L$ or $\c_R$,
		for which we cannot guarantee their acceptance at the current iteration (see Theorem~\ref{thm:sbic}).
		However, if $\c_L$ (resp. $\c_R$) has already been accepted, we 
		discard any candidates in $\mc D \setminus \wh{\mc A}$ which lie to the left (right) of
		the leftmost (rightmost) candidate in $\wh{\mc A}$.
		Similarly, unaccepted candidates in $\mc D \setminus \wh{\mc A}$
		that lie between any two elements of $\wh{\mc A}$ are discarded.
		In addition, we remove $\c_\circ$ from the future consideration regardless of 
		whether it has been accepted by {\tt PrunAlg} or not.
		
		In summary, we denote the set of all the candidates for which a decision has been reached, 
		either because it has been accepted or discarded according to the above consideration,
		by $\mc R$. 
		Then, we add $\wh{\mc A}$ to $\wh{\Cp}$ 
		and remove all the candidates in $\mc R$ from $\mc C$.
		
	\item {\bf Iteration.} Repeat Steps~1~to~4 until $\mc C$ is empty. 
		The set $\wh{\Cp}$ is the final set of estimators and the output of the algorithm.
\end{enumerate}

A pseudo-code of the outer algorithm can be found in 
Algorithm~1 of Appendix~\ref{sec:algs}.

{\tt LocAlg} is guaranteed to terminate since at each iteration, 
Step~4 discards at least one candidate~$\c_\circ$ from the active candidate set.
Under a mild condition on $\C$,
we show that this yields consistent estimation by guaranteeing that
at least one suitable estimators remain in $\mc C$ for all the undetected change points, 
see Assumption~\ref{assum:det:intervals} and the discussion thereafter.

In Step~3 of {\tt LocAlg}, the inner algorithm {\tt PrunAlg} makes a decision 
between competing candidates using $\sic$, 
which are evaluated at each $\mc A \subset \mc D = \mc C \cap (\c_L, \c_R)$ as
\begin{align}
\sic(\mc A|\mc C, \wh{\Cp}, \c_L, \c_R) =&
\frac{n}{2}\log\l\{\frac{\rss(\mc A\cup \wh{\Cp} \cup (\mc C \setminus \mc D))}{n}\r\} 
\nn \\
& + (|\mc A| + \vert \wh{\Cp}\vert + \vert\mc C \setminus \mc D \vert) \cdot \xi_n.
\nn 
\end{align}
By construction, it makes a decision which of the candidates in $\mc D$ to accept 
while treating all other currently surviving candidates outside of $(\c_L, \c_R)$ as given.
Therefore, at any iterations of {\tt LocAlg}, all $X_t, \, 1 \le t \le n$, enter in the computation of $\sic$.
In other words, {\tt LocAlg} has the interpretation of
performing an adaptively selected subset of the exhaustive search over the complete candidate set $\C$
in a localised manner, by utilising the information readily available about the detection intervals 
of change point candidates.

\subsubsection{Inner algorithm: Pruning ({\tt PrunAlg})}
\label{section_PA}

The inner pruning algorithm {\tt PrunAlg} in Step~3 of the outer localisation algorithm {\tt LocAlg}
takes as its input $(\mc D, \mc C, \wh{\Cp}, \c_L, \c_R)$,
and looks for a subset $\wh{\mc A} \subset \mc D$ 
to be added to the finally accepted candidates according to the following rules:

Let $\mc F$ denote the collection of 
all subsets $\mc A \subset \mc D$ for which it holds:
\begin{enumerate}[label=(C\arabic*)]
\setlength\itemsep{0em}
\item \label{eq_c1} 
adding further change point candidates to $\mc A$ monotonically increases the $\sic$,
\end{enumerate}
and denote by $m^* = \min_{\mc A \in \mc F} |\mc A|$.
Then, we select $\wh{\mc A}$ as
\begin{align}
\wh{\mc A} &= \arg\min\l\{\mc A \subset_R \mc A^\prime
\text{ with } \mc A^\prime \in \mc F \text{ and } \r. 
\label{eq_c2}  \tag{C2} \\
& \qquad \qquad \qquad \qquad \qquad
\l. m^* \le |\mc A^\prime| \le m^* + 2: \, \sic(\mc A | \mc C, \wh{\Cp},\c_L, \c_R)\r\}
\nn
\end{align}
where, by $\mc A \subset_R \mc A^\prime = \{\tilde{\c}_{1} < \tilde{\c}_{2} < \ldots < \tilde{\c}_{m} \}$,
we indicate that $\mc A \setminus \mc A^\prime \subset \{\tilde{\c}_{1}, \tilde{\c}_{m}\}$,
i.e., $\mc A$ contains all {\em inner} elements of $\mc A^\prime$ (if exist)
while the first and the last elements of $\mc A^\prime$ may or may not be included in $\mc A$.
If there are multiple subsets yielding the minimum $\sic$ in \eqref{eq_c2}, 
we choose the one with the minimum cardinality.
If there are ties in the cardinality as well, we arbitrarily select one.

\begin{rem}
\label{rem_prun}
By performing a top-down search, 
the condition \ref{eq_c1} typically prunes down the search space quickly: 
If removing $\c \in \mathcal{A}$ from $\mathcal{A}$ leads to an increase in $\sic$, 
no subset of $\mathcal{A} \setminus \{\c\}$ can be an element of $\mc F$. 
For a complete algorithmic description of {\tt PrunAlg}, see Algorithm~2 in 
Appendix~\ref{sec:algs}.
and also \citet{meier2018} for details about its efficient implementation.
\end{rem}

\begin{rem}
\label{rem:direct}
It is possible to apply the search criteria \ref{eq_c1}--\eqref{eq_c2} to $\C$ directly, 
without iteratively going through the steps of the outer algorithm.
In such a case, \eqref{eq_c2} is simplified to
\begin{equation}
\wh{\mc A} = \arg\min\{\mc A \in \mc F
\text{ with } |\mc A| = m^*: \, \sic(\mc A  | \C, \emptyset, 0, n)\},
\label{eq_c2r} 
\tag{C2$^\prime$}
\end{equation}
i.e., search for $\wh{\mc A}$ only among the subsets satisfying \ref{eq_c1}.
This approach still gains computationally
compared to minimising the $\sic$ among all the $2^{|\C|}$ subsets of $\C$ while,
as shown in Corollary~\ref{cor:one}, achieves consistency in multiple change point estimation.
However, it is still to be avoided when there are many candidates to be pruned down,
and {\tt LocAlg} greatly reduces the computational cost by breaking down 
the scope of {\tt PrunAlg} at each iteration.
\end{rem}

\begin{rem}
We highlight the key differences between the use of $\sic$ in {\tt PrunAlg}
and the conventional use of information criteria 
as a model selection tool in the change point literature.
A common approach is to evaluate an information criterion 
at a sequence of nested candidate models with increasing number of change points,
which often requires the maximum allowable number of change points, say $q_{\max}$,
as an input parameter.
However, selection of this tuning parameter is not straightforward especially when $n$ is large,
without pre-supposing the frequency or the sparsity of the change points,
and some approaches require $q_{\max}$ to be fixed 
in their theoretical consideration \citep{fryzlewicz2014, baranowski2016}.
In contrast, our localised pruning method bypasses such a requirement
by identifying local intervals over which the $\sic$-based search is performed.
In the simulation studies, we observe empirical evidence of the sub-optimality of 
sequential evaluation and minimisation of an information criterion,
particularly when there are frequent changes in the signal
(see e.g., Table~\ref{supp:table:sim:dense}), 
which further supports the search criteria \ref{eq_c1}--\eqref{eq_c2}
adopted by {\tt PrunAlg}.
\end{rem}

\subsection{Consistency of the localised pruning algorithm}
\label{sec:theory}

In this section, we show that the localised pruning algorithm combining {\tt LocAlg} and {\tt PrunAlg}
consistently estimates the total number of change points when applied to a suitable set of candidates. 
Furthermore, it `almost' inherits the rate of convergence of 
the change point estimators from the candidate generating mechanisms,
and thus achieves consistency in change point localisation 
under mild conditions on the set of candidates.

We make the following assumption on candidate generation.
\begin{assum}[Candidate generating algorithm] 
\label{assum:cand}
Let $\C = \C_n$ denote the set of candidates obtained from $\{X_t\}_{t = 1}^n$ and $Q_n = |\C|$ 
the total number of candidates.
Then, with $\omegao$, $\omegat$ and $\omega_n$ as in Assumption~\ref{assum:vep}:
\begin{enumerate}[label = (\alph*)]
\item With probability approaching one, 
each change point has at least one candidate in its $(d_j^{-2} \rho_n)$-environment, i.e., as $n \to \infty$,
\begin{align*}
\p(\mc M_n^{(2)}) \to 1 \quad \text{where} \quad
\mc M_n^{(2)} = \l\{\max_{1 \le j \le q_n} \min_{\c \in \C}  d_j^2 \; |\c - \cp_j| \le \rho_n \r\}
\end{align*}
for a sequence $\rho_n$ with $\max(\omegao, \omegat)^2 = O(\rho_n)$ and $\rho_n = O(\omega_n^2)$.
\item The total number of candidates $Q_n$ fulfils 
$n^{-1} \omega_n^2 \, Q_n \to 0$. 
\end{enumerate}
\end{assum}

The sequence $\rho_n$ is the precision associated with the candidate generating method.
We show that the proposed pruning algorithm almost inherits this rate 
in the sense made more precise in Theorem~\ref{thm:sbic}.
We conjecture that typically, $\omegao \asymp \omegat$
as in all of the examples in Proposition~\ref{prop:vep}.
We further conjecture that, if so, $(\omegao)^2$ (or a related term) 
gives a lower bound for the minimax optimal localisation rate:
This agrees with our observations in Propositions~\ref{prop:vep} and~\ref{prop:loc}
under sub-Gaussian errors and when there are a finite number of change points,
and thus indicates that the lower bound $\max(\omegao, \omegat)^2$ is a reasonable one.
The requirement $\rho_n = O(\omega_n^2)$ is a weak one
with $\omega_n$ always dominating $\omegao$ and $\omegat$,
see Remark~\ref{rem:vep:one}~(a).
If the precision attained by a particular candidate generating procedure is worse than $\omega_n^2$,
the localised pruning can still achieve consistency
but with a stronger penalty $\xi_n$ fulfilling $\rho_n/\xi_n \to 0$, 
see~\eqref{eq:rates} and the discussion underneath.

Assumption~\ref{assum:cand}~(b) on the number of candidates 
replaces a more stringent condition requiring $q_n$ to be fixed,
which is found in the literature adopting the information criterion
for determining the number of change points \citep{yao1988, kuhn2001}.
In particular, this rules out applying the localised pruning algorithm
with every possible point as candidate estimators, i.e., $\C = \{1, \ldots, n - 1\}$.
However, a reasonably good candidate generating method ought not to return
too many candidates while meeting Assumption~\ref{assum:cand}~(a), 
and we show that the MOSUM- and CUSUM-based candidate generating methods 
fulfils this requirement
in Section~\ref{sec:mosum} and Appendix~\ref{sec:wbs}.

The following definitions that categorise the candidate estimators in $\C$
are frequently used throughout the paper.
\begin{defn}
\label{def_acceptable} \hfill
\begin{enumerate}[label = (\alph*)]
\item A candidate $\c^* \in \C$ that yields $d_j^2|\c^* - \cp_j| \le \rho_n$ 
with $\rho_n$ as in Assumption~\ref{assum:cand}~(a)
is referred to as a {\it strictly valid} estimator for $\cp_j$,
and the set of such candidates is denoted by $\mc V^*_j$ for each $j = 1, \ldots, q_n$.
\item For $\nu_n \to \infty$ at an arbitrarily slow rate,
a candidate $\c^{\prime} \in \C$ with 
$d_j^2|\c^{\prime} - \cp_j| \le \rho_n\nu_n$ is referred to as an {\it acceptable} estimator for $\cp_j$, 
and the set of such candidates is denoted by $\mc V_j^{\prime}$.
\item The remaining candidates $\c \in \C \setminus \mc V_j^\prime$ are 
{\it unacceptable} for $\cp_j$.
\end{enumerate}
\end{defn}
 
The gap between the best localisation rate $\rho_n$ of the candidate generating procedure 
and what is acceptable for the localised pruning algorithm is unavoidable:
For two very close candidates, 
the $\sic$ evaluated with the one slightly further away from a change point than the other
can end up being smaller simply by chance. 

We now show that {\tt PrunAlg} described in Section~\ref{section_PA}, as a generic pruning algorithm,
achieves consistent estimation of the number of change points 
as well as returning acceptable estimators for all $\cp_j, \, j = 1, \ldots, q_n$.
Although the boundary points $(\c_L, \c_R)$ supplied as input arguments to {\tt PrunAlg}  
are always chosen among the change point candidates (including $0$ and $n$) in Step~2 of {\tt LocAlg}, 
our theory below is applicable to any $(s, e]$ with $0 \le s < e \le n$ as the interval of consideration
and $\mc D = \C \cap (s, e)$ as the set of local candidates to be pruned down.
In this context, it is understood that $\wh{\Cp}$ contains 
candidates lying outside $(s, e)$ only.

It may be the case that some change points are too close to either $s$ or $e$ 
and thus may or may not be detectable by {\tt PrunAlg} within $(s, e]$,
which necessitates the pruning criterion \eqref{eq_c2} instead of the simpler \eqref{eq_c2r}.
We define the following sets of local change points
with universal constants $0 < c^* < C^* < \infty$ 
as in Proposition~\ref{prop:detectability} below:
\begin{align}
\Cp^{(s, e)} &= \l\{\cp_j: \, d_j^2\,\min(\cp_j - s, e - \cp_j) \ge C^*\xi_n\r\},
\label{eq:sure:detect}
\\
\bar{\Cp}^{(s, e)} &= \l\{\cp_j: \, d_j^2\,\min(\cp_j - s, e - \cp_j) \ge c^*\xi_n\r\}.
\label{eq:all:detect}
\end{align}

Theorem~\ref{thm:sbic} establishes the connection between 
the output of {\tt PrunAlg} and the sets defined in \eqref{eq:sure:detect}--\eqref{eq:all:detect}. 
\begin{thm}
\label{thm:sbic}
Let Assumptions~ \ref{assum:vep}, \ref{assum:size}, \ref{assum:penalty} and~\ref{assum:cand} hold,
and denote by $\wh{\Cp}^{(s, e)}$ the output of {\tt PrunAlg}
from applying the criteria \ref{eq_c1}--\eqref{eq_c2} to the local candidates 
$\mc D = \C \cap (s, e)$ within an interval $(s, e]$,
and by $\mc P_n^{(s, e)}$ the following event:
The output set $\wh{\Cp}^{(s, e)}$ contains
\begin{enumerate}[label = (\alph*)]
\item exactly one acceptable candidate  for each $\cp_j \in \Cp^{(s, e)}$, 
i.e., $|\wh{\Cp}^{(s, e)} \cap \mc V_j^\prime| = 1$ for $\cp_j \in \Cp^{(s, e)}$,
\item  at most one acceptable candidate  for each $\cp_j \in \bar{\Cp}^{(s, e)} \setminus \Cp^{(s, e)}$, i.e., $|\wh{\Cp}^{(s, e)} \cap \mc V_j^\prime| \le 1$ 
for $\cp_j \in \bar{\Cp}^{(s, e)} \setminus \Cp^{(s, e)}$, and
\item no other candidates, 
i.e., $\wh{\Cp}^{(s, e)} \setminus \bigcup_{j: \, \cp_j \in \bar{\Cp}^{(s, e)}} \mc V_j^\prime = \emptyset$.
\end{enumerate}
Then, with $\mc M_n := \mc M_n^{(11)} \cap \mc M_n^{(12)} \cap \mc M_n^{(13)} \cap \mc M_n^{(2)}$, we have
\begin{align*}
\p\l( \bigcap_{0 \le s < e \le n} \mc P_n^{(s, e)}, \ \mc M_n \r) \to 1 \quad \text{as} \quad n \to \infty.
\end{align*}
\end{thm}

In view of Theorem~\ref{thm:sbic}, we categorise the change points
according to their detectability within a given interval in the following definition.
\begin{defn}
\label{def:detectable} 
For any $0 \le s < e \le n$, we refer to
\begin{enumerate}[label = (\alph*)]
\item any change points in $\Cp^{(s,e)}$ as {\it surely detectable} within $(s, e]$,
\item any change points in $\bar{\Cp}^{(s,e)}$ as {\it detectable} within $(s, e]$, and
\item any change points in $\{\Cp \cap (s, e)\} \setminus \bar{\Cp}^{(s,e)}$ 
as {\it undetectable} within $(s, e]$.
\end{enumerate}
\end{defn}

The following corollary establishes that {\tt PrunAlg}, 
when applied to the complete candidate set $\C$ directly,
achieves consistency in multiple change point estimation.
\begin{cor}
\label{cor:one}
Under the assumptions of Theorem~\ref{thm:sbic}, 
applying the search criteria \ref{eq_c1} and \eqref{eq_c2r}
to the candidate set $\C$ within $(0, n]$ yields
$\wh{\Cp}^{(0, n)}=\{\wh{\cp}_1 < \ldots < \wh{\cp}_{\wh q_n}\}$
which consistently estimates $\Cp$, i.e.,
\begin{align*}
\p\l\{\wh q_n = q_n; \, 
\max_{1 \le j \le q_n} d_{j}^2 |\wh \cp_{j} \,\bbI_{j \le \wh q_n} - \cp_{j}| \le \rho_n\nu_n\r\} 
\ge \p(\mc M_n) + o(1) \to 1.
\end{align*}
\end{cor}

As pointed out in Remark~\ref{rem:direct},
pruning down $\C$ according to~\ref{eq_c1} and~\eqref{eq_c2r}
is computationally more efficient than 
the exhaustive minimisation of $\sic$ over all subsets of $\C$. 
Nevertheless, the localisation from the outer algorithm {\tt LocAlg} 
results in a considerable computational advantage 
when a large set of candidates needs to be pruned down.

Next, we establish that the consistency achieved by {\tt PrunAlg} within local search environments
(as in Theorem~\ref{thm:sbic}),
is carried over to the entire data set via the outer localisation algorithm {\tt LocAlg}.

\begin{assum}
\label{assum:det:intervals}
Recall that the detection interval of each $\c \in \C$ 
is denoted by $\mc I(\c) = (\c - G_L(\c), \c + G_R(\c)]$. Then, for each $j = 1, \ldots, q_n$,
there exists at least one acceptable candidate $\check{\c}_j \in \mc V_j^\prime$
which is situated well within its own detection interval by satisfying
\begin{align}\label{eq:assum:det:intervals}
\frac{\xi_n}{d_j^2 \, \min\{G_L(\check{\c}_j), G_R(\check{\c}_j)\}} \to 0.
\end{align} 
\end{assum}

Assumption~\ref{assum:det:intervals} justifies the removal of $\c_\circ$ identified in Step~1 of each iteration
from the future consideration,
regardless of whether it is accepted by {\tt PrunAlg} or not:
If $\c_\circ$ is an acceptable estimator for some $\cp_j$ while meeting \eqref{eq:assum:det:intervals},
such $\cp_j$ is surely detectable within $(\c_L, \c_R]$ and either $\c_\circ$ or some $\c \in \mc V_j^\prime$
is accepted by {\tt PrunAlg} at the current iteration;
if not, there still remain at least one acceptable estimators in the active candidate set $\mc C$
for any undetected change points after removing $\c_\circ$.
 We discuss how Assumption~\ref{assum:det:intervals} is met
by the MOSUM-based candidate generating procedure in Remark~\ref{rem_25_mos},
and provide a similar discussion for the CUSUM-based procedure 
in Appendix~\ref{sec:wbs}.


Theorem~\ref{thm:sbic:full} proves that {\tt PrunAlg} combined with
the outer algorithm {\tt LocAlg} 
achieves consistency in multiple change point estimation.

\begin{thm}
\label{thm:sbic:full}
Under the assumptions of Theorem~\ref{thm:sbic} and Assumption~\ref{assum:det:intervals}, 
the localised pruning algorithm {\tt LocAlg} 
outputs $\wh{\Cp} = \{\wh{\cp}_1 < \ldots < \wh{\cp}_{\wh q_n}\}$
which consistently estimates $\Cp$, i.e.,
\begin{align*}
\p\l\{\wh q_n = q_n; \, 
\max_{1 \le j \le q_n} d_{j}^2 |\wh \cp_{j} \,\bbI_{j \le \wh q_n} - \cp_{j}| \le \rho_n\nu_n\r\} 
\ge \p(\mc M_n) + o(1) \to 1,
\end{align*}
for some $\nu_n \to \infty$ at an arbitrarily slow rate.
\end{thm}

Its proof follows from the following two observations:
\begin{itemize}
\item When a change point is surely detectable for the first time 
at some iteration (in the sense of Definition~\ref{def:detectable}~(a)), 
it gets detected by an acceptable estimator by Theorem \ref{thm:sbic}
and consequently is no longer detectable in the subsequent iterations
thanks to how the local environments are defined in Step~2 of {\tt LocAlg}.
\item On the other hand, those change points which are yet to be detected
have corresponding acceptable estimators in the pool of candidates $\mc C$
due to how $\mc C$ is reduced in Step~4 of {\tt LocAlg}.
\end{itemize}

\section{Candidate generation}
\label{sec:mosum}

In this section, we investigate 
 a two-stage procedure combining the localised pruning methodology
with a multiscale extension of the MOSUM procedure of \citet{kirch2014}.
In Appendix~\ref{sec:wbs}, we provide the corresponding results
for a CUSUM-based procedure motivated by the WBS \citep{fryzlewicz2014}. 
Our theoretical analysis indicates that 
both the detection lower bound and the localisation rate
achieved with the MOSUM-based candidate generating procedure
are always better than those achievable with the CUSUM-based one.

\subsection{MOSUM procedure and its multiscale extension}

\cite{kirch2014} analyse the properties of a single-scale MOSUM procedure
which, for a bandwidth $G = G_n$, estimates the locations of the change points
by the locations of {\it significant} local maxima of the MOSUM statistic~\eqref{eq:mosum}
according to two different criteria. 
For the purpose of generating candidates for the localised pruning,
we adopt the method termed $\eta$-criterion 
with a lower false negative rate (see Section~2.2 of \cite{meier2018}).
Let $\C(G, \alpha) = \{\c_{G, j}, \, 1 \le j \le \wh q_G\}$ denote the set of candidates
obtained from bandwidth $G$ and some significance level $\alpha \in (0, 1)$.
\\ [1mm]
{\bf $\eta$-criterion.} Each $\c_{G, j}$ is the local maximiser of the MOSUM detector
within its $\lfloor \eta G\rfloor$-radius for some $\eta > 0$,
and $ \vert T_{\c_{G, j}, n}(G; X) \vert > \tau \, D_n(G; \alpha)$,
where $\tau^2$ is the (long-run) variance of the error sequence $\{\vep_t\}_{t = 1}^n$. 
The threshold $D_n(G; \alpha)$ is chosen such that for a signal with no change points,
there are no false positives reported (uniformly for this given bandwidth) 
with asymptotic probability $(1 - \alpha)$.
\\[1mm]
In the following, we assume that $\tau^2$ is known for simplicity.
Our arguments can readily be adapted to the case where 
a global estimator $\wh{\tau}^2_n$ satisfying 
$\vert \wh{\tau}_n^2 - \tau^2 \vert = o_P(\log^{-1}(n))$ is available.
More complicated arguments, as given in Section~2.3 of \cite{kirch2014}, 
are needed when a scale-dependent, local estimator $\wh{\tau}_{t, G}^2$ 
is adopted in place of $\tau^2$, such as the one implemented in the R package {\tt mosum} \citep{mosum};
however, this estimator needs not be uniformly consistent (in $G \le t \le n - G$).

When a single-scale MOSUM procedure is adopted for estimating
both the number and the locations of the change points,
the parameter $\alpha$ needs to be selected small enough 
in order not to incur any false positives, at the cost of a high false negative rate.
On the other hand, when the MOSUM procedure is adopted solely for
generating a set of candidates to be pruned down by a model selection method,
we can select $\alpha$ generously 
(e.g., $\alpha = 0.1$ is used by default in \citet{mosum}) or even do without thresholding.
In practice, it is recommended to apply a mild threshold 
since setting $D_n(G; \alpha) = 0$ may incur a violation of Assumption~\ref{assum:cand}~(b),
which adds computational burden as well as possibly leading to a loss of estimation accuracy.

The following proposition extends Theorem~3.2 of \cite{kirch2014}.

\begin{prop}
\label{prop:mosum:a}
Let $\eta \in (0, 1)$ for the $\eta$-criterion and suppose:
\begin{enumerate}[label = (\alph*)]
\item For each $j = 1, \ldots, q_n$, there exists $G(j)$ such that 
$2G(j) \le \delta_j$ and
$d_j^2G(j) \ge c_M D_n$ for some constant $c_M > 0$ that does not depend on $j$.

\item  $\p(\mc M_n^{(11)}) \to 1$ with $D_n^{-1} \omega_n^2 \to 0$,
where $\mc M_n^{(11)}$ is as in Assumption~\ref{assum:vep}~(a).

\item $\p(\mc M_n^{(12)} \cap \mc M_n^{(12+)} \cap \mc M_n^{(12-)}) \to 1$ with 
$\mc M_n^{(12)}$ from Assumption~\ref{assum:vep}~(b),
and $\mc M_n^{(12\pm)}$ defined analogously as
\begin{align*}
\mc M_n^{(12\pm)} &= \l\{ \max_{1 \le j \le q_n} \;
\max_{d_j^{-2} a_n \le \ell \le \cp_j - \cp_{j - 1}} 
\frac{\sqrt{d_j^{-2} a_n}}{\ell} \left\vert \sum_{t = \cp_j - \ell \pm G(j) + 1}^{\cp_j \pm G(j)} \vep_t \right\vert 
\le \omegao \r\}
\\
& \bigcap \l\{ \max_{1 \le j \le q_n} \; \max_{d_j^{-2}a_n \le \ell \le \cp_{j + 1} - \cp_j} 
\frac{\sqrt{d_j^{-2}a_n}}{\ell} \left\vert \sum_{t = \cp_j \pm G(j) + 1}^{\cp_j \pm G(j) + \ell} \vep_t \right\vert
 \le \omegao
\r\}.
\end{align*}
\end{enumerate}
Then, for a set $\mc S_n$ (specified in Lemma~\ref{lem_prop_mosum}) 
fulfilling $\p(\mc S_n) \to 1$,
there exists a universal constant $C_M > 0$ 
(not depending on the signal or the distribution of $\{\vep_t\}_{t = 1}^n$)
such that
\begin{align*}
\p\l(\max_{1 \le j \le q_n} \min_{\c \in \C(G(j), \alpha)} d_j^2 |\c - \cp_j| \ge C_M(\omegao)^2,
\; \mc S_n \r) \to 0.
\end{align*}
\end{prop}

\begin{rem}
\label{rem:prop:mosum:a}
\begin{enumerate}[label = (\alph*)]
\item Condition~(a) of Proposition~\ref{prop:mosum:a} requires that for each change point $\cp_j$, 
there exists a bandwidth $G(j)$ suitable for its detection.

\item Condition~(b) is assumed for the consistency of the localised pruning method also. 
Proposition~\ref{prop:mosum:a} continues to hold under the following weaker condition:
\begin{align*}
\max_{1 \le j \le q_n} \frac{1}{|d_j|\sqrt{G(j)}}\,\max_{|\ell - \cp_j|\le\frac 32 G(j)} 
\l\vert \frac{1}{\sqrt{G(j)}} \sum_{t = \lfloor \ell -  G(j)/2 + 1 \rfloor}^{\lfloor \ell + G(j)/2 \rfloor} \vep_t \r\vert = o_P(1).	
\end{align*}

This assertion follows e.g., \
when an invariance principle holds as in Proposition~\ref{prop:vep}~(c.i),
and there are a finite mixture of homogeneous change points 
with an appropriate bandwidth for each of the homogeneous subsets
(see Definition~\ref{def_scenarios}~(b)), in addition to
\begin{align*}
\frac{\lambda_n^2}{\min_{1 \le j \le q_n} d_j^2G(j)^2} = o(1)
\quad \text{and} \quad 
\frac{\log(n)}{\min_{1 \le j \le q_n} d_j^2G(j)} = o(1).		
\end{align*}

\item The assumptions on $\mc M_n^{(12\pm)}$ in Condition~(c) 
do not impose additional constraints in the following cases:
\begin{itemize}
\item When $\{\vep_t\}_{t = 1}^n$ are independent and identically distributed.
\item When $\{\vep_t\}_{t = 1}^n$ are stationary time series errors
and there are a finite mixture of homogeneous change points.
\end{itemize}
\end{enumerate}
\end{rem}

In Corollary~\ref{cor:mosum:a}l, 
we show that the single-scale MOSUM procedure yields consistent estimators 
with optimal localisation rate, 
either under sub-Gaussianity or when there are finitely many change points,
but only under the assumption that the change points are {\it homogeneous}
as defined in Definition~\ref{def_scenarios}~(b).
On the other hand, when the change points are heterogeneous,
it cannot produce consistent estimators by construction.

As noted in Remark~\ref{rem:prop:mosum:a}~(a), 
a natural solution to this lack of adaptivity 
is to apply the MOSUM procedure with a range of bandwidths.
At the same time, scanning the same data at multiple scales introduces 
duplicate estimators and false positives, necessitating the use of a pruning method. 
\cite{messer2014} and \cite{messer2018} propose to 
prune down the estimators from a multiscale MOSUM procedure in a bottom-up manner,
and a similar approach is taken by \cite{chan2017}:
Accepting all the estimators from the smallest bandwidth,
it proceeds to coarser scales and only accepts a change point estimator
if its detection interval does not contain any estimators that are already accepted.
While the bottom-up approach is applicable with multiple symmetric bandwidths,
there is no canonical ordering when asymmetric bandwidths are used.
More importantly, this approach rules out the possibility of removing any spurious estimators
including those detected from the finest bandwidth,
and thus requires the finest bandwidth to be large relative to $n$ 
in order to avoid spurious change point estimators.
In Section~\ref{sec:sim}, we observe on the simulated datasets
that indeed, the bottom-up merging tends to incur a large number of false positives.

\subsection{Localised pruning with MOSUM-based candidate generation}

The localised pruning algorithm proposed in Section~\ref{sec:method} 
is well-suited for pruning down the candidates generated by the multiscale MOSUM procedure.
Let $\mc G$ denote a set of bandwidths.
Each estimator $\c \in \C(G, \alpha)$ for $G \in \mc G$ is associated with
the natural detection interval $\mc I_N(\c) = (\c - G, \c + G]$.
Asymmetric bandwidths $\mbf G = (G_\ell, G_r)$ with $(G_\ell, G_r) \in \mc H \subset \mc G \times \mc G$ 
are readily incorporated into the methodology
using the MOSUM statistics defined as a correctly scaled difference 
between $\bar{X}_{(b - G_\ell + 1):b}$ and $\bar{X}_{(b + 1):(b + G_r)}$ for $b = G_\ell, \ldots, n - G_r$, 
and the corresponding $\mc I_N(\c) = (\c - G_\ell, \c + G_r]$ for 
$\c \in \C(\mbf G, \alpha)$; 
for more details, we refer to \cite{meier2018}.
Then, the collection of all the estimators from the multiscale MOSUM procedure,
$\C(\mc H, \alpha) = \bigcup_{\mbf G \in \mc H}\C(\mbf G, \alpha)$, 
can serve as the set of candidates $\C$. 
For Step~1 of the outer localisation algorithm {\tt LocAlg},
we can sort the candidate change points either according to the size of associated jumps
(see \eqref{eq:jump:sort})
or using the $p$-values derived from the asymptotic null distribution defined for each pair of bandwidths,
although care should be taken in their interpretation across multiple scales.
\\ [1mm]
{\bf Selection of bandwidths.}
We propose to generate the set of bandwidths $\mc G$ as follows.
Selecting a single parameter $G_0$,
which should be smaller than the minimal distance between adjacent change points, 
and setting $G_1 = G_0$,
we iteratively yield $G_m, \, m \ge 2,$ as a Fibonacci sequence, 
i.e., $G_m = G_{m - 1} + G_{m - 2}$.
Equivalently, we set $G_m = F_m \,G_0$ 
where $F_m = F_{m - 1} + F_{m -2}$ with $F_0 = F_1 = 1$ are the Fibonacci numbers.
This is repeated until for some $H = H_n$, it holds that $G_{H} < \lfloor n/\log(n) \rfloor$ 
while $G_{H + 1} \ge \lfloor n/\log(n) \rfloor$.
When using asymmetric bandwidths,
it is advisable to avoid the pairs of bandwidths 
which are too strongly unbalanced,
both in view of the asymptotic theory and  
the finite sample performance as is well-known from the two-sample testing literature.
A similar requirement can also be found in \citet{chan2017}.
For this reason, we only include the pairs of bandwidths $\mbf G = (G_\ell, G_r)$ in $\mc H$ that satisfy
\begin{align}
G_\ell, G_r \in \mc G = \{G_1, \ldots, G_H\} \quad \text{with} \quad
\frac{\max(G_{\ell},G_r)}{\min(G_{\ell},G_r)}\le C_{\text{asym}}
\label{eq_band_asym}
\end{align}
for some constant $C_{\text{asym}} > 0$. 
\\ [1mm]
With the thus-constructed set of asymmetric bandwidths $\mc H$,
Assumption~\ref{assum:cand}~(b) follows.

\begin{prop}
\label{prop:mosum:b}
Suppose that $\omega_n^2/G_0\to 0$ with $\omega_n$ as in Assumption~\ref{assum:vep}~(a).
Then, for $\mc H$ fulfilling~\eqref{eq_band_asym},
we have $n^{-1} \omega_n^2 \,\vert \C(\mc H, \alpha) \vert \to 0$.
\end{prop}

The assumption $\omega_n^2/G_0 \to 0$ is made solely to 
obtain a crude deterministic upper bound on the number of possible candidates from the smallest bandwidth.
We may replace it by a condition that directly 
limits the number of candidates detected at each bandwidth,
or an assumption on $q_n$ in combination 
with a stochastic version of Assumption~\ref{assum:cand}.

\begin{rem}\label{rem_25_mos} \hfill
\begin{enumerate}[label=(\alph*)]
\item For each $\c \in \C(\mc H, \alpha)$,
the natural detection interval $\mc I_N(\c)$ can serve as its detection interval $\mc I(\c) = (\c - G_L, \c + G_R]$,
whereby the detection distances $(G_L, G_R)$
are given by the set of bandwidths $(G_\ell, G_r)$ with which $\c$ has been detected.
Then, we have Assumption~\ref{assum:det:intervals} fulfilled by $\C(\mc H, \alpha)$ provided that
there exists a single bandwidth $G(j) \in \mc G$ satisfying $d_j^2G(j)/\xi_n \to \infty$
for each $j = 1, \ldots, q_n$,
which is readily met
under Condition~(a) of Proposition~\ref{prop:mosum:a} and Assumption~\ref{assum:penalty}.

\item It may be the case that $\C(\mc H, \alpha)$ contains
identical acceptable candidates $\c$ of $\cp_j$ returned at multiple scales,
including some $(G_\ell, G_r)$ that does not satisfy $d_j^2\min(G_\ell, G_r)/\xi_n \to \infty$. 
Against such a contingency, we propose to assign as $\mc I(\c)$ the natural detection interval
that returns the smallest $p$-value for the MOSUM test associated with the detection of $\c$.
Because the $p$-values decrease with the increase of jump size as well as 
that of bandwidths,
this strategy will recommend a reasonably large natural detection interval as $\mc I(\c)$.
In simulation studies, we use an implementation of the algorithm
which simply supposes that Assumption~\ref{assum:det:intervals} is satisfied
by the candidate generating mechanism.
\end{enumerate}
\end{rem}

The consistency of the localised pruning algorithm 
in combination with the MOSUM-based candidate generating mechanism 
follows immediately from Propositions~\ref{prop:mosum:a},~\ref{prop:mosum:b} 
and Theorem~\ref{thm:sbic:full}.

\begin{thm}
\label{cor:mosum}
Let Assumptions~\ref{assum:vep}, \ref{assum:size},  \ref{assum:penalty} and~\ref{assum:det:intervals} hold,
and suppose that the conditions in Propositions~\ref{prop:mosum:a} and \ref{prop:mosum:b} are satisfied.
Then, the localised pruning algorithm {\tt LocAlg} applied to $\C(\mc H, \alpha)$,
yields $\wh{\Cp} = \{\wh\cp_1 < \ldots < \wh\cp_{\wh{q}_n}\}$
which consistently estimates $\Cp$, i.e., 
\begin{align*}
\p\l\{\wh q_n = q_n; \, 
\max_{1 \le j \le q_n} d_{j}^2 |\wh \cp_{j} \bbI_{j \le \wh q_n} - \cp_{j}| \le \nu_n (\omegao)^2 \r\}  \to 1
\end{align*}
for any $\nu_n \to \infty$ arbitrarily slowly.
\end{thm}

The next corollary provides the consistency of $\wh{\Cp}$ in specific settings,
which follows directly from Proposition~\ref{prop:vep} and Theorem~\ref{cor:mosum}.
\begin{cor}
\label{cor:mosum:optimal}
 Let Assumptions~\ref{assum:size}, \ref{assum:penalty}, \ref{assum:det:intervals} 
and Condition~(a) of Proposition~\ref{prop:mosum:a} hold and $\omega_n^2/G_0 \to 0$,
with $\omega_n$ specified below.
\begin{enumerate}[label = (\alph*)]
\item {\bf Sub-Gaussianity.}
Let $\{\vep_t\}_{t = 1}^n$ meet the conditions of Proposition~\ref{prop:vep}~(a).
Then, with $\omega_n \asymp \sqrt{\log(n)}$, we have
\begin{align*}
\p\l\{\wh q_n = q_n; \, 
\max_{1 \le j \le q_n} d_{j}^2 |\wh \cp_{j} \bbI_{j \le \wh q_n} - \cp_{j}| \le \max(\log(q_n), \nu_n) \r\} 
\to 1.
\end{align*}

\item {\bf Heavy tails.} Let $\{\vep_t\}_{t = 1}^n$ meet the conditions of Proposition~\ref{prop:vep}~(b).
Then, with $\omega_n \asymp n^{1/\beta}$ for any $\beta < \alpha$, we have
\begin{align*}
\p\l\{\wh q_n = q_n; \, 
\max_{1 \le j \le q_n} d_{j}^2 |\wh \cp_{j} \bbI_{j \le \wh q_n} - \cp_{j}| \le \max(q_n^{2/\beta}, \nu_n) \r\} 
\to 1.
\end{align*}

\item {\bf Invariance principle and moment conditions.} 
Let $\{\vep_t\}_{t = 1}^n$ meet the conditions of Propositions~\ref{prop:vep}~(c) 
and~\ref{prop:mosum:a}~(c)  with $\omegao \asymp q_n^{1/\gamma} \nu_n$.
Then, with $\omega_n \asymp \max(\lambda_n\nu_n, \sqrt{\log(n)})$, we have 
\begin{align*}
\p\l\{\wh q_n = q_n; \, 
\max_{1 \le j \le q_n} d_{j}^2 |\wh \cp_{j} \bbI_{j \le \wh q_n} - \cp_{j}| \le q_n^{2/\gamma} \nu_n \r\} 
\to 1,
\end{align*}
where these rates are typically not tight.
\end{enumerate}
\end{cor}

In light of Propositions~\ref{prop:lb} and~\ref{prop:loc},
Corollary~\ref{cor:mosum:optimal} shows that under sub-Gaussianity,
the localisation pruning applied with the MOSUM-based candidate generating procedure
yields minimax optimal rates 
both in terms of the detection lower bound in the sublinear change point regime,
and the localisation rate.
Also, even when $\{\vep_t\}_{t = 1}^n$ is heavy-tailed, 
if the number of change points $q_n$ is finite,
the combined methodology achieves the minimax optimal localisation rate.

\section{Numerical results}
\label{sec:numeric}

\subsection{Simulation results}
\label{sec:sim}

We conducted an extensive simulation study
comparing the performance of the proposed localised pruning algorithm
combined with the MOSUM- and CUSUM-based candidate generation 
discussed in Section~\ref{sec:mosum} and Appendix~\ref{sec:wbs}, respectively,
against that of a large number of competitors whose implementations are readily available in R.
We consider the five test signals from \cite{fryzlewicz2014}
and their extensions ($n \ge 2 \times 10^4$) with both frequent and sparse change points, 
in order to assess the scalability of different methods.
As error sequences, we consider i.i.d. random variables following 
Gaussian and $t_5$ distributions, and AR($1$) processes with both weak and strong autocorrelations.

Overall, the proposed localised pruning performs well according to a variety of criteria,
often performing as well as or even better than many competitors 
both in terms of the total number of estimated change points and their locations.
At the same time, the localised pruning is shown to be scalable to long signals with $n \ge 2 \times 10^4$.
Most competing methods are specifically tailored for i.i.d. Gaussian errors
and thus struggle with heavy tails or serial correlations.
On the other hand, the localised pruning applied with theoretically-motivated tuning parameters
is shown to handle such error distributions well.
Between the two different candidate generating methods,
the MOSUM-based method produces estimators of better localisation accuracy
while the CUSUM-based one tends to incur more false positives.
For a complete description of the simulation results, see Appendix~\ref{sec:sim:app}.

\subsection{Real data analysis: Array CGH data}
\label{sec:data}

In this section, we illustrate the performance of the proposed methodology 
using  array comparative genomic hybridisation (CGH) data
that has previously been analysed in the literature. 

\begin{figure}[htbp]
\centering
\includegraphics[width=1\textwidth]{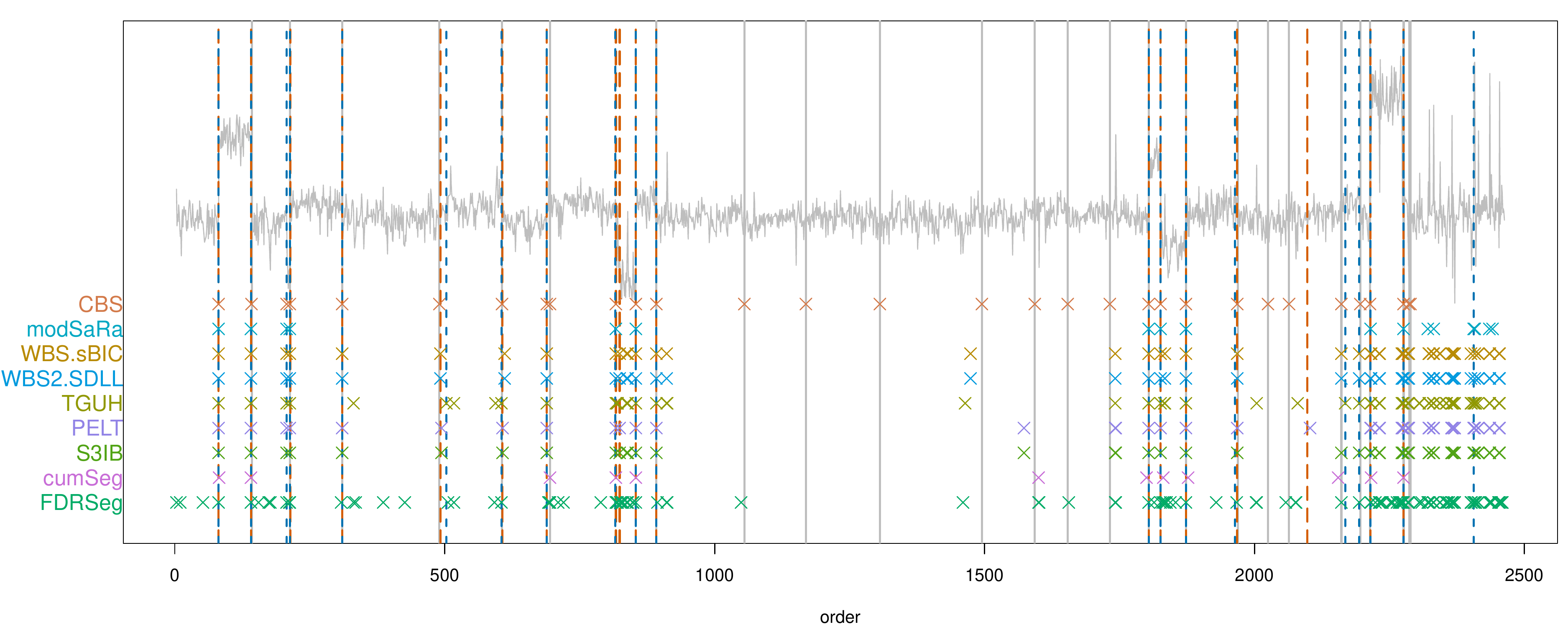}
\caption{Normalised copy number ratios of a comparison of DNA from cell strain S$0034$.
Vertical solid lines indicate the boundaries between chromosomes,
longdashed lines are change points estimated by MoLP
and dashed lines are those estimated by CuLP.
Change-point estimators from different methods are also plotted ($\times$).}
\label{fig:s0034}
\end{figure}

Microarray-based comparative genomic hybridization (array CGH) 
provides a means to quantitatively measure DNA copy number aberrations 
and to map them directly onto genomic sequences \citep{snijders2001}.
We analyse a dataset obtained from a breast tumour specimen (S$0034$) 
described in \citet{snijders2001} ($n = 2227$).
A number of algorithms have been proposed which,
regarding any gains or losses in the copy number from
the normalised copy number ratios between two DNA samples as change points,
identify their total number and locations under the model~\eqref{eq:model},
see e.g., \ \citet{olshen2004}, \citet{li2016} and \citet{niu2012}.

\citet{olshen2004} proposed to smooth the array CGH data for outlier removal 
prior to change point analysis.
Noticing that such a step may introduce serial correlations,
we choose to analyse the raw data and account for possible outliers by
adopting the penalty $\xi_n = \log^{1.1}(n)$ for the localised pruning algorithm,
with $\alpha = 0.4$ and $\eta = 0.4$ for MoLP and $C_\zeta = 0.5$ 
for CuLP.
In addition to the methods included in the comparative simulation study in Section~\ref{sec:sim},
we consider the circular binary segmentation algorithm of \citet{olshen2004} 
(CBS, implemented in \citet{DNAcopy})
and the modified screening and ranking algorithm of \citet{xiao2014} 
(modSaRa, implemented in \citet{modsara}).
It is important to note that the CBS takes all boundary markers between neighbouring chromosomes
as an input unlike any other procedures in consideration,
and automatically marks all of them as change points.

Figure~\ref{fig:s0034} plots the normalised fluorescence ratios from S$0034$
and the change point estimators returned by various methods,
and Table~\ref{table:s0034} reports the number of estimated change points.
Overall, MoLP and CuLP detect fewer number of change points compared to most of the competitors,
and many elements of the two sets of estimators either coincide or lie very close to each other.
Also, many change point estimators coincide with the boundary markers 
although they are detected {\it without} knowing their positions unlike the CBS.

The data exhibits heteroscedasticity particularly
beyond the genome order $2274$
where there is a dramatic increase in the variability.
Both candidate generating methods return a large number of candidates
(MoLP: $167$, CuLP: $93$) 
and our localised approach to pruning manages to reduce the size of the candidate sets
reasonably well.
On the other hand, WBS.sBIC, WBS2.SDLL, TGUH, PELT, S3IB and FDRSeg
are susceptible to returning spurious change point estimators particularly in 
this region of increased volatility.
CumSeg misses some of the change points commonly detected by many methods, 
which is consistent with the findings reported in Section~\ref{sec:sim}.

Interestingly, CuLP, WBS.sBIC, WBS2.SDLL and FDRSeg
are affected by the randomness involved in generating either the candidate estimators or the critical values,
and yield different results on different runs when applied to this data set.
It may be due to that the underlying signal is not exactly piecewise constant,
a phenomenon known as genomic waves \citep{diskin2008}. 
The results for these methods reported here were obtained 
by setting the seed of R's random number generator to be one. 

\begin{table}[htb]
\centering
\caption{Number of change points estimated from the S$0034$ data set.}
\label{table:s0034}
\resizebox{\columnwidth}{!}
{\small
\begin{tabular}{|cc|ccccccccc|}
\hline
MoLP & CuLP & CBS & modSaRa &  WBS.sBIC &  WBS2.SDLL & TGUH   & PELT & S3IB & cumSeg & FDRSeg \\ \hline
18 & 20 & 31 &     17  &     52 &      84 &      65  &     46  &     49   &    12  &    126
\\
\hline
\end{tabular}}
\end{table}

\section{Conclusions and outlook}
\label{sec:conc}

In this paper, we propose the localised pruning algorithm which,
together with a class of multiscale candidate generating procedures,
forms a two-stage methodology to data segmentation.
Adopting a truly multiscale framework,
we prove the consistency of the proposed methodology 
 in multiple change point estimation under mild conditions,
and show that it inherits the localisation property of the candidate generating mechanism. 
Theoretical properties for the second-stage localised pruning algorithm 
are discussed independently from the choice of first-stage candidate generating methods, 
allowing an easy extension of the results to other candidate generating methods.
Two examples for this choice are provided: A multiscale MOSUM procedure and a WBS algorithm.
In particular, combined with the former, the localised pruning algorithm 
achieves minimax rate optimality both in change point localisation and detection lower bound 
in those settings where such optimality results are available.
Importantly, we work with meta-assumptions on the key elements
of the change point structure and the error distribution,
the latter of which only concern the bounds given in Assumption~\ref{assum:vep}
and thus permit both heavy-tailedness and serial dependence.
In doing so, the influence of each element on our theoretical arguments
is made transparent and discussed in details, 
allowing for their easy extension to other error distributions in the future.

A comparison with competitors in terms of 
(a) theoretical properties such as the detection lower bound and the localisation rate, 
(b) computational complexity, speed and scalability to large sample sizes, and
(c) the performance in a variety of simulations and real data examples, 
shows that our proposed methodology performs universally well,
especially when combined with the MOSUM-based candidate generating method,
whose implementation is provided in the R package \texttt{mosum} available on CRAN \citep{mosum}.

While we focus on the univariate mean change point detection problem in this paper, 
there are natural ways for extending the proposed methodology to more general change point problems:
Via an appropriate transformation of the data, e.g., \ by adopting an $M$-estimation framework,
change points in the stochastic properties of interest can be made detectable
as change points in the mean of the transformed time series.
With a suitably modified information criterion,
our methodology becomes applicable to 
a variety of more complex change point scenarios, such as
the detection of changes in the mean of multivariate data; 
regression parameters (e.g., neural-network-based nonparametric (auto-)regression); 
other distributional parameters (e.g., integer-valued time series)
and robust change point detection
\citep{kirch2015detection, kirch2015use, kirch2018}.
Some first results in this direction based on the current paper 
have already been obtained in \citet{reckruhm2019estimating}, 
where the necessity for a model selection strategy in such general change point problems 
is well-motivated (see Chapter~2.4 therein).
Besides, our results can be adapted to detect parameter changes in renewal processes 
\citep{kuhn2001, messer2014}.

In light of these examples, the present work can be seen as an important first step towards 
an extended methodology for more general data segmentation problems, 
for which the literature is much scarcer compared to the literature on
change point detection in the mean of univariate time series.

\section{Main proofs}

In this section, we provide the proofs of Theorems~\ref{thm:sbic}--\ref{thm:sbic:full}
which establish the consistency of the localised pruning algorithm
combining {\tt LocAlg} and {\tt PrunAlg}.
They are based on Propositions~\ref{prop:detectability}--\ref{prop:step:three}, 
whose proofs can be found in Appendix~\ref{sec:pf:main:props}.
Throughout, we assume that Assumptions~\ref{assum:vep} and~\ref{assum:cand}~(a)
(and Assumption~\ref{assum:det:intervals} for Theorem~\ref{thm:sbic:full}) hold.
In addition, we work under the following non-asymptotic bound:
\begin{align}
\max\l(
\frac{\omegao}{\sqrt{\nu_n\,\rho_n}}, \, 
{\frac{\omegat}{\nu_n\sqrt{\rho_n}}}, \, 
\frac{Q_n\omega_n^2}{n}, \,
\frac{\xi_n}{D_n}, \,
\frac{\rho_n\nu_n}{\xi_n}, \,
\frac{\omega_n^2}{\xi_n}, \,
\frac{1}{\nu_n}
\r) \le \frac{1}{M} 
\label{eq:rates}
\end{align}
for some $M > 0$, which holds for all $n \ge n(M)$ for some large enough $n(M)$.
This replaces the asymptotic conditions in Assumptions~\ref{assum:size}, 
\ref{assum:penalty} and~\ref{assum:cand}. 
 Here, we regard $\rho_n$ as the precision originally attained by 
a candidate generating mechanism.
If $\max(\omegao, \omegat)^2 = O(\rho_n)$ as in Assumption~\ref{assum:cand}~(a),
\eqref{eq:rates} is fulfilled by $\nu_n \to \infty$ 
arbitrarily slowly as stated in the theorem.
If not, the assertions still hold for any $\nu_n$ fulfilling the above.
Also, when $\rho_n = O(\omega_n^2)$ is not met, 
the assertions continue to hold but with a penalty parameter
greater than the acceptable precision, which is reflected in~\eqref{eq:rates}.
In the proofs of Propositions~\ref{prop:detectability}--\ref{prop:step:three}, 
we state the precise requirement on the ratios in the LHS of \eqref{eq:rates} 
each instance they appear;
while this allows to make a tighter bound on each term
with which non-asymptotic results are readily derived, we omit such a detailed analysis here and simply state
that the assertion in~\eqref{eq:rates} holds for $n$ large enough.

We write $\sic(\mc A) = \sic(\mc A|\mc C, \wh{\Cp}, s, e)$
where there is no confusion since,
for given $s$ and $e$,
the difference between $\sic(\mc A|\mc C, \wh{\Cp}, s, e)$ 
and $\sic(\mc A^\prime|\mc C, \wh{\Cp}, s, e)$ does not depend on candidates outside $(s, e)$
for any $\mc A, \mc A^\prime \subset \mc C \cap (s, e)$.
For a change point currently under consideration, say $\cp_{\circ}$,
we write its neighbouring change points as $\cp_{\pm}$ 
(i.e., $\Cp \cap (\cp_{-}, \cp_+) = \{\cp_\circ\}$)
allowing for $\cp_- = 0$ and $\cp_+ = n$,
and denote the associated jump sizes by $d_\circ$ and $d_{\pm}$, respectively.

For any given interval $(s, e]$, 
Proposition \ref{prop:detectability} establishes the sure detectability
of any change point in $\Cp^{(s, e)}$ as defined in \eqref{eq:sure:detect},
as well as the {\it un}detectability of any change point not belonging to
$\bar{\Cp}^{(s, e)}$ as defined in \eqref{eq:all:detect}.
\begin{prop}
\label{prop:detectability}
For any $0 \le s < e \le n$ (with $\Cp \cap (s, e) \ne \emptyset$) and $\cp_\circ \in \Cp \cap (s, e)$, 
let $\mc A \subset \mc D = \C \cap (s, e)$ denote a set of candidate estimators where
$\c_\pm \in \mc A \cup \{s, e\}$ satisfy $\cp_\circ\in  (\c_-, \c_+)$
as well as $\mc A \cap (\c_-, \c_+) = \emptyset$.
Then, there exist universal constants
$c^*, C^* \in (0, \infty)$ with $c^* < C^*$, 
with which the following statements hold on $\mc M_n$ 
for $n$ large enough:
Let
\begin{align}
	\max\l\{ d_+^2(\c_+ - \cp_+) \cdot \bbI_{\c_+ \ge \cp_+}, 
	d_-^2(\cp_- - \c_-) \cdot \bbI_{\c_- \le \cp_-} \r\} \le C^*\xi_n.
\notag
\end{align}
\begin{enumerate}[label = (\alph*)]
\item If $d_\circ^2 \min(\cp_\circ - \c_-, \c_+ - \cp_\circ) \ge C^*\xi_n$, we have 
$\sic(\mc A) > \sic(\mc A \cup \{\c_\circ^\prime\})$
for all $\c_\circ^\prime \in \mc V_\circ^\prime$.

\item Suppose $\cp_- < \c_-$ and $d_{\circ}^2(\cp_{\circ} - \c_-) < c^*\xi_n$.
Then, if either $\cp_+ > \c_+$ or $|\c - \cp_{\circ}| < (\cp_+ - \c)$, 
we have $\sic(\mc A) < \sic(\mc A \cup \{\c\})$.

\item Suppose $\c_+ < \cp_+$ and $d_{\circ}^2(\c_+ - \cp_{\circ}) < c^*\xi_n$.
Then, if either $\c_- > \cp_-$ or $|\c - \cp_{\circ}| < (\c - \cp_-)$, 
we have $\sic(\mc A) < \sic(\mc A \cup \{\c\})$.
\end{enumerate}
\end{prop}

Throughout, for any $\c_\pm, \c_\circ \in \C \cup \{0, n\}$ with $\c_- < \c_\circ < \c_+$,
we refer to $\c_\circ$ as detecting $\cp_\circ \in \Cp \cap (\c_-, \c_+]$
if $\cp_\circ = \arg\min_{\cp \in \Cp \cap (\c_-, \c_+]} |\c_\circ - \cp|$,
i.e., its nearest change point within $(\c_-, \c_+]$ is $\cp_\circ$,
even though there may be some $\cp_j \notin (\c_-, \c_+]$ closer to $\c_\circ$ than $\cp_\circ$.

Proposition \ref{prop:step:two} states that
when a given set $\mc A$ already contains an acceptable candidate for a change point in a local environment,
$\sic$ increases if another candidate detecting the same change point is added to $\mc A$, 
as well as that adding spurious candidates increases $\sic$.
\begin{prop}
\label{prop:step:two}
For any $0 \le s < e \le n$ and some $\c_{\circ} \in \mc D = \C \cap (s, e)$, 
let $\mc A \subset \mc D \setminus \{\c_\circ\}$ with $\c_\pm \in \mc A \cup \{s, e\}$ 
chosen such that $\c_- < \c_\circ < \c_+$ and $(\c_-, \c_+) \cap \mc A = \emptyset$.
Further, we suppose that $\c_\pm$ satisfy
\begin{enumerate}[label = (\alph*)]
\item $\Cp \cap (\c_-, \c_+] = \emptyset$, or
\item if $\Cp \cap (\c_-, \c_+] \ne \emptyset$, then for any $\cp_j \in \Cp \cap (\c_-, \c_+]$,
we have $d_j^2\min(\cp_j - \c_-, \c_+ - \cp_j) \le C^*\xi_n$. Additionally, for $\cp_{\circ} \in \Cp \cap (\c_-, \c_+]$ detected by $\c_\circ$, either
\begin{enumerate}[label=(\roman*)]
\item at least one of $\c_\pm$ is acceptable, i.e., $d_\circ^2\min(\cp_\circ - \c_-, \c_+ - \cp_\circ) \le \rho_n\nu_n$, or
\item $d_{\circ}^2|\c_{\circ}-\cp_{\circ}|>\wt{C}\,\xi_n$ for $\wt{C}>\max(C^*,\bar{C}(C^*)^2)$ with 
$\bar{C}$ as defined in Lemma~\ref{lem:sbic}.
\end{enumerate}
\end{enumerate}
Then, adding $\c_\circ$ to $\mc A$ yields an increase of $\sic$,
i.e., for $n$ large enough,
\begin{align*}
\sic(\mathcal{A}) < \sic(\mathcal{A} \cup \{\c_{\circ}\}) \quad \text{on} \quad \mathcal{M}_n.
\end{align*}
\end{prop}

The next proposition asserts that a set containing an unacceptable candidate yields larger $\sic$ 
than the one replacing it with a strictly valid estimator,
when the corresponding change point is detectable in the interval of consideration.
\begin{prop}
\label{prop:step:three}
For any $0 \le s < e \le n$ (with $\bar{\Cp}^{(s, e)} \ne \emptyset$) and $\cp_\circ \in \bar{\Cp}^{(s, e)}$,
let $\mc A \subset \mc D = \C \cap (s, e)$ be any candidate subset with
$\c_\pm \in \mc A \cup \{s, e\}$ satisfying 
$\cp_\circ \in (\c_-, \c_+)$, $\mc A \cap (\c_-, \c_+) = \emptyset$, 
$d_\circ^2 |\c_\pm - \cp_\circ| \ge c^*\xi_n$,
as well as
\begin{align*}
\max\l\{d_+^2(\c_+ - \cp_+) \cdot \bbI_{\c_+ \ge \cp_+}, \, 
d_-^2(\cp_- - \c_-) \cdot \bbI_{\c_- \le \cp_-}\r\} \le C^*\xi_n.
\end{align*}
Denote by $\c_{\circ}^*\in \mc V_\circ^*$ a strictly valid estimator for $\cp_{\circ}$, 
and by $\c_\circ$ an estimator detecting $\cp_\circ$ within $(\c_-, \c_+]$
which satisfies $d_\circ^2 |\c_\circ - \cp_\circ| \le \wt{C}\xi_n$
with $\wt{C}$ as in Proposition~\ref{prop:step:two},
while being unacceptable for $\cp_\circ$.
Then, adding $\c^*_{\circ}$ to $\mc A$ yields a greater reduction in the
$\rss$ than adding $\c_\circ$, 
i.e., for $n$ large enough,
\begin{align*}
\sic(\mathcal{A} \cup \{\c_{\circ}\}) > \sic(\mathcal{A} \cup \{\c_{\circ}^*\})
 \quad \text{on} \quad \mathcal{M}_n.
\end{align*}
\end{prop}

\subsection{Proof of Theorem~\ref{thm:sbic}}

On $\mc M_n$, the following arguments hold uniformly in $0 \le s < e \le n$ 
and the corresponding $\mc D = \C \cap (s, e)$
for $n$ large enough.
First, we note that
\begin{enumerate}[label=(D\arabic*)]
\setlength\itemsep{0em}
\item \label{eq_c1n} 
any set $\mc A \subset \mc D$ fulfilling \ref{eq_c1} contains
at least one estimator 
satisfying $d_j^2 \min_{\c \in \mc A}|\c - \cp_j| \le C^*\xi_n$
for all $\cp_j \in \Cp^{(s, e)}$.
\end{enumerate}
We prove \ref{eq_c1n} by contradiction.
Suppose that for some $\cp_\circ \in \Cp^{(s, e)}$, the set $\mc A$ does not contain 
any candidate within its $(C^* d_\circ^{-2} \xi_n)$-environment.
To such $\mc A$, we can add, if necessary, strictly valid candidates 
until the resultant set contains one strictly valid candidate 
for each $\cp_j \in \Cp \cap (s, e) \setminus \{\cp_\circ\}$.
Then, the conditions of Proposition~\ref{prop:detectability}~(a) are met,
and adding any $\c^\prime_\circ \in \mc V^\prime_\circ$ to such a set results in a decrease of $\sic$.

Also, we can always find a subset of $\mc D$ that fulfils \ref{eq_c1}, since
\begin{enumerate}[label=(D\arabic*)]
\addtocounter{enumi}{1}
\setlength\itemsep{0em}
\item \label{eq_d2}
any $\mc A \subset \mc D$ 
containing exactly one acceptable estimator for all $\bar{\Cp}^{(s, e)}$ 
with $|\mc A| = |\bar{\Cp}^{(s, e)}|$ satisfies \ref{eq_c1}.
\end{enumerate}
To see this, adding candidates detecting $\cp_j \in \bar{\Cp}^{(s, e)}$ to $\mc A$ incurs
monotonic increase of $\sic$ by Proposition~\ref{prop:step:two}
since in each step, either (a) or (b.i) therein is fulfilled for any candidates 
$\c_\circ \in \mc D \setminus \mc A$ 
(since $\rho_n\nu_n < C^*\xi_n$ under \eqref{eq:rates} for $n$ large enough).
Similarly, when adding those detecting $\cp_j \in \Cp \cap (s, e) \setminus \bar{\Cp}^{(s, e)}$ to $\mc A$, 
Proposition~\ref{prop:detectability}~(b)--(c) applies.

Denoting by $\mc F_{[m]}$ the collection of
the subsets of $\mc D$ of cardinality $m$ that fulfil \ref{eq_c1}.
By \ref{eq_c1n}, we have $|\mc F_{[m]}| = 0$ for $m < |\Cp^{(s, e)}|$.
Also, defining $m^* = \min\{1 \le m \le |\mc D|: \, |\mc F_{[m]}| \ne \emptyset\}$, 
we have $m^* \le|\bar{\Cp}^{(s, e)}| \le |\Cp^{(s, e)}| +2\le  m^* + 2$ by \ref{eq_d2}.
Suppose now that there exists $\mc A \in \bigcup_{m^* \le m \le m^* + 2} \mc F_{[m]}$
for which
\begin{enumerate}[label = (\alph*)]
\item $\vert \mc A \cap \mc V_j^\prime \vert \ne 1$ for $\cp_j \in \Cp^{(s, e)}$, or
\item $\vert \mc A \cap \mc V_j^\prime \vert > 1$ for $\cp_j \in \bar{\Cp}^{(s, e)} \setminus \Cp^{(s, e)}$, or
\item $\mc A \setminus \bigcup_{j: \, \cp_j \in \bar{\Cp}^{(s, e)}} \mc V_j^\prime \ne \emptyset$.
\end{enumerate}
We show that such a set $\mc A$ cannot be returned by \eqref{eq_c2}.
To this end, we apply the following operations to $\mc A$. 
Because the set changes after each operation, 
we denote the active set by $\mc A^\prime$ in the following which is initially set as $\mc A^\prime = \mc A$.

\begin{enumerate}[label=\textbf{Step \arabic*:}, itemindent = 25pt]
\item If $\mc A^\prime$ contains any estimator of $\Cp \cap (s, e) \setminus \bar{\Cp}^{(s, e)}$, 
iteratively remove such estimators from $\mc A^\prime$ one at a time which,
by Proposition~\ref{prop:detectability}~(b)--(c) and \ref{eq_c1n}, 
strictly reduces the $\sic$ monotonically.
Also remove any estimator $\c_{\circ} \in \mc A^\prime$ one at a time
which is too far from its nearest change point,
say $\cp_\circ$,
in the sense that $d_\circ^2 |\c_\circ - \cp_\circ| > \wt{C}\xi_n$;
this strictly reduces the $\sic$ by Proposition~\ref{prop:step:two}~(a), (b.ii) and~\ref{eq_c1n}.

\item If $\mc A^\prime \cap \mc V_\circ^\prime = \emptyset$
for some $\cp_\circ \in \Cp^{(s, e)}$,
by \ref{eq_c1n}, we have at least one $\c_\circ \in \mc A^\prime$
satisfying $d_\circ^2|\c_\circ - \cp_\circ| \le C^*\xi_n$.
Let $\c_\circ$ be the closest estimator of $\cp_\circ$ in $\mc A^\prime$ 
and identify $\c_\pm \in \mc A \cup \{s, e\}$ such that
$(\c_-, \c_+) \cap \mc A = \{\c_\circ\}$.
When  $d_\circ^2\min(\cp_\circ - \c_-, \c_+ - \cp_\circ) < c^*\xi_n$, 
we can remove one of $\c_{\pm}$ closer to $\cp_{\circ}$ while decreasing the $\sic$. 
To see this, suppose without loss of generality (otherwise consider the time series in reverse) that this is $\c_+$.
Then, $\c_+ > \cp_{\circ}$ since $\c_{\circ}$ is the estimator closest to $\cp_{\circ}$ in $\mc A^\prime$.
Denote by $\tilde{\c}_{\circ} = \c_+$ 
and define $\tilde{\c}_\pm$ analogously as $\c_\pm$ with regards to $\tilde{\c}_\circ$
(such that $\tilde{\c}_- = \c_\circ$), and let $\tilde{d}_\circ$ denote the jump size associated with
a change point $\tilde{\cp}_\circ$.
Then, one of the followings applies.
\begin{itemize}
\item Conditions of Proposition~\ref{prop:detectability}~(b) are met by $\tilde\c_\pm$
if $\c_\circ = \tilde{\c}_- \le \cp_{\circ} = \tilde{\cp}_{\circ} < \c_+ = \tilde{\c}_{\circ}$.

\item Conditions of Proposition~\ref{prop:step:two}~(a) are met by $\tilde\c_\pm$
if $\cp_{\circ} <\tilde{\c}_- < \tilde{\c}_{\circ} < \tilde{\c}_+ \le \cp_+$.

\item Conditions of Proposition~\ref{prop:step:two}~(b.ii) hold for $\tilde\c_\pm$ and $\tilde\cp_\circ$
if $\cp_{\circ} < \tilde{\c}_- < \tilde{\c}_{\circ} < {\cp}_+ = \tilde{\cp}_{\circ} <\tilde{\c}_+$,
since in this case, 
$\tilde{d}_{\circ}^2(\tilde{\cp}_\circ - \tilde{\c}_{\circ}) =  
\tilde{d}_{\circ}^2(\tilde{\cp}_\circ - \cp_\circ)\{1 - (\tilde{\c}_\circ - \cp_\circ)/(\tilde{\cp}_\circ - \cp_\circ)\}
\ge D_n - c^*\xi_n>\wt{C}\xi_n$ for $n$ large enough.
\end{itemize}

In all cases, removing $\tilde{\c}_{\circ} = \c_+$ results in a decrease of $\sic$.
Iteratively repeat the removal and re-defining of $\c_\circ$ and $\c_\pm$
until $d_\circ^2\min(\cp_\circ - \c_-, \c_+ - \cp_\circ) \ge c^*\xi_n$.
Then, the resultant $\mc A^\prime$ and $\c_\circ$ are such that
$\mc A^\prime \setminus \{\c_\circ\}$ meets the conditions of Proposition~\ref{prop:step:three} for $\cp_\circ$.
Therefore, replacing $\c_\circ$ with any of $\c_\circ^* \in \mc V_\circ^*$
yields a reduction in the $\sic$.
Repeat the above until $|\mc A^\prime \cap \mc V_j^\prime| = 1$ for all $\cp_j \in \Cp^{(s, e)}$,
which strictly decreases $\sic(\mc A^\prime)$ monotonically.

\item If $\mc A^\prime \cap \mc V_j^\prime = \emptyset$
for some $\cp_j \in \bar{\Cp}^{(s, e)} \setminus \Cp^{(s, e)}$
yet $\mc A^\prime$ contains an estimator of $\cp_j$,
we take the same steps as in Step~2 for all such $\cp_j$ so that $|\mc A^\prime \cap \mc V_j^\prime| = 1$,
which strictly decreases $\sic(\mc A^\prime)$ monotonically.

\item If there exists $\cp_j \in \bar{\Cp}^{(s, e)}$ 
for which there are more than one estimator in $\mc A^\prime$, 
through Steps 2--3, we have $\mc A^\prime \cap \mc V_j^\prime \ne \emptyset$.
Remove the duplicate estimators one at a time 
until all $\cp_j$ with $\mc A^\prime \cap \mc V_j^\prime \ne \emptyset$
have exactly one acceptable estimator in $\mc A^\prime$
which, by Proposition~\ref{prop:detectability}~(b)--(c) or
by Proposition~\ref{prop:step:two}~(a) and (b.i), results in a strictly monotonic reduction of $\sic$.
\end{enumerate}

After Steps 1--4, we have $\mc A^\prime$ that satisfies
$\mc A^\prime \setminus \bigcup_{j: \, \cp_j \in \bar{\Cp}^{(s, e)}} \mc V_j^\prime = \emptyset$,
with $|\mc A^\prime \cap \mc V_j^\prime| = 1$ for $\cp_j \in \Cp^{(s, e)}$
and $|\mc A^\prime \cap \mc V_j^\prime| \le 1$ for $\cp_j \in \bar{\Cp}^{(s, e)} \setminus \Cp^{(s, e)}$,
as well as $\sic(\mc A^\prime) < \sic(\mc A)$
because under (a)--(c), at least one of Steps 1--4 above has to take place. 
Further, if necessary, by adding strictly valid candidates to $\mc A^\prime$
for all those $\cp_j \in \bar{\Cp}^{(s, e)} \setminus \Cp^{(s, e)}$ 
with $|\mc A^\prime \cap \mc V_j^\prime| = 0$, 
we yield $\mc A^{\prime\prime} \supset \mc A^\prime$ fulfilling \ref{eq_c1} by \ref{eq_d2} 
and of cardinality $|\bar{\Cp}^{(s, e)}|$, i.e., 
$\mc A^{\prime\prime} \in \bigcup_{m^* \le m \le m^* + 2} \mc F_{[m]}$.
Since $\mc A^\prime \subset_R \mc A^{\prime\prime}$ with $\subset_R$ defined below \eqref{eq_c2}
and $\sic(\mc A^\prime) < \sic(\mc A)$, 
this shows that $\mc A$ with candidates belonging to either of (a)--(c) cannot be returned in \eqref{eq_c2}.
In conclusion, $\wh{\Cp}^{(s, e)}$ obtained from \eqref{eq_c2}
satisfies the assertion of the theorem.

\subsection{Proof of Theorem~\ref{thm:sbic:full}}

Under \eqref{eq:rates}, 
we make the following observations:
For all $j = 1, \ldots, q_n$,
\begin{enumerate}
\item[(a)] $d_j^2 | \wh\cp_j - \cp_j | \le \rho_n \nu_n < c^* \xi_n$ for any $\wh\cp_j \in \mc V_j^\prime$, and
\item[(b)] $d_j^2 \min(\cp_j - \cp_{j - 1}, \cp_{j + 1} - \cp_j) \ge D_n > 2\max(C^* \xi_n, \rho_n\nu_n)$
\end{enumerate}
for $n$ large enough. 

In iteratively applying Steps 1--4 of {\tt LocAlg},
Theorem~\ref{thm:sbic} guarantees that $\wh{\Cp}$ contains 
only acceptable estimators of $\cp_j \in \Cp$. 
Also, each change point can belong to $\Cp^{(s, e)}$ 
defined by the interval of consideration $(s, e] = (\c_L,\c_R]$ at most once:
When $\cp_j \in \Cp^{(s, e)}$ for the first time, it gets detected by some $\wh{\cp}_j \in \mc V_j^\prime$
by Theorem~\ref{thm:sbic}.
Then, in the following iterations, either $\cp_j \notin (s, e)$, 
or some $\c \in (\mc C \cup \wh{\Cp}) \cap [\min(\cp_j, \wh{\cp}_j) \, , \,
\max(\cp_j, \wh{\cp}_j)]$ defines the endpoints of the local environment by Step~2.
In the latter case, $\cp_j$ cannot be a detectable change point 
within the interval of consideration of this particular iteration due to (a),
which guarantees that no further estimator for $\cp_j$ is added to $\wh{\Cp}$.

When there exists $\cp_j \in \bar{\Cp}^{(s, e)} \setminus \Cp^{(s, e)}$ at some iteration,
Theorem~\ref{thm:sbic} indicates that it may or may not get detected at this iteration.
If it does, an acceptable estimator of $\cp_j$ is added to $\wh{\Cp}$
and the same argument as above applies.
If not, without loss of generality, suppose $\cp_j - s \le e - \cp_j$.
By construction, $c^*\xi_n \le d_j^2(\cp_j - s) < C^*\xi_n$ and
thus from (b), we have 
\begin{align*}
d_{j - 1}^2(s - \cp_{j - 1}) = d_{j-1}^2 (\cp_j-\cp_{j-1}) 
\left\{ 1-\frac{d_j^2(\cp_j-s)}{d_j^2(\cp_j-\cp_{j-1})} \right\} 
\ge D_n - C^*\xi_n > \rho_n\nu_n,
\end{align*}
i.e., the boundary point $s$ cannot be an acceptable estimator for either $\cp_{j - 1}$ or $\cp_j$.
Consequently, it cannot have already been added to $\wh{\Cp}$ 
in the previous iterations by Theorem~\ref{thm:sbic}. 
Therefore, all acceptable estimators for $\cp_j$, 
with the possible exception of $\c_{\circ}$ identified in Step~1, 
remain in $\mc C$ by (a)--(b) and how it is reduced in Step~4 of {\tt LocAlg}.

Next, we justify the removal of $\c_\circ$ from $\mc C$ at each iteration.
Clearly, if $\c_{\circ}$ is not acceptable for any change point, 
it can be safely removed from the future consideration.
Next, suppose that $\c_\circ$ is an acceptable estimator of $\cp_j$
and $\cp_j - s \le e - \cp_j$.
\begin{enumerate}[label = (\alph*)]
\item When $\cp_j \in \Cp^{(s, e)}$, we have either $\c_\circ$ or another acceptable estimator of $\cp_j$ accepted by {\tt PrunAlg}, 
and therefore $\c_\circ$ can be removed.

\item When $\cp_j \in \bar{\Cp}^{(s, e)} \setminus \Cp^{(s, e)}$, 
if $\cp_j$ is detected at the current iteration, the same argument as in (a) applies.
If not, as shown above, 
$s$ has not been added to $\wh{\Cp}$ yet and 
by construction of the interval of consideration in Step~2, it follows that
\begin{align*}
d_j^2 G_L(\c_{\circ}) = d_j^2(\c_{\circ} - s) \le C^*\xi_n + \rho_n\nu_n = C^*\xi_n(1 + o(1)),
\end{align*}
which shows that $\c_{\circ}$ cannot fulfil \eqref{eq:assum:det:intervals} for $\cp_j$ 
(nor any other change point as it is acceptable for $\cp_j$). 
Consequently, $\c_{\circ}$ can safely be removed from $\mc C$ 
since by Assumption~\ref{assum:det:intervals} and the construction of $\mc R$ in Step~4, 
there remains at least one acceptable estimator for $\cp_j$ that fulfils \eqref{eq:assum:det:intervals} 
in $\mc C$ after the current iteration.

\item When $\cp_j \notin \bar{\Cp}^{(s, e)}$ (which is not necessarily situated within $(s, e)$),
we first consider the case where $s$ has already been accepted. 
Then by Theorem~\ref{thm:sbic}, $s$ is acceptable for some change point, say $\cp_{j^\prime}$, such that
\begin{align*}
&d_{j^\prime + 1}^2 (\cp_{j^\prime+1} - s) = d_{j^\prime+1}^2(\cp_{j^\prime+1} - \cp_{j^\prime})\,
\left\{ 1 - \frac{d_{j^\prime}^2(s-\cp_{j^\prime})}{d_{j^\prime}^2(\cp_{j^\prime + 1} - \cp_{j^\prime})} \right\}\\&
\ge D_n\left( 1-\frac{\rho_n\nu_n}{D_n} \right) > C^*\xi_n,
\end{align*}
i.e., $\cp_{j^\prime + 1}$ is either surely detectable within $(s, e)$, too close to $e$, or $\cp_{j^\prime + 1} \notin (s, e)$
to have been detected by $\c_\circ$.
Therefore, $j = j^\prime$ and $\c_{\circ}$ can safely be removed as in (a)
since there already exists an acceptable estimator $s$ in $\wh{\Cp}$.
If $s$ has not been accepted, the argument analogous to that in (b) applies.
\end{enumerate}
The case when $\cp_j - s > e - \cp_j$ is similarly handled.

The above (b)--(c) show that under Assumptions~\ref{assum:cand} and~\ref{assum:det:intervals}, 
for each $j = 1, \ldots, q_n$, 
acceptable estimators of $\cp_j$ remain in $\mc C$ until its detection
and at least one of them, when set as $\c_\circ$ in Step~1 of {\tt LocAlg},
leads $\cp_j$ to belong to $\Cp^{(s, e)}$ at some iteration,
from which we conclude that all $\cp_j \in \Cp$ are eventually detected by acceptable estimators.
Finally, $|\mc R| \ge 1$ at all iterations since $\mc R$ contains $\c_\circ$ at least,
which ensures that {\tt LocAlg} terminates eventually.

\subsection*{Acknowledgements} 

Haeran Cho was supported by the Engineering and Physical Sciences Research Council grant no. EP/N024435/1.
The authors would like to thank the Isaac Newton Institute for Mathematical Sciences for support and hospitality during the programme `Statistical scalability'  when work on this paper was undertaken. This work was supported by: EPSRC grant number EP/R014604/1.

\bibliographystyle{asa}
\bibliography{fbib}

\clearpage

\appendix

\numberwithin{equation}{section}
\numberwithin{figure}{section}
\numberwithin{table}{section}

\part*{Appendix}

\section{Data Example: Kepler light curve data}
\label{sec:kepler}

\begin{figure}[htbp]
\centering
\includegraphics[width=\textwidth]{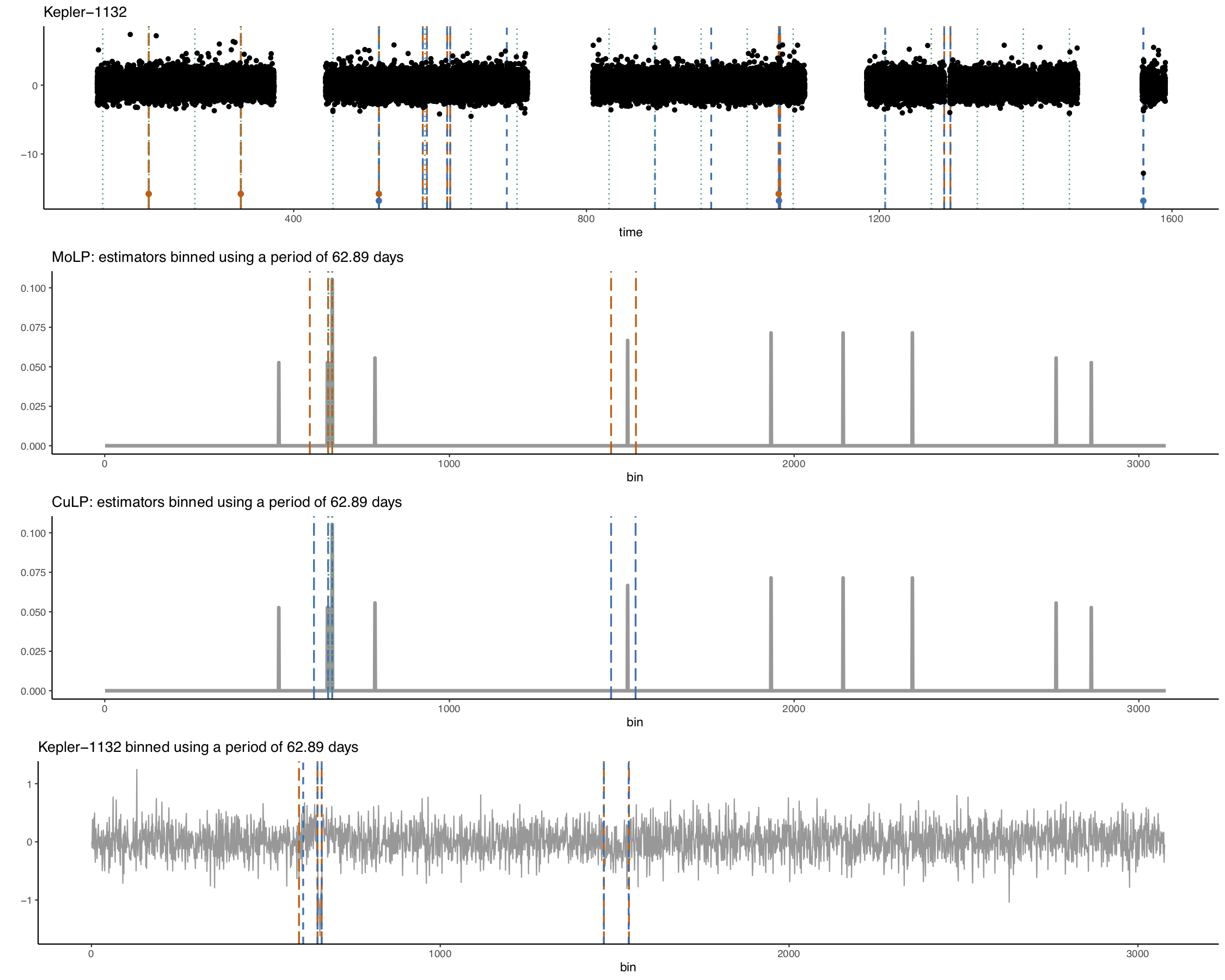}
\caption{Top: Luminosity of Kepler-1132 measured every half an hour (approximately)
with change point estimators (vertical lines; longdashed: MoLP, dashed: CuLP)
and the beginnings of the anomalous intervals detected by \citet{fisch2018}
periodically repeated every $62.89$ days (vertical dotted lines). 
Where two estimators returned by the same method lie too close to each other to be distinguished, a filled circle is added.
Second, third: Change point estimators from the top panel binned using the periodicity of $62.89$.
Bottom: Kepler-1132 data binned and aggregated using the periodicity of $62.89$ days.
In the second, third and bottom panels, 
change point estimators from the aggregated data are also given as vertical lines
(longdashed: MoLP, dashed: CuLP, dotted: \citet{fisch2018}).}
\label{fig:kepler}
\end{figure}

Kepler light curve dataset contains regularly measured luminosity of stars.
The transit of an orbiting planet results in periodically recurring segments of reduced luminosity,
which can be used for detecting exoplanets via the transit method \citep{sartoretti1999}.
Regarding segments of dimmed luminosity as collective anomalies,
\citet{fisch2018} apply their anomaly detection methodology to the light curve data
obtained from Kepler-1132 (available in the R package {\tt anomaly} \citep{anomaly}),
which is known to host at least one orbiting planet \citep{kepler}.
In their paper, the data is pre-processed into equally sized bins 
aggregating the luminosity from different orbits
using the known periodicity ($62.89$ days) of the orbiting planet.
This amplifies the signal and transforms the irregularly sampled time series data into a regular one.
From the aggregated data, they detect a short interval of collective anomalies over $[649, 660]$ 
(at the scale of bins).

We first apply the proposed localised pruning to the raw Kepler-1132 data without aggregation,
the result of which is reported in the top panel of Figure~\ref{fig:kepler}.
Without further information available, 
we simply ignore the presence of missing observations, which yields $n = 51405$.
Considering the possible presence of outliers and heavy tails,
we set the penalty at $\xi_n = \log^{1.1}(n)$.
MoLP (with $\alpha = 0.2$ and $\eta = 0.4$) detects $14$ estimators in total,
while CuLP (with $C_\zeta = 0.5$) 
returns $16$ estimators,
out of which there are $10$ overlapping estimators in the sense that
either they are identical or very close to one another.
Unlike \citet{fisch2018}, we do not use the known periodicity 
to accumulate the information obtained from different orbits,
nor do we utilise the knowledge that the changes are of epidemic nature.
Nonetheless, as demonstrated in Figure~\ref{fig:kepler},
both MoLP and CuLP identify the anomalous interval detected by \citet{fisch2018} at some orbits.
Additionally detected change points may be attributed to the missingness in the data
which is not accounted for by our methodology,
particularly the pair in the vicinity of $1290$ in the observation time scale.

We also analyse the binned and aggregated data of length $N = 3078$
with the penalty $\xi_N = \log^{1.01}(N)$ 
chosen on the basis of Gaussian-like tail behaviour of the binned data.
Both MoLP and CuLP yield $5$ estimators including $648$ and $660$, correctly
identifying the anomalous segment reported in \citet{fisch2018}.

In summary, our methodology is able to detect the periodic reduction in luminosity of Kepler-1132
without aggregating the signal using the extra information of periodicity which,
in the problem of detecting exoplanets, may not be readily available. 

\section{CUSUM-based candidate generation}
\label{sec:wbs}

The CUSUM statistic in \eqref{eq:cusum} is designed to test 
the null hypothesis of no change point ($H_0: q_n = 0$)
against the at-most-one-change alternative ($H_1: q_n = 1$).
It corresponds to the likelihood ratio statistic under i.i.d.\ Gaussian errors
and as such, is particularly appropriate for single change point estimation. 

For multiple change point detection,
\citet{vostrikova1981} and \citet{venkatraman1992} establish
the consistency of the Binary Segmentation algorithm
that makes recursive use of CUSUM-based estimation.
However, its sub-optimality, both in terms of the conditions required for the consistency
and the rate of change point localisation, has been noted in \citet{fryzlewicz2014}.
As an alternative, he proposes the Wild Binary Segmentation (WBS) 
which aims at isolating the change points by drawing a large number of random intervals.
When a sufficient number of random intervals are drawn, with large probability,
there exists at least one interval which is well-suited for the detection and localisation of 
each $\cp_j, \, j = 1, \ldots, q_n$.
Since then, \citet{fryzlewicz2018} proposes its variation (WBS2)
that draws random intervals in a more systematic fashion and 
generates a complete solution path,
while \citet{kovacs2020} propose a `seeded' version of WBS
that constructs the background intervals in a deterministic fashion.
In the WBS and its variants, the candidates are generated by scanning the data multiple times
over a large number of (randomly drawn) intervals,
and various pruning methods have been proposed including
thresholding, sequential application of an information criterion \citep{fryzlewicz2014}
and the steepest-drop to low levels (SDLL) method \citep{fryzlewicz2018}.

We propose the following version of WBS2 as a candidate generating mechanism.
It requires the tuning parameters
$R_n$, the maximal number of random intervals to be drawn at each iteration,
and $\wt{Q}_n$, which relates to the maximal depth of recursion $L_n$ as
$L_n = \lfloor \log_2(\wt{Q}_n + 1) \rfloor$.
The step-by-step description of the WBS2 is provided below.

\begin{enumerate}[label=\textbf{Step \arabic*:}, itemindent = 25pt]
\setcounter{enumi}{-1}
\item Initialise the input arguments:
The set of candidates $\C(R_n, \wt Q_n) = \emptyset$,
$s = 0$, $e = n$ and the recursion depth $\ell = 1$. 

\item Quit the routine if 
$e - s = 1$ or $\ell > L_n$; if not,
let $\wt R = \min\{R_n, (e - s)(e - s - 1)/2\}$. 
If $\wt R \le R_n$, let 
$\mc R_{s, e} =\{(l, r) \in \Z^2: \, s \le l < r \le e \text{ and } r - l > 1\}$ 
serve as $[s_m, e_m], \, m = 1, \ldots, \wt{R}$. 
If not, draw $\wt R$ intervals $[s_m, e_m], \, m = 1, \ldots, \wt R$,
uniformly at random from the set $\mc R_{s, e}$.

\item Identify $(m_\circ, \c_\circ) =
\arg\max_{(m, b): \, 1 \le m \le \wt R, \, s_m < b < e_m} |\mc X_{s_m, b, e_m}|$.
	
\item Update $\C(R_n, \wt Q_n)$ 
by adding $\c_\circ$ and store its natural detection interval $\mc I_N(\c_\circ) = (s_\circ, e_\circ]$.

\item Repeat Steps~1--3 separately with $(s, \c_\circ, \ell + 1)$ and $(\c_\circ, e, \ell + 1)$.
\end{enumerate}

Through implementing the maximal recursion depth into the procedure,
it trivially holds that the size of candidate set satisfies
$\vert \C(R_n, \wt Q_n) \vert \le \wt Q_n$.
We propose to apply the localised pruning to the thus-generated set of candidates 
$\C(R_n, \wt Q_n)$,
which satisfies Assumptions~\ref{assum:cand} and~\ref{assum:det:intervals} 
on the set of candidate estimators.

\begin{prop}
\label{prop:wbs}
\hfill
\begin{enumerate}[label=(\alph*)]
\item Let $\p(\mc M_n^{(11)}) \to 1$ where $\mc M_n^{(11)}$ is defined in Assumption~\ref{assum:vep}~(a).
Also, suppose that there exist some $\beta \in (0, 1]$ and $c_\delta \in (0, 1)$ satisfying
\begin{align}
& \min_{1 \le j \le q_n} \delta_j \ge c_\delta n^\beta 
\quad \text{and} \quad 
\frac{\omega_n^2}{\min_{1 \le j \le q_n} d_j^2 n^{5\beta - 4}} \to 0
\label{cond:wbs:one}
\end{align}
where, as before, $\delta_j = \min(\cp_j - \cp_{j - 1}, \cp_{j + 1} - \cp_j)$
and $\omega_n$ is as in Assumption~\ref{assum:vep}~(a).
In addition, suppose that
\begin{align}
\frac{n^{2 - 2\beta} \log(n)}{R_n} \to 0, \qquad \frac{q_n}{\wt Q_n} \to 0
\label{cond:wbs:two}
\end{align}
and let $\rhow = c_W n^{4 - 4\beta} \omega_n^2$ for some $c_W \in (0, \infty)$.
Then, it holds
\begin{align*}
& \p\Big(\max_{1 \le j \le q_n} 
\min_{\c \in \C(R_n, \wt Q_n)} d_j^2|\c - \cp_j| \le \rhow \Big) \to 1.
\end{align*}
\item Suppose $n^{-1} \omega_n^2 \wt Q_n \to 0$.
Then, for any realisation of the random intervals, we have
$n^{-1} \omega_n^2 \vert\C(R_n, \wt Q_n) \vert \to 0$.
\item Suppose that conditions in (a) hold. Then, for each $j = 1, \ldots, q_n$, 
there exists $\check{\c} \in \{\c \in \C(R_n, \wt Q_n): \, d_j^2|\c - \cp_j| \le \rhow\}$
such that $\min(\check{\c} - \check{s}, \check{e} - \check{\c}) \ge c \delta_j$,
where $\mc I_N(\check{\c}) = (\check{s}, \check{e}]$ represents the natural detection interval of $\check{\c}$
and $c$ is a universal constant satisfying $c \in (0, 1]$.
\end{enumerate}
\end{prop}

\begin{rem}\label{rem_25_wbs}
For each $\c \in \C(R_n, \wt Q_n)$, 
the natural detection interval $\mc I_N(\c) = (s, e]$ 
can serve as its detection interval $\mc I(\c)$
in which case the detection distances are given by $G_L(\c) = \c - s$
and $G_R(\c) = e - \c$.
Proposition~\ref{prop:wbs}~(c) indicates that 
$\C(R_n, \wt Q_n)$ fulfils Assumption~\ref{assum:det:intervals}
under Assumption~\ref{assum:size}.
Besides, by construction, each estimator in 
$\C(R_n, \wt Q_n)$ is distinct and therefore $\C(R_n, \wt Q_n)$
bypasses the issue discussed in Remark~\ref{rem_25_mos}~(b).
\end{rem}

Compared to the condition (a) of Proposition~\ref{prop:mosum:a}
on the minimal size of change, measured by the jump size $d_j$ and spacing $\delta_j$,
for the MOSUM-based candidate generating mechanism,
the corresponding condition in~\eqref{cond:wbs:one} is considerably stronger.
Also, the rate of localisation reported in Proposition~\ref{prop:mosum:a} is always
tighter than $\rhow$ given in the above theorem. 
The bottleneck in our theoretical analysis of the CUSUM-based candidate generation procedure
is the following: 
The WBS-type procedures looking for the largest CUSUM at each iteration,
do not rule out that a change point $\cp_j$ is detected
by its estimator $\c_\circ$ within an interval $(s_\circ, e_\circ)$ 
which also contains $\cp_{j - 1}$ or $\cp_{j + 1}$ (and more) well within the interval.
In such a case, the localisation rate $\vert \c_\circ - \cp_j \vert$ depends not only on $d_j$
but also on the minimum spacing $\min_{1 \le j \le q_n} \delta_j$, 
which results in the sub-optimal localisation rate as well as the detection lower bound
given in Proposition~\ref{prop:wbs}.
Besides, the theoretical guarantee therein is for the homogeneous change points only.
An analogous result is reported in \citet{wang2018}
where the WBS is adopted for high-dimensional change point detection which,
to the best of our knowledge, is the best available result 
on the detection lower bound and the localisation rate of the WBS.
\\ [1mm]
The maximum number of intervals to be drawn at each iteration, $R_n$, is 
required to increase as the minimal spacing $\min_{1 \le j \le q_n} \delta_j$ decreases
(see~\eqref{cond:wbs:two}), thus increasing the total computational complexity of 
the candidate generating procedure as $O(R_n n)$.

The consistency of the localised pruning algorithm 
in combination with the CUSUM-based candidate generating mechanism  
follows from Proposition~\ref{prop:wbs} and Theorem~\ref{thm:sbic:full}.

\begin{thm}
\label{cor:wbs}
Let Assumptions~\ref{assum:vep}--\ref{assum:size} and~\ref{assum:penalty} hold
and additionally, let $\xi_n^{-1} n^{4 - 4\beta}\omega_n^2 \to 0$.
Also suppose that the conditions in Proposition~\ref{prop:wbs} are satisfied.
Then, the localised pruning algorithm {\tt LocAlg} applied to $\C(R_n, \wt Q_n)$
yields $\wh{\Cp} = \{\wh\cp_1 < \ldots < \wh\cp_{\wh{q}_n}\}$
which consistently estimates $\Cp$, i.e.,
\begin{align*}
& \p\l\{\wh q_n = q_n; \, 
\max_{1 \le j \le q_n} d_j^2 |\wh \cp_{j} \bbI_{j \le \wh q_n} - \cp_{j}| \le \rhow\nu_n\r\} 
\to 1,
\end{align*}
with $\rhow$ as in Proposition~\ref{prop:wbs} and $\nu_n \to \infty$ arbitrarily slow.
\end{thm}

The additional requirement on the penalty $\xi_n$ is necessary due to
the localisation rate achieved by the CUSUM-based candidate generation 
always dominating $\omega_n^2$, such that the penalty needs to be chosen accordingly larger;
see also the discussion following~\eqref{eq:rates}.
In view of the discussion below Proposition~\ref{prop:wbs}, 
we believe that such a requirement on the penalty term cannot be lifted
when performing model selection on the candidates generated by a WBS-type method
using an information criterion, unless some modification of the WBS
such as that proposed in \cite{baranowski2016} is adopted.

In practice, it is not straightforward to select $\wt Q_n$ which effectively imposes 
an upper bound on the number of candidates.
For numerical studies in Section~\ref{sec:numeric},
instead of selecting $\wt Q_n$, we choose a weak threshold $\zeta_n$ as a multiple of $\sqrt{\log(n)}$,
and keep only those candidates for which the corresponding CUSUM statistics
(after standardisation) exceed $\zeta_n$.
This approach provides more flexibility to deal with heavy-tailedness or serial dependence
present in the error sequence.


\subsection{Proof of Proposition~\ref{prop:wbs}}

Firstly, (b) follows directly from the construction of WBS2 and the condition on $\wt Q_n$, since
\begin{align*}
\vert \C(R_n, \wt Q_n) \vert \le 
\sum_{j=0}^{L_n} 2^j = 2^{L_n} - 1 \le \wt Q_n.
\end{align*}

The following proof of (a) is an adaptation of the proof of Theorem~3.1~(iii) of \cite{fryzlewicz2018}
and that of Theorem~2 of \cite{wang2018}.
Throughout the proof, we adopt $C_i, \, i \ge 1$ to denote positive constants.
Also, $\mc X_{s, b, e}(f)$ (resp. $\mc X_{s, b, e}(\vep)$)
denotes the CUSUM statistic analogously defined as $\mc X_{s, b, e}$ in \eqref{eq:cusum}
with $f_t$ ($\vep_t$) replacing $X_t$.

We define the following intervals for $j = 0, \ldots, q_n$,
\begin{align*}
&I_j = [r_j, \ell_{j + 1} ]\qquad \text{where} \quad 
r_j = \cp_j + \lceil (\cp_{j + 1} - \cp_j)/3\rceil,  \quad
\ell_{j+1} = \cp_{j+1} - \lceil (\cp_{j+1} - \cp_{j})/3\rceil
\end{align*}
and for $1 \le u + 1 < v \le q_n + 1$,
\begin{align*}
&	I^{t_1, t_2}_{u, v} = [\max(0, \cp_u + t_1), \min(\cp_v + t_2, n)]  
\quad \text{with} \quad t_1, t_2 \in [- \underline{\Delta}_n, \underline{\Delta}_n],
\text{ where } \underline{\Delta}_n =\frac{\rho_n^{(W)}}{\min_{1 \le j \le n} d_j^2}.
\end{align*}

Suppose that on each interval $I^{t_1,t_2}_{u, v}$,
we draw $R_n$ intervals $\{[s_m, e_m], \, m = 1, \ldots, R\}$ randomly and uniformly
from $\left\{(l, r) \in I^{t_1, t_2}_{u, v}  \times I^{t_1, t_2}_{u, v} : \, l + 1 < r\right\}$.
When $R_n \ge |I^{t_1, t_2}_{u, v}|(|I^{t_1, t_2}_{u, v}| - 1)/2$,
we use $\{[s_m, e_m], \, m = 1, \ldots, \wt{R}\}$
with $\wt{R} = |I^{t_1, t_2}_{u, v}|(|I^{t_1, t_2}_{u, v}| - 1)/2$ which contains all feasible sub-intervals of $I^{t_1, t_2}_{u, v}$.
For notational convenience, we do not specify the (stochastic) dependence
of $(s_m, e_m)$ on $u, v, t_1$ or $t_2$.

For each interval $I^{t_1, t_2}_{u, v}$, consider the event
$\mc A^{t_1, t_2}_{u, v} = 
\bigcap_{j = u + 1}^{v - 1} \bigcup_m \{(s_m, e_m) \in I_{j - 1} \times I_j\}$.
If $\wt R \le R_n$, we have $\p((\mc A^{t_1, t_2}_{u, v})^c) = 0$; if not, 
\begin{align*}
\p\l((\mc A^{t_1, t_2}_{u, v})^c\r) 
&\le
q_n \prod_{m = 1}^{R_n}  \max_{u + 1 \le j \le v - 1} \l\{ 1 - \p((s_m, e_m) \in I_{j - 1} \times I_j) \r\}
\le 
q_n\l(1 - \frac{c_{\delta}^2}{9n^{2 - 2\beta}}\r)^{R_n}
\end{align*}
such that for $\Omega_n := \bigcap_{t_1, t_2, u, v} \mc A^{t_1, t_2}_{u, v}$, 
by $\log(1 - x) \le -x$ for $x \in [0, 1)$,
\begin{align*}
\p(\Omega_n)
\ge 
1 - \sum_{t_1, t_2, u, v} \p((\mc A^{t_1, t_2}_{u, v})^c)
\ge 
1 - \frac{1}{2}q_n(q_n + 1)(q_n + 2) (2\underline{\Delta}_n + 1)^2\exp\l(- \frac{c_{\delta}^2\,R_n}{9n^{2 - 2\beta}}\r) \to 1
\end{align*}
under \eqref{cond:wbs:two}.
We claim that on $\Omega_n \cap \mc M_n^{(11)}$,
\begin{enumerate}[label = (W\arabic*)]
\item at some iteration,
if there exist $1 \le u + 1 < v \le q_n +1$ such that
$s$ and $e$ satisfy $\max\{d_u^2|s - \cp_u|, d_v^2|e - \cp_v|\} \le \rhow$,
\item the call of Steps~1--3 of WBS2 with such $s$ and $e$ as its arguments
adds $\c_\circ$ which satisfies
$d_j^2 |\c_\circ - \cp_j| \le \rho_n^{(W)}$ for some $j \in \{u + 1, \ldots, v - 1\}$.
\end{enumerate}

The condition in (W1) trivially holds at the very first iteration of WBS2
with $s = \cp_0 = 0$ and $e = \cp_{q_n + 1} = n$.
Then by induction, each $\cp_j, \, j = 1, \ldots, q_n$ is detected by an estimator 
within $(d_j^{-2}\rhow)$-distance
before the depth exceeds $\lceil \log_2(q_n +1) \rceil + 1$ thanks to (W2),
since we add (at most) $2^{\ell - 1}$ elements to 
$\C(R_n, \wt Q_n)$ at each depth $\ell$, 
which completes the proof of (a).

It remains to show that (W2) holds given that (W1) is met by some $s$ and $e$.
Let
\begin{align*}
(s_\circ, \c_\circ, e_\circ) 
= \arg\max_{(s_m, b, e_m): \, s_m < b < e_m, 1 \le m \le R_n}
|\mc X_{s_m, b, e_m}|.
\end{align*}
On the event $\Omega_n$,
there exists at least one interval $(s_{m(j)}, e_{m(j)}] \in \{(s_m, e_m] \subset (s, e], \, m = 1, \ldots, \wt R\}$
satisfying $(s_{m(j)}, e_{m(j)}) \in I_{j - 1} \times I_j$ for each $j \in \{u + 1, \ldots, v - 1\}$, 
which is non-empty by (W1).
Denoting by 
$\c^*_j = \arg\max_{s_{m(j)} < b < e_{m(j)}} |\mc X_{s_{m(j)}, b, e_{m(j)}}|$,
we have
\begin{align}
|\mc X_{s_\circ, \c_\circ, e_\circ}| 
 \ge 
\max_{u + 1 \le j \le v - 1} |\mc X_{s_{m(j)}, \c^*_j, e_{m(j)}}|
\ge \max_{u + 1 \le j \le v - 1} |\mc X_{s_{m(j)}, \cp_j, e_{m(j)}}|. 
\label{eq:pf:wbs:one}
\end{align}
On $\mc M_n^{(11)}$, it holds as in \eqref{eq_noise_cusum}
\begin{align}
|\mc X_{s, b, e}(\vep)| \le 2\,\omega_n.
\label{eq:pf:wbs:two}
\end{align}
Also, under \eqref{cond:wbs:one}, it follows straightforwardly that
\begin{align}
(d^2_j\delta_j)^{-1} \omega_n^2\to 0.
\label{cond:wbs:three}
\end{align}
Then, we have
\begin{align}
& |\mc X_{s_\circ, \c_\circ, e_\circ}  (f)| 
\ge \max_{u + 1 \le j \le v - 1} |\mc X_{s_{m(j)}, \cp_j, e_{m(j)}}|-2\omega_n
\ge \max_{u + 1 \le j \le v - 1} |\mc X_{s_{m(j)}, \cp_j, e_{m(j)}}(f)|-4\omega_n\notag\\
&\ge \min_{1 \le j \le q_n} \frac{\sqrt{d_j^2 \delta_j}}{\sqrt 6} - 4\omega_n > 
\min_{1 \le j \le q_n} \frac{\sqrt{d_j^2 \delta_j}}{2\sqrt 6} \label{eq_ck_b24a2}
\end{align}
by \eqref{cond:wbs:three} for $n$ large enough,
which shows in particular that there is at least one change point within $(s_{\circ}, e_{\circ})$.

Let $\cp_\pm$ denote the two change points $\cp_- < \c_{\circ} \le \cp_+$ satisfying 
$(\cp_-, \c_\circ) \cap \Cp = \emptyset$ and $(\c_\circ, \cp_+) \cap \Cp = \emptyset$.
From \eqref{eq_ck_b24a2}, at least one of $\cp_\pm$ belongs to $(s_\circ, e_\circ)$.
If $\cp_+ \notin (s_\circ, e_\circ)$, then by Lemma~8~(b) of \cite{wang2018b},
$\mc X_{s_\circ, b, e_\circ}(f)$ does not change sign and has strictly decreasing absolute values for $\cp_- \le b \le \c_\circ$.
In this case, we set $\cp_j = \cp_-$. 
If $\cp_- \notin (s_\circ, e_\circ)$, similarly,
$\mc X_{s_\circ, b, e_\circ}(f)$ does not change sign and 
has strictly increasing absolute values for $\c_\circ \le b \le \cp_+$ (their Lemma~8~(a)), 
and we set $\cp_j = \cp_+$.
If both $\cp_\pm \in (s_\circ, e_\circ)$,
by Lemma~8~(c)--(d) of \cite{wang2018b} and Lemma~2.2 of \cite{venkatraman1992},
$\mc X_{s_\circ, b, e_\circ}(f)$ is either strictly decreasing in modulus without sign change
for $\cp_- \le b \le \c_{\circ}$,  
or strictly increasing in modulus without sign change for $\c_{\circ} \le b \le \cp_+$. 
In the first case, we set $\cp_j = \cp_-$ while in the latter, we set $\cp_j = \cp_+$.

If the thus-identified $\cp_j \ge \c_\circ$, we consider the time series in reverse
such that w.l.o.g., we suppose that $\c_\circ \ge \cp_j$ and 
$\vert \mc X_{s_\circ, b, e_\circ}(f) \vert$ is strictly decreasing between $\cp_j$ and $\c_{\circ}$.
In addition, we assume that $\mc X_{s_\circ, \c_\circ, e_\circ} > 0$;
otherwise, consider $-X_t$ (resp. $-f_t$ and $-\vep_t$) in place of $X_t$ ($f_t$ and $\vep_t$).

Then, by \eqref{eq:pf:wbs:two}, \eqref{cond:wbs:three} and the arguments analogous to
those adopted in \eqref{eq_ck_b24a2}, we yield 
$\vert \mc X_{s_\circ, \c_\circ, e_\circ}(\vep)\vert /\mc X_{s_\circ, \c_\circ, e_\circ} = o(1)$ and in particular,
\begin{align}\label{eq_ck_b24a}
\mc X_{s_\circ, \c_\circ, e_\circ}(f) > 0
\end{align}
for large enough $n$.
Also from \eqref{eq:pf:wbs:one}--\eqref{eq:pf:wbs:two} and by the construction of $(s_{m(j)}, e_{m(j)}]$, we yield
\begin{align}
\mc X_{s_\circ, \cp_j, e_\circ}(f) &\ge 
\mc X_{s_\circ, \c_{\circ}, e_\circ}(f) \ge \mc X_{s_\circ, \c_{\circ}, e_\circ}-2\omega_n
\ge \mc X_{s_{m(j)}, \cp_j, e_{m(j)}}(f) - 4\omega_n\notag\\
&\ge \frac{\sqrt{d_j^2 \delta_j}}{\sqrt 6} - 4\omega_n \ge \frac{\sqrt{d_j^2\delta_j}}{2\sqrt 6}
\label{eq:f:lower}
\end{align}
under \eqref{cond:wbs:three}.
Besides, since $\mc X_{s_\circ, \c_\circ, e_\circ} \ge \mc X_{s_\circ, \cp_j, e_\circ}$, it holds
\begin{align}
\label{eq_ck_wbs_a}
\mc X_{s_\circ, \cp_j, e_\circ}(f) - \mc X_{s_\circ, \c_\circ, e_\circ}(f)
\le 
\mc X_{s_\circ, \c_\circ, e_\circ}(\vep) - \mc X_{s_\circ, \cp_j, e_\circ}(\vep)
\end{align}
and further, the positivity of the LHS of \eqref{eq_ck_wbs_a} implies
\begin{align}
1 \le \frac{|\mc X_{s_\circ, \cp_j, e_\circ}(\vep) - \mc X_{s_\circ, \c_\circ, e_\circ}(\vep)|}
{\mc X_{s_\circ, \cp_j, e_\circ}(f) - \mc X_{s_\circ, \c_\circ, e_\circ}(f)}.
\label{eq_ck_wbs}
\end{align}

Using the notations adopted in the proof of Proposition~\ref{prop:step:three},
we denote $\mc X_{s, b, e}(\vep) = \sqrt{\mc W}_b \mc E_b$ 
(suppressing the dependence on $s$ and $e$).
Then, 
\begin{align*}
|\mc X_{s_\circ, \c_\circ, e_\circ}(\vep) - \mc X_{s_\circ, \cp_j, e_\circ}(\vep)|
\le \l\vert \sqrt{\mc W_{\c_\circ}} - \sqrt{\mc W_{\cp_j}} \r\vert \; \l\vert \mc E_{\cp_j} \r\vert
+ \sqrt{\mc W_{\c_\circ}} \vert \mc E_{\cp_j} - \mc E_{\c_\circ} \vert.
\end{align*}
By the mean value theorem,
\begin{align*}
\l\vert \sqrt{\mc W_{\c_\circ}} - \sqrt{\mc W_{\cp_j}} \r\vert
\le \frac{\sqrt{2} (\c_\circ - \cp_j)}{\min(\cp_j - s_\circ, e_\circ - \cp_j)^{3/2}}.
\end{align*}
Also, on $\mc M_n^{(11)} $,
\begin{align*}
& \vert \mc E_{\cp_j} \vert = \l\vert 
\frac{e_\circ - \cp_j}{e_\circ - s_\circ} \sum_{t = s_\circ + 1}^{\cp_j} \vep_t -
\frac{\cp_j - s_\circ}{e_\circ - s_\circ} \sum_{t = \cp_j + 1}^{e_\circ} \vep_t \r\vert
\le \sqrt{2\min(\cp_j - s_\circ, e_\circ - \cp_j)} \; \omega_n,
\\
& \vert \mc E_{\cp_j} - \mc E_{\c_\circ} \vert = 
\l\vert \sum_{t = \cp_j + 1}^{\c_\circ} \vep_t - 
\frac{\c_\circ - \cp_j}{e _\circ- s_\circ} \sum_{t = s_\circ + 1}^{e_\circ} \vep_t
\r\vert \le
\sqrt{\c_\circ - \cp_j} \; \omega_n + 
\frac{\c_\circ - \cp_j}{\sqrt{e _\circ- s_\circ}} \; \omega_n \le 2\sqrt{\c_\circ - \cp_j} \; \omega_n.
\end{align*}

Combining the above, we arrive at
\begin{align}
|\mc X_{s_\circ, \c_\circ, e_\circ}(\vep) - \mc X_{s_\circ, \cp_j, e_\circ}(\vep)| 
\le \frac{2\,(\c_{\circ}-\cp_j)}{\min(\cp_j - s_{\circ}, e_{\circ} - \cp_j)} \; \omega_n
+ \sqrt{\frac{8(\c_{\circ}-\cp_j)}{\min(\c_\circ - s_{\circ}, e_{\circ} - \c_\circ)}} \; \omega_n. \label{eq_noise_ck}
\end{align}

%

The proof proceeds by considering the following possible scenarios.

{\bf Case 1:} There is at least one change point to the right of $\cp_j$ in $(s_\circ, e_\circ)$,
i.e., $\cp_{j + 1} < e_{\circ}$, and $\mc X_{s_\circ, \cp_j, e_\circ}(f) \ge 
\mc X_{s_\circ, \cp_{j + 1}, e_\circ}(f)$.
Adopting the arguments in the proof of Theorem~2 in \cite{wang2018} under their Case~2~(b),
we can show that $\cp_j - s_\circ \ge c_1\delta_j$ for some universal constant $c_1 \in (0, 1]$;
otherwise, we cannot have $\mc X_{s_\circ, \cp_j, e_\circ}(f) \ge 
\mc X_{s_\circ, \cp_{j + 1}, e_\circ}(f)$.
This ensures that $j \in \{u + 1, \ldots, v - 1\}$.
Then, 
\begin{align*}
\frac{c_1 \delta_j \mc X_{s_\circ, \cp_j, e_\circ}(f) (\c_\circ - \cp_j)}{2n^2}
& \le \mc X_{s_\circ, \cp_j, e_\circ}(f) - \mc X_{s_\circ, \c_\circ, e_\circ}(f) 
\\
& \le \vert \mc X_{s_\circ, \cp_j, e_\circ}(\vep)- \mc X_{s_\circ, \c_\circ, e_\circ}(\vep) \vert
\le 4\omega_n
\end{align*}
where the first inequality follows from Lemma~9 of \cite{wang2018b}
(with $\mc X_{s_\circ, s_\circ + t, e_\circ}(f)$, $e_\circ - s_\circ$, $\cp_j - s_\circ$, $\cp_{j + 1} - s_\circ$ 
and $c_1\delta_j/n$ taking the roles of $g(t)$, $n$, $z$, $z^\prime$ and $\tau$ therein, respectively),
the second from \eqref{eq_ck_wbs_a} and the last from \eqref{eq:pf:wbs:two}.
Together with \eqref{eq:f:lower}
and that $\min(\cp_j - s_\circ, e_\circ - \cp_j) \ge c_1\delta_j$,
we obtain
\begin{align}
\label{eq:rough:bound}
\c_\circ - \cp_j
\le 24 \sqrt{6} (c_1 |d_j| \delta_j^{3/2})^{-1} n^2 \omega_n
< \frac{c_1}{2} \delta_j
\end{align}
under \eqref{cond:wbs:one} for $n$ large enough
which, together with \eqref{eq_noise_ck}, leads to
\begin{align}
\label{eq_noise_ck2}
|\mc X_{s_\circ, \c_\circ, e_\circ}(\vep) - \mc X_{s_\circ, \cp_j, e_\circ}(\vep)| 
\le 8\sqrt{\frac{\c_{\circ}-\cp_j}{\min(\cp_j - s_{\circ}, e_{\circ} - \cp_j)}} \; \omega_n. 
\end{align}
Then, combining this with \eqref{eq:f:lower}--\eqref{eq_ck_wbs}, we yield
\begin{align*}
1 &\le \frac{|\mc X_{s_\circ, \cp_j, e_\circ}(\vep) - \mc X_{s_\circ, \c_\circ, e_\circ}(\vep)|}
{\mc X_{s_\circ, \cp_j, e_\circ}(f) - \mc X_{s_\circ, \c_\circ, e_\circ}(f)}
\le \frac{8\omega_n\sqrt{(\c_\circ - \cp_j)/\min(\cp_j - s_\circ, e_\circ - \cp_j)} }
{c_1 \delta_j \mc X_{s_\circ, \cp_j, e_\circ}(f)(\c_\circ - \cp_j)/(2n^2)}
\le \frac{32\sqrt{6} n^2 \omega_n}
{\sqrt{c_1^3d_j^2 \delta_j^4 (\c_\circ - \cp_j)}}
\end{align*}
such that under \eqref{cond:wbs:one}, we can find some fixed $c_W$ 
for $\rhow = c_W n^{4 - 4\beta} (\omega_n)^2$ satisfying
$d_j^2(\c_\circ - \cp_j) \le 6144 c_1^{-3} \delta_j^{-4} n^4 \omega_n^2 \le \rhow$.
\medskip

{\bf Case 2:} $\cp_{j + 1} < e_\circ$ and 
$\mc X_{s_\circ, \cp_j, e_\circ}(f) < \mc X_{s_\circ, \cp_{j + 1}, e_\circ}(f)$.
In this case, from Lemma~8~(d) of \cite{wang2018b}, 
$\mc X_{s_\circ, b, e_\circ}(f)$ strictly decreases and then increases for $\cp_j \le b \le \cp_{j + 1}$
without changing sign,
and thus we can find $\tau := \max\{\cp_j  + 1 \le b \le \cp_{j + 1}: \,
\mc X_{s_\circ, b, e_\circ}(f) \le \mc X_{s_\circ, \cp_{j + 1}, e_\circ}(f) - 4\omega_n\}$.
Adopting the arguments  in the proof of Theorem~2 in \cite{wang2018}
under their Case~2~(c),  
we have $e_\circ - \cp_{j + 1} \ge c_1\delta_{j + 1}$,
which in turn leads to $\cp_{j + 1} - \tau + 1 < c_1\delta_{j + 1}/2$.
Since by construction and the first line of \eqref{eq:f:lower} we get
\begin{align*}
\mc X_{s_\circ, \cp_j, e_\circ}(f) 
\ge \mc X_{s_\circ, \cp_{j + 1}, e_\circ}(f) - 4\omega_n
\ge \mc X_{s_\circ, \tau, e_\circ}(f),
\end{align*}
we can then adopt the same argument as in Case~1 and prove the claim,
by applying Lemma~9 of \cite{wang2018b}
with $\mc X_{s_\circ, s_\circ + t, e_\circ}(f)$, $e_\circ - s_\circ$, $\cp_j - s_\circ$, $\tau - s_\circ$ 
and $c_2\delta_j/n$ for some $c_2 \in (0, c_1/2]$ taking the roles of $g(t)$, $n$, $z$, $z^\prime$ and $\tau$
in the lemma, respectively.
\medskip

{\bf Case 3:} There is no change point to the right of $\cp_j$ in $(s_\circ, e_\circ)$,
i.e., $\cp_{j + 1} \ge e_\circ$.
We first establish that $\min(\cp_j - s_\circ, e_\circ - \cp_j) \ge \min(c_2, c_3)\delta_j$ for
$c_2$ introduced under Case~2 and some $c_3 \in (0, 1/24]$,
which ensures that $j \in \{u + 1, \ldots, v - 1\}$.
To this end, consider the following two cases: (a) $\cp_{j - 1} \le s_\circ$ and (b) $\cp_{j - 1} > s_\circ$.
Under (a), if $\min(\cp_j - s_\circ, e_\circ - \cp_j) < c_3\delta_j$,
by construction and from \eqref{eq:pf:wbs:two}--\eqref{cond:wbs:three} we yield
\begin{align*}
\mc X_{s_\circ, \c_\circ, e_\circ} \le \sqrt{c_3d_j^2\delta_j} + 2\omega_n <
|\mc X_{s_{m(j)}, \cp_j, e_{m(j)}}(f)| - 2\omega_n \le |\mc X_{s_{m(j)}, \c_j^*, e_{m(j)}}|
\end{align*}
for large enough $n$,
which contradicts \eqref{eq:pf:wbs:one}.
Under (b), we have either $\mc X_{s_\circ, \cp_j, e_\circ}(f) \ge \mc X_{s_\circ, \cp_{j - 1}, e_\circ}(f)$ or not.
In either situations, applying the arguments borrowed from \cite{wang2018b} 
under Cases~1--2 in the reverse direction,
we can establish that $e_\circ - \cp_j \ge c_2\delta_j$.

Next, define $\vartheta = \frac{1}{\cp_j - s_\circ} \sum_{t = s_\circ + 1}^{\cp_j} f_t - f_{\cp_j + 1}$.
Then,
\begin{align*}
\mc X_{s_\circ, \cp_j, e_\circ}(f) \le \vartheta \sqrt{\min(\cp_j - s_\circ, e_\circ - \cp_j)}.
\end{align*}
Applying Lemma~7 of \cite{wang2018b}
with $e_\circ - s_\circ$ and $\cp_j - s_\circ$ taking the roles of $n$ and $z$ in the lemma, respectively,
we obtain
\begin{align}
&\mc X_{s_\circ, \cp_j, e_\circ}(f) - \mc X_{s_\circ, \c_\circ, e_\circ}(f) 
\ge \frac{2\vartheta(\c_\circ - \cp_j)}{3\sqrt{6 \min(\cp_j - s_\circ, e_\circ - \cp_j)}}
\ge \frac{2\,\mc X_{s_\circ, \cp_j, e_\circ}(f)(\c_\circ - \cp_j)}{3\sqrt{6} \min(\cp_j - s_\circ, e_\circ - \cp_j)}.
\label{eq:lem:19}
\end{align}

Combining \eqref{eq:pf:wbs:two}, \eqref{cond:wbs:three}, \eqref{eq:f:lower}, \eqref{eq_ck_wbs_a} and \eqref{eq:lem:19},
\begin{align*}
&  \c_\circ - \cp_j \le 72 (d_j^2\delta_j)^{-1/2} \omega_n 
\min(\cp_j - s_\circ, e_\circ - \cp_j)
\le \frac{1}{2} \min(\cp_j - s_\circ, e_\circ - \cp_j)
\end{align*}
for large enough $n$.
Then, \eqref{eq:f:lower} and \eqref{eq_ck_wbs} with \eqref{eq_noise_ck2} yields
\begin{align*}
1 &\le \frac{|\mc X_{s_\circ, \cp_j, e_\circ}(\vep) - \mc X_{s_\circ, \c_\circ, e_\circ}(\vep)|}{
\mc X_{s_\circ, \cp_j, e_\circ}(f) - \mc X_{s_\circ, \c_\circ, e_\circ}(f)}
\le 
\frac{4\omega_n\sqrt{(\c_\circ - \cp_j)/\min(\cp_j - s_\circ, e_{\circ} - \cp_j)}}
{\mc X_{s_\circ, \cp_j, e_\circ}(f) (\c_\circ - \cp_j)/\{3\sqrt{6} \min(\cp_j - s_\circ, e_\circ - \cp_j)\}}
\\
&\le
\frac{144\omega_n\sqrt{\min(\cp_j - s_\circ, e_\circ - \cp_j)}}
{\sqrt{d_j^2\delta_j(\c_\circ - \cp_j)}}
\le
\frac{144\;\omega_n}
{\sqrt{d_j^2(\c_\circ - \cp_j)}}
\end{align*}
by noting that $\delta_j \ge e_\circ - \cp_j \ge \min(\cp_j - s_\circ, e_\circ - \cp_j)$ under Case~3.
Therefore, there exists some large $c_W > 0$ 
such that
$d_j^2(\c_\circ - \cp_j) \le 144^2 \, \omega_n^2 \le \rhow$.
\medskip

In all Cases~1--3, we have established that 
$\min(\cp_j - s_\circ, e_\circ - \cp_j) \ge \min(c_2, c_3)\delta_j$
(recalling that $c_2 < c_1/2$).
From that $d_j^2|\c_\circ - \cp_j| \le \rhow$ and \eqref{cond:wbs:one},
we yield 
\begin{align*}
\frac{|\c_\circ - \cp_j|}{\delta_j} \le \frac{c_Wd_j^{-2}n^{4 - 4\beta}\omega_n^2}{c_\beta n^\beta} \to 0
\end{align*}
as $n \to \infty$.
Hence, $\min(\c_\circ - s_\circ, e_\circ - \c_\circ) \ge \{\min(c_2, c_3) + o(1)\}\delta_j$
and we conclude that (c) holds with some $c \in (0, \min(c_2, c_3))$.

\section{Proof of the result in Section~\ref{sec:optimal}}

\subsection{Proof of Proposition~\ref{prop:vep}}

We first prove assertion (a):
By Hoeffding's inequality (see Theorem~2.6.3 of \cite{vershynin2018}), we have
\begin{align*}
\p\l(\max_{0 \le s < e \le n} \frac{1}{\sqrt{e - s}} \l\vert \sum_{t = s + 1}^e \vep_t \r\vert \ge 
\omega_n \r)
\le n (n + 1) \exp\l(-c_{\vep}\omega_n^2\r)
\end{align*}
where $c_{\vep}$ is an absolute constant depending on the distribution of $\vep_t$ (via its Orlicz norm).
Consequently, $\omega_n = \sqrt{3\log(n)/c_{\vep}}$ fulfils Assumption~\ref{assum:vep}. 
Next, define $S^{\pm}_{\ell}=\pm\,\sum_{t=1}^{\ell}\vep_t$ 
and note that $\{\exp(S^{\pm}_{\ell})\}$ is a non-negative sub-martingale.
Then following the proof of Lemma~5 of \cite{wang2018b},
by Doob's martingale inequality and Proposition~2.5.2~(v) of \cite{vershynin2018},
we get for some constant $c^\prime_{\vep} > 0$ 
only depending on the distribution of $\vep_t$ (via its Orlicz norm)
and any $0 < l_n < u_n < \infty$,
\begin{align*}
& \p\left(\max_{l_n \le \ell \le u_n} \frac{\sqrt{l_n}}{\ell} S^{\pm}_{\ell} \ge \omegao \right) \le  
\sum_{i = \lfloor \log_2 l_n\rfloor}^{\lceil \log_2 u_n \rceil} 
\p\left( \max_{2^{i - 1} \le \ell < 2^i} S^{\pm}_{\ell} \ge \frac{2^{i - 1} \omegao}{\sqrt{l_n}} \right)
\\
& =\sum_{i = \lfloor \log_2l_n \rfloor}^{\lceil \log_2u_n \rceil} \inf_{\lambda > 0} 
\p\left( \max_{2^{i - 1} \le \ell < 2^i}  e^{\lambda S^{\pm}_{\ell}} \ge 
\exp\l(\frac{\lambda 2^{i - 1}\omegao}{\sqrt{l_n}}\r)\right)
\\
& \le \sum_{i = \lfloor \log_2l_n \rfloor}^{\lceil \log_2u_n \rceil} \inf_{\lambda > 0}
\exp\l(2^i c^{\prime}_{\vep} \lambda^2 - \frac{\lambda 2^{i - 1}\omegao}{\sqrt{l_n}}\r)
= \sum_{i = \lfloor \log_2l_n \rfloor}^{\lceil \log_2u_n \rceil} 
\exp\left( -\frac{2^{i+1}\,(\omegao)^2}{32\, c_{\vep}^{\prime} l_n }\right)
\\
& \le \sum_{j = 0}^{\lceil \log_2u_n \rceil-\lfloor \log_2l_n \rfloor} \exp\left( -\frac{2^j(\omegao)^2}{32\,c_{\vep}^{\prime}} \right)
\le \sum_{j \ge 0}(\max(q_n,\nu_n))^{-2^{j+1}}
\le \frac{2}{(\max(q_n,\nu_n))^2}
\end{align*}
for $\omegao=\sqrt{64 c^{\prime}_{\vep} \log(\max(q_n,\nu_n))}$ provided that $\max(q_n, \nu_n) \ge 2$.
Consequently,
\begin{align*}
& \p\left( \max_{l_n \le \ell \le u_n} \frac{\sqrt{l_n}}{\ell}
\left|\sum_{t = 1}^{\ell} \vep_t \right| \ge \omegao \right) \le \frac{4}{(\max(q_n, \nu_n))^2}.
\end{align*}
By sub-additivity and the i.i.d.\ assumption of the errors (distributional equality)
and since $\nu_n \to \infty$ at an arbitrary rate,
the above choice for $\omegao$ fulfils Assumption~\ref{assum:vep}. 
Furthermore, by L\'{e}vy's reflection principle (see e.g., Theorem~3.1.11 of \cite{gine2016})
and Hoeffding's inequality, it holds for some constant $c_{\vep}^{\prime\prime}$ and any $u_n > 1$,
\begin{align*}
&	\p\left( \max_{1 \le \ell \le  u_n} \frac{1}{\sqrt{u_n}} S_{\ell}^{\pm}
\ge \omegat \right) \le 2\, \p\left( \frac{1}{\sqrt{u_n}} S_{u_n}^{\pm}\ge \omegat \right)
\le 4\,\exp\left( -c_{\vep}^{\prime\prime}(\omegat)^2 \right).
\end{align*}
Therefore, with $\omegat = \sqrt{2\,\log(\max(q_n,\nu_n))/c_{\vep}^{\prime\prime}}$, we have
\begin{align*}
\p\left( \max_{1 \le \ell \le  u_n} \frac{1}{\sqrt{u_n}}
\left|\sum_{j=1}^{\ell}\vep_t\right| \ge \omegat \right) \le \frac{8}{(\max(q_n,\nu_n))^2},
\end{align*}
such that the proof can be completed as for $\omegao$, concluding the proof of (a).

The assertion in (c.i) follows directly from the invariance principle 
and the fact that $\lambda_n = o(\sqrt{n})$ in addition to (a) 
($\omega_n \asymp \sqrt{\log(n)}$ derived for the increments of the Wiener process). 
To prove (c.ii), let 
\begin{align*}
M_n(j) = \max_{d_j^{-2} a_n \le \ell \le \cp_j - \cp_{j - 1}} 
\frac{\sqrt{d_j^{-2}a_n} }{\ell}\left\vert \sum_{t = \cp_j - \ell + 1}^{\cp_j} \vep_t\right \vert.
\end{align*}
Then, by Theorem~B.3 in \citet{kirch2006resampling}, it holds uniformly in $j$ that
\begin{align*}
& \E|M_n(j)|^{\gamma} = O(1) \, \left( \frac{1}{(d_j^{-2}a_n)^{\gamma/2}}
\sum_{i = 1}^{d_j^{-2}a_n} i^{\gamma/2 - 1} + 
(d_j^{-2}a_n)^{\gamma/2} \, \sum_{i > d_j^{-2}a_n}\frac{1}{i^{\gamma/2 + 1}} \right) = O(1),
\end{align*}
where $O(1)$ does not depend on $j$.
From this and Markov's inequality, we yield
\begin{align*}
\p\left(\max_{1 \le j \le q_n} M_n(j) \ge \omegao \right) &\le 
\frac{q_n \, \max_{1 \le j \le q_n} \E\left(\vert M_n(j) \vert^{\gamma} \right)}{(\omegao)^{\gamma}}
= O(1)\,\left(\frac{q_n^{1/\gamma}}{\omegao}\right)^{\gamma}
\end{align*}
such that the claim follows with $\omegao \asymp q_n^{1/\gamma}\,\nu_n$. 
The assertion for $\omegat$ follows analogously.

The assertion in (b) for $\omega_n$ follows directly from Theorem~1.1 of \citet{mikosch2010}.
The assertion for $\omegao$ and $\omegat$ follows analogously as in the proof of (c.ii):
The moments $\E(\vep_t^{\beta^\prime})$ and $\E(\vep_t^\beta)$ 
for all $\beta <\beta^{\prime} < \alpha$ exist
and independent and centred sequences fulfil the moment condition in (c.ii),
see e.g.,\ Theorem~3.7.8 of \cite{stout1974}.
By (c.ii),  it gives $\omegao \asymp q_n^{1/\beta^{\prime}}\,\nu_n$ 
so that the given choice $\omegao \asymp\max(q_n^{1/\beta}, \nu_n)$ is also valid.

\section{Proofs of the results in Section~\ref{sec:loc}}
\label{sec:pf:main:props}

Recall that for the change point currently under consideration, $\cp_{\circ} \in \mc D := \C \cap (s, e)$,
we write its neighbouring change points as $\cp_{\pm}$ 
(i.e., $\Cp \cap [\cp_{-},\cp_{\circ}) = \{\cp_{-}\}$ 
as well as $\Cp \cap (\cp_{\circ}, \cp_{+}] = \{\cp_{+}\}$),
allowing $\cp_- = 0$ and $\cp_+ = n$,
and denote the associated jump sizes by $d_\circ$ and $d_{\pm}$, respectively.
For any candidate $\c_\circ \in \mc D$ and a subset $\mc A \subset \mc D$ in consideration, 
let $\c_\pm \in \mc A \cup \{s, e\}$ satisfy $\c_- < \c_\circ < \c_+$ and 
$(\c_-, \c_+) \cap \mc A \setminus \{\c_\circ\} = \emptyset$.
Also, within the corresponding interval $(\c_-, \c_+]$,
we refer to $\c_\circ$ as detecting $\cp_\circ \in \Cp \cap (\c_-, \c_+]$
if $\cp_\circ = \arg\min_{\cp \in \Cp \cap (\c_-, \c_+]} |\c_\circ - \cp|$,
even though there may be some $\cp_j \notin (\c_-, \c_+]$ closer to $\c_\circ$ than $\cp_\circ$.
In addition, we denote by $\c_{\circ}^*$ ($\c_{\circ}^{\prime}$)
a strictly valid (acceptable) estimator for $\cp_{\circ}$.
We write $\sic(\mc A) = \sic(\mc A|\mc C, \wh{\Cp}, s, e)$
where there is no confusion,
since the difference between $\sic(\mc A|\mc C, \wh{\Cp}, s, e)$ 
and $\sic(\mc A^\prime|\mc C, \wh{\Cp}, s, e)$ does not depend on $s$, $e$, $\mc C \setminus (s, e)$ and $\wh{\Cp}$
for any $\mc A, \mc A^\prime \subset \mc D = \C \cap (s, e)$.

\subsection{Auxiliary lemmas}
\label{sec:prem:res}

In this section, we list some auxiliary lemmas that will be used in  
the proof of Theorem~\ref{thm:sbic}.
Unless stated otherwise, we assume that
the conditions made in Theorem~\ref{thm:sbic} are met throughout.

Recall the definition of the CUSUM statistic computed on $X_t$ as
\begin{align*}
\mc X_{\c_-, \c_\circ, \c_+} \equiv \mc X_{\c_-, \c_\circ, \c_+}(X) := 
\sqrt{\frac{\c_+ - \c_-}{(\c_+ - \c_\circ)(\c_\circ - \c_-)}}
\sum_{t = \c_- + 1}^{\c_\circ} (X_t - \bar{X}_{(\c_- + 1):\c_+})
\end{align*}
for any $1 \le \c_- + 1 \le \c_\circ < \c_+ \le n$,
and analogously define $\mc X_{\c_-, \c_\circ, \c_+}(f)$ and $\mc X_{\c_-, \c_\circ, \c_+}(\vep)$
with $f_t$ and $\vep_t$ in place of $X_t$, respectively.
Also, we use the notation $\bar{f}_{u:v} = (v - u + 1)^{-1}\sum_{t = u}^v f_t$ for any $1 \le u \le v \le n$
and $\bar{\vep}_{u:v}$ is defined analogously.

\begin{lem}
\label{lem_ck_1}
For $\max(\c_-, \cp_-) < \c < \cp_{\circ} < \min(\c_+,\cp_+)$,
it holds with $r_+ := \max(0, \c_+-\cp_{+})$ and $r_- := \max(0,\cp_--\c_-)$,
\begin{align*}
\mathcal{F}_{\c}= \sum_{t=\c_-+1}^{\c} (f_t - \bar{f}_{(\c_-+1):\c_+})
= -\frac{(\c-\c_-)\,(\c_+-\cp_{\circ})}{\c_+-\c_-}\, d_{\circ} -
\frac{\c-\c_-}{\c_+-\c_-}\,d_{+}r_+-\frac{\c_+-k}{\c_+-\c_-}\,d_-r_-
\end{align*}
as well as
\begin{align*}
\mc X_{\c_-,\c,\c_+}(f)
= -\sqrt{\frac{(\c-\c_-)(\c_+-\c)}{\c_+-\c_-}} 
\left(\frac{(\c_{+}-\cp_{\circ})\,d_{\circ}}{\c_+-\c} + 
\frac{r_+\,d_+}{\c_+-\c}+\frac{r_-\,d_-}{\c-\c_-}\right).
\end{align*}

Similarly, for $\max(\c_-, \cp_-) < \cp_{\circ} \le \c < \min(c_+, \cp_+)$, it holds
\begin{align*}
\mathcal{F}_{\c}
= -\frac{(\c_+-\c)\,(\cp_{\circ}-\c_-)}{\c_+-\c_-}\, d_{\circ} -
\frac{\c-\c_-}{\c_+-\c_-}\,d_{+}r_+-\frac{\c_+-k}{\c_+-\c_-}\,d_-r_-
\end{align*}
as well as
\begin{align*}
\mc X_{\c_-,\c,\c_+}(f) =
- \sqrt{\frac{(\c-\c_-)(\c_+-\c)}{\c_+-\c_-}} 
\left(\frac{(\cp_{\circ}-\c_-)\,d_{\circ}}{\c-\c_-} + 
\frac{r_+\,d_+}{\c_+-\c} + \frac{r_-\,d_-}{\c-\c_-}\right).
\end{align*}
\end{lem}
\begin{proof}
The results follow from straightforward calculations.
\end{proof}

\begin{lem}
\label{lem:cusum}
For an arbitrary set $\mc A \subset \C$ and $\c_\circ \in \mc A$,
let $\c_\pm \in \mc A \cup \{0, n\}$ satisfy 
$\c_- < \c_\circ < \c_+$ with $\mc A \cap (\c_-, \c_+) = \emptyset$.
Then,
\begin{align*}
\rss(\mc A \setminus \{\c_\circ\}) - \rss(\mc A) = |\mc X_{\c_-, \c_\circ, \c_+}|^2.
\end{align*}
\end{lem}

\begin{proof}
\begin{align*}
& \rss(\mc A \setminus \{\c_\circ\}) - \rss(\mc A)
\\
=& \sum_{t = \c_- + 1}^{\c_+} (X_t - \bar{X}_{(\c_-+1):\c_+})^2 - 
\l\{\sum_{t = \c_-+1}^{\c_\circ} (X_t - \bar{X}_{(\c_-+1):\c_\circ})^2 + \sum_{t = \c_\circ+1}^{\c_+} (X_t - \bar{X}_{(\c_\circ+1):\c_+})^2\r\}
\\
=& (\c_\circ - \c_-)\bar{X}_{(\c_-+1):\c_\circ}^2+(\c_+ - \c_\circ)\bar{X}_{(\c_\circ+1):\c_+}^2 
\\& \qquad 
-\frac{1}{\c_+ - \c_-}\l\{(\c_\circ - \c_-)\bar{X}_{(\c_-+1):\c_\circ}
+ (\c_+ - \c_\circ)\bar{X}_{(\c_\circ+1):\c_+}\r\}^2
\\
=& \l\{\sqrt{\frac{(\c_\circ - \c_-)(\c_+ - \c_\circ)}{\c_+ - \c_-}} 
\l(\bar{X}_{(\c_-+1):\c_{\circ}}-\bar{X}_{(\c_{\circ}+1):\c_+}
\r)\r\}^2
= |\mc X_{\c_-, \c_\circ, \c_+}|^2.
\end{align*}
\end{proof}

\begin{lem}
\label{lem:rss:n}
Under Assumptions~\ref{assum:vep}, \ref{assum:size} and~\ref{assum:cand}~(b),
there exist fixed $C^{\prime}, C^{\prime\prime} > 0$ for which
we have $C^{\prime}n \le \rss(\mc A) \le C^{\prime\prime} n$ for any $\mc A \subset \C$ on $\mc M_n^{(11)}$.
\end{lem}
\begin{proof}

Firstly, by ergodicity and that $0 < \Var(\vep_t) < \infty$, 
there exist $c_l, c_u \in (0, \infty)$ such that
\begin{align*}
0 < c_l \le \frac{1}{n} \sum_{t = 1}^n \vep_t^2 \le c_u < \infty \quad \text{a.s.}
\end{align*}

From $\sum_{t = s}^e (X_t - \bar{X}_{s:e})^2 = \min_{a \in \R} \sum_{t = s}^e (X_t - a)^2$, 
it holds that $\rss(\mc A) \ge \rss(\mc A^{\prime})$ for any $\mc A \subset \mc A^{\prime}$.
Thus we can find $C^{\prime\prime} \in (0, \infty)$ such that
for any $\mc A \subset \C$,
\begin{align*}
\rss(\mc A)& \le \rss(\emptyset) = \sum_{t = 1}^n (X_t - \bar{X}_{1:n})^2 
\le 2\,\sum_{t = 1}^n (\vep_t - \bar{\vep}_{1:n})^2 + 2\,\sum_{t=1}^n (f_t-\bar{f}_{1:n})^2\\
&
\le 2\,\sum_{t=1}^n\vep_t^2+2n\,\bar{f}^2
\le n\l(2c_u
+2 \bar{f}^2 \r) \le C^{\prime\prime} n,
\end{align*}
where $\bar f = \max_{1 \le j \le q_n} |f_j - \bar{f}_{1:n}|$ is bounded by that $\max_{1\le j \le q_n}|d_j| = O(1)$.

Next, let $\tilde{\C} := \C \cup \Cp
= \{\tilde{\c}_1 < \ldots < \tilde{\c}_{A_n}\}$
with $\tilde{\c}_0 = 0$ and $\tilde{\c}_{A_n + 1} = n$,
where $A_n \le Q_n + q_n$.
Then,  we can find $C^{\prime} \in (0, \infty)$
such that for any $\mc A \subset \C$ and $n$ large enough,
\begin{align*}
\rss(\mc A) 
& \ge \rss(\tilde{\C})
\ge \sum_{j = 0}^{A_n}\sum_{t = \tilde{\c}_j + 1}^{\tilde{\c}_{j + 1}} 
(\vep_t - \bar{\vep}_{(\tilde{\c}_j + 1):\tilde{\c}_{j + 1}})^2
= \sum_{t = 1}^n \vep_t^2 - 
\sum_{j = 0}^{A_n} 
\l(\frac{\sum_{t = \tilde{\c}_j + 1}^{\tilde{\c}_{j + 1}} \vep_t}
{\sqrt{\tilde{\c}_{j + 1} - \tilde{\c}_j}}\r)^2
\\
& \ge n\l(c_l - \frac{(Q_n + q_n) \omega_n^2}{n} \r) \ge C^{\prime} n,
\end{align*}
where the last inequality follows from
that $(\min_{1 \le j \le q_n} \delta_j)^{-1} \omega_n^2 \to 0$ 
under Assumption~\ref{assum:size}
and thus $n^{-1} \omega_n^2 q_n \to 0$, and from Assumption~\ref{assum:cand}~(b).
\end{proof}

\begin{lem}
\label{lem:sbic}
Let the conditions in Lemma~\ref{lem:rss:n} hold.
Then, there exist fixed $\underline{C}, \bar{C} > 0$ such that we have
\begin{align}
\underline{C}\,|\mc X_{\c_-, \c_\circ, \c_+}|^2 - \xi_n
\le  \sic(\mc A \setminus \{\c_\circ\}) - \sic(\mc A) 
\le \bar{C}\,|\mc X_{\c_-, \c_\circ, \c_+}|^2 - \xi_n
\label{eq:lem:sbic}
\end{align}
for any $\c_\circ \in \mc A \subset \C$.
\end{lem}
\begin{proof}
From Lemmas \ref{lem:cusum}--\ref{lem:rss:n} and 
that $\log(1 + x) \le x$ for all $x \ge 0$, we obtain
\begin{align*}
\sic(\mc A \setminus \{\c_\circ\}) - \sic(\mc A) 
&= \frac{n}{2}\log\l\{\frac{\rss(\mc A \setminus \{\c_\circ\})}{\rss(\mc A)}\r\} - \xi_n
= \frac{n}{2}\log\l\{1 + \frac{|\mc X_{\c_-, \c_\circ, \c_+}|^2}{\rss(\mc A)}\r\} - \xi_n
\\
&\le \frac{|\mc X_{\c_-, \c_\circ, \c_+}|^2}{2C^\prime} - \xi_n
\end{align*}
hence the RHS of \eqref{eq:lem:sbic} holds with $\bar{C} = 1/(2C^\prime)$.

Furthermore, by Lemmas \ref{lem:cusum}--\ref{lem:rss:n} it holds 
\begin{align*}
1\le \frac{\rss(\mc A \setminus \{\c_\circ\})}{\rss(\mc A)} \le \frac{C^{\prime\prime}}{C^\prime}.
\end{align*}
Let $g(x) = \log(x)/(x - 1)$.
Since $\lim_{x \downarrow 1} g(x) \to 1$ and from its continuity,
there exists a constant $C^{\prime\prime\prime}>0$ such that $\inf_{1\le x\le C^{\prime\prime}/C^\prime} g(x) \ge C^{\prime\prime\prime}$. 
Hence by Lemma~\ref{lem:rss:n}
\begin{align*}
\sic(\mc A \setminus \{\c_\circ\}) - \sic(\mc A) 
&= \frac{n}{2}\log\l\{\frac{\rss(\mc A \setminus \{\c_\circ\})}{\rss(\mc A)}\r\} - \xi_n
\ge \frac{C^{\prime\prime\prime}}{2C^{\prime\prime}}\,|\mc X_{\c_-, \c_\circ, \c_+}|^2
- \xi_n,
\end{align*}
so that $\underline{C}=C^{\prime\prime\prime}/(2C^{\prime\prime})$ meets \eqref{eq:lem:sbic}.
%
\end{proof}

\subsection{Proofs of the Propositions~\ref{prop:detectability}--\ref{prop:step:three}}

Within the proofs of the propositions, the $o$-notation always refers to $M$ 
in \eqref{eq:rates} being large enough, 
which in turn follows for large enough $n$,
and precise bounds can be given in each instance. 

\subsubsection{Proof of Proposition~\ref{prop:detectability}}
\label{sec:proof:prop:detect}

Choose $C^*>\max(1,2/\underline{C})$ and $c^*<\min(1,1/\bar{C})$ for 
$\underline{C}$ and $\bar{C}$ defined in Lemma~\ref{lem:sbic}.
The tighter the choice is, the larger $M$ in \eqref{eq:rates} is required to be.

Firstly, in the situation of (a),
we have $d_\circ^2\min(\cp_\circ - \c_-, \c_+ - \cp_\circ) \ge C^*\xi_n$
and $\max(d_+^2 r_+, d_-^2 r_-) \le C^*\xi_n$.
Then, when $\c_\circ^{\prime} \ge \cp_\circ$, it holds from Lemma~\ref{lem_ck_1},
\begin{align*}
&	\left|\mc X_{\c_-, \c_\circ^{\prime}, \c_+}\right| \ge \\
& 
\sqrt{\frac{(\c_\circ^{\prime} - \c_-)(\c_+ - \c_\circ^{\prime})}{\c_+ - \c_-}} 
\left\{\frac{(\cp_\circ - \c_-)|d_\circ|}{\c_\circ^{\prime} - \c_-} 
- \frac{r_-|d_-|}{\c_\circ^{\prime} - \c_-}
- \frac{r_+|d_+|}{\c_+ - \c_\circ^{\prime}}\right\} -
|\mc X_{\c_-, \c_\circ^{\prime}, \c_+}(\varepsilon)|.
\end{align*}
For the first summand, note that
\begin{align*}
&\sqrt{\frac{(\c_\circ^{\prime} - \c_-)(\c_+ - \c_\circ^{\prime})}{\c_+ - \c_-}} 
\frac{(\cp_\circ - \c_-)|d_\circ|}{\c_\circ^{\prime} - \c_-} =|d_{\circ}|\,\sqrt{\frac{(\c_+-\cp_{\circ})(\cp_{\circ}-\c_-)}{\c_+-\c_-}}\, \left( \frac{\c_+-\c_{\circ}^\prime}{\c_+-\cp_{\circ}} \right)^{1/2}\,\left( \frac{\cp_{\circ}-\c_-}{\c_{\circ}^{\prime}-\c_-} \right)^{1/2}\\
&
\ge |d_\circ|\,\sqrt{\frac{1}{2}\min\left( \cp_\circ - \c_-, \c_+ - \cp_\circ\right)}\,\l(1 - \frac{\rho_n\nu_n}{\xi_n}\r)\ge \sqrt{\frac{C^*\xi_n}{2}}\,(1+o(1)).
\end{align*}

For the second summand, 
\begin{align*}
\sqrt{\frac{(\c_\circ^{\prime} - \c_-)(\c_+ - \c_\circ^{\prime})}{\c_+ - \c_-}} 
\frac{r_-\;|d_-|}{\c_\circ^{\prime} - \c_-}
= \sqrt{\frac{\c_+ - \c_\circ^{\prime}}{\c_+ - \c_-}} 
\frac{r_-\;|d_-|^2}{\sqrt{\c_\circ^{\prime} - \c_-}|d_-|}
\le \frac{C^*\xi_n}{\sqrt{D_n}} = o(\sqrt{\xi_n}),
\end{align*}
where the inequality follows from noting that when $r_- > 0$, 
we have $\c_\circ^\prime - \c_- \ge \cp_\circ - \cp_-$ as well as  $d_-^2(\cp_\circ - \cp_-) \ge D_n$ 
under Assumption~\ref{assum:size}. An analogous argument applies to the third summand,
noting that $\c_+ - \c_\circ^\prime \ge (\cp_+  - \cp_\circ) \, (1-\rho_n\nu_n/D_n)$ 
when $\cp_+ < \c_+$ (hence $r_+ > 0$).

Finally, on $\mc M_n^{(11)}$, the fourth summand satisfies
\begin{align}\label{eq_noise_cusum}
|\mc X_{\c_-, \c_\circ^{\prime}, \c_+}(\varepsilon)| =
\sqrt{\frac{(\c_\circ - \c_-)(\c_+ - \c_\circ)}{\c_+ - \c_-}} 
\l\vert \bar{X}_{(\c_-+1):\c_{\circ}} - \bar{X}_{(\c_{\circ}+1):\c_+} \r\vert \le 2 \omega_n = o(\sqrt{\xi_n}).
\end{align}
Putting the above together, for $M$ large enough
(whose exact value depends on the choice of $C^*$), 
\begin{align}
\label{eq_detect_one}
\vert \mc X_{\c_-, \c_j^\prime, \c_+} \vert \ge 
\sqrt{\frac{C^*\xi_n}{2}} \l(1 + o(1) +\sqrt{2\,C^*} \, o(1)\r) +o(\sqrt{\xi_n})\ge \sqrt{\frac{\xi_n}{\underline{C}}},
\end{align}
which holds uniformly for any $s, e, \cp_\circ, \c_\pm$ and $\c_\circ^\prime$ meeting the conditions of the proposition. 
By symmetric arguments (reversing time), the same holds 
when $\c_\circ^\prime < \cp_\circ$. 
The assertion of (a) now follows from Lemma~\ref{lem:sbic}.

Next, we suppose that
$d_\circ^2(\cp_\circ - \c_-) \le c^*\xi_n$ as in the case of (b).
Recalling the decomposition of $\mc X_{\c_-, \c, \c_+}(f)$ from Lemma~\ref{lem_ck_1},
the term that does not depend on $r_+$ (note that $r_- = 0$ in the situation of (b)) 
satisfies for $\c\ge \cp_{\circ}$,
\begin{align*}
& \sqrt{\frac{(\c - \c_-)(\c_+ - \c)}{\c_+ - \c_-}} 
\frac{(\cp_\circ - \c_-)|d_\circ|}{\c - \c_-} 
\le \sqrt{d_\circ^2\min(\cp_\circ - \c_-, \c_+ - \c)}  
\\
&\le \sqrt{d_\circ^2\min(\cp_\circ - \c_-, \c_+ - \cp_\circ)} \le \sqrt{c^*\xi_n}
\end{align*}
and analogously for $\c\le \cp_{\circ}$, that 
\begin{align*}
\sqrt{\frac{(\c - \c_-)(\c_+ - \c)}{\c_+ - \c_-}} 
\frac{(\c_+ - \cp_\circ)|d_\circ|}{\c_+ - \c} \le \sqrt{c^*\xi_n}.
\end{align*}

Also, when $r_+ > 0$,
\begin{align}
\sqrt{\frac{(\c - \c_-)(\c_+ - \c)}{\c_+ - \c_-}} \frac{r_+\;|d_+|}{\c_+ - \c}
\le
\frac{r_+\;|d_+|^2}{\sqrt{|d_+|^2 (\c_+-\c)}}
\le \frac{\sqrt{2}\,C^*\xi_n}{\sqrt{D_n}}, 
\label{eq:r:plus}
\end{align}
by noting that $\c - \c_- \le \c_+ - \c_-$ and 
$\c_+ - \c \ge (\cp_+ - \cp_\circ)/2$ when $r_+ > 0$, 
because $\c$ is closer to $\cp_{\circ}$ than to $\cp_+$.
Together with \eqref{eq_noise_cusum}, this leads to
\begin{align}
\label{eq_detect_two}
\vert \mc X_{\c_-, \c, \c_+} \vert \le 
\sqrt{c^*\,\xi_n} + \frac{\sqrt{2}C^*\xi_n}{\sqrt{D_n}} + 2 \omega_n
= \sqrt{c^*\xi_n} \left(1+\frac{C^*}{\sqrt{c^*}} \, o(1) + o(1)\right) < \sqrt{\frac{\xi_n}{\bar{C}}},
\end{align}
for $M$ sufficiently large (depending on $\bar{C},\underline{C}$), with the inequality holding uniformly for any $s, e, \cp_\circ, \c_\pm$ and $\c$
meeting the conditions.
Hence the conclusion of (b) follows from Lemma~\ref{lem:sbic}.

The proof of (c) follows by symmetry (reversing time).

\subsubsection{Proof of Proposition~\ref{prop:step:two}}
\label{sec:proof:prop:step:two}

Under (a), i.e.,\ when there is no change point contained within this interval, 
by \eqref{eq_noise_cusum} we get 
\begin{align}
\left|\mc X_{\c_-, \c_{\circ}, \c_+}\right| =	\left|\mc X_{\c_-, \c_{\circ}, \c_+}(\varepsilon)\right|
\le 2 \omega_n = o(\sqrt{\xi_n})
\nn 
\end{align}
on $\mc M_n^{(11)}$, so that assertion (a) follows from Lemma~\ref{lem:sbic}.

In the case of (b.i), w.l.o.g., we assume that  
$d_{\circ}^2\, |\c_- - \cp_{\circ}| \le \rho_n\nu_n$, which in particular implies that $r_-=0$ (for $M$ large enough);	otherwise consider the series in reversed time.
Then, by assumption, $d_+^2r_+ \le C^*\xi_n$ as well as 
$\c_+ - \c_\circ \ge (\cp_+ - \cp_\circ)/2$ when $\c_+ > \cp_+$,
since $\c_{\circ}$ is closer to $\cp_{\circ}$ than any other change point within $(\c_-,\c_+)$.
We now distinguish the two cases: (I) $\c_- < \cp_{\circ} \le \c_{\circ}$ 
and (II) $\c_- < \c_{\circ} < \cp_{\circ}$.
If (I) holds, Lemma~\ref{lem_ck_1} leads to
\begin{align*}
\left|\mc X_{\c_-, \c_{\circ}, \c_+} \right| 
&	\le 
\sqrt{\frac{(\c_{\circ} - \c_-)(\c_+ - \c_{\circ})}{\c_+ - \c_-}} 
\left\{ \frac{(\cp_{\circ} - \c_-)|d_{\circ}|}{\c_{\circ} - \c_-}
+ \frac{r_+\,|d_+|}{\c_+ - \c_{\circ}} \right\}
+ |\mc X_{\c_-, \c_{\circ}, \c_+}(\varepsilon)|
\\
& \le \sqrt{\cp_\circ - \c_-}\,|d_{\circ}| + \frac{\sqrt{2}C^*\xi_n}{\sqrt{D_n}}+ 2\omega_n
=\sqrt{\xi_n} \left( o(1)+C^* o(1) \right)
< \sqrt{\frac{\xi_n}{\bar{C}}},
\end{align*}
for $M$ large enough. 
The assertion follows from Lemma~\ref{lem:sbic}. The case of (II) can be dealt with analogously. 

Similarly, in the case of (b.ii), w.l.o.g., suppose 
$d_{\circ}^2\, |\c_- - \cp_{\circ}| \le C^*\xi_n$
such that in particular, $r_- = 0$ (for $M$ large enough). 
Also, as in (b.i), it holds $\c_+ - \c_\circ \ge (\cp_+ - \cp_\circ)/2$ 
when $\c_+ > \cp_+$. 
In this case, necessarily $\c_- < \cp_\circ \le \c_\circ$
and thus by Lemma ~\ref{lem_ck_1},
\begin{align*}
\left|\mc X_{\c_-, \c_{\circ}, \c_+} \right| 
&	\le 
\sqrt{\frac{(\c_{\circ} - \c_-)(\c_+ - \c_{\circ})}{\c_+ - \c_-}} 
\left\{ \frac{(\cp_{\circ} - \c_-)|d_{\circ}|}{\c_{\circ} - \c_-}
+ \frac{r_+\,|d_+|}{\c_+ - \c_{\circ}} \right\}
+ |\mc X_{\c_-, \c_{\circ}, \c_+}(\varepsilon)|
\\
& \le \frac{(\cp_{\circ}-\c_-)|d_{\circ}|}{\sqrt{\c_{\circ}-\cp_{\circ}}} + \frac{\sqrt{2}C^*\xi_n}{\sqrt{D_n}}+ 2\omega_n
=\sqrt{\xi_n} \left(\frac{C^*}{\sqrt{\wt{C}}}+ o(1)+C^* o(1) \right)
< \sqrt{\frac{\xi_n}{\bar{C}}},
\end{align*}
for $M$ large enough, completing the proof.

\subsubsection{Proof of Proposition~\ref{prop:step:three}}
\label{sec:proof:prop:step:three}

We start with some preliminary numerical calculations that will be used throughout the proof.
We use the notations
\begin{align*}
\mc W_k = \frac{\c_+ - \c_-}{(k - \c_-)(\c_+ - k)}, \,
\mc F_k = \sum_{t = \c_- + 1}^k (f_t - \bar{f}_{(\c_- + 1):\c_+})
\text{ and }
\mc E_k = \sum_{t=\c_- + 1}^k (\vep_t - \bar{\vep}_{(\c_- + 1):\c_+});
\end{align*}
we suppress the dependence of the above definitions on $\c_{\pm}$ for brevity.

From Lemma~\ref{lem_ck_1}, we get $\mc F_{\c} = \wt{\mc F}_{\c} - R_{\c}^+ - R_{\c}^-$ with
\begin{align*}
&	\wt{\mc F}_{\c} = -d_{\circ} \,
\begin{cases}
\frac{(\c - \c_-) (\c_+ - \cp_{\circ})}{\c_+ - \c_-}, & \c \le \cp_{\circ},
\\
\frac{(\c_+ - \c)(\cp_{\circ} - \c_-)}{\c_+ - \c_-},  & \c \ge \cp_{\circ},
\end{cases}
\\[2mm]
& R_{\c}^+ = \frac{\c - \c_-}{\c_+ - \c_-}\,d_+r_+, \qquad 
R_{\c}^- = \frac{\c_+ - \c}{\c_+ - \c_-}\,d_-r_-.
\end{align*}

Note that
\begin{align}
	\mc W_k\wt{\mc F}_k=-d_{\circ}\,
	\begin{cases}
		\frac{\c_+-\cp_{\circ}}{\c_+-k}, &k\le \cp_{\circ}, \\
		\frac{\cp_{\circ}-\c_-}{k-\c_-}, &k\ge\cp_{\circ},
	\end{cases} \quad 
	\mc W_k\wt{\mc F}^2_k=d_{\circ}^2\,
	\begin{cases}
		\frac{(\c_+-\cp_{\circ})^2\,(k-\c_-)}{(\c_+-k)(\c_+-\c_-)} &k\le \cp_{\circ}, \\
		\frac{(\cp_{\circ}-\c_-)^2\,(\c_+-k)}{(k-\c_-)\,(\c_+-\c_-)}, & k\ge \cp_{\circ},
	\end{cases} \label{eq_4_term}
\end{align}
which yields
\begin{align}
&	\mc W_k\wt{\mc F}_k-\mc W_{\cp_{\circ}}\wt{\mc F}_{\cp_{\circ}}=d_{\circ}
	\begin{cases}
		\frac{\cp_{\circ}-k}{\c_+-k}, &k\le \cp_{\circ},\\
		\frac{k-\cp_{\circ}}{k-\c_-}, &k\ge\cp_{\circ},
	\end{cases} \label{eq_6_term}
\\
&	\mc W_{\cp_{\circ}}\wt{\mc F}^2_{\cp_{\circ}}- \mc W_k\wt{\mc F}^2_k=d^2_{\circ}
	\begin{cases}
		\frac{(\cp_{\circ}-k)(\c_+-\cp_{\circ})}{\c_+-k}, &k\le \cp_{\circ},\\
		\frac{(k-\cp_{\circ})(\cp_{\circ}-\c_-)}{k-\c_-}, &k\ge\cp_{\circ}.
	\end{cases} \label{eq_1_term}
\end{align}

Concerning the remainder term, we get
\begin{align}\label{eq_5_term}
	\mc W_k R_k^+ = \frac{d_+r_+}{\c_+-k},\qquad \mc W_k R_k^- = \frac{d_-r_-}{k-\c_-}, \qquad
	\mc W_k R_k^+ R_k^- = \frac{d_+r_+ \cdot d_-r_-}{k_+ - k_-},
\end{align}
as well as 
\begin{align}\label{eq_7_term}
	&	\mc W_k R_k^+-\mc W_{\cp_{\circ}} R_{\cp_{\circ}}^+=d_+r_+\frac{\c-\cp_{\circ}}{(\c_+-\c)( \c_+-\cp_{\circ}) },\notag\\
	&\mc W_k R_k^--\mc W_{\cp_{\circ}} R_{\cp_{\circ}}^-=d_-r_-\frac{\cp_{\circ}-\c}{(\c-\c_-)( \cp_{\circ}-\c_-) }.
\end{align}
Furthermore,
\begin{align*}
	&	\mc W_k(R_k^+)^2=d_+^2\,r_+^2\,\frac{k-\c_-}{(\c_+-\c_-)(\c_+-k)},\qquad \mc W_k(R_k^-)^2=d_-^2\,r_-^2\,\frac{\c_+-k}{(\c_+-\c_-)(k-\c_-)}
\end{align*}
and thus
\begin{align}\label{eq_2_term}
	&\mc W_{k} (R_k^+)^2-\mc W_{\cp_{\circ}}(R_{\cp_{\circ}}^+)^2=d_+^2r_+^2\,\frac{\c-\cp_{\circ}}{(\c_+-\c)(\c_+-\cp_{\circ})},\notag\\
	&\mc W_{\cp_{\circ}}(R_{\cp_{\circ}}^-)^2-\mc W_{\c} (R_k^-)^2
	=d_-^2r_-^2\,\frac{\c-\cp_{\circ}}{(\c-\c_{-})(\cp_{\circ}-\c_{-})}.
\end{align}
Finally, for the terms involving both $\mc{\wt F}_k$ and $\mc R_k$, we get
\begin{align}\label{eq_3_term}
	&\mc W_{\cp_{\circ}}\wt{\mc F}_{\cp_{\circ}}R^+_{\cp_{\circ}}-\mc W_{\c_{\circ}}\wt{\mc F}_{\c_{\circ}}R^+_{\c_{\circ}}
	=d_+r_+d_{\circ}\frac{\cp_{\circ}-\c_{\circ}}{\c_+-\c_{\circ}}\,\bbI_{\c_{\circ}\le \cp_{\circ}}\notag\\
	&\mc W_{\cp_{\circ}}\wt{\mc F}_{\cp_{\circ}}R^-_{\cp_{\circ}}-\mc W_{\c_{\circ}}\wt{\mc F}_{\c_{\circ}}R^-_{\c_{\circ}}
	=d_-r_-d_{\circ}\frac{\c_{\circ}-\cp_{\circ}}{\c_{\circ}-\c_-}\,\bbI_{\c_{\circ}\ge \cp_{\circ}}
\end{align}

Concerning the error terms,  
on $\mathcal{M}_n^{(12)} \cap \mathcal{M}_n^{(13)}$, 
it holds uniformly in $\c_{\circ} \in \mc V_\circ$ and $\c^*_{\circ} \in \mc V^*_\circ$,
\begin{align}
& |\mc E_{\c_\circ} - \mc E_{\c_\circ^*}|
\le \left| \sum_{t = \c_\circ + 1}^{\cp_\circ} \vep_t \right|
+ \left| \sum_{t = \min(\c_\circ^*, \cp_\circ) + 1}^{\max(\c_\circ^*, \cp_\circ)} \vep_t \right|
+ \frac{\c_\circ^* - \c_\circ}{\c_+ - \c_-}\left| \sum_{t = \c_- + 1}^{\c_+} \vep_t \right|
\notag
\\
&\qquad \le \frac{(\cp_\circ - \c_\circ)\omegao}{\sqrt{d_\circ^{-2}\rho_n\nu_n}}
+ \sqrt{d_\circ^{-2}\rho_n} \; \omegat 
+ \frac{2(\cp_\circ - \c_\circ)\omega_n}{\sqrt{\c_+-\c_-}}
\label{eq_noise_diff}
\end{align}
as well as
\begin{align*}
	&\left|\mc E_{\c}\right|\le \left|\sum_{t=\c_-+1}^{\c}\vep_t\right|+\left|\frac{\c-\c_-}{\c_+-\c_-}\sum_{t=\c_-+1}^{\c_+}\vep_t\right|\le \sqrt{\c-\c_-}\,\omega_n +\frac{\c-\c_-}{\sqrt{\c_+-\c_-}}\,\omega_n\le 2\, \sqrt{\c-\c_-}\,\omega_n
\end{align*}
and by symmetry of $\mc E_{\c}$ also $\left|\mc E_{\c}\right|\le 2\,\sqrt{\c_+-\c}\,\omega_n$, such that
\begin{align}
	\left|\mc E_{\c}\right|\le 2 \, \sqrt{\min(\c-\c_-,\c_+-\c)}\,\omega_n.
\label{eq_noise_one}
\end{align}

In what follows, we consider the following two cases: 
When $\c_{\circ}$ is closer to one of the boundary points than to $\cp_{\circ}$, i.e.,\  
$\left|\cp_{\circ}-\c_{\circ}\right| \ge \min(\c_{\circ}-\c_-, \c_+-\c_{\circ})$, 
and when this is not so.

\medskip
\textbf{Case 1: $|\cp_{\circ}-\c_{\circ}| \ge \min(\c_{\circ}-\c_-, \c_+-\c_{\circ})$.}
\medskip

We further distinguish the following two cases:
\begin{enumerate}[label = (\alph*)]
\item $\vert \cp_{\circ} - \c_{\circ} \vert \ge \c_{\circ} - \c_-$, 
which can occur only if $\c_{\circ}<\cp_{\circ}$ and $r_- = 0$ 
(otherwise $\c_{\circ}$ is closer to $\cp_{-}$ than $\cp_{\circ}$ 
which contradicts that $\c_\circ$ detects $\cp_{\circ}$). 
In particular, this implies that 
\begin{align}
\label{eq:prop:3:case1a}
\c_{\circ} - \c_- < (\cp_{\circ}-\c_-)/2.
\end{align}

\item $\vert \cp_{\circ}-\c_{\circ} \vert \ge \c_+ - \c_{\circ}$, 
which can occur only if $\c_{\circ} < \c_+$ and $r_+ = 0$.
\end{enumerate}

We detail the proof of (a) below; the assertion under (b) follow by symmetry (reversing time).

First, by \eqref{eq_4_term} and \eqref{eq_noise_one}, it holds for any $\c_-<\c<\c_+$
\begin{align*}
&	\mc W_{\c} \mc E^2_{\c}\le 2\,\mc W_{\cp_{\circ}}\wt{\mc F}^2_{\cp_{\circ}}\, \frac{(\c_+-\c_-)^2\min(\c-\c_-,\c_+-\c)\,\omega_n^2}{d_{\circ}^2\,(\c_+-\cp_{\circ})(\cp_{\circ}-\c_-)(\c_+-\c)(\c-\c_-)}\\
&\le 8 \,\mc W_{\cp_{\circ}}\wt{\mc F}^2_{\cp_{\circ}}\, \frac{\omega_n^2}{d_{\circ}^2\min(\c_+-\cp_{\circ},\cp_{\circ}-\c_-)}\le 8 \,\mc W_{\cp_{\circ}}\wt{\mc F}^2_{\cp_{\circ}}\, \frac{\omega_n^2}{c^*\xi_n}=o\left(\mc W_{\cp_{\circ}}\wt{\mc F}^2_{\cp_{\circ}}  \right).
\end{align*}
Concerning the remainder term (keeping in mind that in this situation, $r_-=0$), we get by \eqref{eq_2_term}
\begin{align}\label{eq_new_ck}
	&\mc W_{_{\c_{\circ}}} (R_{\c_{\circ}}^+)^2
	\le \mc W_{\cp_{\circ}}\wt{\mc F}^2_{\cp_{\circ}} \,\frac{d_+^2r_+^2}{d_{\circ}^2(\c_+-\cp_{\circ})^2}\,\frac{(\c_{\circ}-\c_-)(\c_+-\cp_{\circ})}{\left( \cp_{\circ}-\c_- \right)(\c_+-\c_{\circ})}\le  \mc W_{\cp_{\circ}}\wt{\mc F}^2_{\cp_{\circ}} \,\frac{(C^*)^2\xi_n^2}{D_n^2}\notag\\
	&=o\left(\mc W_{\cp_{\circ}}\wt{\mc F}^2_{\cp_{\circ}}   \right).
\end{align}
Furthermore, by \eqref{eq_4_term} it holds for all $\c_-<\c<\c_+$
\begin{align*}
	\mc W_{\c} \wt{\mc F}^2_{\c}\le \mc W_{\cp_{\circ}}\wt{\mc F}^2_{\cp_{\circ}},
\end{align*}
which is used to deal with the mixed terms to arrive at
\begin{align*}
	|\mc X_{\c_-, \c_\circ, \c_+}|^2=\mc W_{\c_{\circ}} \wt{\mc F}^2_{\c_{\circ}}+o\left(\mc W_{\cp_{\circ}}\wt{\mc F}^2_{\cp_{\circ}}  \right).
\end{align*}
For $\c_{\circ}$ replaced by $\c_{\circ}^*$ in \eqref{eq_new_ck} we get
\begin{align*}
		&\mc W_{_{\c_{\circ}^*}} (R_{\c_{\circ}^*}^+)^2
		\le  \mc W_{\cp_{\circ}}\wt{\mc F}^2_{\cp_{\circ}} \,\frac{(C^*)^2\xi_n^2}{D_n^2}\,\frac{1+\frac{\rho_n}{c^*\xi_n}}{1-\frac{\rho_n}{c^*\xi_n}}
	=o\left(\mc W_{\cp_{\circ}}\wt{\mc F}^2_{\cp_{\circ}}   \right),
\end{align*}
resulting in
\begin{align*}
	|\mc X_{\c_-, \c_\circ^*, \c_+}|^2=\mc W_{\c_{\circ}^*} \wt{\mc F}^2_{\c_{\circ}^*}+o\left(\mc W_{\cp_{\circ}}\wt{\mc F}^2_{\cp_{\circ}}  \right).
\end{align*}

By \eqref{eq_4_term} and \eqref{eq_1_term} we get
\begin{align*}
	\mc W_{\cp_{\circ}}\wt{\mc F}^2_{\cp_{\circ}}- \mc W_{\c_{\circ}}\wt{\mc F}^2_{\c_{\circ}}
	=\mc W_{\cp_{\circ}}\wt{\mc F}^2_{\cp_{\circ}} \, \frac{\cp_{\circ}-\c_{\circ}}{\cp_{\circ}-\c_-}\,\frac{\c_+-\c_-}{\c_+-\c_{\circ}}
	\ge \mc W_{\cp_{\circ}}\wt{\mc F}^2_{\cp_{\circ}} \, \left(1-\frac{\c_{\circ}-\c_-}{\cp_{\circ}-\c_-}\right)\ge \frac{1}{2}\mc W_{\cp_{\circ}}\wt{\mc F}^2_{\cp_{\circ}},
\end{align*}
where the last inequality follows from \eqref{eq:prop:3:case1a}.
Similarly, 
\begin{align}
& \mc W_{\cp_{\circ}}\wt{\mc F}^2_{\cp_{\circ}} - \mc W_{\c^*_{\circ}}\wt{\mc F}^2_{\c^*_{\circ}}
= \mc W_{\cp_{\circ}}\wt{\mc F}^2_{\cp_{\circ}} \, |\cp_{\circ} - \c^*_{\circ}|\, 
\frac{\c_+-\c_-}{(\c_+-\c^*_{\circ})(\cp_{\circ}-\c_-)}\notag
\\
&	\le \mc W_{\cp_{\circ}}\wt{\mc F}^2_{\cp_{\circ}}\, 
\left(\frac{d_{\circ}^2|\cp_{\circ}-\c^*_{\circ}|}{d_{\circ}^2\min(\c_+-\cp_{\circ},\cp_{\circ}-\c_-) } \right)
(2+o(1))\notag\\
& \le \mc W_{\cp_{\circ}}\wt{\mc F}^2_{\cp_{\circ}}\, \frac{\rho_n}{c^*\xi_n} (2+o(1))=o\left(  \mc W_{\cp_{\circ}}\wt{\mc F}^2_{\cp_{\circ}} \right).\label{eq_B5_1}
\end{align}

Putting the above together and by Lemma~\ref{lem:cusum},
\begin{align*}
&\rss(\mc A \cup \{\c_\circ\}) - \rss(\mc A\cup \{\c_\circ^*\})
= |\mc X_{\c_-, \c_\circ^*, \c_+}|^2 - |\mc X_{\c_-, \c_\circ, \c_+}|^2\\
&=\mc W_{\cp_{\circ}}\wt{\mc F}^2_{\cp_{\circ}}- \mc W_{\c_{\circ}}\wt{\mc F}^2_{\c_{\circ}}+o\left(\mc W_{\cp_{\circ}}\wt{\mc F}^2_{\cp_{\circ}}   \right)\ge \mc W_{\cp_{\circ}}\wt{\mc F}^2_{\cp_{\circ}} \left( \frac{1}{2}+o(1) \right)>0,
\end{align*}
which proves the claim.

\medskip
\textbf{Case 2: $\min(\c_\circ - \c_-, \c_+ - \c_\circ) > |\c_\circ - \cp_\circ|$.}
\medskip

In this case, we have
\begin{align}
\label{eq:prop:three:k}
\c_\circ - \c_- > (\cp_\circ - \c_-)/2
\quad \text{and} \quad
\c_+ - \c_\circ > (\c_+ - \cp_\circ)/2.
\end{align}

By Lemma~\ref{lem:cusum}, the following decomposition holds:
\begin{align}
& \rss(\mc A \cup \{\c_{\circ}\}) - \rss(\mc A \cup \{\c_{\circ}^{*}\})  
= |\mc X_{\c_-, \c_{\circ}^*, \c_+}|^2 - |\mc X_{\c_-, \c_{\circ}, \c_+}|^2 \notag \\
&=  \mc W_{\c_{\circ}^*}(\mc F_{\c_{\circ}^*} + \mc E_{\c_{\circ}^*})^2 - 
\mc W_{\c_{\circ}}(\mc F_{\c_{\circ}} + \mc E_{\c_{\circ}})^2 \notag \\
&= 
\left(\mc W_{\cp_{\circ}}\mc F_{\cp_{\circ}}^2 -\mc W_{\c_{\circ}}\mc F_{\c_{\circ}}^2\right)
+ \left(\mc W_{\c_{\circ}^*}\mc F_{\c_{\circ}^*}^2 -\mc W_{\cp_{\circ}}\mc F_{\cp_{\circ}}^2\right)
+ 2\mc W_{\c_{\circ}}\mc F_{\c_{\circ}}(\mc E_{\c_{\circ}^*} - \mc E_{\c_{\circ}})
\notag\\*
&\qquad
+ 2 (\mc W_{\c_{\circ}^*}\mc F_{\c_{\circ}^*} -\mc W_{\cp_{\circ}}\mc F_{\cp_{\circ}} )\mc E_{\c_{\circ}^*}
+ 2 (\mc W_{\cp_{\circ}}\mc F_{\cp_{\circ}} -\mc W_{\c_{\circ}}\mc F_{\c_{\circ}} )\mc E_{\c_{\circ}^*} 
\notag \\*
&\qquad + \mc W_{\c_{\circ}}(\mc E_{\c_{\circ}^*}^2 - \mc E_{\c_{\circ}}^2)
+ (\mc W_{\c_{\circ}^*} - \mc W_{\c_{\circ}})\mc E_{\c_{\circ}^*}^2
\notag\\
&=: A_1(\mathcal{F}) + A_2(\mathcal{F}) + A_3(\mathcal{F}) + A_4(\mathcal{F}) + A_5(\mathcal{F}) + A_6 + A_7.
\label{eq:decomposition}
\end{align}

We now show that for $n$ large enough,
\begin{align*}
& A_1(\wt{\mc F}) > 0, \qquad 
\frac{|A_1(\wt{\mc F})-A_1(\mc F)|}{A_1(\wt{\mc F})} = o(1), \\
& \frac{|A_j(\mc F)|}{A_1(\wt{\mc F})} = o(1)
\quad \text{for } j = 2, \ldots, 5, \quad \text{ and } \quad 
\frac{|A_j|}{A_1(\wt{\mc F})} = o(1) \quad \text{for } j = 6, 7
\end{align*}
on $\mathcal{M}_n$,
uniformly in $\c_{\pm}, \c_{\circ}$ and $\c^*_{\circ}$ meeting the conditions of the proposition.
Consequently, 
$\rss\left(\mathcal{A} \cup \{\c_{\circ}\}\right) > \rss\left(\mathcal{A} \cup \{\c_{\circ}^{*}\}\right)$,
which proves the assertion. 


W.l.o.g., let $\c_{\circ} < \cp_{\circ}$ (otherwise consider the time series in reverse).
In what follows, all the inequalities are uniform in the sense that
they hold provided that the conditions of the proposition are met.

Firstly, from \eqref{eq_1_term}, 
\begin{align}
A_1(\wt{\mc F}) &= \mc W_{\cp_{\circ}}\wt{\mc F}^2_{\cp_{\circ}} - 
\mc W_{\c_{\circ}}\wt{\mc F}^2_{\c_{\circ}} 
= \frac{d_\circ^2(\cp_\circ - \c_\circ)(\c_+ - \cp_\circ)}{\c_+ - \c_\circ}
\ge 
\frac{d_\circ^2}{2} \min(\cp_\circ - \c_\circ, \c_+ - \cp_\circ)
\notag
\\
&\ge \l\{\begin{array}{ll} 
\frac{\rho_n\nu_n}{2} > 0 & \text{when } \cp_\circ - \c_\circ \le \c_+ - \cp_\circ,
\\
\frac{c^*\xi_n}{2} > 0 & \text{when } \cp_\circ - \c_\circ > \c_+ - \cp_\circ.
\end{array}\r.
\label{eq_A1_main}
\end{align}

Next, under the conditions imposed on $\c_\pm$ and from \eqref{eq_2_term},
when $\cp_+ < \c_+$ such that $r_+ \ne 0$, 
\begin{align}
& \left|\mc W_{\c_{\circ}} (R_{\c_{\circ}}^+)^2- \mc W_{\cp_{\circ}} (R_{\cp_{\circ}}^+)^2\right| 
= \frac{d_+^2r_+^2(\cp_\circ - \c_\circ)}{(\c_+ - \c_\circ)(\c_+ - \cp_\circ)}
\le A_1(\wt{ \mc F}) \cdot \frac{d_+^4r_+^2}{d_\circ^2d_+^2(\c_+ - \cp_\circ)^2}\notag \\
& \qquad
\le A_1(\wt{\mc F}) \cdot \l(\frac{C^*\xi_n}{D_n}\r)^2 = o(A_1(\wt{\mc F})).
\label{eq_remainder1_1}
\end{align}
Similarly, when $\c_-<\cp_-$ such that $r_-\neq 0$, by \eqref{eq:prop:three:k}
\begin{align}
& \left|\mc W_{\c_{\circ}} (R_{\c_{\circ}}^{-})^2- \mc W_{\cp_{\circ}} (R_{\cp_{\circ}}^{-})^2\right| 
= \frac{d_-^2r_-^2(\cp_\circ - \c_\circ)}{(\c_\circ - \c_-)(\cp_\circ - \c_-)}
= A_1(\wt {\mc F}) \cdot \frac{d_-^4r_-^2}
{d_\circ^2(\c_\circ - \c_-)\,d_-^2(\cp_\circ - \c_-)}\,\frac{\c_+ - \c_\circ}{\c_+ - \cp_\circ}\notag \\
&\qquad
\le 2A_1(\wt{\mc F}) \cdot \l(\frac{C^*\xi_n}{D_n}\r)^2\, \left( 1+\frac{d_{\circ}^2(\cp_{\circ}-\c_{\circ})}{d_{\circ}^2(\c_+-\cp_{\circ})} \right)=o\left( A_1(\wt{\mc F}) \right)
\label{eq_remainder1_2}
\end{align}
since $d_{\circ}^2(\cp_{\circ}-\c_{\circ})\le \wt{C}\xi_n$. 
Furthermore, from \eqref{eq_3_term}, if $r_+ \ne0$,
\begin{align}
& \l\vert \mc W_{\cp_{\circ}}\wt{\mc F}_{\cp_{\circ}}R^+_{\cp_{\circ}} - 
\mc W_{\c_{\circ}}\wt{\mc F}_{\c_{\circ}}R^+_{\c_{\circ}} \r\vert + 
\l\vert \mc W_{\cp_{\circ}}\wt{\mc F}_{\cp_{\circ}}R^-_{\cp_{\circ}} - 
\mc W_{\c_{\circ}}\wt{\mc F}_{\c_{\circ}}R^-_{\c_{\circ}} \r\vert
= \frac{|d_+| r_+ \cdot |d_\circ|(\cp_\circ - \c_\circ)}{\c_+ - \c_\circ} \nn
\\
& \qquad = A_1(\wt{\mc F}) \cdot \frac{d_+^2r_+}{|d_\circ d_+|\,(\c_+ - \cp_\circ)}
\le A_1(\wt{\mc F}) \cdot \frac{C^*\xi_n}{D_n}=o\left( A_1(\wt{\mc F}) \right).
\label{eq_remainder2}
\end{align}
Together, \eqref{eq_A1_main}--\eqref{eq_remainder2} 
establish that 
$|A_1(\wt{\mc F}) - A_1({\mc F})| = o\l(A_1(\wt{\mc F})\r)$. 

For $A_2(\mc F)$, first note that by \eqref{eq_1_term} it holds 
\begin{align*}
A_2(\wt{\mc F}) \le d_{\circ}^2 |\c_{\circ}^*-\cp_{\circ}|
\le \rho_n  = o\l( A_1(\wt{\mc F}) \r).
\end{align*}
By analogous arguments to those
adopted in \eqref{eq_remainder1_1}--\eqref{eq_remainder2},
we also obtain $A_2(\mc F) = o\left( A_1(\wt{\mc F}) \right)$.

From \eqref{eq_4_term}, \eqref{eq_A1_main}
and \eqref{eq_noise_diff}, we yield
\begin{align*}
& |A_3(\wt{\mc F})|
= 2|\mc W_{\c_\circ}\wt{\mc F}_{\c_\circ}| \; |\mc E_{\c_\circ^*} - \mc E_{\c_\circ}|
= \frac{2(\c_+ - \cp_\circ)|d_\circ|}{\c_+ - \c_\circ} \; |\mc E_{\c_\circ^*} - \mc E_{\c_\circ}|
\\
& \le 2A_1(\wt{\mc F}) \left( \frac{ \omegao}{\sqrt{\rho_n\nu_n}}+ \frac{\sqrt{\rho_n}\omegat}{d_{\circ}^2(\cp_{\circ}-\c_{\circ})}+\frac{2\omega_n}{|d_{\circ}|\, \sqrt{\c_+-\c_-}} \right)\le  2A_1(\wt{\mc F}) \left( \frac{ \omegao}{\sqrt{\rho_n\nu_n}}+ \frac{\omegat}{\nu_n\sqrt{\rho_n}}+\frac{2\omega_n}{\sqrt{c^*\xi_n}} \right)\\
&=o\left( A_1(\wt{\mc F}) \right).
\end{align*}
Also, by \eqref{eq_5_term} and because
$\c_+ - \cp_{\circ} \ge \cp_+ - \cp_{\circ}$ when $r_+ > 0$, we get
\begin{align*}
& |A_3(R^+)| = \frac{2r_+|d_+|}{\c_+ - \c_\circ} \;  |\mc E_{\c_\circ^*} - \mc E_{\c_\circ}|
\\
& \le 2 A_1(\wt{\mc F})\, \left(\frac{\omegao\,r_+ |d_+|}{\sqrt{\rho_n\nu_n}\,|d_{\circ}|\,(\c_+-\cp_{\circ})}+\frac{\sqrt{\rho_n}\omegat\,r_+\,|d_+|}{|d_{\circ}|^3\, (\cp_{\circ}-\c_{\circ})\,(\c_+-\cp_{\circ})}  + \frac{2\omega_n\,r_+\,|d_+|}{d_{\circ}^2\sqrt{\c_+-\c_-}\,(\c_+-\cp_{\circ})}\right)\\
&\le 2\,C^*\, A_1(\wt{\mc F})\, {\frac{\xi_n}{D_n}}\,\left(\frac{\omegao}{\sqrt{\rho_n\nu_n}}+\frac{\omegat}{\nu_n\,\sqrt{\rho_n}}+\frac{2\omega_n}{\sqrt{c^*\xi_n}}\right)=o\left(A_1(\wt{\mc F})  \right).
\end{align*}
Similarly, using \eqref{eq:prop:three:k} for the last inequality,
\begin{align*}
& |A_3(R^-)| = \frac{2r_-|d_-|}{\c_\circ - \c_-} \;  |\mc E_{\c_\circ^*} - \mc E_{\c_\circ}|\\
&\le 2 A_1(\wt{\mc F})\,\frac{\c_+-\c_{\circ}}{\c_+-\cp_{\circ}}\, \left( \frac{|d_-|\,r_-\,\omegao}{|d_{\circ}|\,(\c_{\circ}-\c_-)\,\sqrt{\rho_n\nu_n}}
+\frac{|d_-|\,r_- \,\sqrt{\rho_n}\,\omegat}{|d_{\circ}|^3(\c_{\circ}-\c_-)\,(\cp_{\circ}-\c_{\circ})}
+ \frac{2\,|d_-|\,r_-\,\omega_n}{d_{\circ}^2(\c_{\circ}-\c_-)\,\sqrt{\c_+-\c_-}}
\right)\\
&\le 4 A_1(\wt{\mc F})\,C^*\,\left( 1+\frac{\wt{C}}{c^*} \right)\,\frac{\xi_n}{D_n}\,\left( \frac{\omegao}{\sqrt{\rho_n\nu_n}}+\frac{\omegat}{\nu_n\sqrt{\rho_n}}+\frac{2\omega_n}{\sqrt{c^*\xi_n}} \right)=o\left( A_1(\wt{\mc F}) \right).
\end{align*}

As for $A_4(\mathcal{F})$,
 \eqref{eq_6_term}, \eqref{eq_A1_main} and \eqref{eq_noise_one} lead to
\begin{align*}
& |A_4(\wt{\mc F})| \le \frac{2|d_\circ||\c_\circ^* - \cp_\circ|}{\min(\c_+ - \c_\circ^*, \c_\circ^* - \c_-)}
\; |\mc E_{\c_\circ^*}|\le \frac{4\,|d_{\circ}|\,|\c_{\circ}^*-\cp_{\circ}|\,\omega_n}{\sqrt{\min(\c_+ - \c_\circ^*, \c_\circ^* - \c_-)}}
\\
&\le 8\,A_1(\wt{\mc F})\,\frac{\omega_n}{\sqrt{c^*\xi_n\left( 1-\frac{\rho_n}{c^*\xi_n} \right)}}\,\frac{\rho_n}{\min(\rho_n\nu_n,c^*\xi_n)}=o\left( A_1(\wt{\mc F}) \right),
\end{align*}
while from \eqref{eq_7_term},
\begin{align*}
& |A_4(R^+)| \le \frac{4|d_+|r_+|\c_\circ^* - \cp_\circ|\,\omega_n\,\sqrt{\min(\c_+-\c_{\circ}^*,\c_{\circ}^*-\c_-)}}{(\c_+ - \c_\circ^*)(\c_+ - \cp_\circ)}
\\
&\le 4 \,A_1(\wt{\mc F})\,\frac{|d_+|\,r_+}{|d_{\circ}|\,(\c_+-\cp_{\circ})}\, \frac{\omega_n\,|\c_{\circ}^*-\cp_{\circ}|}{|d_{\circ}|\,\sqrt{\c_+-\c^*_{\circ}}\,\min(\cp_{\circ}-\c_{\circ},\c_+-\cp_{\circ})}\\
&\le 4 \,A_1(\wt{\mc F})\,\frac{C^*\xi_n}{D_n}\,\frac{\omega_n}{\sqrt{c^*\xi_n\left( 1-\frac{\rho_n}{c^*\xi_n} \right)}}\,\frac{\rho_n}{\min(\rho_n\nu_n,c^*\xi_n)} =o\left( A_1(\wt{\mc F}) \right).
\end{align*}
Analogously, we obtain the bound of the same order for $|A_4(R^-)|$.

For $A_5(\mc F)$, from \eqref{eq_6_term}, \eqref{eq_7_term} and \eqref{eq:prop:three:k},
\begin{align*}
	&|A_5(\wt{\mc F})| =2\,  |\mc E_{\c_\circ^*}| \frac{|d_{\circ}|\,(\cp_{\circ}-\c_{\circ})}{\c_+-\c_{\circ}}
	\le 4\,A_1(\wt{\mc{ F}})\,\frac{\sqrt{\min(\c_+-\c_{\circ}^*,\c_{\circ}^*-\c_-)}\,\omega_n}{|d_\circ|(\c_+ - \cp_\circ)}\\
	&\le 4\,A_1(\wt{\mc{ F}})\, \frac{\omega_n}{|d_{\circ}| \sqrt{\c_+-\cp_{\circ}}}\,\sqrt{\frac{\c_+-\c_{\circ}^*}{\c_+-\cp_{\circ}}}
	\le 4 \, A_1(\wt{\mc{ F}})\,\frac{\omega_n}{\sqrt{c^*\xi_n}}\,\sqrt{1+\frac{\rho_n}{c^*\xi_n}}=o\left( A_1(\wt{\mc{ F}}) \right).
\end{align*}

Furthermore, by  \eqref{eq_7_term}, \eqref{eq:prop:three:k}, \eqref{eq_A1_main} and \eqref{eq_noise_one} it holds
\begin{align*}
	&|A_5(R^-)| = 2\,\frac{|d_-|\,r_-\, (\cp_{\circ}-\c_{\circ})}{(\c_{\circ}-\c_-)\,(\cp_{\circ}-\c_-)}\,\left|\mc E_{\c_\circ^*}\right|\\
	&\le 4\, A_1(\wt{\mc{ F}})\,\frac{|d_-|\,r_-}{d_{\circ}^2\,(\c_{\circ}-\c_-)}\, \frac{\c_+-\c_{\circ}}{(\c_+-\cp_{\circ})\,(\cp_{\circ}-\c_-)}\, \omega_n\,\sqrt{\min(\c_+-\c_{\circ}^*,\c_{\circ}^*-\c_-)}\\
	&\le 16 \,  A_1(\wt{\mc{ F}})\,\frac{|d_-|\,r_-}{|d_{\circ}|\,(\cp_{\circ}-\c_-)}\,\frac{\omega_n}{|d_{\circ}| \sqrt{\min(\c_+-\cp_{\circ},\cp_{\circ}-\c_-)}}\,\sqrt{\frac{\min(\c_+-\c^*_{\circ},\c^*_{\circ}-\c_-)}{\min(\c_+-\cp_{\circ},\cp_{\circ}-\c_-)}}\\
	&\le 16 \,  A_1(\wt{\mc{ F}})\,\frac{C^*\xi_n}{D_n}\,\frac{\omega_n}{\sqrt{c^*\xi_n}}\, \sqrt{1+\frac{\rho_n}{c^*\xi_n}}=o\left(A_1(\wt{\mc{ F}})  \right).
\end{align*}
Similar but slightly easier arguments give the same bound (with the factor 16 replaced by 4) for $|A_5(R^+)|$.
Since
\begin{align*}
&	\frac{\c_+-\c_-}{(\c_+-\cp_{\circ})(\cp_{\circ}-\c_-)}\,
\left\vert \mc E_{\c_{\circ}^*}+\mc E_{\c_{\circ}} \right\vert
\le 2\,\omega_n\,\frac{\sqrt{\min(\c_+-\c_{\circ},\c_{\circ}-\c_-)+\min(\c_+-\c^*_{\circ},\c^*_{\circ}-\c_-)}}{\min(\c_+-\cp_{\circ},\cp_{\circ}-\c_-)}\\
&\le 2\,\omega_n\,\frac{\sqrt{2+\frac{\wt{C}}{c^*}+\frac{\rho_n}{c^*\xi_n}}}{\sqrt{\min(\c_+-\cp_{\circ},\cp_{\circ}-\c_-)}},
\end{align*}
we yield from \eqref{eq:prop:three:k} and \eqref{eq_noise_diff}--\eqref{eq_noise_one}
\begin{align*}
	&|A_6|=\frac{\c_+-\c_-}{(\c_{\circ}-\c_-)(\c_+-\c_{\circ})}\,
	\left\vert \mc E_{\c_{\circ}^*}+\mc E_{\c_{\circ}} \right\vert\,
	\left|\mc E_{\c_{\circ}^*}-\mc E_{\c_{\circ}}\right|\\
	&\le 8 A_1(\wt{\mc{ F}})\, \sqrt{2+\frac{\wt{C}}{c^*}+\frac{\rho_n}{c^*\xi_n}}\,\frac{\omega_n}{|d_{\circ}|\,\sqrt{\min(\c_+-\cp_{\circ},\cp_{\circ}-\c_-)}}
	\\&\qquad \times
	\left( \frac{\omegao}{\sqrt{\rho_n\nu_n}}+\frac{\sqrt{\rho_n}\,\omegat}{|d_{\circ}|^2(\cp_{\circ}-\c_{\circ})}+2\frac{\omega_n}{|d_{\circ}|\sqrt{\c_+-\c_-} }\right)\\
	&\le 8 A_1(\wt{\mc{ F}})\, \sqrt{2+\frac{\wt{C}}{c^*}+\frac{\rho_n}{c^*\xi_n}}\,\frac{\omega_n}{\sqrt{c^*\xi_n}}\,\left(\frac{\omegao}{\sqrt{\rho_n\nu_n}}+\frac{\omegat}{\nu_n\,\sqrt{\rho_n}}+\frac{2\omega_n}{\sqrt{c^*\xi_n}}  \right)=o\left(A_1(\wt{\mc{ F}})  \right).
\end{align*}

Finally, noting that by \eqref{eq:prop:three:k}
\begin{align*}
|\mc W_{\c_\circ^*} - \mc W_{\c_\circ}|
&\le \frac{(\c_+ - \c_-)}{(\c_+ - \c_\circ^*)(\c_\circ^* - \c_-)}
\,\frac{|\c_\circ^* - \c_\circ|\,\{(\c_+ - \c_\circ^*) + (\c_\circ - \c_-)\}}
{(\c_+ - \c_\circ)(\c_\circ - \c_-)}\\
&\le 2\, A_1(\wt{\mc{ F}})\,\frac{1}{\min(\c_+-\c_{\circ}^*,\c_\circ^*-\c_-)}\, \frac{|\c_{\circ}-\c_{\circ}^*|}{|\cp_{\circ}-\c_{\circ}|}\,\left( \frac{1}{d_{\circ}^2\,(\c_{\circ}-\c_-)}\,\frac{\c_+-\c_{\circ}^*}{\c_+-\cp_{\circ}}+\frac{1}{d_{\circ}^2(\c_+-\cp_{\circ})} \right)
\\
&\le 2\, A_1(\wt{\mc{ F}})\,\frac{1}{\min(\c_+-\c_{\circ}^*,\c_\circ^*-\c_-)}\,\left(1+\frac{1}{\nu_n}\right)\,\left( 3+\frac{2\rho_n}{c^*\xi_n} \right)\,\frac{1}{c^*\xi_n}
\end{align*}
we bound $A_7$ as
\begin{align*}
|A_7| \le 8\, A_1(\wt{\mc{ F}})\,\left(1+\frac{1}{\nu_n}\right)\,\left( 3+\frac{2\rho_n}{c^*\xi_n} \right)\,\frac{\omega_n^2}{c^*\xi_n}=o\left(A_1(\wt{\mc{ F}})  \right),
\end{align*}
which concludes the proof.

\section{Proof of the results in Section~\ref{sec:mosum}}
\label{sec:proof:mosum}

\subsection{Proof of Proposition~\ref{prop:mosum:a}}

The following lemma is used for the proofs of Proposition~\ref{prop:mosum:a} 
and Corollary~\ref{cor:mosum:a}.

\begin{lem}
\label{lem_prop_mosum}
\begin{enumerate}[label = (\alph*)] 
\item Under the assumption of Proposition~\ref{prop:mosum:a}, consider 
\begin{align*}
\mc S_n(j) = \l\{ \l\vert T_{\cp_j, n}(G(j)) \r\vert 
\ge \max\l(\max_{|\c - \cp_j| > (1 - \eta) G(j)} \l\vert T_{k, n}(G(j)) \r\vert, 
\tau\, D_n(G(j), \alpha) \r) \r\},
\end{align*}
and $\mc S_n = \bigcap_{1 \le j \le q_n} \mc S_n(j)$.
Then for any $\alpha, \eta \in (0, 1)$, we have
\begin{align*}
\p(\mc S_n(j)) \to 1 \text{ for any } j = 1, \ldots, q_n \quad \text{and} \quad 
\p(\mc S_n) \to 1.
\end{align*}
\item Under the assumptions of Corollary~\ref{cor:mosum:a} below, analogous assertions hold with
$\wt{\mc S}_n(j)$ replacing $\mc S_n(j)$, where
\begin{align*}
\wt{\mc S}_n(j) = \bigcap_{0 \le r \le 2/\eta - 2}
\l[\l\{\l\vert T_{\cp_j + r \eta G/2, n}(G) \r\vert 
\ge \max_{k \in [\cp_j + (r + 1) \eta G/2, \cp_j + (r + 2) \eta G/2]}
\l\vert T_{k, n}(G) \r\vert \r\} \r.
\\
\qquad \qquad \qquad \l.
\bigcap \l\{\l\vert T_{\cp_j - r \eta G/2, n}(G) \r\vert 
\ge \max_{k \in [\cp_j - (r + 2) \eta G/2, \cp_j - (r + 1) \eta G/2}
\l\vert T_{k, n}(G) \r\vert \r\}
\r].
\end{align*}
\end{enumerate}
\end{lem}

\begin{proof}
Adopting the arguments analogous to those used in the proof of Lemma~5.1~(a) of \cite{kirch2014}, we get
\begin{align*}
\sqrt{2}\,|T_{\cp_j,n}(G(j))| \ge |d_j|\sqrt{G(j)} + O_P\left(\omega_n\right) = |d_j|\sqrt{G(j)} \,(1+o_P(1)),
\\
\max_{|\c - \cp_j| > (1 - \eta) G(j)} \sqrt{2}\,\vert T_{k, n}(G(j)) \vert 
\le \eta\, |d_j|\sqrt{G(j)} + O_P(\omega_n) = \eta\,|d_j|\sqrt{G(j)}\,(1 + o_P(1)).
\end{align*}
Also, noting that $D_n(G(j), \alpha) = O(\sqrt{\log(n)})$ and $D_n/\sqrt{\log(n)} \to \infty$,
the `significance' of $\vert T_{\cp_j, n}(G(j)) \vert$ follows,
and so does the assertion for $S_n(j)$. 
For the set $S_n$, the assertion follows because all $O_P$-terms hold uniformly in $j$.
The assertion of (b) follows analogously.
\end{proof}

With the help of Lemma~\ref{lem_prop_mosum},
we can now prove Proposition~\ref{prop:mosum:a} 
by adopting the arguments of the proof of Theorem~3.2 of \cite{kirch2014}. 
Therefore we only sketch the proof by emphasizing the differences using the notations adopted therein. 
In particular the quantities $V_{l, n}^{(j)}(G(j))$ and $A_i(l, n; G(j)) = A_i(l, n), \, i = 1, 2$ are defined as in that proof.

On $S_n(j)$ defined in Lemma~\ref{lem_prop_mosum}~(a),
the maximiser of $\vert T_{b, n}(G(j)) \vert$ over $b$ satisfying $\vert b - \cp_j \vert \le (1 - \eta)G(j)$,
fulfils the $\eta$-criterion and as such 
is a candidate produced by the MOSUM procedure which we denote by $\c_j$ in the following. 
For this candidate, it holds:
\begin{align*}
& \l\{ \c_j - \cp_j < -C_M(\omegao/d_j)^2 \r\} \subset  
\\
& \qquad \qquad \qquad 
\l\{ \max_{\cp_j - G(j) + 1 \le l < \cp_j - C_M(\omegao/d_j)^2} V_{l, n}^{(j)}(G(j)) \ge 
\max_{\cp_j - C_M(\omegao/d_j)^2 \le l \le \cp_j + G(j)} V_{l, n}^{(j)}(G(j)) \r\}.
\end{align*}

Furthermore, by Condition~(b) we obtain
\begin{align*}
\max_{1 \le j \le q_n} \frac{1}{|d_j|\sqrt{G(j)}}\,\max_{|l-\cp_j| < G(j)} \vert A_2(l, n; G(j)) \vert = o_P(1).
\end{align*}
Also by Condition~(c), we can find a suitable constant $\wt{C}_M > 0$ such that for all $C_M>0$, it holds
\begin{align*}
\p\l(	\max_{1 \le j \le q_n} \sqrt{2G(j)}\,\max_{\cp_j - G(j) \le l \le \cp_j - C_M(\omegao/d_j)^2}
\frac{\sqrt{C_M(\omegao/d_j)^2}}{|\cp_j - l|} \, \vert A_1(l, n; G(j)) \vert
> \wt{C}_M\, \omegao \r) \to 0,
\end{align*}
from which we can find a suitable choice of $C_M$ depending only on $\wt C_M$ such that 
\begin{align*}
\p\left( \max_{1\le j\le q_n} \frac{\sqrt{2G(j)}}{|d_j|} \max_{\cp_j - G(j) \le l \le \cp_j - C_M(\omegao/d_j)^2}\frac{\left| A_1(l, n; G(j))\right|}{|\cp_j-l|} \ge \frac{1}{3}
\right) \to 0.
\end{align*}
Consequently,
\begin{align*}
\p\l( \min_{1 \le j \le q_n} d_j^2(\c_j - \cp_j) < - C_M (\omegao)^2, \mc S_n \r) = o(1).
\end{align*}
The case $\c_j - \cp_j > C_M(\omegao/d_j)^2$  
can be dealt with analogously, 
which concludes the proof. 

\subsection{Proof of Proposition~\ref{prop:mosum:b}}
\label{sec:prop:mosum:b}

Firstly note that the $\eta$-criterion employed by the MOSUM procedure
implicitly imposes an upper bound on the number of estimators returned:
At bandwidths $\mbf G = (G_\ell, G_r)$,
for each local maximiser $\c$ of the MOSUM detector, it is checked whether
the local maximum corresponds to the maximum absolute MOSUM value 
within the interval $(\c - \eta \, G_\ell, \c + \eta \, G_r]$ and, if so, $\c$ is marked as a candidate change point.
Therefore, the maximal number of possible candidates detectable at scale $(G_\ell, G_r)$ is
$(\eta \min(G_\ell, G_r))^{-1} n$.
Then, by~\eqref{eq_band_asym}, it holds 
\begin{align*}
\min(G_\ell, G_r) \ge \frac{\max(G_\ell, G_r)}{C_{\text{asym}}} \ge \frac{G_\ell + G_r}{2\,C_{\text{asym}}}.
\end{align*}
From this and by the construction of $\mc G$ with $G_{\ell} = F_{\ell}\,G_0$, it holds
\begin{align}
\vert \C(\mc H, \alpha) \vert \le \sum_{\ell, r = 1}^{H_n} \vert \C(G_\ell, G_r, \alpha) \vert
\le 2\,C_{\text{asym}} \, \frac{n}{\eta G_0}\,\sum_{\ell, r = 1}^{H_n} \frac{1}{F_\ell + F_r} 
\le 2\,C_{\text{asym}} \, \frac{\psi}{\eta}\,\frac{n}{G_0},
\nn 
\end{align}
for some universal constant $\psi$ satisfying
\begin{align*}
\sum_{\ell, r = 1}^{\infty} \frac{1}{F_{\ell} + F_r} \le \psi <\infty.
\end{align*}
This holds as the Fibonacci numbers are asymptotically bounded from below by an exponentially decreasing sequence,
i.e., $F_\ell \ge (3/2)^\ell$ for all $\ell \ge 10$ which is easily seen by induction. 
Then, the conclusion follows from $\omega_n^2/G_0 \to 0$.

\subsection{Single-scale MOSUM procedure}
\label{sec:cor:mosum:a}

As a corollary, we show that the single-bandwidth MOSUM procedure yields consistent estimators 
with optimal localisation rate either under sub-Gaussianity, or when there are finitely many change points,
but only under the assumption that the change points are {\it homogeneous}
as defined in Definition~\ref{def_scenarios}~(a.i). 
It improves upon Theorem~3.2 of \cite{kirch2014}
where the optimal rate is obtained only in the case when $q_n$ is finite. 
By construction, when the change points are heterogeneous as in Definition~\ref{def_scenarios}~(a.iii),
the single-bandwidth MOSUM procedure cannot produce consistent estimators.

\begin{cor}
\label{cor:mosum:a}
Let $\C(G, \alpha_n) = \{\c_{G, j}: \, 1 \le j \le \wh q_G\}$ denote the set of 
estimated change points from a single-bandwidth MOSUM procedure,
obtained according to either the $\eta$- or $\epsilon$-criterion (see \cite{meier2018} for their description)
with $\eta, \epsilon \in (0, 1)$,
where the bandwidth $G$ and the significance level $\alpha_n$ satisfy
\begin{align*}
& \min_{0 \le j \le q_n} (\cp_{j + 1} - \cp_j) > 2G, \quad
\min_{1 \le j \le q_n} \frac{d_j^2G}{\log(n/G)} \to \infty, \quad \text{and} 
\\
& \alpha_n \to 0 \quad \text{with} \quad D_n(G; \alpha_n) = O(\sqrt{\log(n/G)}).
\end{align*}
We further assume that the invariance principle holds as in Proposition~\ref{prop:vep}~(c.i) with
\begin{align*}
\frac{\lambda_n^2 \log(n/G)}{G} \to 0.
\end{align*}
Then, there exists a universal constant $C_M > 0$ such that
\begin{align*}
\p\l( \wh q_G = q_n; \, \max_{1 \le j \le q_n} d_j^2 
\vert \c_{G, j} \bbI_{j \le \wh q_G} - \cp_j \vert \le C_M (\omegao)^2 \r) \to 1.
\end{align*}
\end{cor}

\begin{proof}
First, we need to show that asymptotically,
(i) there is exactly one significant local maximum in the $G$-environment of each change point, 
and (ii) there are no other significant local maxima. 
The second assertion follows by~Lemma 5.1~(b) of \cite{kirch2014}.
Concerning (i), by Lemma~\ref{lem_prop_mosum}~(b), 
there is only one (and significant by (a)) local maximum
within a $G$-environment of every change point on an asymptotic one set,
which also fulfils the $\epsilon$-criterion by Lemma~5.1~(a) in \cite{kirch2014}.
Then, the localisation rate follows by the same arguments as 
in the proof of Proposition~\ref{prop:mosum:a}, completing the proof.
\end{proof}

\section{Computational complexity}
\label{sec:comp}

Analysing the computational complexity of the localised pruning is challenging without further assumption
on the number of the candidates analysed at each iteration (denoted by $\mc D$ in Step~2 of {\tt LocAlg}).
In implementing the algorithm, we impose a fixed upper bound of $N = 24$ on $\vert \mc D \vert$;
if $\vert \mc D \vert > N$ at a particular iteration,
we modify the order in which the candidates remaining in $\mc C$ are processed which often resolves the issue.
Theorem~\ref{thm:sbic} holds irrespective of which candidate is chosen in Step~1 of  {\tt LocAlg},
and thus this step does not harm the theoretical guarantee.
To guard against the contingency where all other candidates in $\mc C$ 
also have more than $N$ conflicting candidates,
a manual thinning step for the set $\mc D$ is implemented in the R package {\tt mosum} 
which triggers a warning message,
see Appendix~A of \citet{meier2018} for further details. 
In practice, this manual thinning step is rarely activated; for example,
for the {\bf dense} test signals with frequent change points and $n \ge 2 \times 10^4$
considered in simulations (see Section~\ref{sec:sim:setup}),
we did not encounter a single occurrence over $1000$ realisations for each test signal. 
Since there are at most $O(n\log^{-1}(n))$ candidates in total (see Assumption~\ref{assum:cand}~(b)),
the localised pruning requires $O(n + 2^N n\log^{-1}(n))$ operations in the worst case.

The MOSUM-based candidate generating procedure discussed in Section~\ref{sec:mosum} 
requires $O(n\vert \mc H \vert)$ operations,
where $\mc H$ denotes the set of asymmetric bandwidths.
In Section~\ref{sec:mosum}, we propose a scheme for bandwidth generation 
which ensures the adaptivity of the multiscale MOSUM procedure
while bounding the total number of bandwidths to be considered at
$\vert \mc H \vert = O(\log(n))$ through a condition on 
the balancedness of asymmetric bandwidths (see~\eqref{eq_band_asym}),
which amounts to the computation time of $O(n\log(n))$ for the multiscale MOSUM procedure.
The computational complexity of the CUSUM-based candidate generation
depends on the number $R_n$ of the random intervals drawn as $O(R_n n)$,
which in turn needs to increase in $(\min_j \delta_j)^{-1}$ as
$n^2(\min_j \delta_j)^{-2}\log(n)/R_n \to 0$ (see~\eqref{cond:wbs:two}) for adaptivity.

In summary, with the MOSUM-based candidate generating mechanism,
the combined two-stage methodology requires $O(n\log(n) + 2^N n\log^{-1}(n))$ operations in total,
which is much faster than many competitors requiring dynamic programming-type solutions
(such as those proposed in \citet{frick2014}, \citet{wang2018d} and \citet{fromont2020})
whose computational complexity is $O(n^2)$, 
see also Table~\ref{table:overview} for the summary of the computational complexity
of various methods for univariate data segmentation. 

\section{Complete simulation results}
\label{sec:sim:app}

\subsection{Set-up}
\label{sec:sim:setup}

We consider the five test signals from \cite{fryzlewicz2014}
referred to as {\tt blocks}, {\tt fms}, {\tt mix}, {\tt teeth10} and {\tt stairs10},
see Appendix~B therein for further details.
In addition, we include the following test signals extending the original ones
in order to investigate the scalability of the localised pruning algorithm:
\begin{description}[leftmargin = 0cm]
\setlength\itemsep{0em} 
\item[Dense test signals.] Each test signal is concatenated
until the length of the resultant signal exceeds $2 \times 10^4$.
\item[Sparse test signals:] Each test signal is embedded 
in the series of i.i.d.\ random variables of length $n = 2 \times 10^4$ at $t = 500$.
\end{description}

For $\vep_t$, we consider 
\begin{enumerate}[label=($\mc E$\arabic*)]
\setlength\itemsep{0em}
\item \label{eq:n1} independent Gaussian random variables as in \cite{fryzlewicz2014},
\item \label{eq:n2} independent random variables following the $t_5$ distribution, and
\item \label{eq:n3} AR($1$) processes with Gaussian innovations 
and the AR parameter $\varrho \in \{0.3, 0.9\}$,
\end{enumerate}
while keeping the signal-to-noise ratio defined by
$\{\Var(\vep_t)\}^{-1/2}\min_{1 \le j \le q_n}|d_j|$ constant
across different error distributions.
Under \ref{eq:n3}, in order to account for the information loss due to the serial dependence,
the length of each segment between adjacent change points is 
increased by the factor of $\lfloor 1/(1-\varrho) \rfloor$. 

We apply the localised pruning algorithm outlined in Section~\ref{section_LA}
together with the two candidate generating mechanisms described in 
Section~\ref{sec:mosum} and Section~\ref{sec:wbs}.
\begin{description}[leftmargin = 0cm]
\item[MOSUM-based candidate generation (`MoLP'):] 
We use the multiscale extension of the MOSUM procedure 
with the asymmetric bandwidths $\mc H$ selected as described in Section~\ref{sec:mosum},
setting $G_0 = 10$ in the case of \ref{eq:n1}--\ref{eq:n2}
and $G_0 = \max(10, \lfloor 8/(1 - \varrho) \rfloor)$ in the case of \ref{eq:n3},
and $C_{\text{asym}} = 4$.
In deriving the asymptotic critical value, we use $\alpha \in \{0.1, 0.2\}$
except when the change points are dense, we consider $\alpha \in \{0.2, 0.4\}$
to ensure that $\C(\mc H, \alpha)$ meets Assumption~\ref{assum:cand}~(a).
Also, we set $\eta = 0.4$ for locating the change points according to the $\eta$-criterion,
and consider both the inverse of $p$-value ($h_{\mc P}$) and 
the jump size ($h_{\mc J}$, see \eqref{eq:jump:sort}) associated with their detection 
for sorting the change point candidates in Step~1 of {\tt LocAlg}.
For variance estimation, we adopt the MOSUM variance estimator of \cite{kirch2014}
in the case of the independent errors in \ref{eq:n1}--\ref{eq:n2}.

When serial dependence is present under \ref{eq:n3},
we use the MOSUM variance estimator inflated by the factor of $(1 + \wh\varrho) / (1 - \wh\varrho)$
with an estimator of the AR parameter $\wh\varrho$.
For this, we first generate candidates via the multiscale MOSUM procedure.
Here, through using the MOSUM variance estimator without any correction,
the procedure is expected not to under-estimate the number of change points.
Then, $\wh\varrho$ is obtained as the Yule-Walker estimator from the resultant residuals,
which is fed into correct the MOSUM variance estimator as above.

MoLP is implemented in the R package {\tt mosum} \citep{mosum}.

\item[CUSUM-based candidate generation (`CuLP'):]
We select the number of random intervals for each recursion of WBS2 as recommended in \cite{fryzlewicz2018}.
Instead of selecting an upper bound $\wt Q_n$ on its cardinality,
we use the following subset of $\C(R_n, \wt{Q}_n = n)$
\begin{align*}
\C(R_n, n, \zeta_n) = \l\{\c_\circ: \, 
\c_\circ \in \C(R_n, \wt{Q}_n = n) \text{ with the corresponding } 
|\mc X_{s_\circ, \c_\circ, e_\circ}| \ge \zeta_n \r\},
\end{align*}
which provides more flexibility with respect to the choice of the threshold $\zeta_n$.
In addition, for numerical stability in local variance estimation (described below), 
we consider only those $\c_\circ \in \C(R_n, n, \zeta_n)$
with $\min(\c_\circ - s_\circ, e_\circ - \c_\circ) \ge 5$.
As the thresholding needs not remove all false positives,
we can choose $\zeta_n$ generously even in the presence of heavy-tailed or serially dependent errors.
We use $\zeta_n = C_\zeta \cdot K \wh\tau_n \sqrt{2\log(n)}$,
where $K$ is a constant chosen as per \cite{fryzlewicz2018}.
The deflation factor $C_\zeta$ is set at $C_\zeta = 0.9$, except when change points are dense,
in which case we also consider stronger deflation by $C_\zeta = 0.5$ to ensure that 
$\C(R_n, n, \zeta_n)$ meets Assumption~\ref{assum:cand}~(a).

As in MoLP, we estimate the (long-run) variance 
using a local estimator extending the MOSUM variance estimator of \cite{kirch2014}:
for \ref{eq:n1}--\ref{eq:n2}, it is obtained as the sample variance of the residuals over $\mc I(\c)$
after fitting a stump function with a break at the candidate change point $\c$;
for \ref{eq:n3}, we inflate the local variance estimator by 
the factor of $(1 + \wh\varrho) / (1 - \wh\varrho)$.
\end{description}
For the penalty of $\sic$, we consider $\xi_n \in \{\log^{1.01}(n), \log^{1.1}(n)\}$ for \ref{eq:n1};
$\xi_n = \{\log^{1.1}(n), n^{2/4.99}\}$ for \ref{eq:n2};
$\xi_n = \{\log^{1.1}(n), \log^2(n)\}$ for \ref{eq:n3},
respectively referred to as the `light' and `heavy' penalties.

The following competitors are considered for the comparative study.
\begin{enumerate}
\item The multiscale MOSUM procedure with the `bottom-up' merging ({\tt bottom.up}) 
implemented in the R package {\tt mosum} \citep{mosum}
(see \cite{messer2014} and also \cite{meier2018}).
\item WBS \citep{fryzlewicz2014} applied with 
the strengthened Bayesian information criterion (WBS.sBIC,
implemented in the R package {\tt breakfast} \citep{breakfast}).
For the generation of random intervals, the same approach as that in CuLP is taken.
\item WBS2.SDLL proposed in \cite{fryzlewicz2018},
whose implementation is available on \url{https://github.com/pfryz/wild-binary-segmentation-2.0}).
\item Pruned exact linear time (PELT) algorithm of \cite{killick2012} (R package {\tt changepoint} \citep{changepoint}).
\item The dynamic programming algorithm based on functional pruning (S3IB) proposed in \cite{rigaill2010} 
(R package {\tt Segmentor3IsBack} \citep{s3ib}).
\item Tail-greedy unbalanced Haar (TGUH) algorithm of \cite{fryzlewicz2017}
(R package {\tt breakfast} \citep{breakfast}).
\item FDRSeg \citep{li2016}, the multiscale segmentation method controlling the false discovery rate 
(R package {\tt FDRSeg} \citep{fdrseg}).
\item cumSeg \citep{muggeo2010}, the method based on transforming the data and 
iteratively fitting a linear model (R package {\tt cumSeg} \citep{cumSeg}).
\end{enumerate}
Unless stated otherwise, we apply the above methods with default choices of parameters
recommended by the authors.
Additionally, we consider:
\begin{enumerate}
\setcounter{enumi}{8}
\item Functional pruning optimal partitioning (FPOP) algorithm of \cite{maidstone2017}
(R package {\tt FPOP} \citep{fpop}) with the penalty set at $\sqrt{2\log(n)}$ 
is considered for the test signals with $n \ge 2 \times 10^4$.

\item Jump segmentation for dependent data (JUSD) of \cite{tecuapetla2017} 
and DepSMUCE of \cite{dette2018} are considered for \ref{eq:n3}.
the latter extending the simultaneous multiscale change point estimator (SMUCE) \citep{frick2014} 
to the dependent case
(both implemented using the R package {\tt stepR} \citep{stepR}).
For JUSD, the estimator of the long-run variance relies on the assumption of $m$-dependence
yet there does not exist an automatic way of determining $m$;
instead we use $m = [\log(0.1)/\log(\varrho)]$ utilising the typically unavailable knowledge of $\varrho$;
For DepSMUCE, the recommended choice of block length $K = 10$ often severely under-estimates
the long-run variance, 
and thus we supply $K = [\log(0.1)/\log(\varrho)]$.
\end{enumerate}

Many of the algorithms mentioned above are specifically tailored for the data
with i.i.d.\ innovations following sub-Gaussian distributions, with the exception of cumSeg, JUSD and DepSMUCE.

\subsection{Results}
\label{sec:sim:res}

All simulations are based on $1000$ replications.

We define that a change point $\cp_j$ is detected
if there exists at least one estimator that falls between
$\max\{(\cp_j + \cp_{j-1})/2, \cp_j - \bar{\delta}\}$ and 
$\min\{(\cp_j + \cp_{j + 1})/2, \cp_j + \bar{\delta}\}$,
where $\bar{\delta} = \min_{1 \le j \le q_n - 1} (\cp_{j + 1} - \cp_j)$.
Based on this, we report the true positive rate 
(TPR, the proportion of the correctly identified change point out of the $q_n$ true change points) and
false positive rate (FPR, the proportion of the spurious estimators out of the $\wh q$ estimated change points).
Also reported are
the Adjusted Rand Index (ARI) measuring the similarity between the estimated and true segmentations
\citep{rand1971, hubert1985},
the relative mean squared error (MSE) of the estimated piecewise constant signal
to that of the signals estimated using the true change points,
Bayesian information criterion (BIC) with the penalty term $\log(n)$,
and the weighted average of trimmed distances 
$\delta_{\trim} = (\sum_{j=1}^{q_n} d_j^2)^{-1} \sum_{j=1}^{q_n} d_j^2 \cdot \delta_{\trim, j}$
where
\begin{align}
\delta_{\trim, j} = \min\l\{ \frac{\cp_{j + 1} - \cp_j}{2}, \frac{\cp_j - \cp_{j - 1}}{2},
\min_{1 \le j' \le \wh q} |\wh \cp_{j'} - \cp_j| \r\},
\end{align}
averaged over $1000$ replications.
Also, we provide $v_{\trim} = q_n^{-1} \sum_{j=1}^{q_n} \mbox{MAD}(\delta_{\trim, j})$,
where the MAD operator is taken over $1000$ replications for each change point $\cp_j$.
Finally, for the dense and sparse test signals, we report the average execution time.


\subsection*{\ref{eq:n1} Independent Gaussian errors}

Tables \ref{supp:table:sim:orig}--\ref{supp:table:sim:sparse} report the simulation results 
in the presence of independent Gaussian errors for the original five test signals
and their dense and sparse versions.
Figures~\ref{fig:sim:blocks}--\ref{fig:sim:sparse:stairs10} visualise
the performance of various methods 
by plotting the weighted densities of estimated change point(s) 
falling between two adjacent change points
$[(\cp_{j - 1} + \cp_j)/2 + 1, (\cp_j + \cp_{j + 1})/2]$ for $j = 1, \ldots, q_n$.

Table~\ref{supp:table:sim:orig} indicate that
choices of $\alpha$ for the MoLP or the sorting function $h$ and the penalty $\xi_n$
for the localised pruning algorithm do not greatly influence the results.
In particular, with $n$ relatively small ($\le 2048$),
the choice of penalty $\xi_n$ does not alter the results much.
Difference in performance due to these choices are more apparent 
when $n$ is large ($\ge 2 \times 10^4$),
see Tables~\ref{supp:table:sim:dense}--\ref{supp:table:sim:sparse}.

When change points are dense, a lighter penalty $\xi_n$ and a generous choice of the critical value
for the candidate generation method
(larger $\alpha$ for the MoLP, smaller $C$ for the CuLP)
are preferable for some test signals such as {\tt teeth10}, 
which ensures that the candidate set contains
at least one valid estimator for each $\cp_j$ (Assumption~\ref{assum:cand}~(a)).
On the other hand, when the change points are sparse, a heavier penalty $\xi_n$
is successful in removing spurious false positives over a long stretch of stationary observations
without harming the TPR much.
Between the two methods equipped with different candidate generating methods
CuLP tends to incur more false positives than the MoLP.
Overall, the MoLP produces estimators of better localisation accuracy,
possibly benefiting from the systematic approach to candidate generation adopted
by the multiscale MOSUM procedure. 
This is also reflected on the execution time of the two methods when the change points are dense.

{\tt bottom.up}, compared to the MoLP, tends to return many false positives.
This reflects the corresponding theoretical requirements on the MOSUM procedure,
that the significance level is small ($\alpha = \alpha_n \to 0$, \citet{kirch2014}) and 
that the bandwidths are in the order of $n$ \citep{messer2014},
and the problem is further amplified with increasing $n$ 
(see Table~\ref{supp:table:sim:sparse})
and heavy-tailed errors as observed under \ref{eq:n2}.
An interesting phenomenon is observed in Figure~\ref{fig:sim:fms} 
which plots the weighted densities of estimated change points
for the {\tt fms} test signal, where {\tt bottom.up} incurs several false positives systematically.
This is attributed to spurious estimators detected with large bandwidths
between the first and the second change points.

There is no single method that outperforms the rest universally for all test signals and evaluation criteria. 
While S3IB marginally outperforms other competitors in terms of TPR, it is at the price of larger FPR.
FPOP, another functional pruning algorithm, is computationally fast and generally performs well,
but fails at handling the teeth-like jump structure of {\tt teeth10} 
(see Tables~\ref{supp:table:sim:dense}--\ref{supp:table:sim:sparse}).
WBS2.SDLL shows its strength in handling frequent changes, 
although returning marginally more false positives
compared to other methods achieving comparable TPR.
Both PELT and cumSeg tend to under-estimate the number of change points across all test signals
and so does WBS.sBIC. The latter result indicates
that minimisation of an information criterion along a solution path
is not as efficient as the pruning criteria \ref{eq_c1}--\eqref{eq_c2} adopted by {\tt PrunAlg},
both computationally or empirical performance-wise.
Interestingly, when the frequent changes in {\tt teeth10} are repeated over 
$n \ge 2 \times 10^4$ observations,
the BIC is minimised at the null model (Table~\ref{supp:table:sim:dense}), 
further suggesting that the sequential minimisation of BIC often leads to less favourable results 
compared to {\tt PrunAlg}.


In terms of computation time, 
FPOP, PELT and {\tt bottom.up} take less than $0.1$ seconds to process a long signal.
It is followed by the MoLP and TGUH,
demonstrating that the localised pruning is scalable to long signals.
While CuLP tends to be slower than MoLP,
it still surpasses WBS.sBIC and WBS2.SDLL in this respect (except for the dense {\tt block} signal),
which demonstrates the computational gain achievable by the localised exhaustive search
adopted in the proposed methodology.
FDRSeg and S3IB, while showing good performance for short test signals,
are computationally too expensive for long signals and, along with cumSeg, 
are omitted in these situations.

\cite{meier2018} observed that for the MoLP,
the ordering function $h_{\mc P}$ incurs many ties as the $p$-values
associated with candidates detected at larger bandwidths 
are set exactly to be zero by the machine,
which increases the search space for the inner algorithm {\tt PrunAlg} 
and consequently slows down the pruning procedure.
As there is no meaningful difference in terms of change point detection accuracy,
we recommend the use of $h_{\mc J}$.

{\small
\setlength{\tabcolsep}{3pt}
\begin{longtable}{c|ccc|c|c|c|c|c|cc}
\caption{Summary of change point estimation over $1000$ realisations for 
the test signals with Gaussian errors:
we use $\xi_n \in \{\log^{1.01}(n), \log^{1.1}(n)\}$ as the `light' and `heavy' penalties
for the localised pruning.} 
\label{supp:table:sim:orig} 
\endfirsthead
\endhead
\hline\hline
model &	$\alpha$ &	penalty &	method &	TPR &	FPR &	ARI &	MSE &	BIC &	$\delta_{\trim}$ &	$v_{\trim}$ 	\\	\hline
{\tt blocks} &	0.1 &	light &	MoLP-$h_{\mc P}$ &	0.954 &	0.009 &	0.977 &	5.155 &	4784.242 &	351.119 &	262.916	\\	
&	&	&	MoLP-$h_{\mc J}$ &	0.955 &	0.009 &	0.978 &	5.003 &	4783.629 &	332.611 &	246.996	\\	
&	&	heavy &	MoLP-$h_{\mc P}$ &	0.944 &	0.004 &	0.977 &	5.231 &	4784.409 &	377.367 &	195.167	\\	
&	&	&	MoLP-$h_{\mc J}$ &	0.945 &	0.004 &	0.978 &	5.1 &	4783.634 &	347.61 &	179.247	\\	
&	0.2 &	light &	MoLP-$h_{\mc P}$ &	0.96 &	0.014 &	0.975 &	5.217 &	4784.947 &	365.521 &	258.931	\\	
&	&	&	MoLP-$h_{\mc J}$ &	0.961 &	0.014 &	0.977 &	4.911 &	4783.747 &	327.638 &	227.091	\\	
&	&	heavy &	MoLP-$h_{\mc P}$ &	0.949 &	0.006 &	0.977 &	5.08 &	4784.11 &	356.655 &	195.167	\\	
&	&	&	MoLP-$h_{\mc J}$ &	0.949 &	0.005 &	0.978 &	5.01 &	4783.597 &	338.395 &	163.327	\\	
&	&	light &	CuLP &	0.934 &	0.095 &	0.919 &	10.091 &	4805.107 &	1001.623 &	284.352	\\	
&	&	heavy &	CuLP &	0.936 &	0.033 &	0.95 &	7.631 &	4794.349 &	708.949 &	197.160	\\	\cline{5-11}
&	0.2 &	- &	{\tt bottom.up} &	0.958 &	0.278 &	0.877 &	6.308 &	4812.993 &	372.152 &	309.686	\\	
&	- &	- &	WBS.sBIC &	0.938 &	0.032 &	0.962 &	7.304 &	4795.326 &	694.854 &	264.426	\\	
&	- &	- &	WBS2.SDLL &	0.94 &	0.027 &	0.971 &	5.51 &	4785.457 &	359.36 &	195.167	\\	
&	- &	- &	PELT &	0.878 &	0.001 &	0.961 &	6.413 &	4785.848 &	588.644 &	220.480	\\	
&	- &	- &	S3IB &	0.974 &	0.019 &	0.979 &	4.773 &	4782.984 &	306.041 &	186.930	\\	
&	- &	- &	cumSeg &	0.772 &	0.002 &	0.914 &	13.119 &	4818.988 &	1743.155 &	555.444	\\	
&	- &	- &	TGUH &	0.948 &	0.023 &	0.967 &	6.589 &	4788.875 &	488.462 &	342.805	\\	
&	0.2 &	- &	FDRSeg &	0.975 &	0.081 &	0.956 &	5.367 &	4788.694 &	328.394 &	235.020	\\	\hline
{\tt fms} &	0.1 &	light &	MoLP-$h_{\mc P}$ &	0.982 &	0.015 &	0.954 &	4.402 &	-564.883 &	0.175 &	0.168	\\	
&	&	&	MoLP-$h_{\mc J}$ &	0.981 &	0.015 &	0.955 &	4.356 &	-564.894 &	0.175 &	0.151	\\	
&	&	heavy &	MoLP-$h_{\mc P}$ &	0.98 &	0.009 &	0.955 &	4.407 &	-564.878 &	0.178 &	0.168	\\	
&	&	&	MoLP-$h_{\mc J}$ &	0.979 &	0.01 &	0.955 &	4.354 &	-564.897 &	0.178 &	0.168	\\	
&	0.2 &	light &	MoLP-$h_{\mc P}$ &	0.99 &	0.02 &	0.958 &	4.138 &	-565.219 &	0.148 &	0.151	\\	
&	&	&	MoLP-$h_{\mc J}$ &	0.99 &	0.021 &	0.957 &	4.119 &	-565.183 &	0.148 &	0.151	\\	
&	&	heavy &	MoLP-$h_{\mc P}$ &	0.989 &	0.012 &	0.958 &	4.129 &	-565.219 &	0.152 &	0.151	\\	
&	&	&	MoLP-$h_{\mc J}$ &	0.988 &	0.012 &	0.959 &	4.064 &	-565.258 &	0.15 &	0.151	\\	
&	&	light &	CuLP &	0.997 &	0.149 &	0.905 &	5.379 &	-562.971 &	0.137 &	0.033	\\	
&	&	heavy &	CuLP &	0.997 &	0.074 &	0.937 &	4.446 &	-565.116 &	0.139 &	0.033	\\	\cline{5-11}
&	0.2 &	- &	{\tt bottom.up} &	0.976 &	0.32 &	0.836 &	6.266 &	-548.575 &	0.312 &	0.151	\\	
&	- &	- &	WBS.sBIC &	0.975 &	0.014 &	0.96 &	4.747 &	-564.272 &	0.235 &	0.103	\\	
&	- &	- &	WBS2.SDLL &	0.995 &	0.032 &	0.955 &	4.187 &	-566.031 &	0.139 &	0.033	\\	
&	- &	- &	PELT &	0.934 &	0.001 &	0.954 &	5.016 &	-565.769 &	0.389 &	0.033	\\	
&	- &	- &	S3IB &	0.999 &	0.1 &	0.944 &	4.98 &	-566.096 &	0.101 &	0.033	\\	
&	- &	- &	cumSeg &	0.754 &	0.012 &	0.918 &	14.05 &	-549.512 &	1.841 &	0.103	\\	
&	- &	- &	TGUH &	0.995 &	0.04 &	0.945 &	4.822 &	-565.036 &	0.15 &	0.067	\\	
&	0.2 &	- &	FDRSeg &	0.998 &	0.086 &	0.953 &	4.441 &	-564.844 &	0.113 &	0.033	\\	\hline
{\tt mix} &	0.1 &	light &	MoLP-$h_{\mc P}$ &	0.911 &	0.007 &	0.738 &	4.195 &	842.617 &	30.552 &	18.818	\\	
&	&	&	MoLP-$h_{\mc J}$ &	0.913 &	0.007 &	0.74 &	4.178 &	842.562 &	29.944 &	18.818	\\	
&	&	heavy &	MoLP-$h_{\mc P}$ &	0.9 &	0.003 &	0.717 &	4.262 &	842.654 &	30.849 &	23.436	\\	
&	&	&	MoLP-$h_{\mc J}$ &	0.901 &	0.004 &	0.716 &	4.238 &	842.554 &	30.361 &	23.436	\\	
&	0.2 &	light &	MoLP-$h_{\mc P}$ &	0.929 &	0.009 &	0.772 &	4.096 &	842.634 &	30.252 &	16.765	\\	
&	&	&	MoLP-$h_{\mc J}$ &	0.93 &	0.009 &	0.772 &	4.083 &	842.564 &	29.729 &	16.765	\\	
&	&	heavy &	MoLP-$h_{\mc P}$ &	0.916 &	0.005 &	0.749 &	4.178 &	842.633 &	30.459 &	17.791	\\	
&	&	&	MoLP-$h_{\mc J}$ &	0.916 &	0.005 &	0.748 &	4.14 &	842.547 &	29.933 &	18.818	\\	
&	&	light &	CuLP &	0.937 &	0.054 &	0.788 &	4.844 &	845.765 &	43.738 &	17.905	\\	
&	&	heavy &	CuLP &	0.926 &	0.026 &	0.77 &	4.609 &	844.302 &	40.569 &	15.738	\\	\cline{5-11}
&	0.2 &	- &	{\tt bottom.up} &	0.951 &	0.064 &	0.805 &	4.326 &	848.366 &	32.832 &	19.160	\\	
&	- &	- &	WBS.sBIC &	0.817 &	0.034 &	0.638 &	9.916 &	869.485 &	131.693 &	18.533	\\	
&	- &	- &	WBS2.SDLL &	0.91 &	0.021 &	0.735 &	4.562 &	843.571 &	35.944 &	18.304	\\	
&	- &	- &	PELT &	0.771 &	0.002 &	0.461 &	6.148 &	846.354 &	48.85 &	12.659	\\	
&	- &	- &	S3IB &	0.96 &	0.074 &	0.815 &	4.774 &	843.513 &	33.146 &	20.642	\\	
&	- &	- &	cumSeg &	0.333 &	0 &	0.273 &	25.195 &	904.25 &	752.167 &	87.473	\\	
&	- &	- &	TGUH &	0.902 &	0.026 &	0.702 &	5.374 &	845.653 &	47.727 &	30.336	\\	
&	0.2 &	- &	FDRSeg &	0.936 &	0.075 &	0.775 &	4.951 &	846.699 &	36.313 &	16.765	\\	\hline
{\tt teeth10} &	0.1 &	light &	MoLP-$h_{\mc P}$ &	0.95 &	0.001 &	0.92 &	2.337 &	-73.202 &	0.333 &	0.000	\\	
&	&	&	MoLP-$h_{\mc J}$ &	0.95 &	0.001 &	0.92 &	2.337 &	-73.202 &	0.333 &	0.000	\\	
&	&	heavy &	MoLP-$h_{\mc P}$ &	0.944 &	0 &	0.912 &	2.421 &	-73.173 &	0.362 &	0.000	\\	
&	&	&	MoLP-$h_{\mc J}$ &	0.944 &	0 &	0.912 &	2.421 &	-73.173 &	0.362 &	0.000	\\	
&	0.2 &	light &	MoLP-$h_{\mc P}$ &	0.97 &	0.001 &	0.945 &	1.986 &	-73.584 &	0.235 &	0.000	\\	
&	&	&	MoLP-$h_{\mc J}$ &	0.97 &	0.001 &	0.945 &	1.986 &	-73.584 &	0.235 &	0.000	\\	
&	&	heavy &	MoLP-$h_{\mc P}$ &	0.965 &	0.001 &	0.938 &	2.077 &	-73.552 &	0.263 &	0.000	\\	
&	&	&	MoLP-$h_{\mc J}$ &	0.965 &	0.001 &	0.938 &	2.077 &	-73.552 &	0.263 &	0.000	\\	
&	&	light &	CuLP &	0.985 &	0.017 &	0.904 &	3.65 &	-76.47 &	0.463 &	0.000	\\	
&	&	heavy &	CuLP &	0.979 &	0.011 &	0.899 &	3.702 &	-76.476 &	0.488 &	0.000	\\	\cline{5-11}
&	0.2 &	- &	{\tt bottom.up} &	0.983 &	0.004 &	0.965 &	1.813 &	-73.084 &	0.164 &	0.000	\\	
&	- &	- &	WBS.sBIC &	0.644 &	0.02 &	0.579 &	9.065 &	-71.534 &	2.029 &	1.140	\\	
&	- &	- &	WBS2.SDLL &	0.977 &	0.023 &	0.896 &	3.879 &	-76.254 &	0.501 &	0.000	\\	
&	- &	- &	PELT &	0.391 &	0.007 &	0.287 &	13.038 &	-69.041 &	3.194 &	0.342	\\	
&	- &	- &	S3IB &	0.997 &	0.101 &	0.902 &	4.039 &	-76.144 &	0.392 &	0.000	\\	
&	- &	- &	cumSeg &	0.001 &	0 &	0 &	18.287 &	-63.097 &	4.995 &	0.000	\\	
&	- &	- &	TGUH &	0.961 &	0.018 &	0.867 &	4.385 &	-75.328 &	0.631 &	0.000	\\	
&	0.2 &	- &	FDRSeg &	0.958 &	0.061 &	0.859 &	4.511 &	-75.015 &	0.623 &	0.000	\\	\hline	
{\tt stairs10} &	0.1 &	light &	MoLP-$h_{\mc P}$ &	0.998 &	0.002 &	0.979 &	2.097 &	-120.634 &	0.103 &	0.000	\\	
&	&	&	MoLP-$h_{\mc J}$ &	0.998 &	0.002 &	0.979 &	2.097 &	-120.634 &	0.103 &	0.000	\\	
&	&	heavy &	MoLP-$h_{\mc P}$ &	0.998 &	0.001 &	0.979 &	2.091 &	-120.63 &	0.103 &	0.000	\\	
&	&	&	MoLP-$h_{\mc J}$ &	0.998 &	0.001 &	0.979 &	2.096 &	-120.629 &	0.103 &	0.000	\\	
&	0.2 &	light &	MoLP-$h_{\mc P}$ &	0.998 &	0.002 &	0.979 &	2.097 &	-120.634 &	0.103 &	0.000	\\	
&	&	&	MoLP-$h_{\mc J}$ &	0.998 &	0.002 &	0.979 &	2.097 &	-120.634 &	0.103 &	0.000	\\	
&	&	heavy &	MoLP-$h_{\mc P}$ &	0.998 &	0.001 &	0.979 &	2.091 &	-120.63 &	0.103 &	0.000	\\	
&	&	&	MoLP-$h_{\mc J}$ &	0.998 &	0.001 &	0.979 &	2.096 &	-120.629 &	0.103 &	0.000	\\	
&	&	light &	CuLP &	0.999 &	0.021 &	0.961 &	2.924 &	-120.123 &	0.172 &	0.000	\\	
&	&	heavy &	CuLP &	0.999 &	0.012 &	0.963 &	2.874 &	-120.195 &	0.174 &	0.000	\\	\cline{5-11}
&	0.2 &	- &	{\tt bottom.up} &	0.997 &	0.005 &	0.978 &	2.094 &	-119.81 &	0.104 &	0.000	\\	
&	- &	- &	WBS.sBIC &	1 &	0.034 &	0.959 &	2.95 &	-120.052 &	0.165 &	0.000	\\	
&	- &	- &	WBS2.SDLL &	0.998 &	0.014 &	0.958 &	3.085 &	-119.333 &	0.196 &	0.000	\\	
&	- &	- &	PELT &	0.993 &	0.001 &	0.966 &	2.729 &	-120.597 &	0.175 &	0.000	\\	
&	- &	- &	S3IB &	1 &	0.09 &	0.953 &	3.165 &	-120.419 &	0.134 &	0.000	\\	
&	- &	- &	cumSeg &	0.986 &	0.006 &	0.878 &	7.533 &	-95.203 &	0.639 &	0.424	\\	
&	- &	- &	TGUH &	0.999 &	0.009 &	0.963 &	2.93 &	-120.113 &	0.178 &	0.000	\\	
&	0.2 &	- &	FDRSeg &	1 &	0.059 &	0.957 &	3.013 &	-119.874 &	0.146 &	0.000	\\	\hline
\hline
\end{longtable}}

{\small
\setlength{\tabcolsep}{3pt}
\begin{longtable}{c|c|c|c|c|c|c|c|c|cc|c}
\caption{Summary of change point estimation over $1000$ realisations for the 
test signals with {\bf dense} change points and Gaussian errors;
we use $h = h_{\mc J}$ and $\xi_n \in \{\log^{1.01}(n), \log^{1.1}(n)\}$ for the localised pruning.}
\label{supp:table:sim:dense} 
\endfirsthead
\endhead
\hline\hline
model & $\alpha$/$C$ & penalty &	method &	TPR &	FPR &	ARI &	MSE &	BIC &	$\delta_{\trim}$ &	$v_{\trim}$ &	speed	\\	\hline
{\tt blocks} &	0.2 &	light &	MoLP &	0.935 &	0.005 &	0.981 &	5.011 &	48093.8 &	2.313 &	220.54 &	0.660	\\	
&	0.2 &	heavy &	MoLP &	0.91 &	0.001 &	0.979 &	5.473 &	48095.6 &	2.743 &	235.905 &	0.675	\\	
&	0.4 &	light &	MoLP &	0.937 &	0.007 &	0.98 &	5.051 &	48098.09 &	2.308 &	218.629 &	0.778	\\	
&	0.4 &	heavy &	MoLP &	0.912 &	0.002 &	0.979 &	5.472 &	48098.12 &	2.674 &	228.336 &	0.808	\\	
&	0.5 &	light &	CuLP &	0.863 &	0.018 &	0.899 &	15.378 &	48513.02 &	9.978 &	290.71 &	16.819	\\	
&	0.5 &	heavy &	CuLP &	0.873 &	0.003 &	0.934 &	11.312 &	48333.58 &	7.139 &	246.124 &	15.557	\\	
&	0.9 &	light &	CuLP &	0.933 &	0.007 &	0.977 &	5.408 &	48105.5 &	2.528 &	200.072 &	4.978	\\	
&	0.9 &	heavy &	CuLP &	0.904 &	0.002 &	0.97 &	6.285 &	48124.13 &	3.362 &	211.04 &	5.090	\\	\cline{5-12}
&	0.2 &	- &	{\tt bottom.up} &	0.914 &	0.206 &	0.887 &	6.468 &	48360.28 &	2.76 &	274.792 &	0.050	\\	
&	- &	- &	WBS.sBIC &	0.908 &	0.034 &	0.955 &	7.892 &	48242.17 &	4.946 &	212.22 &	77.772	\\	
&	- &	- &	WBS2.SDLL &	0.951 &	0.063 &	0.962 &	5.319 &	48143.77 &	2.101 &	204.006 &	5.028	\\	
&	- &	- &	PELT &	0.81 &	0 &	0.955 &	8.098 &	48128.61 &	6.336 &	293.799 &	0.029	\\	
&	- &	- &	TGUH &	0.919 &	0.005 &	0.974 &	6.32 &	48131.39 &	3.241 &	287.898 &	1.497	\\	
&	- &	- &	FPOP &	0.931 &	0.002 &	0.983 &	4.782 &	48076.92 &	2.331 &	198.34 &	0.010	\\	\hline
{\tt fms} &	0.2 &	light &	MoLP &	0.98 &	0.003 &	0.97 &	4.222 &	-22251.66 &	0.539 &	0.124 &	1.186	\\	
&	0.2 &	heavy &	MoLP &	0.968 &	0.001 &	0.968 &	4.563 &	-22248.33 &	0.676 &	0.155 &	1.187	\\	
&	0.4 &	light &	MoLP &	0.985 &	0.003 &	0.971 &	4.032 &	-22257.81 &	0.48 &	0.093 &	1.301	\\	
&	0.4 &	heavy &	MoLP &	0.973 &	0.001 &	0.969 &	4.392 &	-22253.62 &	0.623 &	0.142 &	1.325	\\	
&	0.5 &	light &	CuLP &	0.973 &	0.006 &	0.962 &	4.291 &	-22238.88 &	0.533 &	0.055 &	4.571	\\	
&	0.5 &	heavy &	CuLP &	0.941 &	0.001 &	0.95 &	5.403 &	-22175.93 &	0.823 &	0.086 &	4.568	\\	
&	0.9 &	light &	CuLP &	0.986 &	0.002 &	0.976 &	3.625 &	-22300.42 &	0.41 &	0.046 &	3.992	\\	
&	0.9 &	heavy &	CuLP &	0.973 &	0.001 &	0.973 &	3.988 &	-22294.87 &	0.547 &	0.049 &	3.990	\\	\cline{5-12}
&	0.2 &	- &	{\tt bottom.up} &	0.906 &	0.28 &	0.728 &	7.563 &	-21260.04 &	1.615 &	0.392 &	0.054	\\	
&	- &	- &	WBS.sBIC &	0.923 &	0.004 &	0.965 &	15.452 &	-21006.43 &	1.892 &	0.079 &	74.857	\\	
&	- &	- &	WBS2.SDLL &	0.997 &	0.017 &	0.973 &	3.605 &	-22278.99 &	0.321 &	0.053 &	6.443	\\	
&	- &	- &	PELT &	0.74 &	0 &	0.945 &	9.446 &	-22159.09 &	2.756 &	0.057 &	0.026	\\	
&	- &	- &	TGUH &	0.986 &	0.002 &	0.963 &	4.099 &	-22258.01 &	0.459 &	0.103 &	1.433	\\	
&	- &	- &	FPOP &	0.958 &	0.001 &	0.977 &	3.859 &	-22335.57 &	0.627 &	0.033 &	0.010	\\	\hline
{\tt mix} &	0.2 &	light &	MoLP &	0.879 &	0.002 &	0.678 &	4.211 &	31852.09 &	0.786 &	11.113 &	1.004	\\	
&	0.2 &	heavy &	MoLP &	0.852 &	0.001 &	0.634 &	4.431 &	31835.61 &	0.871 &	12.482 &	1.014	\\	
&	0.4 &	light &	MoLP &	0.887 &	0.002 &	0.695 &	4.154 &	31863.58 &	0.76 &	10.792 &	1.105	\\	
&	0.4 &	heavy &	MoLP &	0.858 &	0.001 &	0.647 &	4.401 &	31843.89 &	0.855 &	12.38 &	1.134	\\	
&	0.5 &	light &	CuLP &	0.906 &	0.005 &	0.733 &	4.364 &	31905.35 &	0.786 &	19.431 &	8.807	\\	
&	0.5 &	heavy &	CuLP &	0.868 &	0.003 &	0.667 &	4.779 &	31878.13 &	0.941 &	11.417 &	9.170	\\	
&	0.9 &	light &	CuLP &	0.837 &	0.003 &	0.631 &	6.9 &	32114.51 &	1.671 &	12.574 &	7.334	\\	
&	0.9 &	heavy &	CuLP &	0.739 &	0.002 &	0.511 &	11.338 &	32547.75 &	3.198 &	21.751 &	11.187	\\	\cline{5-12}
&	0.2 &	- &	{\tt bottom.up} &	0.887 &	0.025 &	0.705 &	4.385 &	32016.71 &	0.784 &	21.159 &	0.060	\\	
&	- &	- &	WBS.sBIC &	0.676 &	0 &	0.464 &	11.937 &	32905.97 &	3.445 &	22.272 &	73.173	\\	
&	- &	- &	WBS2.SDLL &	0.908 &	0.013 &	0.735 &	4.338 &	31917.1 &	0.751 &	18.889 &	9.400	\\	
&	- &	- &	PELT &	0.625 &	0 &	0.34 &	9.266 &	32013.18 &	2.766 &	5.743 &	0.032	\\	
&	- &	- &	TGUH &	0.821 &	0.002 &	0.564 &	5.593 &	31877.43 &	1.212 &	29.959 &	1.420	\\	
&	- &	- &	FPOP &	0.803 &	0.002 &	0.541 &	5.105 &	31766.6 &	1.142 &	31.329 &	0.011	\\	\hline
{\tt teeth10} &	0.2 &	light &	MoLP &	0.784 &	0 &	0.694 &	5.785 &	-2301.625 &	1.187 &	0 &	1.068	\\	
&	0.2 &	heavy &	MoLP &	0.592 &	0 &	0.492 &	9.512 &	-3259.876 &	2.177 &	0.047 &	1.054	\\	
&	0.4 &	light &	MoLP &	0.821 &	0 &	0.743 &	5.124 &	-2138.238 &	1.004 &	0 &	1.121	\\	
&	0.4 &	heavy &	MoLP &	0.639 &	0 &	0.542 &	8.704 &	-2987.338 &	1.936 &	0 &	1.108	\\	
&	0.5 &	light &	CuLP &	0.903 &	0.004 &	0.814 &	4.904 &	-1962.88 &	0.84 &	0 &	4.319	\\	
&	0.5 &	heavy &	CuLP &	0.751 &	0.002 &	0.631 &	7.505 &	-2478.257 &	1.56 &	0 &	4.349	\\	
&	0.9 &	light &	CuLP &	0.688 &	0.003 &	0.532 &	8.439 &	-2617.247 &	1.885 &	0.414 &	4.156	\\	
&	0.9 &	heavy &	CuLP &	0.438 &	0.001 &	0.309 &	12.302 &	-4266.959 &	3.111 &	0.234 &	4.185	\\	\cline{5-12}
&	0.2 &	- &	{\tt bottom.up} &	0.847 &	0 &	0.799 &	5.01 &	-1478.868 &	0.879 &	0 &	0.091	\\	
&	- &	- &	WBS.sBIC &	0 &	0 &	0 &	16.854 &	-8928.946 &	5.382 &	0 &	69.941	\\	
&	- &	- &	WBS2.SDLL &	0.932 &	0.005 &	0.859 &	4.612 &	-1831.672 &	0.712 &	0 &	10.689	\\	
&	- &	- &	PELT &	0 &	0.001 &	0 &	16.851 &	-8927.371 &	5.379 &	0 &	0.018	\\	
&	- &	- &	TGUH &	0.594 &	0.003 &	0.329 &	9.436 &	-3818.443 &	2.331 &	1.594 &	1.389	\\	
&	- &	- &	FPOP &	0.114 &	0.004 &	0.036 &	15.588 &	-7955.71 &	4.747 &	0 &	0.012	\\	
\hline
{\tt stairs10} &	0.2 &	light &	MoLP &	0.998 &	0 &	0.977 &	3.103 &	-6871.387 &	0.13 &	0 &	8.376	\\	
&	0.2 &	heavy &	MoLP &	0.997 &	0 &	0.976 &	2.339 &	-6902.29 &	0.139 &	0 &	8.566	\\	
&	0.4 &	light &	MoLP &	0.998 &	0 &	0.977 &	3.103 &	-6871.408 &	0.13 &	0 &	8.416	\\	
&	0.4 &	heavy &	MoLP &	0.997 &	0 &	0.976 &	2.339 &	-6902.319 &	0.139 &	0 &	8.574	\\	
&	0.5 &	light &	CuLP &	0.989 &	0.001 &	0.948 &	3.966 &	-6282.889 &	0.311 &	0 &	9.660	\\	
&	0.5 &	heavy &	CuLP &	0.983 &	0.001 &	0.944 &	4.222 &	-6265.687 &	0.34 &	0 &	9.607	\\	
&	0.9 &	light &	CuLP &	0.988 &	0.001 &	0.948 &	3.982 &	-6279.911 &	0.313 &	0 &	9.768	\\	
&	0.9 &	heavy &	CuLP &	0.982 &	0.001 &	0.944 &	4.239 &	-6262.338 &	0.342 &	0 &	9.777	\\	\cline{5-12}
&	0.2 &	- &	{\tt bottom.up} &	0.994 &	0.004 &	0.976 &	2.243 &	-6683.881 &	0.126 &	0 &	0.170	\\	
&	- &	- &	WBS.sBIC &	0.985 &	0.011 &	0.946 &	4.016 &	-5946.003 &	0.303 &	0 &	70.928	\\	
&	- &	- &	WBS2.SDLL &	0.992 &	0.009 &	0.95 &	3.775 &	-6354.621 &	0.293 &	0 &	7.340	\\	
&	- &	- &	PELT &	0.87 &	0 &	0.866 &	9.5 &	-5624.565 &	0.888 &	0 &	0.028	\\	
&	- &	- &	TGUH &	0.991 &	0 &	0.961 &	3.165 &	-6725.094 &	0.225 &	0 &	1.377	\\	
&	- &	- &	FPOP &	0.99 &	0 &	0.968 &	2.803 &	-6892.929 &	0.189 &	0 &	0.010	\\	
\hline\hline
\end{longtable}}

{\small
\setlength{\tabcolsep}{3pt}
\begin{longtable}{c|c|c|c|c|c|c|c|cc|c}
\caption{Summary of change point estimation over $1000$ realisations for the 
test signals with {\bf sparse} change points and Gaussian errors;
we set $\alpha = 0.2$ (for MoLP and {\tt bottom.up}) and $C_\zeta = 0.9$ for CuLP,
and use $h = h_{\mc J}$ and $\xi_n \in \{\log^{1.01}(n), \log^{1.1}(n)\}$ for the localised pruning.}
\label{supp:table:sim:sparse} 
\endfirsthead
\endhead
\hline\hline
model &	penalty & method &	TPR &	FPR &	ARI &	MSE &	BIC &	$\delta_{\trim}$ &	$v_{\trim}$ &	speed	\\	\hline
{\tt blocks} &	light &	MoLP &	0.93 &	0.019 &	0.928 &	5.501 &	46,137.59 &	2.361 &	204.58 &	0.262	\\	
&	heavy &	MoLP &	0.906 &	0.004 &	0.986 &	5.909 &	46,137.37 &	2.882 &	204.58 &	0.264	\\	
&	light &	CuLP &	0.936 &	0.043 &	0.862 &	5.85 &	46,139.66 &	2.25 &	178.918 &	1.399	\\	
&	heavy &	CuLP &	0.913 &	0.006 &	0.977 &	5.796 &	46,137.67 &	2.561 &	178.918 &	1.435	\\	\cline{4-11}
&	- &	{\tt bottom.up} &	0.918 &	0.454 &	0.146 &	8.264 &	46,218.74 &	2.831 &	262.171 &	0.042	\\	
&	- &	WBS.sBIC &	0.91 &	0.004 &	0.998 &	6.512 &	46,141.24 &	4.047 &	178.918 &	62.783	\\	
&	- &	WBS2.SDLL &	0.915 &	0.018 &	0.945 &	5.982 &	46,139.19 &	2.38 &	178.918 &	8.163	\\	
&	- &	PELT &	0.811 &	0.001 &	0.999 &	8.588 &	46,140.96 &	6.272 &	615.151 &	0.022	\\	
&	- &	TGUH &	0.92 &	0.007 &	0.998 &	6.849 &	46,141.3 &	3.236 &	324.322 &	1.361	\\	
&	- &	FPOP &	0.931 &	0.002 &	0.999 &	5.1 &	46,135.74 &	2.336 &	188.331 &	0.013	\\	\hline
{\tt fms} &	light &	MoLP &	0.954 &	0.031 &	0.907 &	5.569 &	-24,026.74 &	0.805 &	0.151 &	0.272	\\	
&	heavy &	MoLP &	0.941 &	0.007 &	0.981 &	5.527 &	-24,027.52 &	0.949 &	0.151 &	0.272	\\	
&	light &	CuLP &	0.984 &	0.063 &	0.801 &	5.205 &	-24,025.83 &	0.465 &	0.033 &	1.145	\\	
&	heavy &	CuLP &	0.969 &	0.008 &	0.973 &	4.658 &	-24,028.14 &	0.618 &	0.033 &	1.295	\\	\cline{4-11}
&	- &	{\tt bottom.up} &	0.909 &	0.635 &	0.063 &	11.633 &	-23,932.78 &	1.595 &	0.372 &	0.047	\\	
&	- &	WBS.sBIC &	0.757 &	0.013 &	0.933 &	50.855 &	-23,928.6 &	4.764 &	0.338 &	70.233	\\	
&	- &	WBS2.SDLL &	0.978 &	0.023 &	0.92 &	4.857 &	-24,026.59 &	0.552 &	0.103 &	9.007	\\	
&	- &	PELT &	0.751 &	0 &	0.999 &	10.924 &	-24,025.52 &	2.722 &	0.068 &	0.025	\\	
&	- &	TGUH &	0.965 &	0.002 &	0.995 &	5.548 &	-24,024.01 &	0.652 &	0.136 &	1.560	\\	
&	- &	FPOP &	0.96 &	0.001 &	0.998 &	4.066 &	-24,029.76 &	0.619 &	0.033 &	0.013	\\	\hline
{\tt mix} &	light &	MoLP &	0.885 &	0.015 &	0.883 &	4.724 &	27,834.89 &	0.797 &	12.431 &	0.218	\\	
&	heavy &	MoLP &	0.862 &	0.004 &	0.922 &	4.934 &	27,834.06 &	0.882 &	12.431 &	0.218	\\	
&	light &	CuLP &	0.879 &	0.041 &	0.763 &	5.434 &	27,837.56 &	0.906 &	12.431 &	2.057	\\	
&	heavy &	CuLP &	0.845 &	0.006 &	0.901 &	5.554 &	27,834.88 &	1.04 &	14.256 &	2.759	\\	\cline{4-11}
&	- &	{\tt bottom.up} &	0.905 &	0.372 &	0.063 &	6.041 &	27,902.7 &	0.781 &	16.252 &	0.043	\\	
&	- &	WBS.sBIC &	0.638 &	0.002 &	0.863 &	14.56 &	27,873.07 &	3.43 &	21.099 &	62.190	\\	
&	- &	WBS2.SDLL &	0.847 &	0.021 &	0.835 &	6.169 &	27,840.44 &	1.087 &	16.081 &	8.049	\\	
&	- &	PELT &	0.665 &	0 &	0.844 &	9.118 &	27,839.65 &	2.255 &	13.8 &	0.039	\\	
&	- &	TGUH &	0.535 &	0.006 &	0.682 &	48.116 &	28,025.52 &	12.786 &	35.24 &	1.377	\\	
&	- &	FPOP &	0.834 &	0.002 &	0.922 &	5.133 &	27,832.15 &	0.992 &	12.431 &	0.017	\\	\hline
{\tt teeth10} &	light &	MoLP &	0.738 &	0.02 &	0.914 &	6.857 &	-18,210.56 &	1.745 &	0 &	0.213	\\	
&	heavy &	MoLP &	0.639 &	0.006 &	0.981 &	8.82 &	-18,212.2 &	2.601 &	0.114 &	0.217	\\	
&	light &	CuLP &	0.783 &	0.045 &	0.788 &	7.588 &	-18,207.1 &	1.848 &	0 &	0.981	\\	
&	heavy &	CuLP &	0.671 &	0.011 &	0.957 &	9.464 &	-18,210.74 &	3.091 &	1.026 &	1.294	\\	\cline{4-11}
&	- &	{\tt bottom.up} &	0.848 &	0.386 &	0.048 &	6.646 &	-18,139.38 &	2.621 &	0 &	0.041	\\	
&	- &	WBS.sBIC &	0.42 &	0.006 &	0.995 &	13.294 &	-18,215.21 &	3.806 &	0 &	63.106	\\	
&	- &	WBS2.SDLL &	0.825 &	0.038 &	0.811 &	6.99 &	-18,206.21 &	1.483 &	0 &	7.996	\\	
&	- &	PELT &	0.164 &	0.015 &	0.999 &	16.11 &	-18,213.93 &	4.682 &	0 &	0.022	\\	
&	- &	TGUH &	0.8 &	0.007 &	0.994 &	6.922 &	-18,210.9 &	1.652 &	0 &	1.402	\\	
&	- &	FPOP &	0.444 &	0.01 &	0.999 &	12.511 &	-18,219.81 &	3.354 &	0.342 &	0.013	\\	\hline
{\tt stairs10} &	light &	MoLP &	0.996 &	0.016 &	0.898 &	2.599 &	-23,948.07 &	0.139 &	0 &	0.227	\\	
&	heavy &	MoLP &	0.989 &	0.003 &	0.973 &	2.707 &	-23,948.82 &	0.174 &	0 &	0.229	\\	
&	light &	CuLP &	0.974 &	0.035 &	0.792 &	6.56 &	-23,934.51 &	0.628 &	0 &	0.990	\\	
&	heavy &	CuLP &	0.966 &	0.006 &	0.964 &	6.536 &	-23,936.69 &	0.668 &	0 &	1.055	\\	\cline{4-11}
&	- &	{\tt bottom.up} &	0.994 &	0.32 &	0.05 &	3.375 &	-23,886.67 &	0.161 &	0 &	0.042	\\	
&	- &	WBS.sBIC &	0.988 &	0.01 &	0.987 &	4.655 &	-23,943.36 &	0.385 &	0 &	63.537	\\	
&	- &	WBS2.SDLL &	0.981 &	0.018 &	0.906 &	5.169 &	-23,941.01 &	0.524 &	0 &	8.097	\\	
&	- &	PELT &	0.955 &	0 &	1 &	4.625 &	-23,947.28 &	0.396 &	0 &	0.013	\\	
&	- &	TGUH &	0.943 &	0.002 &	0.991 &	9.605 &	-23,915.32 &	7.433 &	0 &	1.382	\\	
&	- &	FPOP &	0.998 &	0 &	1 &	2.521 &	-23,949.57 &	0.152 &	0 &	0.013	\\	\hline\hline
\end{longtable}}

\begin{figure}[htbp]
\centering
\includegraphics[width=\textwidth]{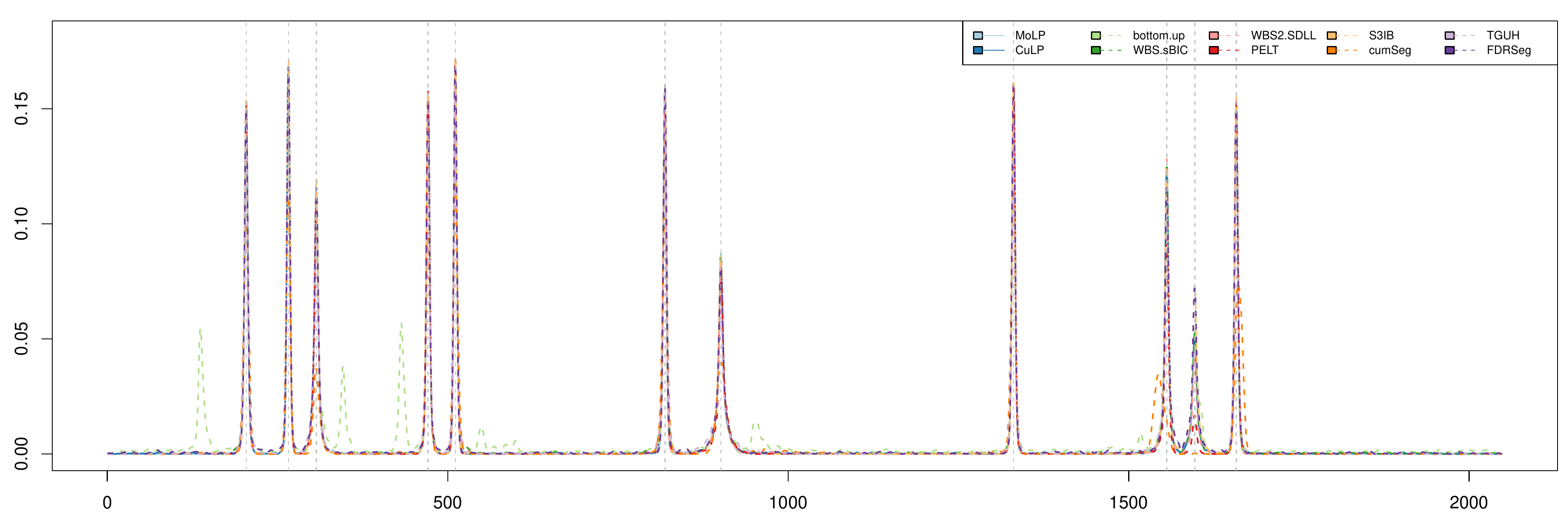}
\caption{Test signal {\tt blocks} with Gaussian errors: 
weighted density of estimated change points over $[(\cp_{j-1}+\cp_j)/2, (\cp_j+\cp_{j+1})/2]$,
$j = 1, \ldots, q_n$,
with the vertical lines indicating the locations of true change points.
We set $\alpha = 0.2$ for MoLP and {\tt bottom.up} and $C_\zeta = 0.9$ for CuLP, 
and use $h = h_{\mc J}$ and $\xi_n = \log^{1.01}(n)$ for the localised pruning.}
\label{fig:sim:blocks}
\end{figure}

\begin{figure}[htb]
\centering
\includegraphics[width=\textwidth]{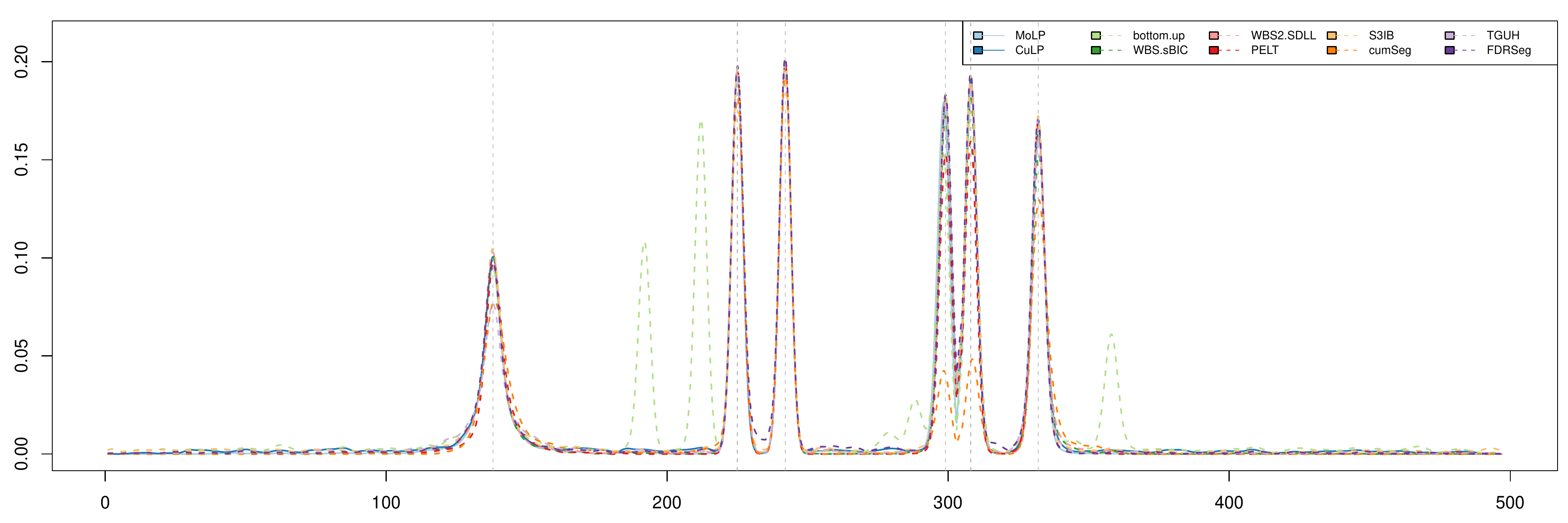}
\caption{Test signal {\tt fms} with Gaussian errors: 
weighted density of estimated change points.}
\label{fig:sim:fms}
\end{figure}

\begin{figure}[htbp]
\centering
\includegraphics[width=\textwidth]{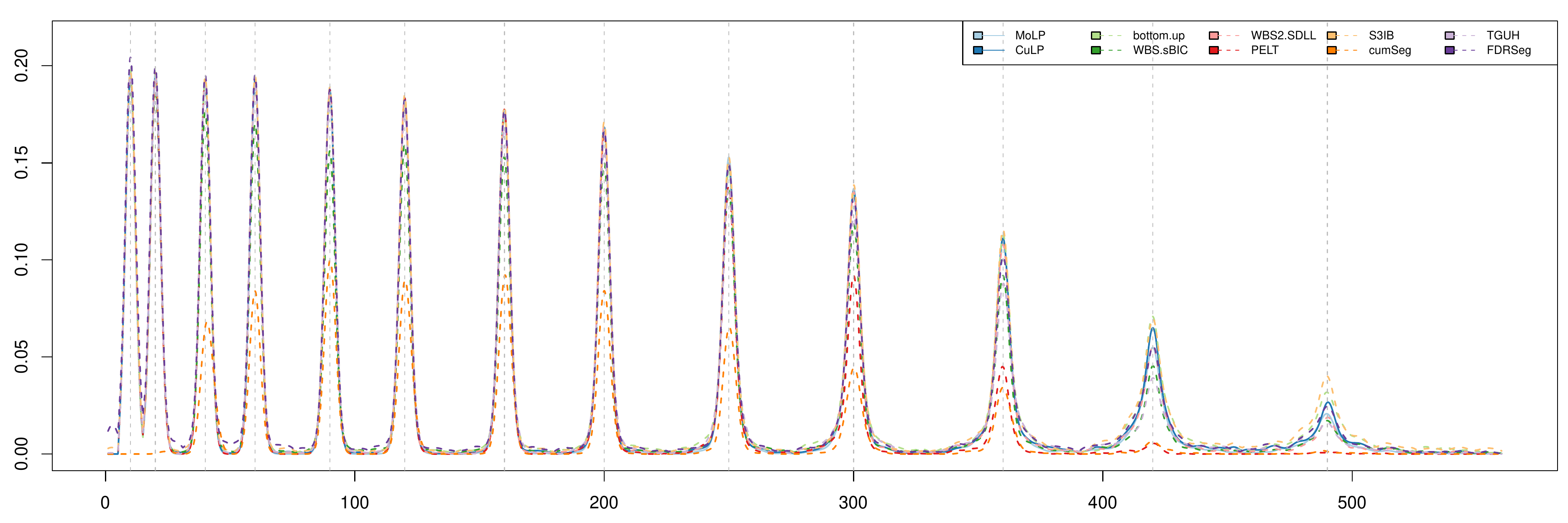}
\caption{Test signal {\tt mix} with Gaussian errors: 
weighted density of estimated change points.}
\label{fig:sim:mix}
\end{figure}

\begin{figure}[htbp]
\centering
\includegraphics[width=\textwidth]{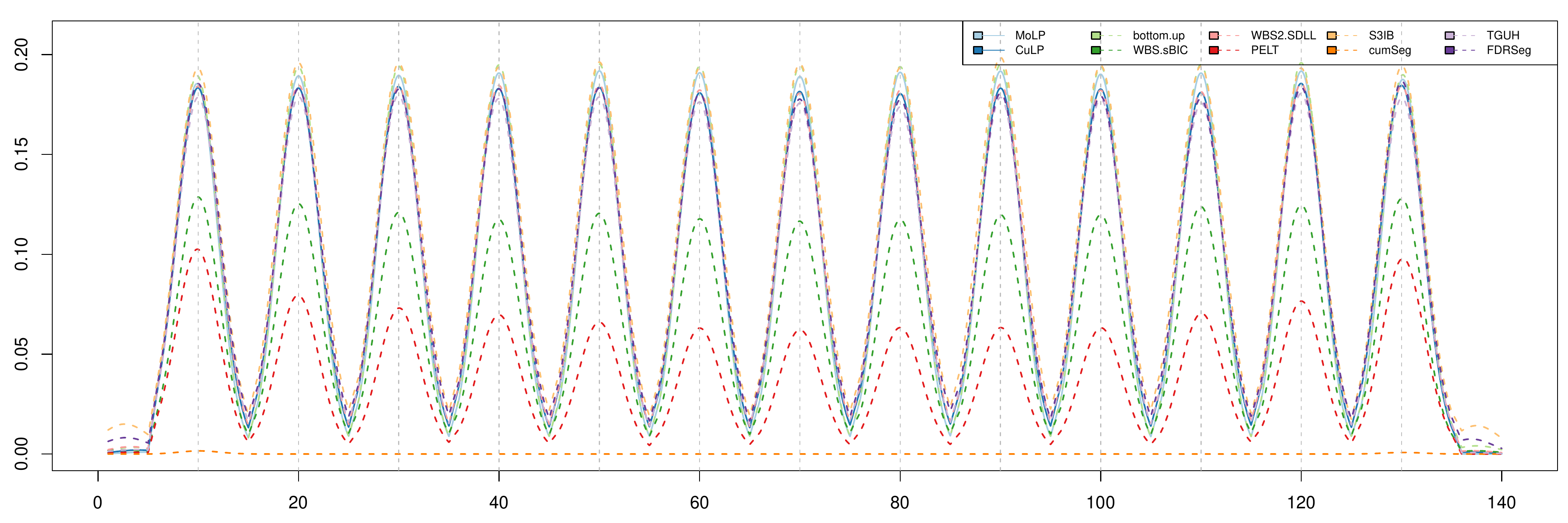}
\caption{Test signal {\tt teeth10} with Gaussian errors: 
weighted density of estimated change points.}
\label{fig:sim:teeth10}
\end{figure}

\begin{figure}[htbp]
\centering
\includegraphics[width=\textwidth]{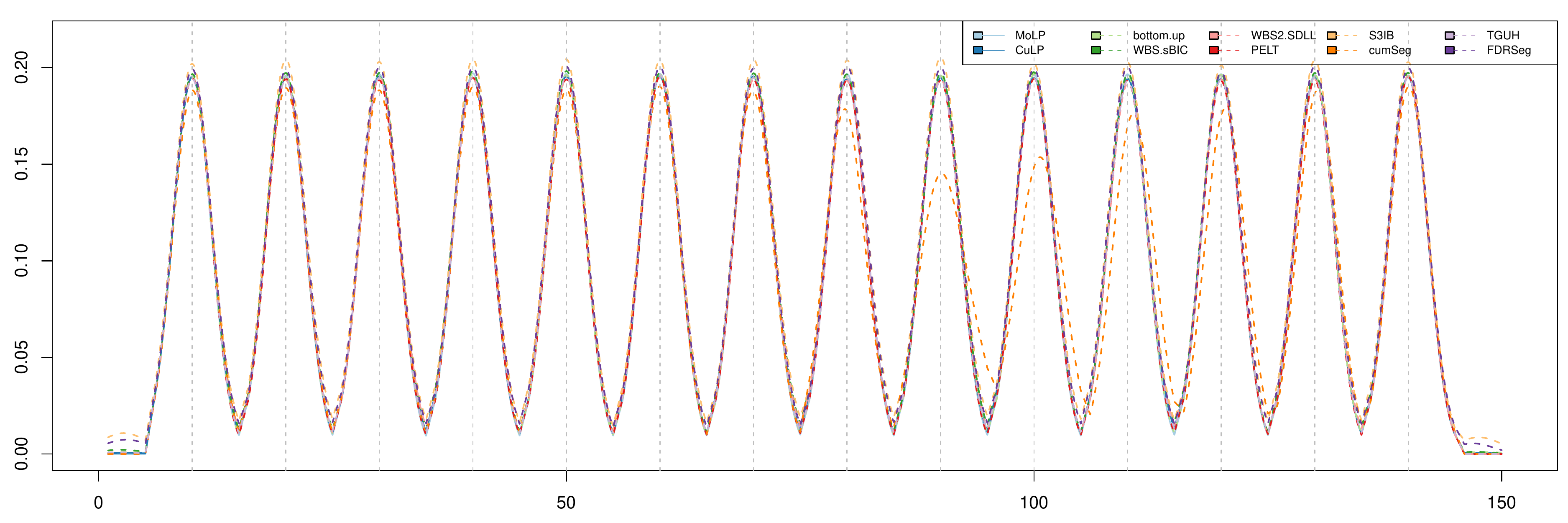}
\caption{Test signal {\tt stairs10} with Gaussian errors: 
weighted density of estimated change points.}
\label{fig:sim:stairs10}
\end{figure}

\begin{figure}[htbp]
\centering
\includegraphics[width=\textwidth]{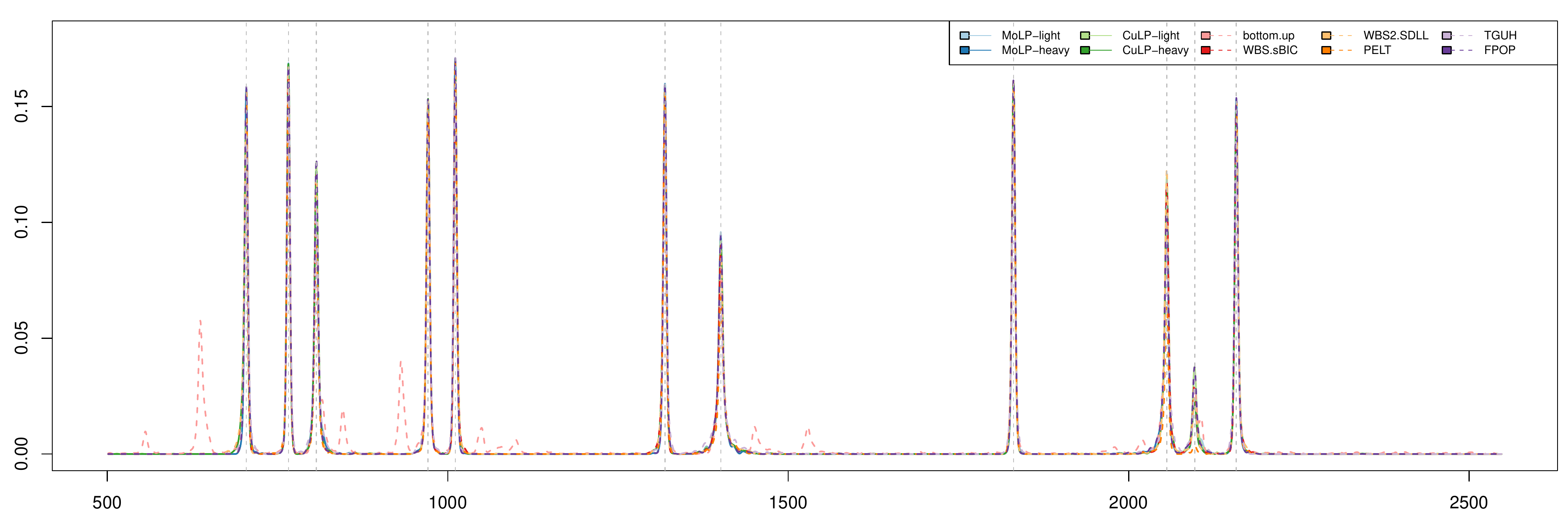}
\caption{Long test signal {\tt blocks} with sparse change points and Gaussian errors: 
weighted density of estimated change points 
with the vertical lines indicating the locations of true change points.
We set $\alpha = 0.2$ for MoLP and {\tt bottom.up} and $C_\zeta = 0.9$ for CuLP, 
and use $h = h_{\mc J}$ and $\xi_n \in \{\log^{1.01}(n), \log^{1.1}(n)\}$ for the localised pruning.}
\label{fig:sim:sparse:blocks}
\end{figure}

\begin{figure}[htbp]
\centering
\includegraphics[width=\textwidth]{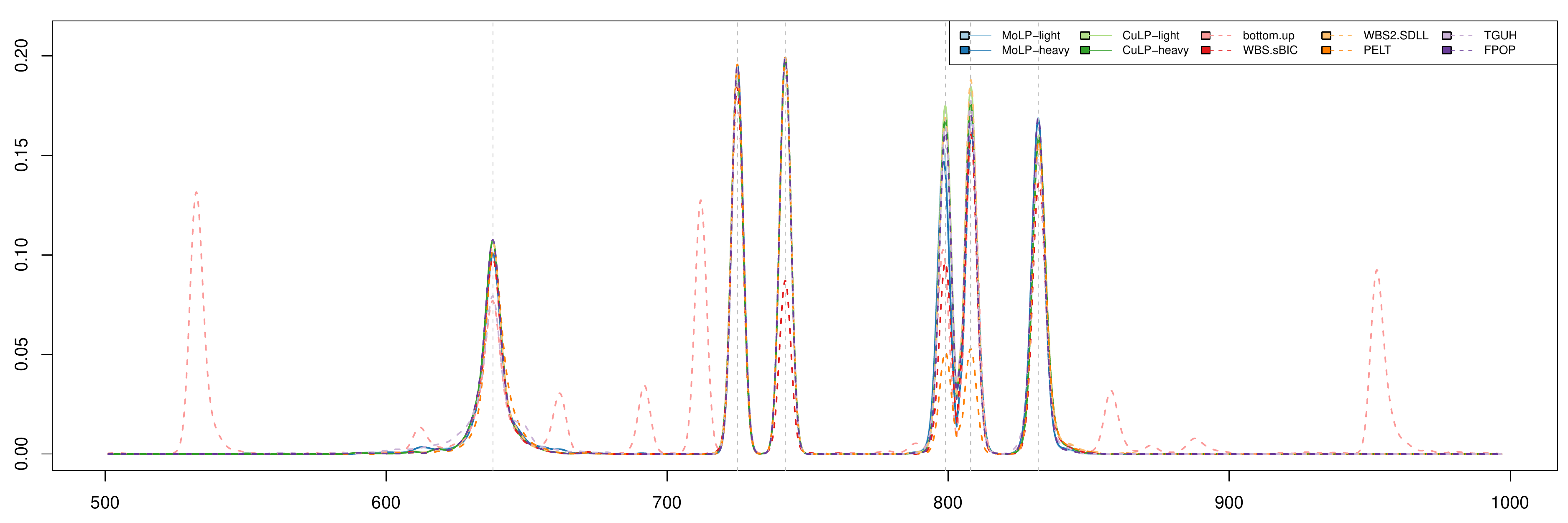}
\caption{Long test signal {\tt fms} with sparse change points and Gaussian errors: 
weighted density of estimated change points.}
\label{fig:sim:sparse:fms}
\end{figure}

\begin{figure}[htbp]
\centering
\includegraphics[width=\textwidth]{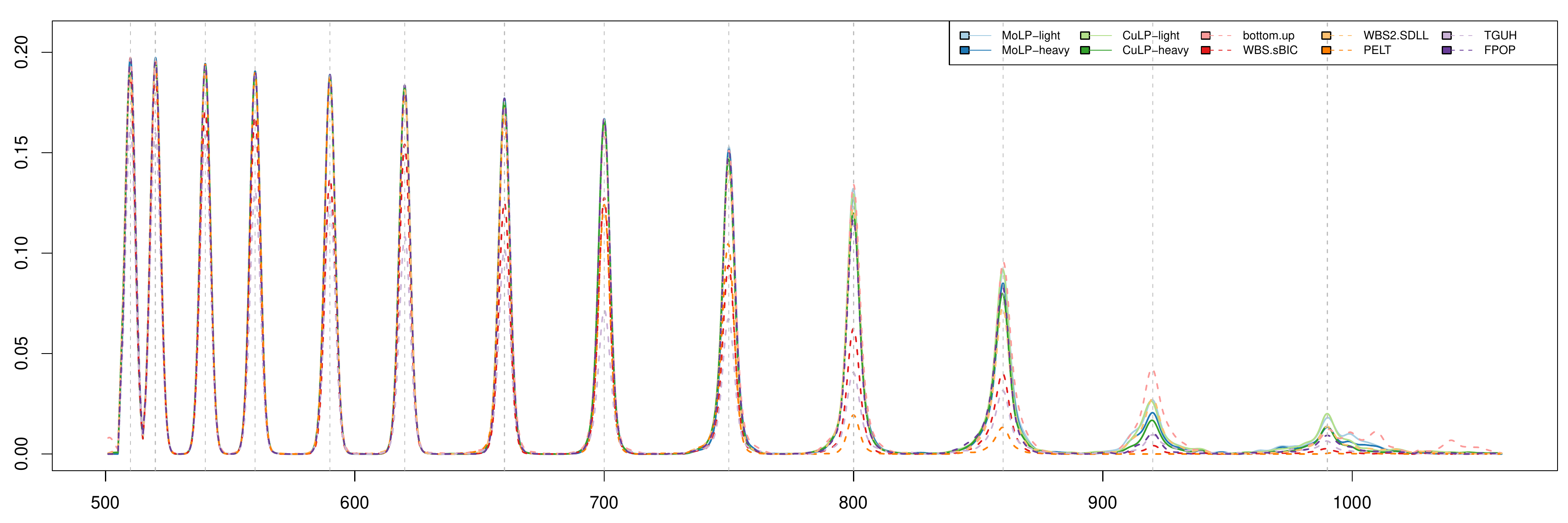}
\caption{Long test signal {\tt mix} with {\bf sparse} change points and Gaussian errors: 
weighted density of estimated change points.}
\label{fig:sim:sparse:mix}
\end{figure}

\begin{figure}[htbp]
\centering
\includegraphics[width=\textwidth]{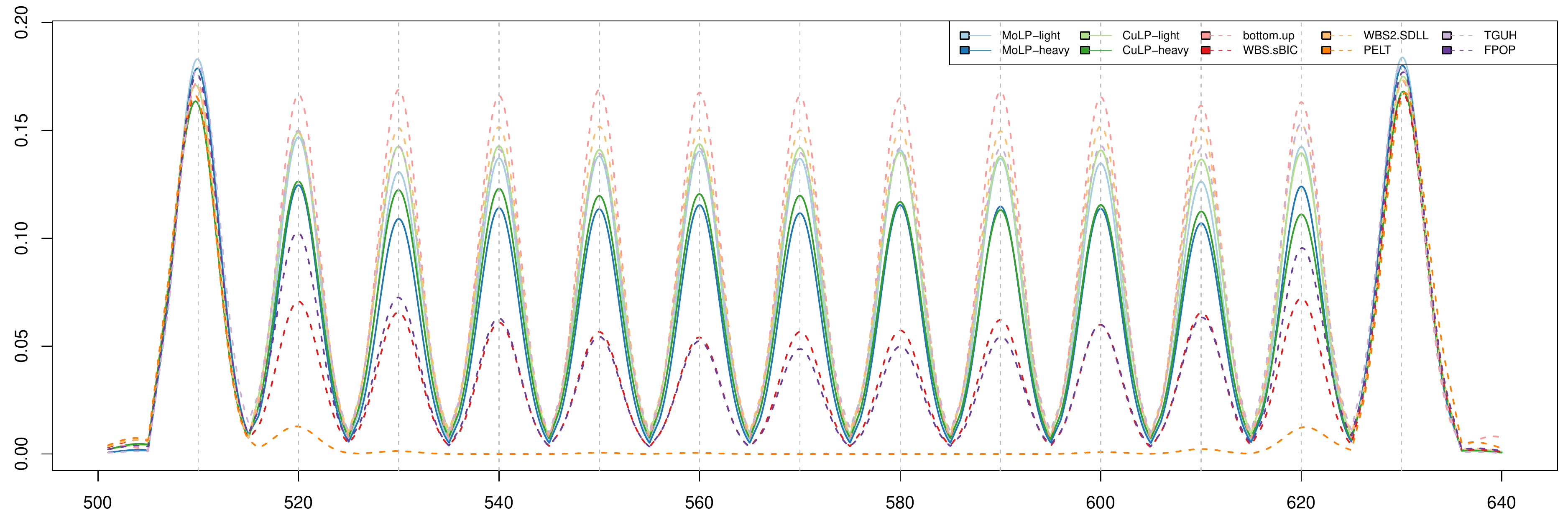}
\caption{Long test signal {\tt teeth10} with sparse change points and Gaussian errors.}
\label{fig:sim:sparse:teeth10}
\end{figure}

\begin{figure}[htbp]
\centering
\includegraphics[width=\textwidth]{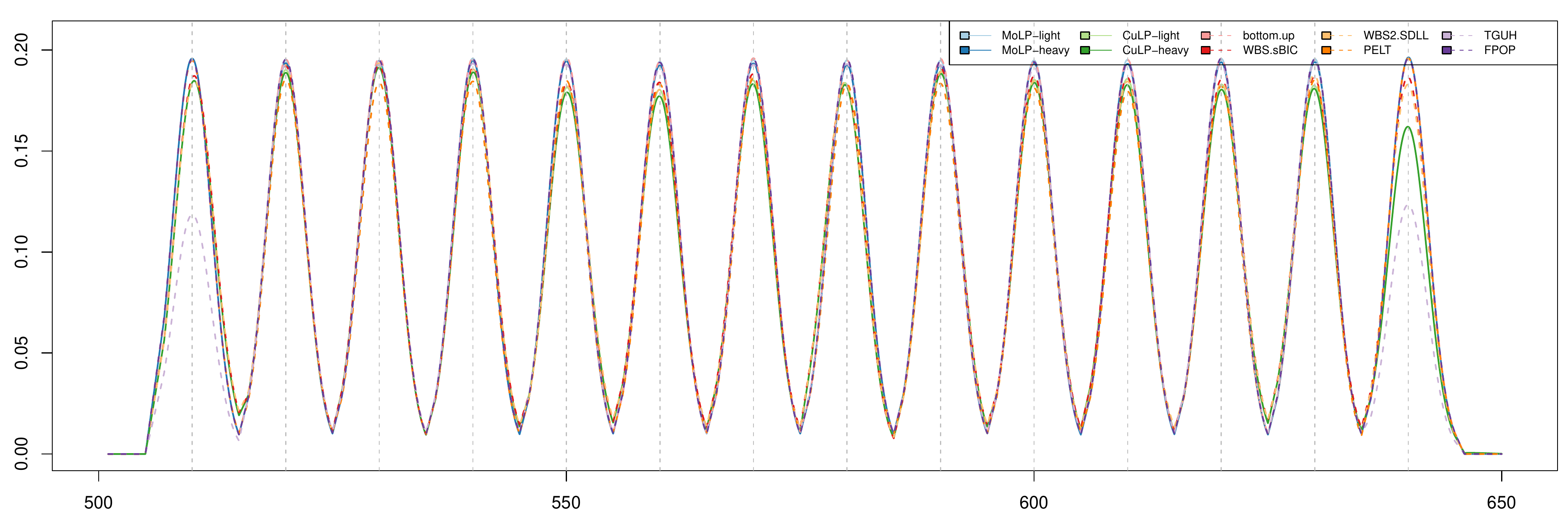}
\caption{Long test signal {\tt stairs10} with sparse change points and Gaussian errors.}
\label{fig:sim:sparse:stairs10}
\end{figure}

\pagebreak

\subsection*{\ref{eq:n2} Independent heavy-tailed errors}

Tables \ref{supp:table:sim:orig:t}--\ref{supp:table:sim:sparse:t}
report the results when $\vep_t \sim_{\iid} t_5$.
Figures~\ref{fig:sim:t:blocks}--\ref{fig:sim:t:stairs10} visualise
the performance of various methods 
by plotting the weighted densities of estimated change point(s).

The heavy penalty $\xi_n = n^{2/4.99}$ is a theoretically valid choice
conforming to Assumption~\ref{assum:penalty} in light of Remark~\ref{rem:vep:two}~(a).
When $n$ is small (Table~\ref{supp:table:sim:orig:t}),
this penalty successfully prevents false positives but 
the resulting procedure lacks power.
The light penalty $\xi_n = \log^{1.1}(n)$ works
reasonably well in not causing false positives while attaining high TPR,
yielding the performance comparable to that observed with 
Gaussian errors (see Table~\ref{supp:table:sim:orig}).
When $n$ is large and change points are sparse (Table~\ref{supp:table:sim:sparse:t}), 
the localised pruning under-estimates the number of change points 
for some test signals such as {\tt fms}, {\tt teeth10} and {\tt stairs10}.
This, in part, is due to that the candidate generating method fails to produce
at least one valid estimator for each true change point,
e.g., compare the TPR for MoLP (resp. CuLP) and {\tt bottom.up} (WBS2.SDLL),
thus failing Assumption~\ref{assum:cand}~(a).
In addition, for theoretical consistency, Assumption~\ref{assum:size} 
requires the magnitude of changes
to be larger for their detection in the presence of heavy-tailed errors 
whereas it is kept at the same level as in ($\mc E 1$) with Gaussian errors.
Most of the competitors are tailored for sub-Gaussian errors,
and they incur considerable false positives, 
a phenomenon that is amplified in Table~\ref{supp:table:sim:sparse:t}
as $n$ is large and change points sparse.

{\small
\setlength{\tabcolsep}{3pt}
\begin{longtable}{c|c|c|c|c|c|c|c|c|cc}
\caption{Summary of change point estimation over $1000$ realisations for 
the test signals with $t_5$ errors;
we set $\alpha = 0.2$ for MoLP and {\tt bottom.up} and $C_\zeta = 0.9$ for CuLP,  
and use $h = h_{\mc J}$ and $\xi_n \in \{\log^{1.1}(n), n^{2/4.99}\}$ for the localised pruning.}
\label{supp:table:sim:orig:t} 
\endfirsthead
\endhead
\hline\hline
model &	penalty &	method &	TPR &	FPR &	ARI &	MSE &	BIC &	$\delta_{\trim}$ &	$v_{\trim}$ 	\\	\hline
{\tt blocks} &	light &	MoLP &	0.948 &	0.005 &	0.98 &	4.76 &	4781.907 &	312.528 &	49.512	\\	
&	heavy &	MoLP &	0.775 &	0 &	0.928 &	10.072 &	4796.438 &	808.116 &	261.309	\\	
&	light &	CuLP &	0.943 &	0.013 &	0.974 &	5.255 &	4783.241 &	366.154 &	110.24	\\	
&	heavy &	CuLP &	0.743 &	0 &	0.908 &	11.676 &	4800.744 &	1360.405 &	208.862	\\	\cline{4-10}
&	- &	{\tt bottom.up} &	0.95 &	0.217 &	0.921 &	6.295 &	4805.876 &	390.389 &	297.751	\\	
&	- &	WBS.sBIC &	0.902 &	0.248 &	0.891 &	16.965 &	4785.934 &	852.116 &	195.167	\\	
&	- &	WBS2.SDLL &	0.974 &	0.611 &	0.694 &	28.436 &	4808.934 &	288.258 &	111.779	\\	
&	- &	PELT &	0.927 &	0.17 &	0.935 &	13.346 &	4761.393 &	398.671 &	57.462	\\	
&	- &	S3IB &	0.965 &	0.265 &	0.913 &	15.858 &	4758.723 &	303.6 &	121.246	\\	
&	- &	cumSeg &	0.768 &	0.001 &	0.915 &	13.462 &	4817.388 &	1758.934 &	592.281	\\	
&	- &	TGUH &	0.961 &	0.35 &	0.885 &	15.58 &	4794.246 &	467.44 &	353.441	\\	
&	- &	FDRSeg &	0.992 &	0.678 &	0.713 &	33.873 &	4822.753 &	312.756 &	252.471	\\	\hline
{\tt fms} &	light &	MoLP &	0.984 &	0.009 &	0.961 &	4.429 &	-567.376 &	0.204 &	0.151	\\	
&	heavy &	MoLP &	0.948 &	0.003 &	0.937 &	5.543 &	-566.093 &	0.312 &	0.168	\\	
&	light &	CuLP &	0.988 &	0.012 &	0.963 &	3.871 &	-568.467 &	0.175 &	0.033	\\	
&	heavy &	CuLP &	0.967 &	0.001 &	0.956 &	4.239 &	-568.043 &	0.249 &	0.033	\\	\cline{4-10}
&	- &	{\tt bottom.up} &	0.983 &	0.302 &	0.869 &	6.099 &	-550.932 &	0.325 &	0.237	\\	
&	- &	WBS.sBIC &	0.978 &	0.174 &	0.91 &	10.382 &	-572.498 &	0.198 &	0.033	\\	
&	- &	WBS2.SDLL &	0.999 &	0.464 &	0.749 &	17.344 &	-561.768 &	0.108 &	0.033	\\	
&	- &	PELT &	0.984 &	0.108 &	0.941 &	8.202 &	-575.528 &	0.185 &	0.033	\\	
&	- &	S3IB &	1 &	0.427 &	0.813 &	16.341 &	-572.761 &	0.111 &	0.033	\\	
&	- &	cumSeg &	0.757 &	0.012 &	0.916 &	13.595 &	-552.024 &	1.777 &	0.120	\\	
&	- &	TGUH &	0.998 &	0.306 &	0.866 &	11.119 &	-565.622 &	0.141 &	0.067	\\	
&	- &	FDRSeg &	1 &	0.523 &	0.777 &	18.759 &	-562.495 &	0.116 &	0.050	\\	\hline
{\tt mix} &	light &	MoLP &	0.916 &	0.005 &	0.753 &	3.874 &	839.603 &	30.371 &	12.659	\\	
&	heavy &	MoLP &	0.854 &	0.001 &	0.634 &	4.428 &	841.166 &	34.099 &	6.501	\\	
&	light &	CuLP &	0.905 &	0.007 &	0.736 &	4.202 &	840.603 &	38.31 &	13.686	\\	
&	heavy &	CuLP &	0.822 &	0.001 &	0.577 &	5.118 &	843.018 &	46.415 &	14.028	\\	\cline{4-10}
&	- &	{\tt bottom.up} &	0.941 &	0.034 &	0.795 &	4.37 &	845.202 &	39.285 &	19.274	\\	
&	- &	WBS.sBIC &	0.835 &	0.14 &	0.652 &	11.853 &	860.71 &	123.284 &	23.722	\\	
&	- &	WBS2.SDLL &	0.961 &	0.303 &	0.769 &	9.536 &	843.276 &	29.806 &	15.624	\\	
&	- &	PELT &	0.837 &	0.065 &	0.582 &	6.918 &	833.81 &	37.444 &	13.8	\\	
&	- &	S3IB &	0.973 &	0.261 &	0.812 &	9.019 &	833.356 &	29.335 &	13.8	\\	
&	- &	cumSeg &	0.346 &	0 &	0.286 &	23.734 &	900.692 &	742.379 &	75.442	\\	
&	- &	TGUH &	0.93 &	0.229 &	0.724 &	8.453 &	842.304 &	40.007 &	27.257	\\	
&	- &	FDRSeg &	0.971 &	0.409 &	0.775 &	11.321 &	850.546 &	32.429 &	19.274	\\	\hline
{\tt teeth10} &	light &	MoLP &	0.937 &	0.001 &	0.905 &	2.49 &	-75.183 &	0.394 &	0	\\	
&	heavy &	MoLP &	0.908 &	0 &	0.873 &	2.941 &	-74.966 &	0.531 &	0	\\	
&	light &	CuLP &	0.843 &	0.001 &	0.76 &	5.507 &	-73.049 &	1.044 &	0	\\	
&	heavy &	CuLP &	0.781 &	0.001 &	0.696 &	6.431 &	-72.824 &	1.33 &	0	\\	\cline{4-10}
&	- &	{\tt bottom.up} &	0.986 &	0.003 &	0.969 &	1.687 &	-75.091 &	0.147 &	0	\\	
&	- &	WBS.sBIC &	0.722 &	0.064 &	0.655 &	8.248 &	-74.897 &	1.63 &	0	\\	
&	- &	WBS2.SDLL &	0.988 &	0.092 &	0.903 &	4.353 &	-80.031 &	0.387 &	0	\\	
&	- &	PELT &	0.643 &	0.036 &	0.539 &	9.114 &	-75.593 &	1.994 &	0.912	\\	
&	- &	S3IB &	0.998 &	0.172 &	0.904 &	4.802 &	-81.069 &	0.317 &	0	\\	
&	- &	cumSeg &	0.001 &	0 &	0 &	18.377 &	-63.278 &	4.996 &	0	\\	
&	- &	TGUH &	0.985 &	0.082 &	0.896 &	4.628 &	-80.295 &	0.457 &	0	\\	
&	- &	FDRSeg &	0.99 &	0.187 &	0.88 &	5.184 &	-78.135 &	0.408 &	0	\\	
\hline
{\tt stairs10} &	light &	MoLP &	0.994 &	0.001 &	0.972 &	2.352 &	-122.885 &	0.14 &	0.000	\\	
&	heavy &	MoLP &	0.993 &	0.001 &	0.971 &	2.382 &	-122.851 &	0.145 &	0	\\	
&	light &	CuLP &	0.99 &	0.001 &	0.952 &	3.424 &	-120.553 &	0.248 &	0	\\	
&	heavy &	CuLP &	0.989 &	0.001 &	0.951 &	3.461 &	-120.514 &	0.255 &	0	\\	\cline{4-10}
&	- &	{\tt bottom.up} &	0.686 &	0.099 &	0.555 &	36.307 &	-44.435 &	3.031 &	2.859	\\	
&	- &	WBS.sBIC &	1 &	0.082 &	0.953 &	3.643 &	-125.895 &	0.155 &	0	\\	
&	- &	WBS2.SDLL &	0.999 &	0.07 &	0.951 &	3.682 &	-124.76 &	0.172 &	0	\\	
&	- &	PELT &	0.997 &	0.015 &	0.967 &	2.905 &	-125.858 &	0.154 &	0	\\	
&	- &	S3IB &	1 &	0.158 &	0.943 &	4.088 &	-126.209 &	0.124 &	0	\\	
&	- &	cumSeg &	0.981 &	0.007 &	0.881 &	7.23 &	-98.274 &	0.625 &	0.424	\\	
&	- &	TGUH &	0.999 &	0.074 &	0.956 &	3.636 &	-125.944 &	0.157 &	0	\\	
&	- &	FDRSeg &	1 &	0.199 &	0.929 &	4.504 &	-122.972 &	0.138 &	0	\\	\hline\hline
\end{longtable}}

{\small
\setlength{\tabcolsep}{3pt}
\begin{longtable}{c|c|c|c|c|c|c|c|cc|c}
\caption{Summary of change point estimation over $1000$ realisations for 
the test signals with {\bf sparse} change points and $t_5$ errors;
we set $\alpha = 0.2$ for MoLP and {\tt bottom.up} and $C_\zeta = 0.9$ for CuLP,  
and use $h = h_{\mc J}$ and $\xi_n \in \{\log^{1.1}(n), n^{2/4.99}\}$ for the localised pruning.}
\label{supp:table:sim:sparse:t} 
\endfirsthead
\endhead
\hline\hline
model &	penalty & method &	TPR &	FPR &	ARI &	MSE &	BIC & $\delta_{\trim}$ &	$v_{\trim}$ & 	speed 	\\	\hline
{\tt blocks} &	light &	MoLP &	1 &	0.007 &	0.979 &	2.271 &	23,128.27 &	369.841 &	0 &	0.174	\\	
&	heavy &	MoLP &	0.995 &	0 &	1 &	2.749 &	23,130.49 &	3.024 &	0 &	0.174	\\	
&	light &	CuLP &	1 &	0.02 &	0.934 &	3 &	23,129.62 &	1,047.866 &	0 &	5.167	\\	
&	heavy &	CuLP &	0.996 &	0 &	1 &	3.05 &	23,130.73 &	2.753 &	0 &	5.170	\\	\cline{4-11}
&	- &	{\tt bottom.up} &	1 &	0.383 &	0.326 &	5.17 &	23,189.64 &	11,697.4 &	0 &	0.021	\\	
&	- &	WBS2.SDLL &	1 &	0.899 &	0.016 &	155.253 &	23,393.5 &	17,507.5 &	0 &	3.735	\\	
&	- &	PELT &	1 &	0.509 &	0.156 &	53.467 &	22,995.83 &	14,698.61 &	0 &	0.007	\\	
&	- &	TGUH &	1 &	0.647 &	0.174 &	47.33 &	23,151.79 &	14,031.65 &	0 &	0.816	\\	
&	- &	FPOP &	1 &	0.755 &	0.053 &	97.408 &	22,997.7 &	16,684.8 &	0 &	0.080	\\	\hline
{\tt fms} &	light &	MoLP &	0.525 &	0.019 &	0.955 &	9.429 &	-11,996.67 &	497.164 &	0.208 &	0.174	\\	
&	heavy &	MoLP &	0.113 &	0 &	0.588 &	29.304 &	-11,959.72 &	64.733 &	0 &	0.179	\\	
&	light &	CuLP &	0.525 &	0.05 &	0.893 &	9.741 &	-11,996.39 &	1,314.255 &	0.208 &	5.075	\\	
&	heavy &	CuLP &	0.012 &	0 &	0.056 &	29.968 &	-11,962.25 &	6.072 &	0 &	7.484	\\	\cline{4-11}
&	- &	{\tt bottom.up} &	0.644 &	0.523 &	0.183 &	12.73 &	-11,953.4 &	12,939.04 &	0.865 &	0.021	\\	
&	- &	WBS2.SDLL &	0.73 &	0.958 &	0.006 &	300.3 &	-11,724.1 &	18,837.29 &	0.54 &	3.774	\\	
&	- &	PELT &	0.493 &	0.784 &	0.067 &	107.746 &	-12,130.48 &	16,040.72 &	0.243 &	0.007	\\	
&	- &	TGUH &	0.501 &	0.874 &	0.048 &	114.887 &	-11,972.02 &	15,887.4 &	1.211 &	0.813	\\	
&	- &	FPOP &	0.634 &	0.899 &	0.02 &	188.625 &	-12,128.99 &	18,015.36 &	0.634 &	0.081
\\	\hline
{\tt mix} &	light &	MoLP &	1 &	0.005 &	0.969 &	2.468 &	13,982.47 &	425.41 &	0.456 &	0.205	\\	
&	heavy &	MoLP &	0.806 &	0 &	0.908 &	14.593 &	14,033.51 &	157.413 &	16.423 &	0.234	\\	
&	light &	CuLP &	0.999 &	0.106 &	0.535 &	4.964 &	13,986.6 &	6,802.449 &	0.456 &	12.851	\\	
&	heavy &	CuLP &	0.819 &	0.001 &	0.909 &	12.7 &	14,022.57 &	201.298 &	1.825 &	19.848	\\	\cline{4-11}
&	- &	{\tt bottom.up} &	1 &	0.246 &	0.205 &	4.795 &	14,020.46 &	12,786.47 &	1.939 &	0.023	\\	
&	- &	WBS2.SDLL &	0.999 &	0.884 &	0.007 &	128.533 &	14,248.03 &	18,671.8 &	0.456 &	4.028	\\	
&	- &	PELT &	0.998 &	0.469 &	0.08 &	44.492 &	13,849.79 &	15,860.93 &	0.456 &	0.007	\\	
&	- &	TGUH &	0.951 &	0.637 &	0.072 &	52.071 &	14,027.56 &	15,578.4 &	3.307 &	0.874	\\	
&	- &	FPOP &	1 &	0.723 &	0.024 &	80.491 &	13,851.49 &	17,852.09 &	0.456 &	0.088	\\	\hline
{\tt teeth10} &	light &	MoLP &	0.165 &	0.088 &	0.963 &	7.109 &	-9,105.34 &	480.773 &	0.228 &	0.188	\\	
&	heavy &	MoLP &	0.075 &	0.028 &	0.953 &	12.442 &	-9,088.252 &	106.865 &	0 &	0.191	\\	
&	light &	CuLP &	0.177 &	0.35 &	0.493 &	9.382 &	-9,100.71 &	7,072.111 &	0.228 &	15.962	\\	
&	heavy &	CuLP &	0.077 &	0.009 &	0.954 &	12.195 &	-9,090.402 &	140.83 &	0 &	21.479	\\	\cline{4-11}
&	- &	{\tt bottom.up} &	0.286 &	0.474 &	0.161 &	8.533 &	-9,069.533 &	13,124.67 &	0.228 &	0.022	\\	
&	- &	WBS2.SDLL &	0.398 &	0.952 &	0.004 &	133.288 &	-8,820.831 &	19,022.91 &	0.228 &	3.851	\\	
&	- &	PELT &	0.143 &	0.85 &	0.054 &	49.521 &	-9,239.104 &	16,239.11 &	0.342 &	0.007	\\	
&	- &	TGUH &	0.257 &	0.853 &	0.065 &	47.879 &	-9,086.805 &	15,828 &	0.342 &	0.857	\\	
&	- &	FPOP &	0.183 &	0.934 &	0.015 &	85.667 &	-9,236.122 &	18,214.47 &	0.228 &	0.083	\\	\hline
{\tt stairs10} &	light &	MoLP &	0.59 &	0.007 &	0.966 &	10.283 &	-11,936.32 &	436.29 &	2.012 &	0.481	\\	
&	heavy &	MoLP &	0.377 &	0 &	0.993 &	23.415 &	-11,883.82 &	12.535 &	0.212 &	0.625	\\	
&	light &	CuLP &	0.599 &	0.025 &	0.894 &	10.521 &	-11,932.84 &	1,353.719 &	2.012 &	4.987	\\	
&	heavy &	CuLP &	0.396 &	0 &	0.994 &	22.688 &	-11,886.22 &	13.098 &	0.424 &	5.221	\\	\cline{4-11}
&	- &	{\tt bottom.up} &	0.585 &	0.314 &	0.162 &	20.418 &	-11,861.37 &	13,113.7 &	1.906 &	0.023	\\	
&	- &	WBS2.SDLL &	0.725 &	0.907 &	0.004 &	124.732 &	-11,662.26 &	19,024.46 &	1.059 &	4.216	\\	
&	- &	PELT &	0.591 &	0.576 &	0.055 &	49.277 &	-12,069.15 &	16,222.22 &	2.012 &	0.007	\\	
&	- &	TGUH &	0.651 &	0.68 &	0.086 &	47.761 &	-11,902.81 &	15,607.59 &	2.012 &	0.889	\\	
&	- &	FPOP &	0.691 &	0.777 &	0.016 &	80.377 &	-12,067.29 &	18,207.36 &	1.589 &	0.085	\\	\hline\hline
\end{longtable}}

\begin{figure}[htbp]
\centering
\includegraphics[width=\textwidth]{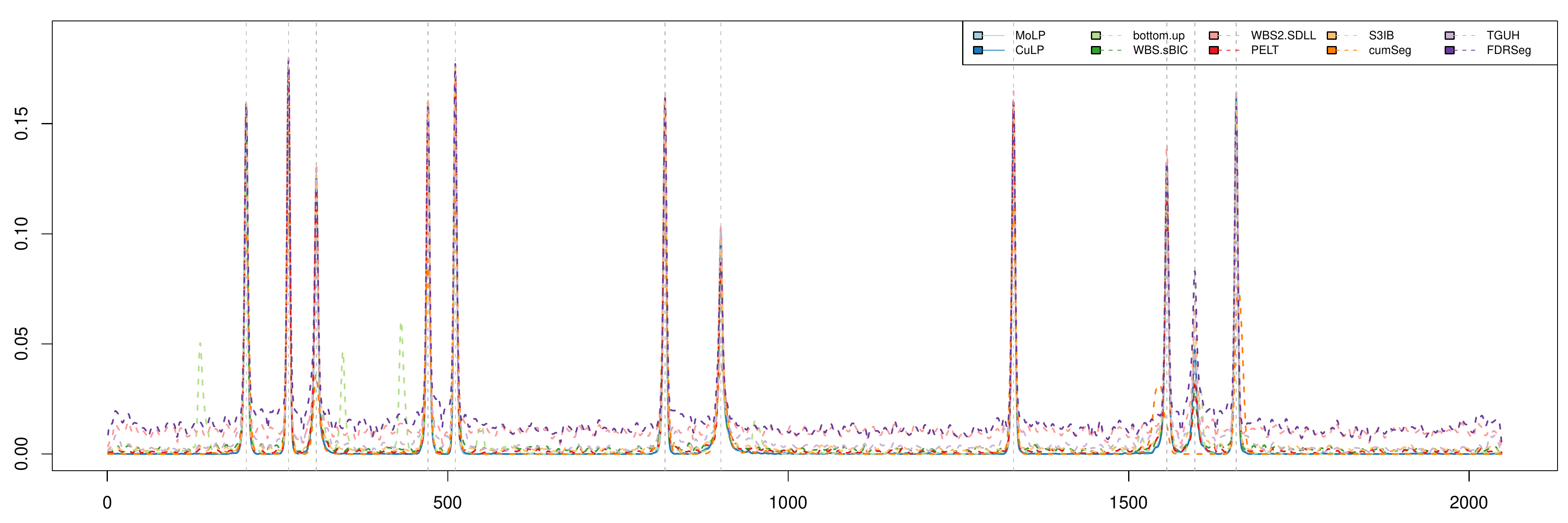}
\caption{Test signal {\tt blocks} with $t_5$ errors: 
weighted density of estimated change points
with the vertical lines indicating the locations of true change points.
We set $\alpha = 0.2$ for MoLP and {\tt bottom.up} and $C_\zeta = 0.9$ for CuLP,  and use $h = h_{\mc J}$ and $\xi_n = \log^{1.1}(n)$ for the localised pruning.}
\label{fig:sim:t:blocks}
\end{figure}

\begin{figure}[htbp]
\centering
\includegraphics[width=\textwidth]{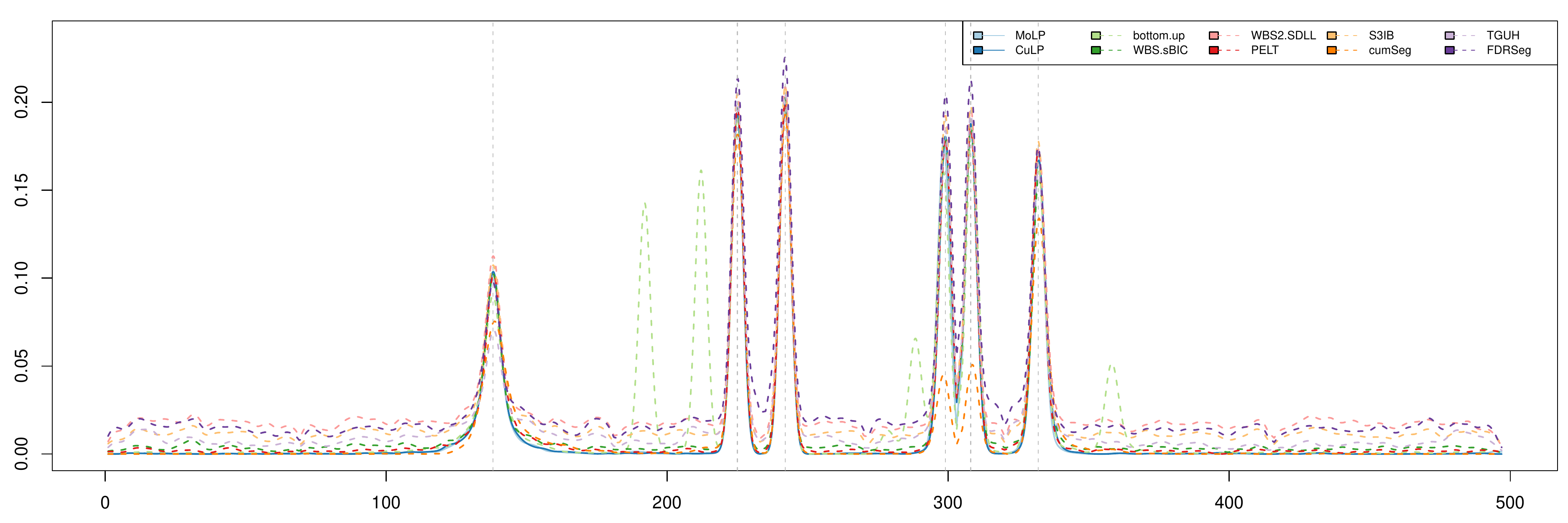}
\caption{Test signal {\tt fms} with $t_5$ errors: weighted density of estimated change points.}
\label{fig:sim:t:fms}
\end{figure}

\begin{figure}[htbp]
\centering
\includegraphics[width=\textwidth]{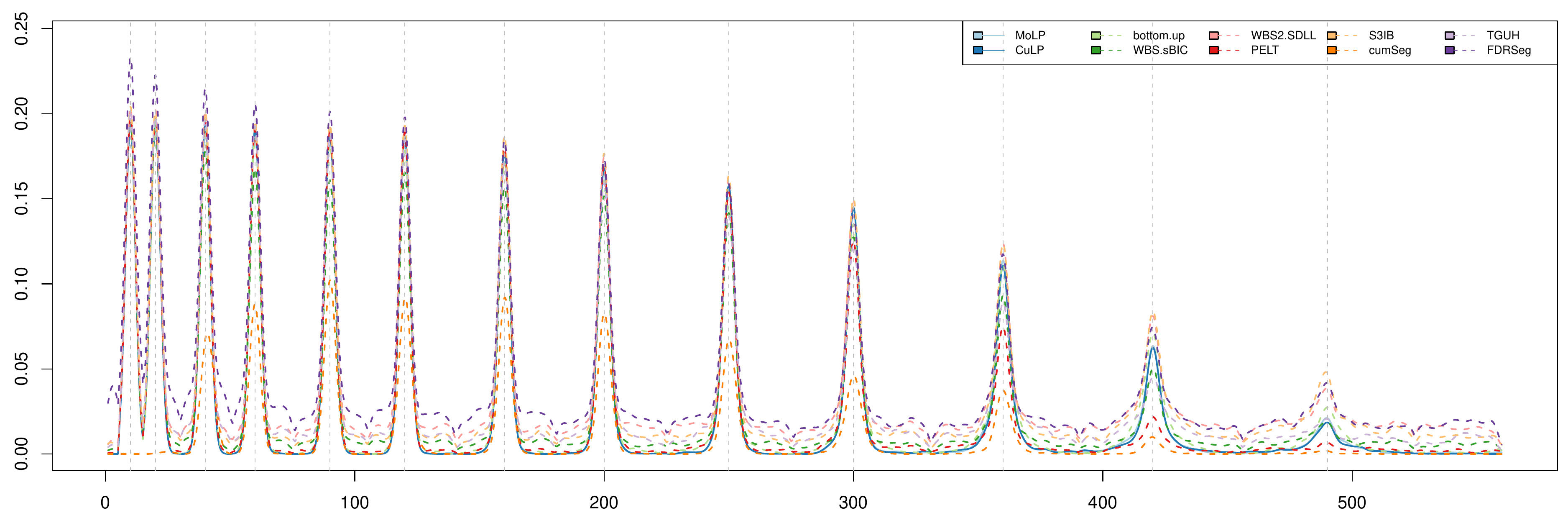}
\caption{Test signal {\tt mix} with $t_5$ errors: weighted density of estimated change points.}
\label{fig:sim:t:mix}
\end{figure}

\begin{figure}[htbp]
\centering
\includegraphics[width=\textwidth]{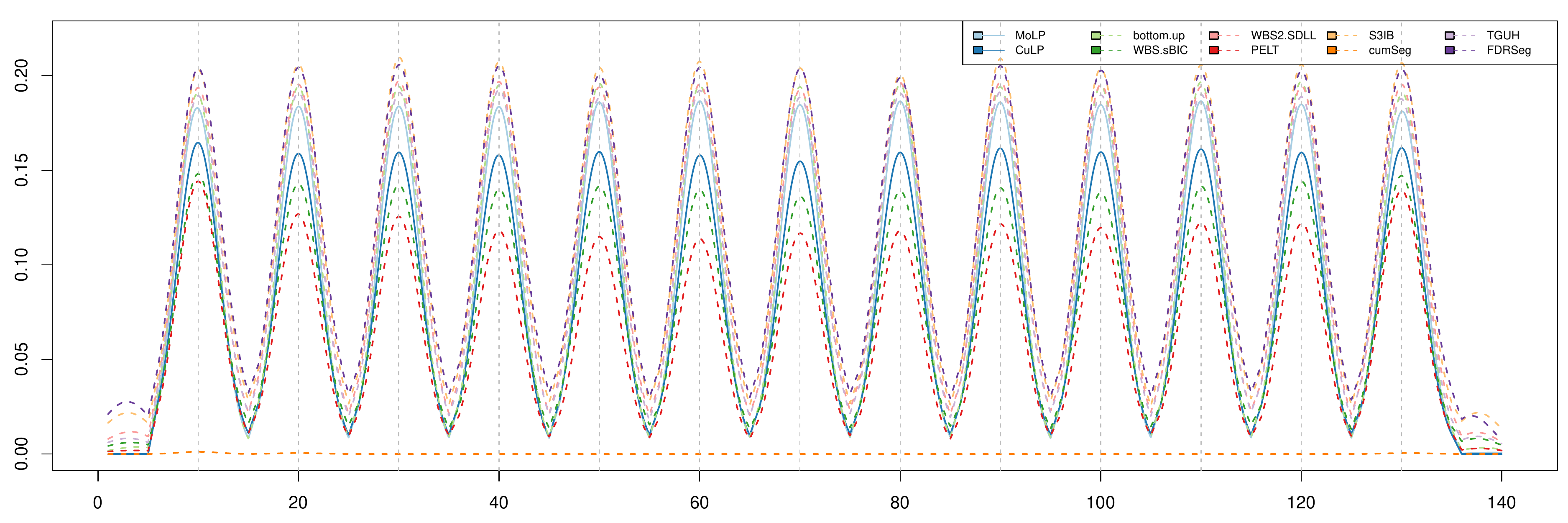}
\caption{Test signal {\tt teeth10} with $t_5$ errors: weighted density of estimated change points.}
\label{fig:sim:t:teeth10}
\end{figure}

\begin{figure}[htbp]
\centering
\includegraphics[width=\textwidth]{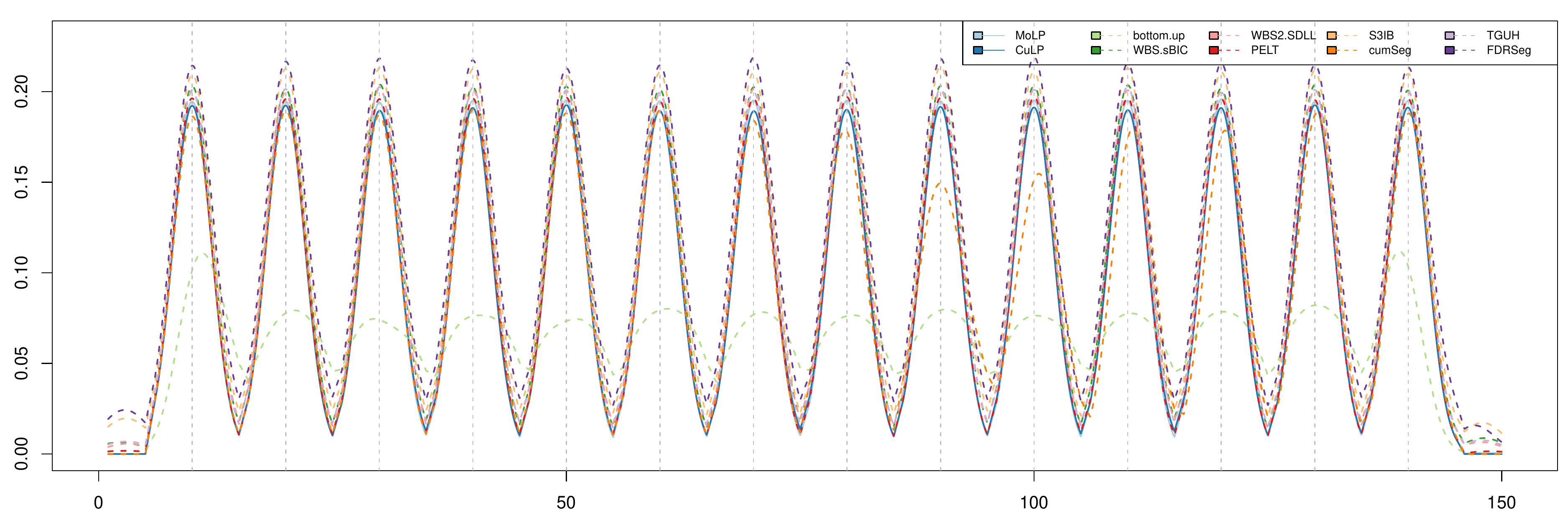}
\caption{Test signal {\tt stairs10} with $t_5$ errors: weighted density of estimated change point.}
\label{fig:sim:t:stairs10}
\end{figure}

%
%
%
%

\pagebreak

\subsection*{\ref{eq:n3} Serially correlated errors}

Tables \ref{table:sim:ar:0.3} --\ref{table:sim:ar:0.9} report the simulation results
obtained from the test signals generated with serially correlated errors
following AR($1$) processes. 
Figures~\ref{fig:sim:ar:0.3:blocks}--\ref{fig:sim:ar:0.9:stairs10} visualise
the performance of various methods 
by plotting the weighted densities of estimated change point(s).

In the presence of week serial dependence (AR parameter $\varrho = 0.3$),
the choice of light penalty $\xi_n = \log^{1.1}(n)$ is observed to be
effective in suppressing the false positives 
in the localised pruning procedure,
while attaining the TPR close to $90\%$.
When the serial dependence is strong ($\varrho = 0.9$),
a heavier penalty of $\xi_n = \log^2(n)$ is required to control the FPR.
Overall the proposed localised pruning is successful in handling serial dependence.

JUSD tends to over-estimate $\tau^2$
even with an informed choice of the parameter for its estimation.
DepSMUCE shows weakness in detecting frequent jumps
as in {\tt teeth10}, whether the serial correlations are small or large,
due to the block-based approach to the estimation of $\tau^2$.

We also consider those methods that
do not require an explicit estimation of $\tau^2$ (WBS.sBIC, cumSeg),
or use a threshold involving its estimator only as a secondary check (WBS2.SDLL),
to which we supply the estimator of $\tau^2$ used by DepSMUCE;
for {\tt bottom.up}, we supplied the true $\tau^2$.
WBS.sBIC tends to over-estimate the number of change points
due to the inadequacy of the chosen penalty when the serial dependence is strong,
which is confirmed by that this set of many spurious estimators returns the minimum BIC.
Although the final model returned by WBS2.SDLL does not critically depend on 
the estimator of $\tau^2$, its performance appears to be heavily dependent on its estimator 
in some settings.
The cumSeg, when the serial correlations are weak, 
tends to under-estimate the number of change points as in \ref{eq:n1}
whereas when the AR parameter is large, it returns many false positives.

{\small
\setlength{\tabcolsep}{3pt}
\begin{longtable}{c|cc|c|c|c|c|c|cc}
\caption{Summary of change point estimation over $1000$ realisations for the 
test signals with Gaussian AR($1$) process
as $\vep_t$ where $\varrho = 0.3$ is used as the AR parameter;
we set $\alpha = 0.2$ for MoLP, {\tt bottom.up}, JUDS and DepSMUCE and $C_\zeta = 0.9$ for CuLP,
and use $h = h_{\mc J}$ and $\xi_n \in \{\log^{1.1}(n), \log^2(n)\}$ for the localised pruning.}
\label{table:sim:ar:0.3} 
\endfirsthead
\endhead
\hline\hline
model &	penalty &	method &	TPR &	FPR &	ARI &	MSE &	BIC &	$\delta_{\trim}$ &	$v_{\trim}$ 	\\	\hline
{\tt blocks} &	light &	MoLP &	0.887 &	0.027 &	0.943 &	5.479 &	4774.904 &	740.203 &	319.633	\\	
&	heavy &	MoLP &	0.348 &	0 &	0.66 &	20.411 &	4884.454 &	6922.511 &	444.873	\\	
&	light &	CuLP &	0.878 &	0.04 &	0.931 &	6.193 &	4779.81 &	995.415 &	434.106	\\	
&	heavy &	CuLP &	0.316 &	0 &	0.61 &	20.57 &	4885.94 &	7516.739 &	128.913	\\	\cline{4-10}
&	- &	{\tt bottom.up} &	0.851 &	0.148 &	0.903 &	6.75 &	4795.408 &	986.112 &	598.441	\\	
&	- &	WBS.sBIC &	0.909 &	0.094 &	0.92 &	6.921 &	4785.532 &	1059.468 &	414.596	\\	
&	- &	WBS2.SDLL &	0.943 &	0.226 &	0.864 &	6.832 &	4781.659 &	606.197 &	332.621	\\	
&	- &	cumSeg &	0.742 &	0.007 &	0.887 &	9.407 &	4814.161 &	2347.276 &	961.464	\\	
&	- &	JUSD &	0.757 &	0.007 &	0.921 &	8.575 &	4800.773 &	1889.446 &	658.822	\\	
&	- &	DepSMUCE &	0.804 &	0.013 &	0.932 &	7.46 &	4792.994 &	1461.211 &	506.747	\\	\hline
{\tt fms} &	light &	MoLP &	0.892 &	0.064 &	0.88 &	5.812 &	-566.721 &	0.576 &	0.287	\\	
&	heavy &	MoLP &	0.411 &	0 &	0.595 &	17.571 &	-527.843 &	2.096 &	0	\\	
&	light &	CuLP &	0.95 &	0.084 &	0.903 &	4.729 &	-570.251 &	0.392 &	0.136	\\	
&	heavy &	CuLP &	0.42 &	0.001 &	0.605 &	16.863 &	-530.367 &	1.917 &	0	\\	\cline{4-10}
&	- &	{\tt bottom.up} &	0.834 &	0.261 &	0.825 &	7.536 &	-551.757 &	0.975 &	1.692	\\	
&	- &	WBS.sBIC &	0.962 &	0.104 &	0.892 &	5.166 &	-569.772 &	0.33 &	0.136	\\	
&	- &	WBS2.SDLL &	0.975 &	0.103 &	0.885 &	4.702 &	-570.251 &	0.287 &	0.136	\\	
&	- &	cumSeg &	0.74 &	0.029 &	0.884 &	9.685 &	-551.414 &	2.167 &	0.465	\\	
&	- &	JUSD &	0.397 &	0.001 &	0.586 &	17.694 &	-525.871 &	2.151 &	0.000	\\	
&	- &	DepSMUCE &	0.824 &	0.009 &	0.912 &	7.265 &	-558.459 &	1.15 &	0.225	\\	\hline
{\tt mix} &	light &	MoLP &	0.864 &	0.023 &	0.637 &	3.915 &	833.581 &	56.885 &	24.52	\\	
&	heavy &	MoLP &	0.24 &	0 &	0.123 &	15.707 &	907.876 &	706.508 &	34.214	\\	
&	light &	CuLP &	0.851 &	0.042 &	0.618 &	4.499 &	836.491 &	75.568 &	29.994	\\	
&	heavy &	CuLP &	0.157 &	0 &	0.082 &	17.166 &	918.815 &	842.511 &	0	\\	\cline{4-10}
&	- &	{\tt bottom.up} &	0.863 &	0.014 &	0.649 &	4.173 &	840.582 &	72.021 &	30.907	\\	
&	- &	WBS.sBIC &	0.857 &	0.101 &	0.658 &	6.551 &	856.824 &	128.398 &	35.810	\\	
&	- &	WBS2.SDLL &	0.89 &	0.043 &	0.677 &	4.2 &	833.53 &	59.276 &	36.837	\\	
&	- &	cumSeg &	0.399 &	0 &	0.33 &	13.36 &	893.913 &	703.169 &	184.755	\\	
&	- &	JUSD &	0.423 &	0.004 &	0.299 &	15.69 &	919.363 &	695.591 &	290.019	\\	
&	- &	DepSMUCE &	0.587 &	0.006 &	0.409 &	11.989 &	895.764 &	484.526 &	167.534	\\	\hline
{\tt teeth10} &	light &	MoLP &	0.873 &	0.002 &	0.83 &	2.611 &	-79.075 &	0.718 &	0	\\	
&	heavy &	MoLP &	0.084 &	0 &	0.075 &	10.437 &	-63.117 &	4.581 &	0	\\	
&	light &	CuLP &	0.874 &	0.035 &	0.773 &	3.882 &	-80.887 &	1.042 &	0	\\	
&	heavy &	CuLP &	0.081 &	0 &	0.067 &	10.4 &	-62.294 &	4.615 &	0	\\	\cline{4-10}
&	- &	{\tt bottom.up} &	0.78 &	0.003 &	0.736 &	3.979 &	-74.989 &	1.163 &	0	\\	
&	- &	WBS.sBIC &	0.8 &	0.07 &	0.704 &	4.737 &	-79.29 &	1.37 &	0.000	\\	
&	- &	WBS2.SDLL &	0.098 &	0.004 &	0.087 &	10.081 &	-65.161 &	4.558 &	0.000	\\	
&	- &	cumSeg &	0.01 &	0 &	0.005 &	10.693 &	-63.908 &	4.956 &	0.000	\\	
&	- &	JUSD &	0 &	0 &	0 &	10.727 &	-63.434 &	5 &	0.000	\\	
&	- &	DepSMUCE &	0.002 &	0 &	0.001 &	10.724 &	-63.469 &	4.992 &	0.000	\\	 \hline
{\tt stairs10} &	light &	MoLP &	0.989 &	0.005 &	0.966 &	1.977 &	-127.028 &	0.174 &	0	\\	
&	heavy &	MoLP &	0.616 &	0 &	0.669 &	17.122 &	-71.541 &	2.04 &	0.318	\\	
&	light &	CuLP &	0.994 &	0.041 &	0.944 &	2.568 &	-127.349 &	0.235 &	0	\\	
&	heavy &	CuLP &	0.688 &	0 &	0.709 &	13.441 &	-82.097 &	1.864 &	0	\\	\cline{4-10}
&	- &	{\tt bottom.up} &	0.651 &	0.083 &	0.543 &	23.182 &	-48.558 &	3.155 &	2.012	\\	
&	- &	WBS.sBIC &	0.998 &	0.084 &	0.94 &	2.556 &	-128.074 &	0.193 &	0.000	\\	
&	- &	WBS2.SDLL &	0.984 &	0.022 &	0.936 &	2.798 &	-125.513 &	0.319 &	0.000	\\	
&	- &	cumSeg &	0.968 &	0.008 &	0.84 &	6.048 &	-98.312 &	0.862 &	0.847	\\	
&	- &	JUSD &	0.524 &	0 &	0.616 &	19.1 &	-63.866 &	2.651 &	0.741	\\	
&	- &	DepSMUCE &	0.551 &	0 &	0.627 &	18.169 &	-67.326 &	2.576 &	1.906	\\	
\hline\hline
\end{longtable}}

{\small
\setlength{\tabcolsep}{3pt}
\begin{longtable}{c|cc|c|c|c|c|c|cc}
\caption{Summary of change point estimation over $1000$ realisations for the 
test signals with Gaussian AR($1$) process as $\vep_t$ where $\varrho = 0.9$ is used as the AR parameter;
we set $\alpha = 0.2$ for MoLP, {\tt bottom.up}, JUDS and DepSMUCE and $C_\zeta = 0.9$ for CuLP,
and use $h = h_{\mc J}$ and $\xi_n \in \{\log^{1.1}(n), \log^2(n)\}$ for the localised pruning.}
\label{table:sim:ar:0.9} 
\endfirsthead
\endhead
\hline\hline
model &	penalty &	method &	TPR &	FPR &	ARI &	MSE &	BIC &	$\delta_{\trim}$ &	$v_{\trim}$ 	\\	\hline
{\tt blocks} &	light &	MoLP &	0.965 &	0.555 &	0.738 &	8.856 &	46624.3 &	4598.286 &	2825.82	\\	
&	heavy &	MoLP &	0.942 &	0.14 &	0.905 &	5.792 &	46941.57 &	6350.733 &	3242.88	\\	
&	light &	CuLP &	0.941 &	0.316 &	0.849 &	6.529 &	46908.09 &	6314.565 &	2814.783	\\	
&	heavy &	CuLP &	0.916 &	0.059 &	0.94 &	5.445 &	47029.13 &	7342.356 &	3070.561	\\	\cline{4-10}
&	- &	{\tt bottom.up} &	0.824 &	0.067 &	0.918 &	6.786 &	47293.08 &	10658.44 &	5986.544	\\	
&	- &	WBS.sBIC &	1 &	0.985 &	0.023 &	78.455 &	41629.19 &	764.912 &	176.043	\\	
&	- &	WBS2.SDLL &	0.97 &	0.443 &	0.723 &	9.548 &	46582.41 &	5485.552 &	3300.133	\\	
&	- &	cumSeg &	0.929 &	0.343 &	0.722 &	7.914 &	47207.89 &	10499.69 &	6610.925	\\	
&	- &	JUSD &	0.782 &	0.009 &	0.928 &	7.675 &	47399.17 &	16699.38 &	5104.590	\\	
&	- &	DepSMUCE &	0.917 &	0.137 &	0.876 &	5.602 &	47045.95 &	7352.367 &	3336.377	\\	\hline
{\tt fms} &	light &	MoLP &	0.962 &	0.398 &	0.758 &	6.1 &	-6128.199 &	1.448 &	0.993	\\	
&	heavy &	MoLP &	0.951 &	0.116 &	0.856 &	4.743 &	-6050.056 &	2.128 &	1.288	\\	
&	light &	CuLP &	0.976 &	0.302 &	0.815 &	5.463 &	-6111.884 &	2.702 &	0.862	\\	
&	heavy &	CuLP &	0.962 &	0.079 &	0.906 &	4.312 &	-6060.462 &	3.078 &	0.965	\\	\cline{4-10}
&	- &	{\tt bottom.up} &	0.74 &	0.04 &	0.771 &	9.335 &	-5761.415 &	10.171 &	15.326	\\	
&	- &	WBS.sBIC &	1 &	0.973 &	0.039 &	40.62 &	-7663.268 &	0.133 &	0.050	\\	
&	- &	WBS2.SDLL &	0.982 &	0.176 &	0.84 &	5.383 &	-6122.909 &	2.406 &	0.923	\\	
&	- &	cumSeg &	0.959 &	0.565 &	0.518 &	11.825 &	-6044.72 &	12.896 &	8.774	\\	
&	- &	JUSD &	0.762 &	0.01 &	0.878 &	8.301 &	-5811.716 &	12.849 &	1.678	\\	
&	- &	DepSMUCE &	0.876 &	0.041 &	0.902 &	6.363 &	-5920.303 &	8.368 &	8.838	\\	\hline
{\tt mix} &	light &	MoLP &	0.902 &	0.3 &	0.671 &	4.467 &	7585.826 &	341.03 &	284.26	\\	
&	heavy &	MoLP &	0.868 &	0.046 &	0.64 &	4.034 &	7695.666 &	500.375 &	256.946	\\	
&	light &	CuLP &	0.927 &	0.179 &	0.718 &	4.323 &	7598.369 &	464.118 &	306.1	\\	
&	heavy &	CuLP &	0.886 &	0.038 &	0.673 &	3.934 &	7669.327 &	509.42 &	337.063	\\	\cline{4-10}
&	- &	{\tt bottom.up} &	0.739 &	0.013 &	0.461 &	5.539 &	7981.962 &	1040.016 &	406.518	\\	
&	- &	WBS.sBIC &	1 &	0.948 &	0.104 &	19.779 &	5946.152 &	88.406 &	14.256	\\	
&	- &	WBS2.SDLL &	0.937 &	0.093 &	0.742 &	4.128 &	7586.685 &	464.132 &	291.616	\\	
&	- &	cumSeg &	0.864 &	0.28 &	0.738 &	8.866 &	8083.526 &	4268.618 &	728.527	\\	
&	- &	JUSD &	0.624 &	0.007 &	0.454 &	11.194 &	8450.092 &	4347.482 &	1298.358	\\	
&	- &	DepSMUCE &	0.828 &	0.009 &	0.643 &	5.659 &	7906.877 &	1559.279 &	447.061	\\	\hline
{\tt teeth10} &	light &	MoLP &	0.924 &	0.003 &	0.877 &	2.75 &	-1301.465 &	6.007 &	0	\\	
&	heavy &	MoLP &	0.898 &	0.001 &	0.843 &	2.901 &	-1290.286 &	7.262 &	0	\\	
&	light &	CuLP &	0.887 &	0.113 &	0.803 &	3.878 &	-1262.316 &	9.108 &	0	\\	
&	heavy &	CuLP &	0.807 &	0.02 &	0.726 &	4.243 &	-1215.827 &	12.82 &	0	\\	\cline{4-10}
&	- &	{\tt bottom.up} &	0.689 &	0.002 &	0.632 &	5.783 &	-1086.38 &	17.678 &	6.273	\\	
&	- &	WBS.sBIC &	1 &	0.844 &	0.345 &	5.45 &	-1793.232 &	0.468 &	0.000	\\	
&	- &	WBS2.SDLL &	0.651 &	0.017 &	0.605 &	5.568 &	-1138.749 &	19.842 &	4.505	\\	
&	- &	cumSeg &	0.997 &	0.292 &	0.801 &	3.976 &	-1375.372 &	6.059 &	3.878	\\	
&	- &	JUSD &	0.151 &	0.003 &	0.156 &	10.515 &	-680.996 &	44.887 &	0.000	\\	
&	- &	DepSMUCE &	0.079 &	0.002 &	0.081 &	10.607 &	-658.432 &	47.344 &	0.000	\\	\hline
{\tt stairs10} &	light &	MoLP &	0.99 &	0.008 &	0.968 &	1.959 &	-1854.062 &	1.632 &	0	\\	
&	heavy &	MoLP &	0.99 &	0.002 &	0.97 &	1.943 &	-1851.876 &	1.65 &	0	\\	
&	light &	CuLP &	0.991 &	0.128 &	0.936 &	2.507 &	-1893.471 &	1.895 &	0	\\	
&	heavy &	CuLP &	0.99 &	0.024 &	0.958 &	2.314 &	-1862.168 &	2.083 &	0	\\	\cline{4-10}
&	- &	{\tt bottom.up} &	0.771 &	0.083 &	0.708 &	13.097 &	-1263.68 &	19.22 &	5.825	\\	
&	- &	WBS.sBIC &	1 &	0.843 &	0.351 &	5.299 &	-2399.815 &	0.011 &	0.000	\\	
&	- &	WBS2.SDLL &	0.864 &	0.01 &	0.871 &	6.744 &	-1612.956 &	8.71 &	0.000	\\	
&	- &	cumSeg &	1 &	0.262 &	0.848 &	4.329 &	-1793.323 &	3.954 &	2.700	\\	
&	- &	JUSD &	0.535 &	0 &	0.648 &	18.334 &	-975.106 &	25.541 &	8.419	\\	
&	- &	DepSMUCE &	0.469 &	0 &	0.59 &	25.603 &	-772.494 &	29.519 &	5.242	\\	\hline\hline
\end{longtable}}

\begin{figure}[htbp]
\centering
\includegraphics[width=\textwidth]{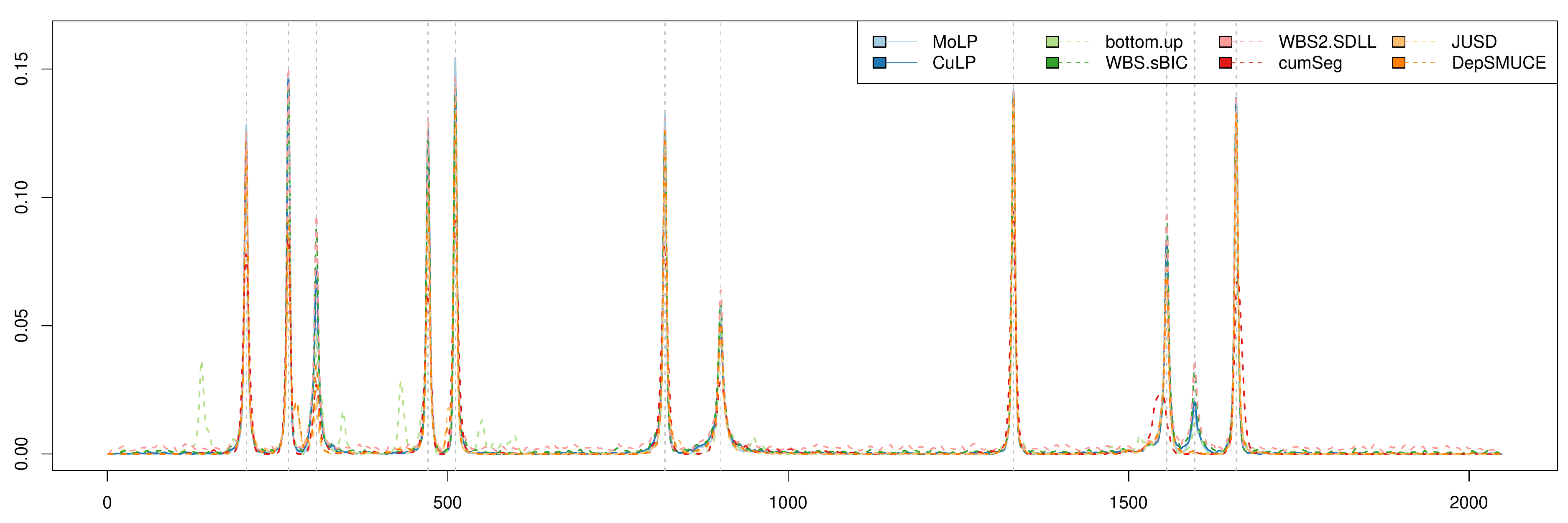}
\caption{Test signal {\tt blocks} with AR($1$) process as $\vep_t$ where $\varrho = 0.3$:
weighted density of estimated change points.
We set $\alpha = 0.2$ for MoLP, {\tt bottom.up}, JUSD and DepSMUCE and $C_\zeta = 0.9$ for CuLP, 
and use $h = h_{\mc J}$ and $\xi_n = \log^{1.1}(n)$ for the localised pruning.}
\label{fig:sim:ar:0.3:blocks}
\end{figure}

\begin{figure}[htbp]
\centering
\includegraphics[width=\textwidth]{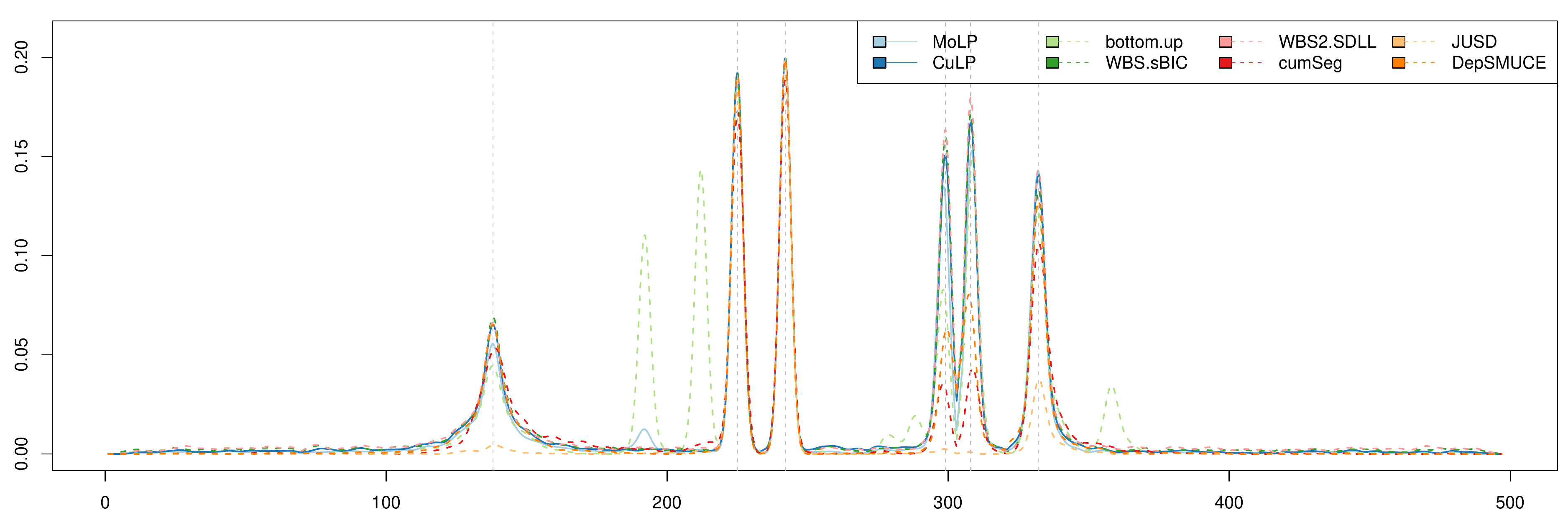}
\caption{Test signal {\tt fms} with AR($1$) process as $\vep_t$ where $\varrho = 0.3$:
weighted density of estimated change points.}
\label{fig:sim:ar:0.3:fms}
\end{figure}

\begin{figure}[htbp]
\centering
\includegraphics[width=\textwidth]{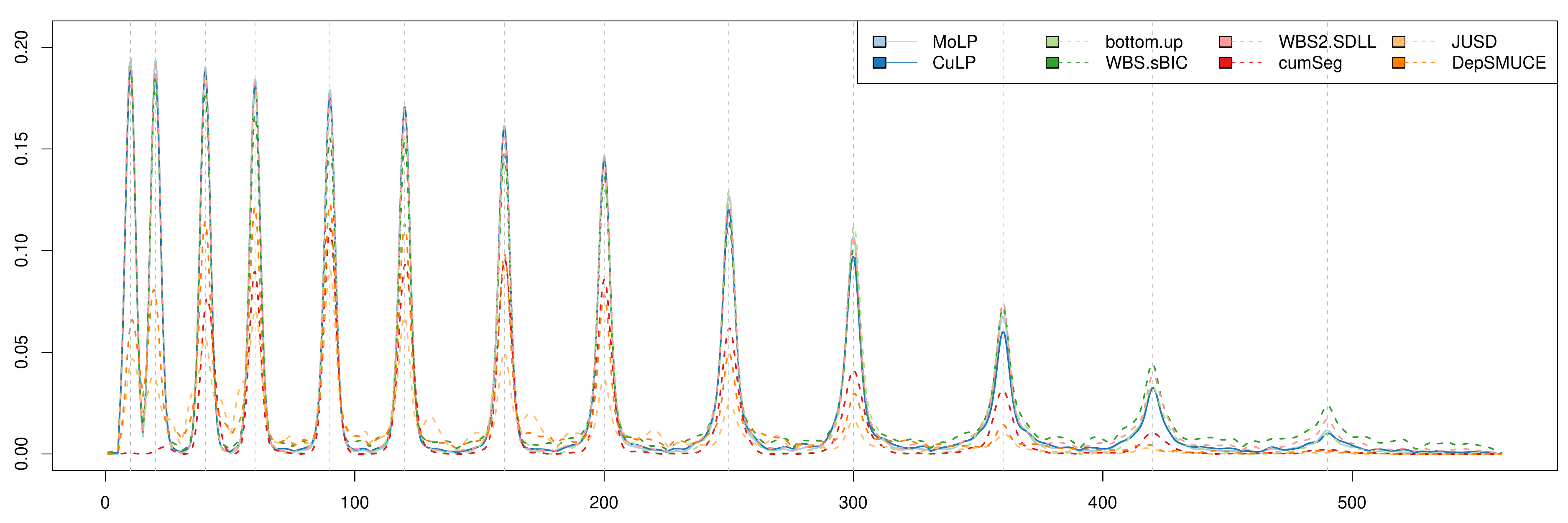}
\caption{Test signal {\tt mix} with AR($1$) process as $\vep_t$ where $\varrho = 0.3$:
weighted density of estimated change points.}
\label{fig:sim:ar:0.3:mix}
\end{figure}

\begin{figure}[htbp]
\centering
\includegraphics[width=\textwidth]{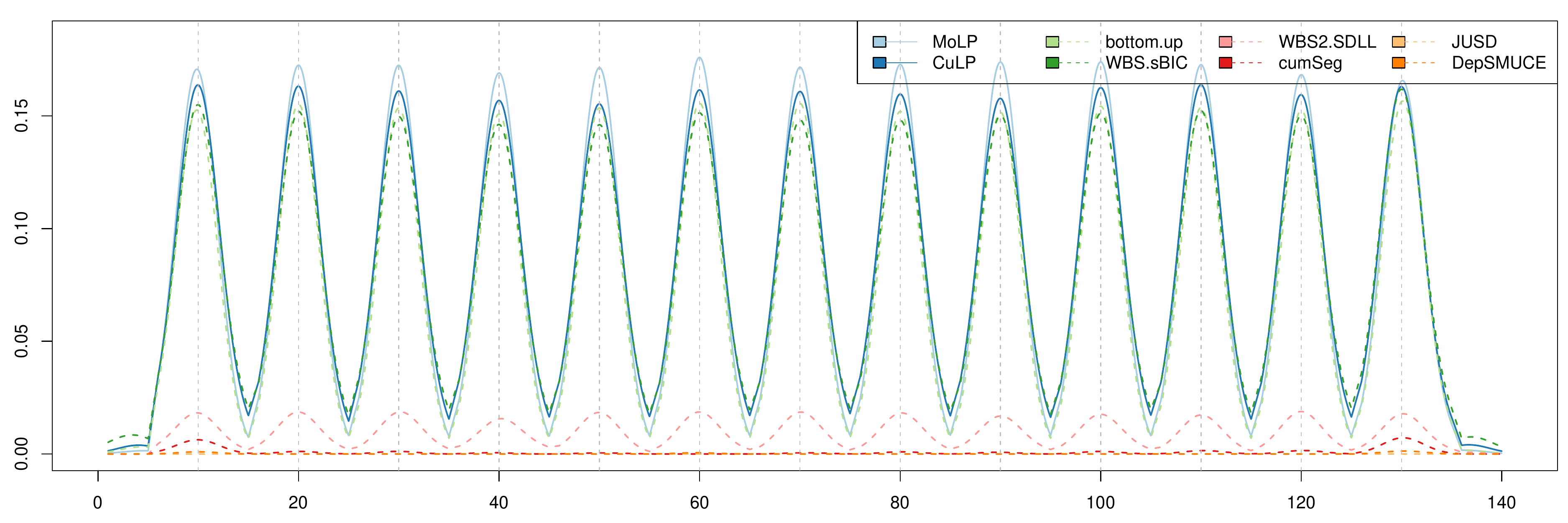}
\caption{Test signal {\tt teeth10} with AR($1$) process as $\vep_t$ where $\varrho = 0.3$:
weighted density of estimated change points.}
\label{fig:sim:ar:0.3:teeth10}
\end{figure}

\begin{figure}[htbp]
\centering
\includegraphics[width=\textwidth]{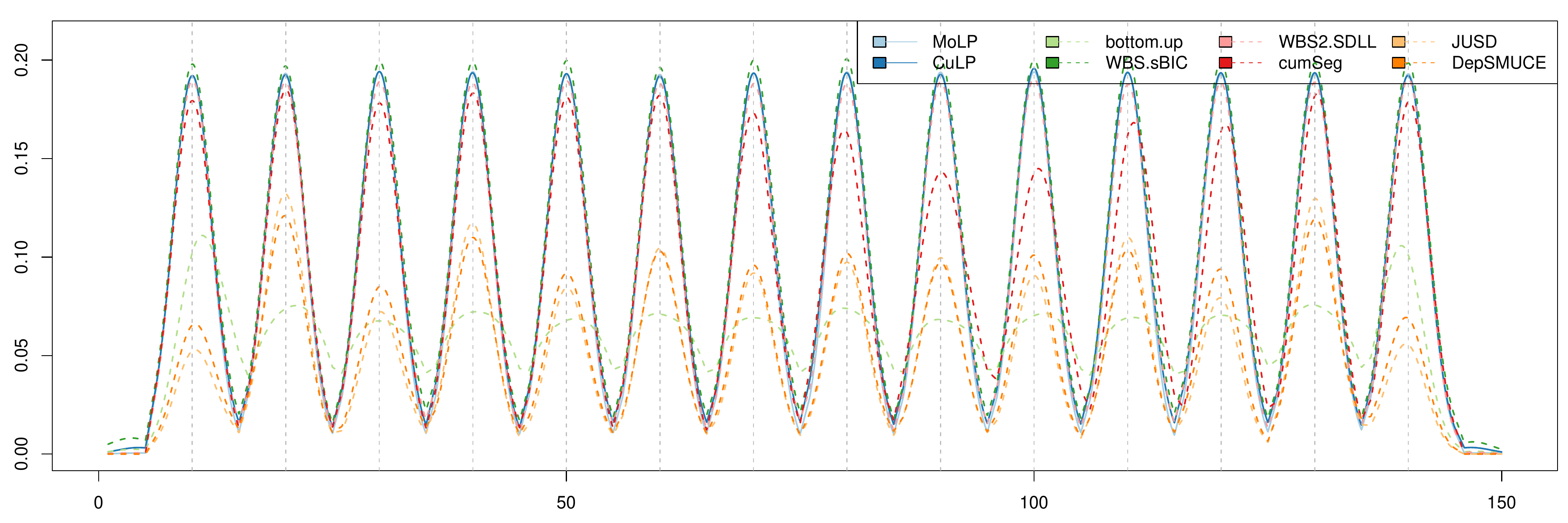}
\caption{Test signal {\tt stairs10} with AR($1$) process as $\vep_t$ where $\varrho = 0.3$:
weighted density of estimated change points.}
\label{fig:sim:ar:0.3:stairs10}
\end{figure}

\begin{figure}[htbp]
\centering
\includegraphics[width=\textwidth]{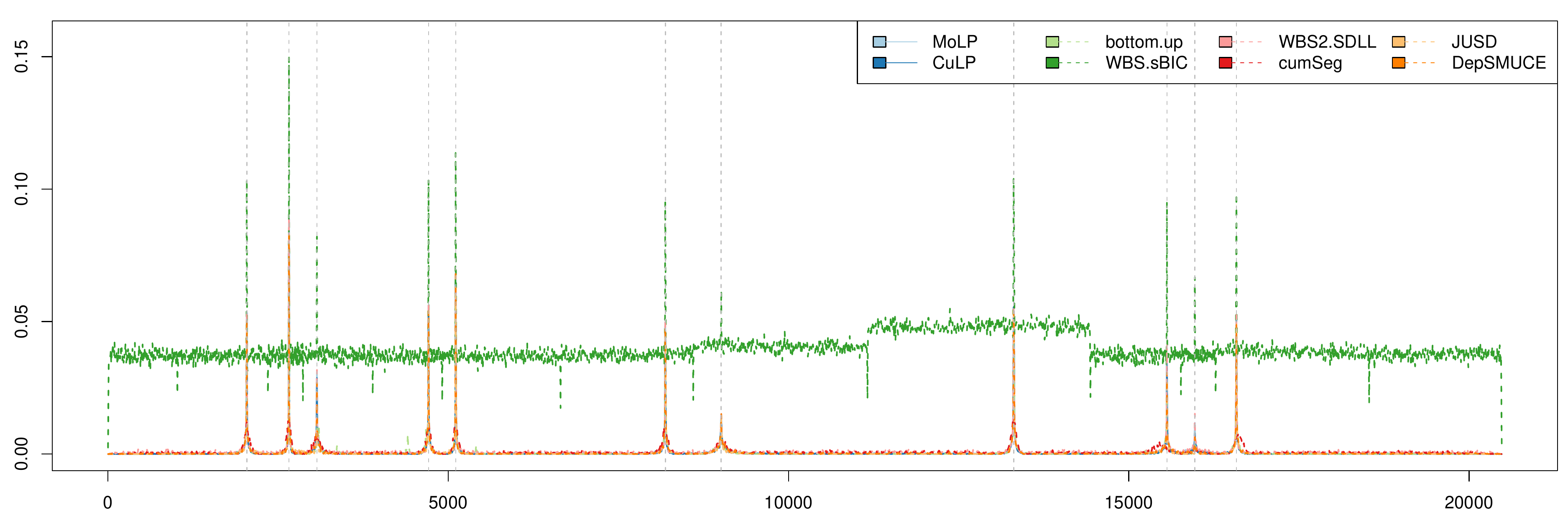}
\caption{Test signal {\tt blocks} with AR($1$) process as $\vep_t$ where $\varrho = 0.9$:
weighted density of estimated change points.
We set $\alpha = 0.2$ for MoLP, {\tt bottom.up}, JUSD and DepSMUCE and $C_\zeta = 0.9$ for CuLP, 
and use $h = h_{\mc J}$ and $\xi_n = \log^{2}(n)$ for the localised pruning.}
\label{fig:sim:ar:0.9:blocks}
\end{figure}

\begin{figure}[htbp]
\centering
\includegraphics[width=\textwidth]{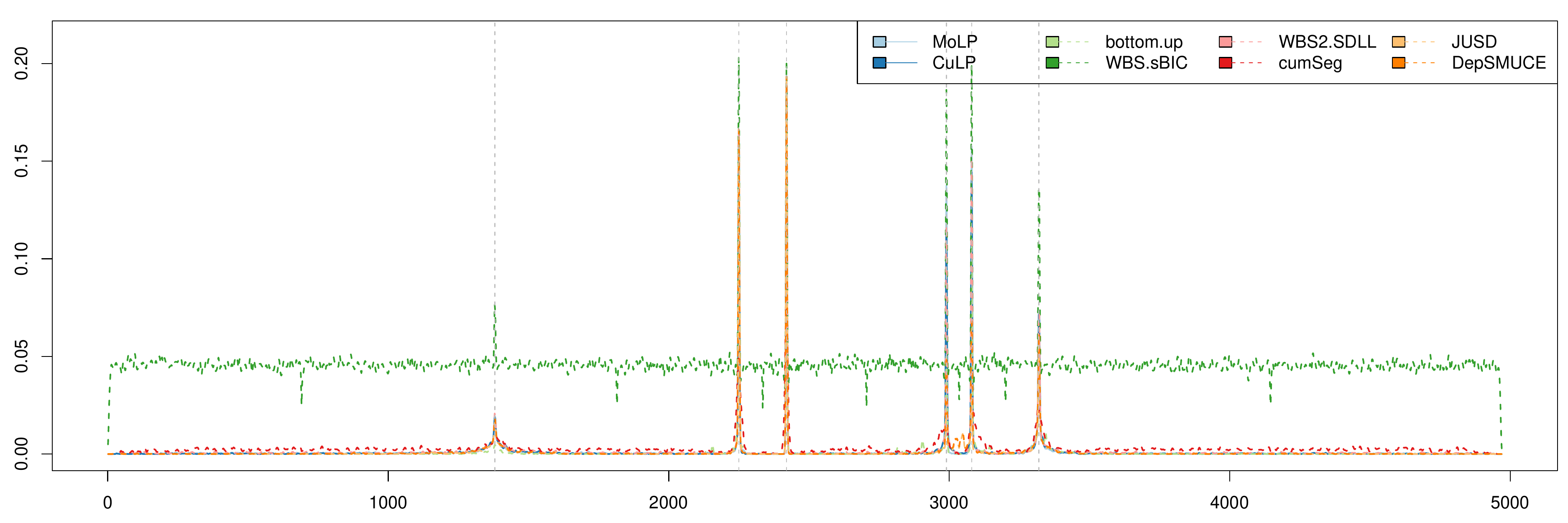}
\caption{Test signal {\tt fms} with AR($1$) process as $\vep_t$ where $\varrho = 0.9$:
weighted density of estimated change points.}
\label{fig:sim:ar:0.9:fms}
\end{figure}

\begin{figure}[htbp]
\centering
\includegraphics[width=\textwidth]{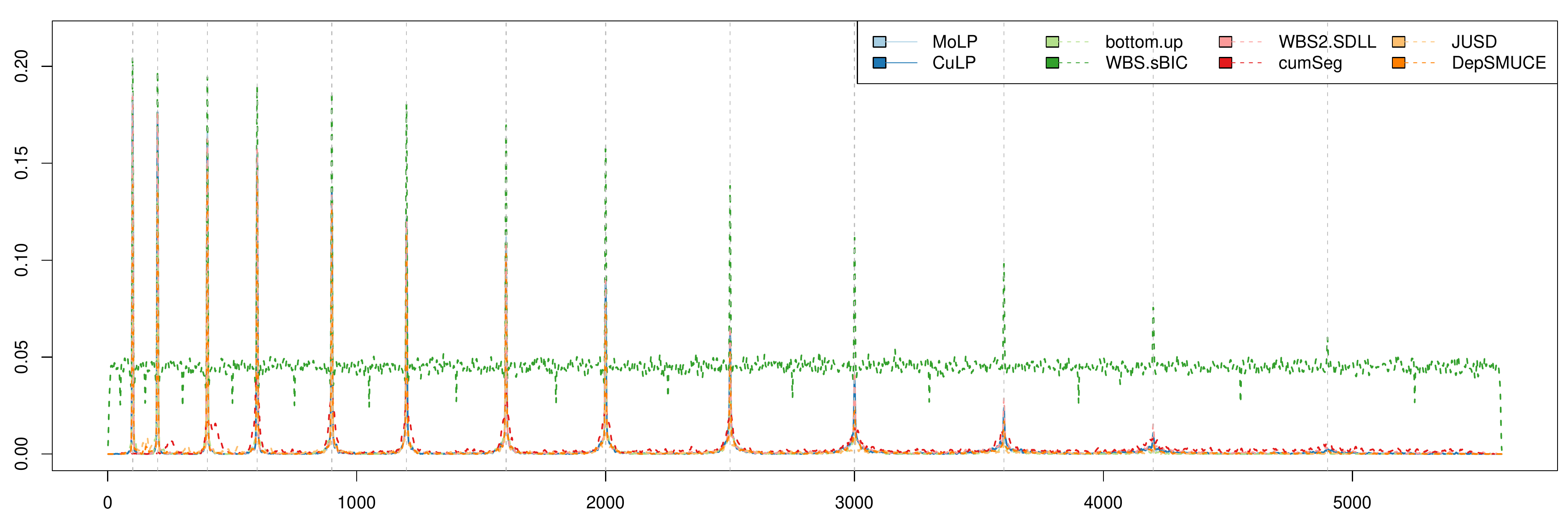}
\caption{Test signal {\tt mix} with AR($1$) process as $\vep_t$ where $\varrho = 0.9$:
weighted density of estimated change points.}
\label{fig:sim:ar:0.9:mix}
\end{figure}

\clearpage
\begin{figure}[htbp]
\centering
\includegraphics[width=\textwidth]{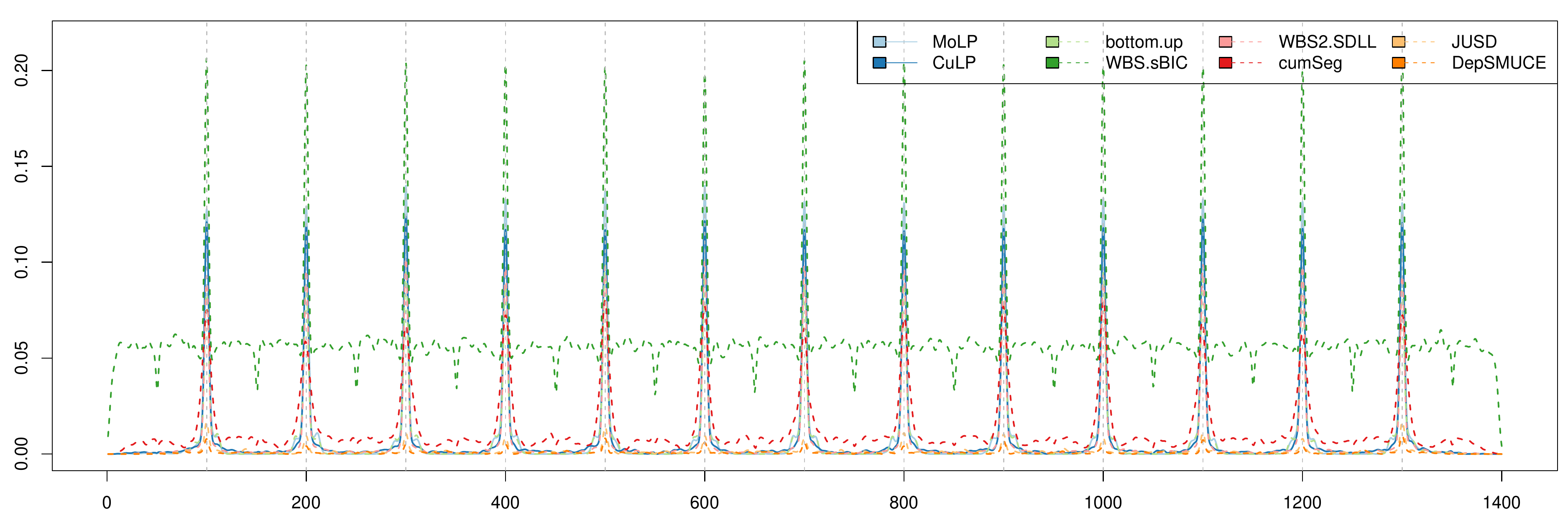}
\caption{Test signal {\tt teeth10} with AR($1$) process as $\vep_t$ where $\varrho = 0.9$:
weighted density of estimated change points.}
\label{fig:sim:ar:0.9:teeth10}
\end{figure}
\begin{figure}[htbp]
\centering
\includegraphics[width=\textwidth]{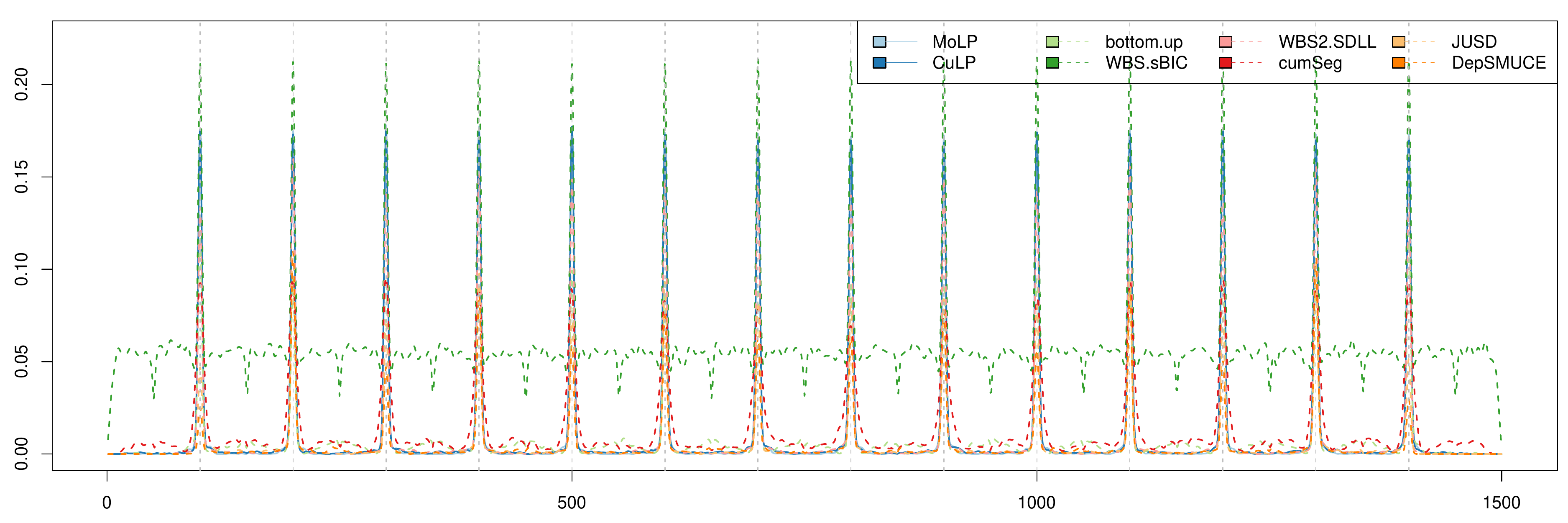}
\caption{Test signal {\tt stairs10} with AR($1$) process as $\vep_t$ where $\varrho = 0.9$:
weighted density of estimated change points.}
\label{fig:sim:ar:0.9:stairs10}
\end{figure}

\clearpage
\section{Algorithms}
\label{sec:algs}

Algorithm~\ref{alg:one} provides the pseudo code for the outer algorithm
of the proposed localised pruning methodology,
which iteratively identifies the local interval over which pruning is to be performed.

\begin{algorithm}[htp]
\caption{Outer algorithm for localisation (\tt LocAlg)}
\label{alg:one}
\DontPrintSemicolon
\SetAlgoLined
\SetKwFunction{exhaustive}{{\tt PrunAlg}}
\KwIn{Data $\{X_t\}_{t = 1}^n$, a set of candidate change point estimators $\C$, 
a candidate sorting function $h(\cdot)$}
\BlankLine
{\bf Step 0:} set $\wh{\Cp} = \emptyset$ and $\mc C \leftarrow \C$\;
\Repeat{$\mc C$ is empty}{
{\bf Step 1:} find $\mc C_\circ$ as $\mc C_\circ \leftarrow \{\c \in \mc C:\, h(\c) = \max_{\c' \in \mc C} h(\c')\}$\;
\qquad \lIf{$|\mc C_\circ| = 1$}{$\c_\circ \leftarrow \mc C_\circ$} 
\qquad \lElse{$\c_\circ \leftarrow \arg\min_{\c \in \mc C_\circ} |\mc I(\c)|$}

{\bf Step 2:} find
\begin{align*}
& \c_L \leftarrow \max\{ \c < \c_\circ \colon 
	\c \in \wh{\Cp} \cup \{0\} \text{ or } (\c \in \mc C \text{ and }\mc I(\c) \cap \mc I(\c_\circ) = \emptyset )\}, \\
& \c_R \leftarrow \min\{ \c > \c_\circ \colon 
	\c \in \wh{\Cp} \cup \{n\} \text{ or } ( \c \in \mc C \text{ and }\mc I(\c) \cap \mc I(\c_\circ) = \emptyset) \}
\end{align*}
and set $\mc D \leftarrow (\c_L, \c_R) \cap \mc C$\;

{\bf Step 3:} $\wh{\mc A} \leftarrow \exhaustive(\mc D, \mc C, \wh{\Cp}, \c_L, \c_R)$\;

{\bf Step 4:} set 
$\mc R \leftarrow \{\c_\circ\} \cup (\mc D \cap [\min\wh{\mc A}, \max\wh{\mc A}])$\;
\qquad \lIf{$\c_L \in \wh{\Cp} \cup \{0\}$}{$\mc R \leftarrow \mc R \cup \{\mc D \cap (\c_L, \min\wh{\mc A}) \}$ }
\qquad \lIf{$\c_R \in \wh{\Cp} \cup \{n\}$}{$\mc R \leftarrow \mc R \cup \{\mc D \cap (\max\wh{\mc A}, \c_R) \}$ }

{\bf Step 5:} set $\wh{\Cp} \leftarrow \wh{\Cp} \cup \wh{\mc A}$ and
$\mc C \leftarrow \mc C \setminus \mc R$
}
\BlankLine
\KwOut{$\wh{\Cp}$}
\end{algorithm}

\clearpage\newpage
Algorithm~\ref{alg:two} outlines the efficient implementation of 
the inner algorithm employed in Step~3 of the outer algorithm (Algorithm~\ref{alg:one}).
For further details on its implementation, see \cite{meier2018}.

\begin{algorithm}[htbp]
\caption{Inner algorithm for pruning ({\tt PrunAlg})}
\label{alg:two}
\DontPrintSemicolon
\SetAlgoLined
\SetKwFunction{exhaustive}{{\tt PrunAlg}}
\SetKwData{flag}{flag}
\SetKwData{true}{true} 
\SetKwData{false}{false}
\SetKwData{child}{child}
\SetKwData{inf}{Inf}

\SetKwProg{Fn}{Function}{:}{}
\Fn{$\exhaustive(\mc D, \mc C, \wh{\Cp}, s, e)$}{
	\BlankLine
	Enumerate all $M = 2^{|\mc D|}$ subsets of $\mc D$ (including $\emptyset$) denoted by $\mc D_i, \, i = 1, \ldots, M$. 

	Set $\mc F \leftarrow \emptyset$, $\wh{\mc A} \leftarrow \emptyset$, $\ell \leftarrow |\mc D|$,
	and assign $\flag_i \leftarrow \true$ for all $i = 1, \ldots, M$.

	\Repeat{$\ell = 0$}{
		\For{\textup{$\mc D_i$ with $|\mc D_i| = \ell$ and $\flag_i = \false$}}{
			identify
			\begin{center}
			$\child(\mc D_i) = \{j: \, \mc D_j \subset \mc D_i  \text{ with } |\mc D_j| = \ell-1 \text{ and }
			\flag_j = \true \}$
			\end{center}
			\lFor{$j \in \child(\mc D_i)$}{ $\flag_j \leftarrow \false$ }
		}
		\For{\textup{$\mc D_i$ with $|\mc D_i| = \ell$ and $\flag_i = \true$}}{
			update $\mc F \leftarrow \mc F \cup \{i\}$ and
			identify $\child(\mc D_i)$\;
			\For{$j \in \child(\mc D_i)$}{
				\lIf{\textup{$\sic(\mc D_i | \mc C, \wh{\Cp}, s, e) < \sic(\mc D_j  | \mc C, \wh{\Cp}, s, e)$}}{
					$\flag_j \leftarrow \false$ }
			}

		}
		$\ell \leftarrow \ell - 1$
	}

	\If{$\mc F \ne \emptyset$}{
		find $m^* \leftarrow \min_{i \in \mc F} \, |\mc D_i|$\;
		identify $i^* \leftarrow \arg\min_{i: \, \mc D_i \subset_R \mc D_{i^\prime}, i^\prime \in \mc F, \, 
		m^* \le |\mc D_{i^\prime}| \le m^* + 2} \sic(\mc D_i| \mc C, \wh{\Cp}, s, e)$\;
		set $\wh{\mc A} \leftarrow \mc D_{i^*}$
	}

  	\KwRet{$\wh{\mc A}$}\;
}
\end{algorithm}

%
%
%
%
%
%
%
%

\end{document}